\documentclass[accepted,booklet]{itpthesis}
\usepackage[T1]{fontenc}
\usepackage[ascii]{inputenc}
\usepackage[USenglish]{babel}
\usepackage[hyphens]{url}
\usepackage{aliascnt,doi,microtype,amsmath,amsthm,amsfonts,amssymb,bbm,graphicx,hyperref,minibox,braket,wasysym,mathtools,tikz,tikz-3dplot,suffix,enumitem,makeidx,algpseudocode,booktabs,mathdots,rotating,multicol,caption,setspace,tabularx,ifthen}
\usepackage[vcentermath]{youngtab}
\usepackage[refpage,noprefix,noncompatible]{nomencl}

\usetikzlibrary{matrix,arrows,external,shapes,positioning,decorations.markings}
\allowdisplaybreaks[4]
\definecolor{polytopecolor}{RGB}{46,111,253}
\hypersetup{pdftitle={Multipartite Quantum States and their Marginals}, pdfauthor={Michael Walter}}

\def\ZZ{\mathbbm{Z}}
\def\RR{\mathbbm{R}}
\def\CC{\mathbbm{C}}
\def\FF{\mathbbm{F}}
\def\NN{\mathbbm{N}}
\def\PP{\mathbbm{P}}

\def\QQ{\mathbbm{Q}}
\def\XX{\mathbbm{X}}
\def\YY{\mathbbm{Y}}

\def\calV{\mathcal{V}}
\def\calW{\mathcal{W}}
\def\calH{\mathcal{H}}
\def\calC{\mathcal{C}}

\def\calO{\mathcal{O}}
\def\calQ{\mathcal{Q}}
\def\calL{\mathcal{L}}
\def\calK{\mathcal{K}}

\def\calP{\mathcal{P}}
\def\calS{\mathcal{S}}
\def\calX{\mathcal{X}}
\def\calY{\mathcal{Y}}

\DeclareMathOperator{\relint}{relint}
\DeclareMathOperator{\tr}{tr}
\DeclareMathOperator{\Hom}{Hom}

\DeclareMathOperator{\Span}{span}

\DeclareMathOperator{\Ad}{Ad}
\DeclareMathOperator{\ad}{ad}

\DeclareMathOperator{\Ha}{H}
\DeclareMathOperator{\Ind}{Ind}

\DeclareMathOperator{\poly}{poly}
\DeclareMathOperator{\spec}{spec}
\DeclareMathOperator{\diag}{diag}
\DeclareMathOperator{\SU}{SU}

\DeclareMathOperator{\U}{U}
\DeclareMathOperator{\GL}{GL}
\DeclareMathOperator{\SL}{SL}
\DeclareMathOperator{\vol}{vol}
\DeclareMathOperator{\Sym}{Sym}
\DeclareMathOperator{\Res}{Res}
\DeclareMathOperator{\conv}{conv}

\DeclareMathOperator{\sign}{sign}
\DeclareMathOperator{\ch}{ch}
\DeclareMathOperator{\aff}{aff}
\DeclareMathOperator{\grad}{grad}

\renewcommand{\matrix}[1]{\left(\begin{smallmatrix}#1\end{smallmatrix}\right)}
\newcommand{\Matrix}[1]{\begin{pmatrix}#1\end{pmatrix}}
\newcommand{\Alt}{\bigwedge}
\newcommand{\su}{\mathfrak{su}}
\newcommand{\gl}{\mathfrak{gl}}
\newcommand{\abs}[1]{\lvert#1\rvert}
\newcommand{\norm}[1]{\lVert#1\rVert}
\newcommand{\normHS}[1]{\norm{#1}_{\operatorname{HS}}}
\newcommand{\ketbra}[2]{\lvert#1\rangle\langle#2\rvert}
\newcommand{\proj}[1]{\lvert#1\rangle\langle#1\rvert}
\def\Id{\mathbbm{1}}

\newcommand{\newjointcountertheorem}[3]{\newaliascnt{#1}{#2}\newtheorem{#1}[#1]{#3}\aliascntresetthe{#1}}
\newtheorem{thm}{Theorem}[chapter]
\newjointcountertheorem{lem}{thm}{Lemma}
\newjointcountertheorem{cor}{thm}{Corollary}
\newjointcountertheorem{prp}{thm}{Proposition}
\newjointcountertheorem{cnj}{thm}{Conjecture}
\newjointcountertheorem{que}{thm}{Question}
\newjointcountertheorem{pro}{thm}{Problem}
\newjointcountertheorem{fct}{thm}{Fact}
\newjointcountertheorem{obs}{thm}{Observation}
\newjointcountertheorem{alg}{thm}{Algorithm}
\newjointcountertheorem{asm}{thm}{Assumption}
\theoremstyle{definition}
\newjointcountertheorem{dfn}{thm}{Definition}
\newjointcountertheorem{ntn}{thm}{Notation}
\theoremstyle{remark}
\newjointcountertheorem{rem}{thm}{Remark}
\newjointcountertheorem{nte}{thm}{Note}
\newjointcountertheorem{exl}{thm}{Example}

\AtBeginDocument{

}

\newcommand{\QMA}{\mathbf{QMA}}
\newcommand{\yket}[1]{\ket{\text{\scriptsize$\young(#1)$}}} %
\Yboxdim{8pt}
\newcommand{\Ytwoone}{\Yboxdim{5pt}\yng(2,1)\Yboxdim{8pt}}
\DeclareMathOperator{\GHZ}{GHZ} %
\DeclareMathOperator{\EPR}{EPR}
\DeclareMathOperator{\SEP}{SEP}
\newcommand{\fourqubitcutsheader}{\begin{tabularx}{0.9\linewidth}{*4{>{\centering\arraybackslash}X}}
  {\footnotesize $\hphantom{x}\lambda_{4,1} = 0.5$} &
  {\footnotesize $\hphantom{xx}\lambda_{4,1} = 2/3$} &
  {\footnotesize $\hphantom{xxx.}\lambda_{4,1} = 5/6$} &
  {\footnotesize $\hphantom{xxxx.}\lambda_{4,1} = 1$} \\
  \end{tabularx}
  \medskip}
\newcommand{\fourqubitcutsheadertwo}{\begin{tabularx}{0.9\linewidth}{*4{>{\centering\arraybackslash}X}}
  {\footnotesize $\lambda_{4,1} = 0.5$} &
  {\footnotesize $\hphantom{x}\lambda_{4,1} = 2/3$} &
  {\footnotesize $\hphantom{xx}\lambda_{4,1} = 5/6$} &
  {\footnotesize $\hphantom{xxx}\lambda_{4,1} = 1$} \\
  \end{tabularx}
  \medskip}
\newcommand{\Pspec}{\mathbf P_{\spec}} %
\newcommand{\Pdiag}{\mathbf P_{\diag}}
\newcommand{\Prob}{\mathbf P}
\newcommand{\Qrob}{\mathbf Q}

\newcommand{\Pclass}{\mathbf P}  %
\newcommand{\NP}{\mathbf{NP}}
\newcommand{\SharpP}{\#\mathbf P}
\newcommand{\GapP}{\mathbf{GapP}}
\newcommand{\VP}{\mathbf{VP}}
\newcommand{\VNP}{\mathbf{VNP}}
\DeclareMathOperator{\sesi}{ss}

\newcommand{\sixj}[6]{\begin{bmatrix}#1 & #2 & #4\\#3 & #6 & #5\end{bmatrix}} %
\newcommand{\smallsixj}[6]{\mbox{\scriptsize$\begin{bmatrix}#1 & #2 & #4\\#3 & #6 & #5\end{bmatrix}$}}
\DeclareMathOperator*{\wlim}{w-lim}
\newcommand*\circled[1]{\tikz[baseline=(char.base)]{\node[shape=circle,draw,inner sep=1pt] (char) {#1};}}
\def\centerarc[#1](#2)(#3:#4:#5){\draw[#1]($(#2)+({#5*cos(#3)},{#5*sin(#3)})$) arc (#3:#4:#5);}
\newcommand{\vecr}{\vec{r{}}}
\newcommand{\vecs}{\vec{s{}}}

\newcommand{\emphindex}[1]{\emph{#1}\index{#1}}

\renewcommand{\nomgroup}[1]{%
\ifthenelse{\equal{#1}{Q}}{\bigskip\item[\textbf{Quantum Mechanics and Information Theory}]}{%
\ifthenelse{\equal{#1}{R}}{\bigskip\item[\textbf{Groups and Representations}]}{%
\ifthenelse{\equal{#1}{T}}{\bigskip\item[\textbf{Symplectic and Algebraic Geometry}]}{}}}} %
\makenomenclature
\makeindex
\nomenclature[A S]{$\#S$, $\abs S$}{number of elements in a set $S$\nomnorefpage}

\addtocontents{toc}{~\hfill\textbf{Page}\par}
\addtocontents{lof}{~\hfill\textbf{Page}\par}
\addtocontents{lot}{~\hfill\textbf{Page}\par}

\begin{document}

\title{\Huge Multipartite Quantum States \\ \vspace{0.2cm} and their Marginals}

\author{Michael Walter}
\previousdegree{Diplom-Mathematiker \\ Georg-August-Universit\"at G\"ottingen}
\authorinfo{born May 3, 1985 \\ citizen of Germany}
\referees{Prof.~Dr.~Matthias Christandl, examiner\\ Prof.~Dr.~Gian Michele Graf, co-examiner\\ Prof.~Dr.~Aram Harrow, co-examiner}
\degreeyear{2014}
\ethdissnumber{22051}
\maketitle

\begin{abstract}{Abstract}
  Subsystems of composite quantum systems are described by reduced density matrices, or \emph{quantum marginals}.
  Important physical properties %
  often do not depend on the whole wave function but rather only on the marginals.
  Not every collection of reduced density matrices can arise as the marginals of a quantum state.
  Instead, there are profound compatibility conditions -- such as Pauli's exclusion principle or the monogamy of quantum entanglement -- which fundamentally influence the physics of many-body quantum systems and the structure of quantum information.
  The aim of this thesis is %
  a systematic and rigorous study of the general relation between \emph{multipartite quantum states}, i.e., states of quantum systems that are composed of several subsystems, and their marginals.

  In the first part of this thesis (Chapters~\ref{ch:onebody}--\ref{ch:multiplicities}) we focus on the one-body marginals of multipartite quantum states.
  Starting from a novel geometric solution of the compatibility problem,
  we then turn towards the phenomenon of quantum entanglement.
  We find that the one-body marginals through their local eigenvalues can characterize the entanglement of multipartite quantum states, and we propose the notion of an entanglement polytope for its systematic study.
  Next, we consider random quantum states, where we describe a method for computing the joint probability distribution of the marginals.
  As an illustration of its versatility, we show that a discrete variant gives an efficient algorithm for the branching problem of compact connected Lie groups.
  A recurring theme throughout the first part is the reduction of quantum-physical problems to classical problems
  of largely combinatorial nature.

  In the second part of this thesis (Chapters~\ref{ch:entropy}--\ref{ch:strong6j}), we study general quantum marginals from the perspective of the von Neumann entropy.
  We contribute two novel techniques for establishing entropy inequalities.
  The first technique is based on phase-space methods; it allows us to establish an infinite number of entropy inequalities for the class of stabilizer states.
  The second technique is based on a novel characterization of compatible marginals in terms of representation theory.
  We show how entropy inequalities can be understood in terms of symmetries, and illustrate the technique by giving a new, concise proof of the strong subadditivity of the von Neumann entropy.
\end{abstract}

\begin{abstract}{Zusammenfassung}

Teilsysteme zusammengesetzter Quantensysteme werden durch reduzierte Dichtematrizen, auch bekannt als \emph{Quantenmarginale}, beschrieben.
We\-sent\-liche physikalische Eigenschaften h\"angen oft nicht von der gesamten Wellenfunktion ab, sondern nur von den reduzierten Dichtematrizen weniger Teilchen.
Nicht alle S\"atze von Dichtematrizen k\"onnen als Mar\-gi\-nale eines globalen Quantenzustands auftreten. %
Tiefgreifende Kom\-pa\-ti\-bi\-li\-t\"ats\-be\-din\-gung\-en wie das Pauli'sche Ausschlussprinzip oder die Monogamie der Quantenverschr\"ankung beeinflussen auf fundamentale Weise die physikalischen Eigenschaften von Vielteilchensystemen und die Struktur von Quanteninformation. %
Ziel dieser Dissertation ist eine systematische und rigorose Untersuchung der allgemeinen Beziehung zwischen \emph{multipartiten Quantenzust\"anden}, d.h.\ Quantenzust\"anden von Vielteilchensystemen, und ihren Marginalen. %

Im ersten Teil dieser Dissertation (Kapitel~\ref{ch:onebody}--\ref{ch:multiplicities}) betrachten wir Einteilchenmarginale multipartiter Quantenzust\"ande.
Wir beginnen mit einer neuen, geometrischen L\"osung des Kompatibilit\"atsproblems.
Danach wenden wir uns dem Ph\"anomen der Quantenverschr\"ankung zu.
Wir zeigen, dass sich mittels der Eigenwerte der Einteilchendichtematrizen Aussagen \"uber die Verschr\"ankungseigenschaften des Gesamtzustands treffen lassen.
Zur systematischen Untersuchung dieses Ph\"anomens f\"uhren wir das Konzept eines ``Verschr\"ankungspolytops'' ein. %
Des Weiteren betrachten wir zuf\"allige Quantenzust\"ande.
Hierzu entwickeln wir ein Verfahren, mit dessen Hilfe sich die gemeinsame Verteilung der Einteilchenmarginale exakt berechnen l\"asst.
Wir illustrieren die Vielseitigkeit unseres Ansatzes, indem wir aus einer diskreten Variante einen effizienten Algorithmus f\"ur das Verzweigungsproblem kompakter, zusammenh\"angender Lie-Gruppen konstruieren. %
Ein wiederkehrendes Motiv im ersten Teil ist die Reduktion quantenphysikalischer Probleme auf im Wesentlichen kombinatorische, klassische Probleme.

Im zweiten Teil dieser Dissertation (Kapitel~\ref{ch:entropy}--\ref{ch:strong6j}) untersuchen wir allgemeine Marginale aus dem Blickwinkel der von Neumann'schen Entropie.
Hier tragen wir zwei neue Techniken zum Beweis von Entropieungleichungen bei.
Die erste Technik basiert auf Phasenraummethoden und erlaubt es uns, eine unendliche Zahl an Entropieungleichungen f\"ur die Klasse der Stabilisatorzust\"ande zu beweisen.
Die zweite Technik beruht auf einer neuen Charakterisierung kompatibler Marginale mittels Darstellungstheorie. %
Wir zeigen, dass Entropieungleichungen im Sinne von Symmetrien verstanden werden k\"onnen, und geben zur Illustration einen neuen, pr\"agnanten Beweis der starken Subadditivit\"at der von Neumann'schen Entropie.

\end{abstract}

\begin{acknowledgements}
  I am greatly indebted to %
  my advisor Matthias Christandl for his support and enthusiasm in research and otherwise.
  I sincerely thank Gian Michele Graf and Aram Harrow for accepting to be my co-examiners and for their careful reading of this thesis.
  I also thank J\"urg Fr\"ohlich for his support. %

  Much of my research work has been done in collaborations, some of whose results are incorporated in this thesis.
  It is a pleasure to thank my collaborators
  Matthias Christandl,
  Brent Doran,
  David Gross,
  Stavros Kousidis,
  Marek Ku\'{s},
  Joseph M.~Renes,
  Burak \c{S}ahino\u{g}lu,
  Adam Sawicki,
  Sebastian Seehars,
  Roman Schmied,
  Volkher Scholz,
  Mich\`{e}le Vergne, and
  Konstantin Wernli
  for making research such an enjoyable enterprise.
  I would also like to extend my gratitude to
  Philippe Biane,
  Alonso Botero,
  Fernando Brand\~{a}o,
  Peter B\"urgisser,
  Mikhail Gromov,
  Patrick Hayden,
  Christian Ikenmeyer,
  Wojciech Kami\'{n}ski,
  Eckhard Meinrenken,
  Graeme Mitchison,
  Renato Renner,
  Nicolas Ressayre,
  Mary Beth Ruskai,
  Philipp Treutlein,
  Stephanie Wehner,
  Reinhard Werner, and
  Andreas Winter
  for many inspiring discussions, helpful advice and well-appreciated support.

  I am grateful to the Institute for Theoretical Physics at ETH Z\"urich for providing an excellent environment for research.
  I also thank all the members of the Quantum Information Theory group for creating such a pleasant and productive atmosphere.
  In the course of my graduate studies, I have had the good fortune of a number of travel opportunities.
  In particular, I want to thank the Isaac Newton Institute and the Institut des Hautes \'Etudes Scientifiques for their hospitality during my research visits; I am grateful to Richard Josza, Noah Linden, Peter Shor, and Andreas Winter and to Mikhail Gromov for making these visits possible.

  Finally, I want to thank Christian Majenz, Burak \c{S}ahino\u{g}lu, Volkher Scholz, and Jonathan Skowera for their careful reading of parts of a preliminary version of this thesis.
\end{acknowledgements}

\cleardoublepage\pagestyle{empty}\phantomsection\pdfbookmark{\contentsname}{toc} %
\tableofcontents

\cleardoublepage\pagestyle{fancy}\startnumbering

\chapter{Introduction}
\label{ch:intro}

The pure state of a quantum system is described by a vector in a Hilbert space, or, more precisely, by a point in the corresponding projective space. Since the Hilbert space for multiple particles is given by the tensor product of the Hilbert spaces of the individual particles, its dimension grows exponentially with the number of particles.
This exponential behavior is the key obstruction to the classical modeling of quantum systems.
The observation is as old as quantum theory itself, and physicists ever since have tried to find ways around it.
One way to address the aforementioned exponential complexity is to make use of the following simple yet powerful observation:
Important physical properties %
often do not depend on the whole wave function but rather only on a small part, namely the \emph{reduced density matrix}, or \emph{quantum marginal}, of a few particles \cite{Loewdin55}.
For instance, the ground state energy of a spin chain is given by a minimization over nearest-neighbor reduced density matrices (\autoref{fig:intro/spin chain}).
In quantum chemistry, the binding energy of a molecule is similarly given by a minimization over two-electron reduced density matrices arising from many-electron wave functions.
Mathematically, the reduced density matrix is the contraction of (or trace over) the indices of the projection operator onto the wave function over the remaining particles.

Not every collection of reduced density matrices can arise as the marginals of a quantum state---there are profound ``kinematic'' constraints that are purely due to the geometry of the quantum state space.
The fundamental problem of characterizing the compatibility of reduced density matrices is known as the \emphindex{quantum marginal problem} in quantum information theory and as the \emph{$n$-representability problem}\index{n-representability problem@$n$-representability problem} in quantum chemistry (\autoref{fig:intro/qmp}).
It has been long recognized for its importance in many-body quantum physics and quantum chemistry \cite{Coleman63, Ruskai69, ColemanYukalov00, Coleman01}.
Unfortunately, the general problem is $\QMA$-complete and therefore $\NP$-hard, and so believed to be computationally intractable, even on a quantum computer \cite{Liu06, LiuChristandlVerstraete07}.
However, even a partial understanding of the problem has proved to be immensely useful.
Entropy inequalities such as the strong subadditivity of the von Neumann entropy \cite{LiebRuskai73}, which constrain the reduced density matrices of a quantum state, are indispensable tools in quantum statistical physics and quantum information theory \cite{OhyaPetz93}.
In computational quantum physics, the power of variational methods can be explained by their ability to reproduce the marginals of the ground state \cite{VerstraeteCirac06}.
The fundamental \emphindex{Pauli exclusion principle} \cite{Pauli25, Pauli46}, which states that the occupation numbers of a fermionic quantum state cannot exceed one, can be understood as a constraint on the one-body reduced density matrix.

\begin{figure}[t]
  \centering
  \includegraphics[width=0.9\linewidth]{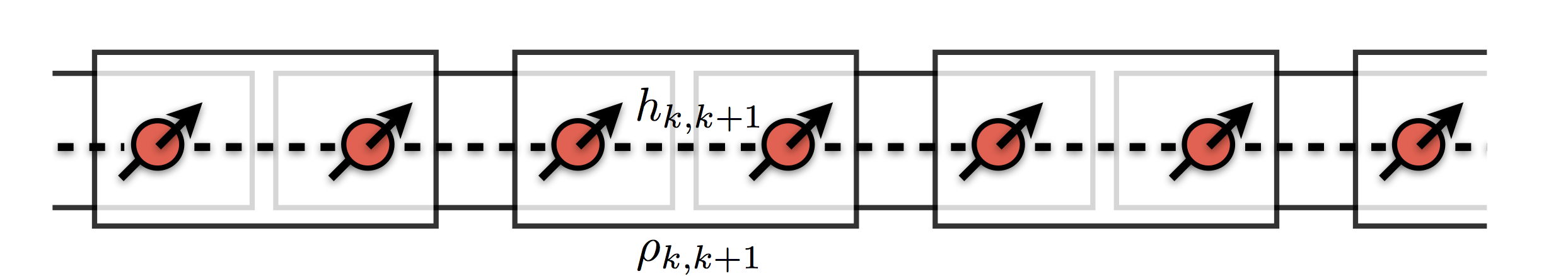}
  \begin{equation*}
    E_0
    = \min_\rho \tr H \rho
    = \min_\rho \sum_k \tr h_{k,k+1} \rho_{k,k+1}
    = \min_{\substack{\{\rho_{k,k+1}\} \\ \text{compatible}}} \sum_k \tr h_{k,k+1} \rho_{k,k+1}
  \end{equation*}
  \caption[The ground state energy of a quantum spin chain]{\emph{The ground state energy of a quantum spin chain} with nearest-neighbor interactions, $H = \sum_k h_{k,k+1}$, only depends on the two-body reduced density matrices $\rho_{k,k+1}$ that are compatible with a global state $\rho$.}
  \label{fig:intro/spin chain}
\end{figure}

\bigskip

The aim of this thesis is a systematic and rigorous study of the relation between multipartite quantum states and their marginals,
which we carry out %
by using a diverse set of mathematical tools.
It is naturally divided into two parts:
Chapters~\ref{ch:onebody}--\ref{ch:multiplicities} are concerned with one-body reduced density matrices and Chapters~\ref{ch:entropy}--\ref{ch:strong6j} with general marginals.
Each part starts with an initial chapter that introduces background material. %
The subsequent chapters then present our research contributions; each begins with a summary of the main results that are obtained in the chapter and concludes with a discussion of the results presented.
We now give a brief overview of the contents of the individual chapters.

In \autoref{ch:onebody} we formally introduce the one-body quantum marginal problem, i.e., the problem characterizing the one-body reduced density matrices that are compatible with a global pure state.
We describe the fundamental connection of this problem to geometric invariant theory, which is an appropriate mathematical framework for its study, and explain the physical consequences of the mathematical theory.
For any given number of particles, local dimensions and statistics, there exists a finite set of linear inequalities that constrain the eigenvalues of compatible one-body reduced density matrices; these inequalities together cut out a convex polytope, known as a moment polytope in mathematics.
The facets of this polytope acquire a physical interpretation through associated ``selection rules''.
We also discuss a dual, representation-theoretic description of the polytope.
Many aspects are clarified greatly by using the appropriate perspective.

\begin{figure}[t]
  \centering
  \includegraphics[width=0.7\linewidth]{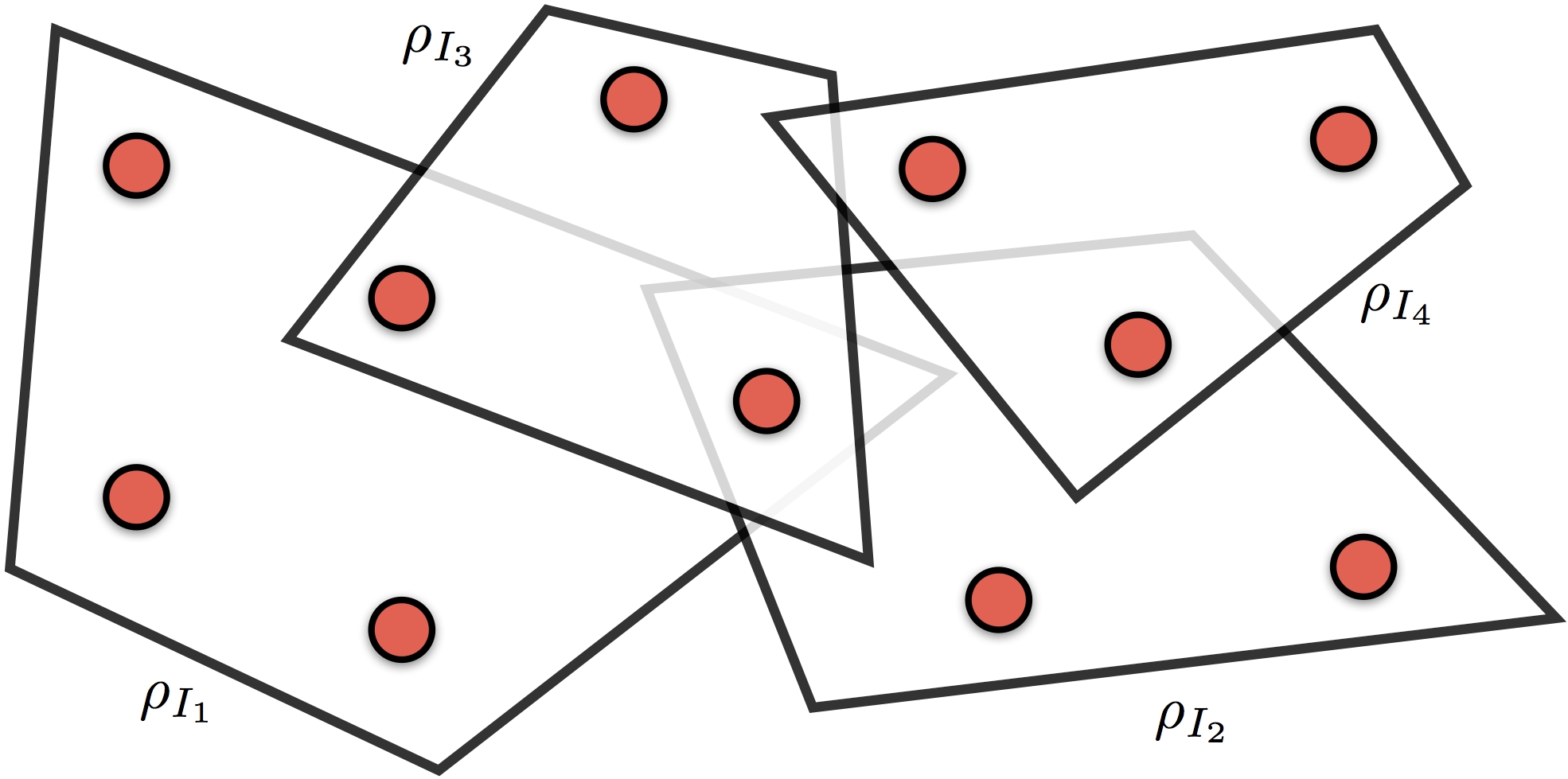}
  \caption[The general quantum marginal problem]{\emph{The general quantum marginal problem.} Given reduced density matrices $\rho_I$ for subsets $I$ of the particles, are they compatible with a global state $\rho$?}
  \label{fig:intro/qmp}
\end{figure}

In \autoref{ch:kirwan} we first review some of the history of the one-body quantum marginal problem, which has seen some significant progress in recent years, culminating in Klyachko's general solution. Along the way we give some concrete examples.
We then present a different, geometric approach to the computation of moment polytopes, which is inspired by recent work of Ressayre.
Significantly, our approach completely avoids many technicalities that have appeared in previous solutions to the problem, and it can be readily implemented algorithmically.

\autoref{ch:slocc} is devoted to the phenomenon of quantum entanglement, which profoundly influences the relation between a quantum system and its parts.
We find that in the case of multipartite pure states, features of the entanglement can already be extracted from the local eigenvalues---the natural generalization of the Schmidt coefficients or entanglement spectrum.
To study this systematically, we associate with any given class of entanglement an entanglement polytope, formed by the eigenvalues of the one-body marginals compatible with the class.
In this way we obtain local witnesses for the multipartite entanglement of a global pure state.
Our construction is applicable to systems of arbitrary size and statistics, and we explain how it can be adapted to states that are affected by low levels of noise.

In \autoref{ch:dhmeasure} we consider the following quantitative version of the one-body quantum marginal problem:
Given a pure state chosen uniformly at random, what is the joint probability distribution of its one-body reduced density matrices?
We obtain the exact probability distribution by reducing to the corresponding distribution of diagonal entries, which corresponds to a quantitative version of a classical marginal problem.
This reduction is an instance of a more general ``derivative principle'' for Duistermaat--Heckman measures in symplectic geometry.

In \autoref{ch:multiplicities} we digress in a brief interlude into a study of multiplicities of irreducible representations of compact, connected Lie groups.
The asymptotic growth of such multiplicities in a ``semiclassical limit'' is directly related to the probability measures considered in the preceding chapter.
We show that the ideas of the preceding chapter can be discretized, or ``quantized'', to give an efficient algorithm for the branching problem, which asks for the multiplicity of an irreducible representation of a subgroup $K \subseteq K'$ in the restriction of an irreducible representation of $K'$.
In particular, we obtain the first polynomial-time algorithm for computing Kronecker coefficients for Young diagrams of bounded height.
There is a surprising connection between our results on entanglement polytopes and multiplicities to recent efforts in the geometric complexity approach to the $\Pclass$ vs.\ $\NP$ problem in computer science.
We sketch this connection and explain some additional observations regarding the relevance of asymptotics.

In \autoref{ch:entropy} we initiate our study of general quantum marginals, motivated by the fundamental role of entropy in physics and information theory.
Like the marginals themselves, these entropies are not independent; instead, they are constrained by linear entropy inequalities -- the ``laws of information theory'' -- such as the strong subadditivity of the von Neumann entropy, which is an indispensable tool in the analysis of quantum systems.
A major open question is to decide if there are any further entropy inequalities satisfied by the von Neumann entropy that are not a consequence of strong subadditivity.
Classically, such entropy inequalities have been found for the Shannon entropy, and the discovery of any further entropy inequality would be considered a major breakthrough.

In \autoref{ch:stabs} we describe a first approach to the study of entropy inequalities.
We consider two classes of quantum states -- stabilizer states and Gaussian states -- which are versatile enough to exhibit intrinsically quantum features, such as multipartite entanglement, but possess enough structure to allow for a concise and computationally efficient description.
Quantum phase-space methods have been built around both classes of states, and we show how they can be used to construct a classical model that can be used to lift entropy inequalities for the Shannon entropy to quantum entropies.
In particular, our technique immediately implies that the von Neumann entropy of stabilizer states satisfies all conjectured entropy inequalities.

In \autoref{ch:strong6j} we introduce a second approach, which is applicable to general quantum states.
To this end, we unveil a novel connection between the existence of multipartite quantum states with given marginal eigenvalues and the representation theory of the symmetric group.
We use this connection to give a new proof of the strong subadditivity and weak monotonicity of the von Neumann entropy, and propose a general approach to finding further entropy inequalities based on studying representation-theoretic symbols and their symmetry properties.

\bigskip

The list of symbols (pp.~\pageref{list of symbols first}--\pageref{list of symbols last}) summarizes the most important notation used throughout this thesis.
This introduction has been adapted from \cite{ChristandlDoranKousidisEtAl12}.
Earlier versions of Figures~\ref{fig:intro/spin chain} and \ref{fig:intro/qmp} have been used in several presentations by Matthias Christandl and the author.
Most of the material in this thesis has been assembled from the works \cite{WalterDoranGrossEtAl13, ChristandlDoranKousidisEtAl12,ChristandlDoranWalter12, GrossWalter13, ChristandlSahinogluWalter12}, and we give the corresponding references at the beginning of each chapter.
\chapter{The One-Body Quantum Marginal Problem}
\label{ch:onebody}

In this chapter we formally introduce the one-body quantum marginal problem and discuss some fundamental properties.
We recall some basic concepts from the theory of Lie groups and their representations that are used throughout this thesis.
Next, we introduce the connection to geometric invariant theory, which is the appropriate mathematical framework for the study of the one-body quantum marginal problem and its variants.
We then explain the physical consequences of the mathematical theory for the quantum marginal problem and conclude by discussing the dual, representation-theoretic description in terms of Kronecker coefficients.
None of the results in this chapter are new; in each section we give pointers to relevant background literature.

\subsection*{Distinguishable Particles}

Composite quantum systems\index{composite quantum systems} are modeled by the tensor product of the Hilbert spaces describing their constituents.
Throughout this thesis, we will assume that all Hilbert spaces are finite-dimensional unless stated otherwise.
It is useful to think of the constituents as individual particles, although they can be of more general nature; for instance, the subsystems can describe different degrees of freedom such as position and spin.
Depending on whether the particles are in principle distinguishable or indistinguishable, we distinguish two basic classes of composite systems, which are of fundamentally different nature.

In the case of \emphindex{distinguishable particles}, the system is described by the tensor-product $\calH = \bigotimes_{k=1}^n \calH_k$%
\nomenclature[QH, K, L]{$\calH, \calK, \calL, \dots$}{Hilbert spaces\nomnorefpage}%
\nomenclature[Qketpsi]{$\ket\psi$, $\bra\psi$}{vector in Hilbert space, and its dual\nomnorefpage}%
\nomenclature[Qketpsibrapsi]{$\proj\psi$}{orthogonal projector onto $\CC\ket\psi$\nomnorefpage}%
\nomenclature[Q<psi,phi>]{$\langle\psi\vert\phi\rangle$}{inner product of two vectors in Hilbert space\nomnorefpage}%
\nomenclature[Q|psi|]{$\norm{\psi}$}{norm of vector in Hilbert space\nomnorefpage}%
\nomenclature[QH_k]{$\calH_k$}{Hilbert space of $k$-th particle}%
\nomenclature[Qspec]{$\spec X$}{spectrum of Hermitian operator $X$}
of the Hilbert spaces describing the individual particles. Given a density matrix $\rho$ on $\calH$, the \emph{one-body reduced density matrices}\index{reduced density matrix!one-body}%
\nomenclature[Qrho]{$\rho$}{density matrix\nomnorefpage}%
\nomenclature[Qrho_k]{$\rho_k$}{one-body reduced density matrix of $k$-th particle}
$\rho_k$ are defined by taking the partial trace of $\rho$ over all subsystems other than $k$.
In other words,
\begin{equation}
\label{eq:onebody/rdm}
  \tr \rho_k O_k = \tr \rho (\Id^{\otimes k-1} \otimes O_k \otimes \Id^{\otimes n-k}) \qquad (\forall O_k = O_k^\dagger).
\end{equation}
In physical terms, \eqref{eq:onebody/rdm} asserts that $\rho_k$ reproduces faithfully the expectation values of all local observables $O_k$.
Hence $\rho_k$ describes the effective state of the $k$-th particle.
The one-body quantum marginal problem then is the following compatibility problem:

\begin{pro}[One-Body Quantum Marginal Problem]\index{quantum marginal problem!one-body}
\label{pro:onebody/qmp}
  Let $\calH = \bigotimes_{k=1}^n \calH_k$.
  Given density matrices $\rho_k$ on $\calH_k$ for all $k=1, \dots, n$,
  does there exist a pure state $\rho = \proj\psi$ on $\calH$ such that $\rho_k$ are its one-body reduced density matrices?
\end{pro}

We will call such density matrices $\rho_1, \dots, \rho_n$ \emph{compatible (with a global pure state)}\index{compatible quantum marginals}.
The term ``quantum marginal problem'' has been coined by Klyachko in analogy to the classical marginal problem in probability theory, which asks for the existence of a joint probability distribution for a given set of marginal distributions \cite{Klyachko04}. Its one-body version was first solved in the paper \cite{Klyachko04}; cf.\ \cite{DaftuarHayden04}.

It is easy to see that the compatibility of a given set of density matrices only depends on their eigenvalues.
Indeed, suppose that $\rho$ is a pure state with one-body reduced density matrices $\rho_k$.
Then, for any collection of unitaries $U_k$ on the Hilbert spaces $\calH_k$, the state $\tilde\rho := (U_1 \otimes \dots \otimes U_k) \rho (U_1^\dagger \otimes \dots \otimes U_n^\dagger)$ is a pure state with $\tilde\rho_k = U_k \rho_k U^\dagger$.
Thus in the formulation of the one-body quantum marginal problem we may equivalently ask which collections of real numbers $\vec\lambda_1, \dots, \vec\lambda_n$ -- which we will always take to be ordered non-increasingly and call \emph{spectra}\index{spectrum}\nomenclature[Qlambda_k]{$\vec\lambda_k$}{vector of eigenvalues $\lambda_{k,1} \geq \lambda_{k,2} \geq \dots$ of $k$-th one-body reduced density matrix} -- can arise as the eigenvalues of the one-body reduced density matrices of a global pure state $\rho$.
We will likewise call such spectra $\vec\lambda_1, \dots, \vec\lambda_n$ \emph{compatible}\index{compatible spectra}.
In many ways, this unitary invariance is at the root of the solvability of \autoref{pro:onebody/qmp}.
No analogous property holds for the general quantum marginal problem.

For two particles, $n=2$, the one-body quantum marginal problem is rather straightforward to solve.
For this, recall that any vector $\ket\psi$ on a tensor product $\calH_1 \otimes \calH_2$ can be expanded in the form
\begin{equation}
\label{eq:onebody/schmidt}
  \ket\psi = \sum_{i=1}^r \sqrt{p_i} \ket i \otimes \ket{\widetilde i}
\end{equation}
for orthonormal sets of vectors $\ket i$ in $\calH_1$ and $\ket{\widetilde i}$ in $\calH_2$ and positive numbers $p_j > 0$.
In quantum information theory this is called the \emphindex{Schmidt decomposition}; it is a simple consequence of the singular value decomposition in linear algebra.
Thus if $\rho = \proj\psi$ is a pure state then it follows that $\rho_1 = \sum_i p_i \proj i$ and $\rho_2 = \sum_i p_i \proj {\widetilde i}$ have the same non-zero eigenvalues, including multiplicities (and indeed the same spectrum if $\calH_1$ and $\calH_2$ are of the same dimension).
Conversely, for any two such density matrices $\rho_1$ and $\rho_2$, we can always use \eqref{eq:onebody/schmidt} with the respective eigenbases to define a corresponding global pure state. We record for future reference:

\begin{lem}
\label{lem:onebody/two}
  Any two density matrices $\rho_1$ and $\rho_2$ are compatible with a global pure state
  if and only if
  $\rho_1$ and $\rho_2$ have the same non-zero eigenvalues (including multiplicities),
  i.e., if and only if
  $\vec\lambda_1 \setminus \{0\} = \vec\lambda_2 \setminus \{0\}$.
\end{lem}

In particular, any density matrix $\rho_1$ can be realized as the reduced density matrix of a pure state.
In quantum information, such a pure state is called a \emphindex{purification} of the density matrix $\rho_1$.

An important consequence of \autoref{lem:onebody/two} is that $n+1$ spectra $\vec\lambda_1$, \dots, $\vec\lambda_n$, $\vec\mu$ are compatible with a pure state if and only if $\vec\lambda_1$, \dots, $\vec\lambda_n$ are compatible with a global state of spectrum $\vec\mu$.
Therefore there is no loss of generality in restricting to pure states in our formulation of \autoref{pro:onebody/qmp}.

Another useful corollary is that the one-body quantum marginal problem for an arbitrary number of particles can always be reduced to the case $n=3$:
A given collection of spectra $\vec\lambda_1, \dots, \vec\lambda_n$ is compatible if and only if there exists a spectrum $\vec\mu$ such that both $\vec\lambda_1, \vec\lambda_2, \vec\mu$ as well as $\vec\mu, \vec\lambda_3, \dots, \vec\lambda_n$ are compatible, and this process can be iterated.
This is immediate from the preceding and \autoref{lem:onebody/two}, which also shows that the rank of $\vec\mu$ can be bounded by the minimum of $\dim \calH_1 \times \dim \calH_2$ and $\prod_{k=3}^n \dim \calH_k$.

\subsection*{Identical Particles}

In the case of $n$ \emphindex{identical particles}, the system is described by the $n$-th symmetric or antisymmetric tensor power of the single-particle Hilbert space, $\calH = \Sym^n \calH_1$\nomenclature[QSym^n H_1]{$\Sym^n \calH_1$}{symmetric subspace of $\calH_1^{\otimes n}$} or $\calH = \Alt^n \calH_1$\nomenclature[QWedge^n H_1]{$\Alt^n \calH_1$}{antisymmetric subspace of $\calH_1^{\otimes n}$} depending on whether the particles are \emphindex{bosons} or \emphindex{fermions}. By considering $\calH$ as a subspace of $\calH_1^{\otimes n}$, we can define the one-body reduced density matrices as in the case of distinguishable particles. Of course, $\rho_1 = \dots = \rho_n$, since for bosons as well as for fermions the global state $\rho$ is permutation-invariant. In summary,
\begin{equation}
\label{eq:onebody/rdm fermions}
  \tr \rho_1 O_1
  = \frac 1 n \tr \rho (\sum_{k=1}^n \Id^{\otimes k-1} \otimes O_1 \otimes \Id^{\otimes n-k})
  = \frac 1 n \tr \rho (\sum_{i,j} \braket{i | O | j} a_i^\dagger a_j),
\end{equation}
where in the last expression $a_i^\dagger$ and $a_j$ denote the creation and annihilation operators with respect to an arbitrary basis $\ket i$ of the single-particle Hilbert space.

For fermions, we thus arrive at the following variant of \autoref{pro:onebody/qmp}:

\begin{pro}[One-Body $N$-Representability Problem]\index{n-representability problem@$n$-representability problem!one-body}
\label{pro:onebody/nrep}
  Given a density matrix $\rho_1$ on $\calH_1$, does there exist a pure state $\rho = \proj\psi$ on $\Alt^n \calH_1$ such that $\rho_1$ is its one-body reduced density matrix?
\end{pro}

We will call such a density matrix $\rho_1$ \emph{$n$-representable}\index{n-representable density matrix@$n$-representable density matrix}.
In the context of second quantization, it is often more convenient to normalize the one-body marginal to trace $n$.
Following quantum chemistry conventions, we correspondingly set $\gamma_1 := n \rho_1$ and call it the \emphindex{first-order density matrix}\nomenclature[Qgamma_1]{$\gamma_1$}{first-order density matrix of fermionic state} \cite{Loewdin55} (but remark that it is not a density matrix in the strict sense).
In quantum chemistry, the diagonal entries of $\gamma_1$ are called \emphindex{occupation numbers}, while its eigenvalues are called the \emph{natural occupation numbers}\index{occupation numbers!natural}.
As in the case of distinguishable particles, \autoref{pro:onebody/nrep} depends only on the eigenvalues of $\rho_1$, or, equivalently, on the natural occupation numbers of $\gamma_1$.
In this language, the \emph{Pauli exclusion principle}\index{Pauli exclusion principle|textbf} asserts that the natural occupation numbers, and hence all occupation numbers, never exceed one \cite{Pauli25}.
This is obvious from second quantization, since $\braket{i | \gamma_1 | i} = \tr \rho \, a_i^\dagger a_i \leq 1$.
Equivalently, the largest eigenvalue of an $n$-representable one-body density matrix is at most $1/n$.
However, there are many more constraints on the natural occupation numbers of a pure state of $n$ fermions \cite{BorlandDennis72,KlyachkoAltunbulak08}.

\bigskip

\autoref{pro:onebody/nrep} can also be formulated for bosons. Here it can be shown that the resulting problem is in fact trivial: Any density matrix $\rho_1 = \sum_i p_i \proj i$ arises as the one-body reduced density matrix of a pure $\rho = \proj\psi$ state on the symmetric subspace, e.g., $\ket\psi = \sum_i \sqrt{p_i} \ket i^{\otimes n}$ \cite{KlyachkoAltunbulak08}.

Further variants of the one-body quantum marginal problem may arise from physical or mathematical considerations\index{quantum marginal problem!variants}.
For instance, the analysis of fermionic systems with several internal degrees of freedom leads to the study of other irreducible representations besides the symmetric or antisymmetric subspace \cite{KlyachkoAltunbulak08}.
We will discuss one such example at the end of \autoref{sec:kirwan/examples}.
We may also combine systems composed of different species of particles, some of them indistinguishable among each other.
On a mathematical level, this situation also arises when the ``purification trick'' that we used to restrict to global pure states in the formulation of \autoref{pro:onebody/qmp} is applied to systems of identical particles.
The mathematical framework that we outline in the subsequent sections subsumes all these variants of the marginal problem.

\section{Lie Groups and their Representations}
\label{sec:onebody/lie}

Before we proceed it will be useful to recall some fundamental notions from the theory of Lie groups and their representations.
We illustrate the general theory in the important case of the unitary groups and their complexification, the general linear groups, and summarize the notation in \autoref{tab:onebody/unitary group}.
We refer to \cite{Knapp86, FultonHarris91, CarterSegalMacDonald95, Knapp02, Procesi07, Brion10} for comprehensive introductions to the subject.

\subsection*{Structure Theory}

Let $G$ be a connected reductive algebraic group $G$ with Lie algebra $\mathfrak g$.%
\index{algebraic group}\index{reductive algebraic group}\index{Lie group}\nomenclature[RG]{$G$, $\mathfrak g$}{connected reductive algebraic group and its Lie algebra}
We denote the Lie bracket by $[-,-]$ and the exponential map by $\exp \colon \mathfrak g \rightarrow G$.%
\nomenclature[{R[-,-]}]{$[-,-]$}{Lie bracket}%
\nomenclature[Rexp]{$\exp$}{exponential map}
Let $K \subseteq G$ be a maximal compact subgroup with Lie algebra $\mathfrak k$.%
\nomenclature[RK]{$K$, $\mathfrak k$}{compact connected Lie group and its Lie algebra}
Then $G$ is the complexification of $K$, and $\mathfrak g = \mathfrak k \oplus i \mathfrak k$.
Let $T \subseteq K$ be a maximal Abelian subgroup, called a \emphindex{maximal torus}, with Lie algebra $\mathfrak t \subseteq \mathfrak k$.%
\nomenclature[RT]{$T$, $\mathfrak t$}{maximal torus of $K$ and its Lie algebra}
Its complexification is a maximal Abelian subgroup $T_\CC \subseteq G$ with Lie algebra $\mathfrak h = \mathfrak t \oplus i \mathfrak t$.%
\nomenclature[RT_C]{$T_\CC$, $\mathfrak h$}{complexification of $T$ and its Lie algebra}
Since $T_\CC$ is Abelian, its irreducible representations are one-dimensional and can be be written in the form
\begin{equation*}
  \Pi \colon T_\CC \rightarrow \CC^*,
  \quad
  \Pi(\exp X) = e^{\omega(X)},
\end{equation*}
where $\CC^* := \GL(1) = \CC \setminus \{0\}$ and where $\omega$ is a complex-linear functional in $\mathfrak h^* = \Hom_\CC(\mathfrak h, \CC)$\nomenclature[Rhstar]{$\mathfrak h^*$}{complex-linear functionals on $\mathfrak h$}, called a \emphindex{weight}.
Note that the weight is simply the Lie algebra representation corresponding to $\Pi$; it encodes the eigenvalues by which the elements of the Lie algebra act on the representation.
Since $\Pi(T) \subseteq \U(1)$, the weights attain imaginary values on $\mathfrak t$.
We define the real subspace%
\nomenclature[Ritstar]{$i \mathfrak t^*$}{functionals $\omega \in \mathfrak h^*$ with $\omega(i\mathfrak t) \subseteq \RR$}%
\nomenclature[Romega]{$\omega$}{element of $i \mathfrak t^*$}
\begin{equation*}
  i \mathfrak t^*
  = \{ \omega \in \mathfrak h^* : \omega(\mathfrak t) \in i \RR \}
  = \{ \omega \in \mathfrak h^* : \omega(i \mathfrak t) \in \RR \}.
\end{equation*}
Then the set of weights forms a lattice $\Lambda^*_G \subseteq i \mathfrak t^*$, called the \emphindex{weight lattice};%
\nomenclature[RLambdaG]{$\Lambda^*_G$}{weight lattice of $G$}
its rank is equal to $r_G = \dim_{\RR} T = \dim_{\CC} T_\CC$, called the \emph{rank}\index{rank of Lie group} of the group $G$.
\nomenclature[RrG]{$r_G$}{rank of $G$}
Now let $\Pi \colon G \rightarrow \GL(V)$ be an arbitrary representation of $G$ on a finite-dimensional complex vector space $V$.
Its differential $\pi \colon \mathfrak g \rightarrow \gl(V)$ is a representation of the Lie algebra.%
\nomenclature[RPi]{$\Pi$, $\pi$}{representation of $G$ and corresponding representation of $\mathfrak g$}
If we restrict to $T_\CC \subseteq G$ then we obtain a decomposition
\begin{equation*}
  V = \bigoplus_{\omega \in \Lambda^*_G} V_\omega,
\end{equation*}
where each
\[V_\omega := \{ \ket\psi \in V : \pi(H) \ket\psi = \omega(H) \ket\psi \quad (\forall H \in \mathfrak h) \}\]
is called a \emphindex{weight space}.%
\nomenclature[RV_omega]{$V_\omega$}{weight spaces of a representation $V$}%
\nomenclature[RV_omega]{$V, W, \dots$}{representations\nomnorefpage}
The elements of $V_\omega$ are called \emph{weight vectors}\index{weight vector}; each weight vector spans an irreducible representation of $T_\CC$ of weight $\omega$.

\bigskip

The Lie group $G$ acts on itself by conjugation, $g . h := g h g^{-1}$.
By taking the derivative, we obtain the \emphindex{adjoint representation} $\Ad \colon G \rightarrow \GL(\mathfrak g)$ of $G$ on $\mathfrak g$.%
\nomenclature[Rad]{$\Ad$, $\ad$}{adjoint representation of $G$ on $\mathfrak g$ and its differential}
Its differential is the representation of the Lie algebra $\mathfrak g$ on itself by the Lie bracket, $\ad(X)Y = [X,Y]$.
By decomposing the adjoint representation into weight spaces $\mathfrak g_\alpha$ and observing that $\mathfrak g_0 = \mathfrak h$, we obtain
\begin{equation*}
  \mathfrak g = \mathfrak h \oplus \bigoplus_{\alpha \in R_G} \mathfrak g_{\alpha},
\end{equation*}
where $R_G := \{ \alpha \neq 0 : \mathfrak g_\alpha \neq 0 \} \subseteq \Lambda^*_G$ is the set of non-trivial weights of the adjoint representation, called the \emphindex{roots}.%
\nomenclature[RR_G]{$R_G$}{set of roots of $G$}
\nomenclature[Ralpha]{$\alpha$}{root}
The corresponding weight spaces
\[\mathfrak g_\alpha = \{ X \in \mathfrak g : [H,X] = \alpha(H) X \quad (\forall H \in \mathfrak h) \}\]
are called \emphindex{root spaces}.%
\nomenclature[Rg_alpha]{$\mathfrak g_\alpha$}{root spaces of $\mathfrak g$}
For an arbitrary representation $V$ we have that $\pi(\mathfrak g_{\alpha}) V_\omega \subseteq V_{\omega+\alpha}$;
in particular, $[\mathfrak g_{\alpha}, \mathfrak g_{\beta}] \subseteq \mathfrak g_{\alpha+\beta}$. %
All root spaces $\mathfrak g_\alpha$ are one-dimensional, and for each root $\alpha$, $-\alpha$ is also a root.
We can find basis vectors $E_\alpha \in \mathfrak g_\alpha$ and elements $Z_\alpha \in i\mathfrak t$, called \emphindex{co-roots},
such that
\begin{equation}
\label{eq:onebody/sl2}
  [E_\alpha, E_{-\alpha}] = Z_\alpha
  \quad\text{and}\quad
  [Z_\alpha, E_{\pm\alpha}] = \pm 2 E_{\pm\alpha}.
\end{equation}
Thus for each pair of roots $\{\pm\alpha\} \subseteq R$, $\mathfrak g_\alpha \oplus \mathfrak g_{-\alpha} \oplus \CC Z_\alpha$ is a complex Lie algebra isomorphic to $\mathfrak{sl}(2)$.
Set $\mathfrak k_\alpha := (\mathfrak g_\alpha \oplus \mathfrak g_{-\alpha}) \cap \mathfrak k$.%
\nomenclature[Rk_alpha]{$\mathfrak k_\alpha$}{root spaces of $\mathfrak k$}
Then $\mathfrak k_\alpha \oplus i\RR Z_\alpha$ is a real Lie algebra isomorphic to $\mathfrak{su}(2)$, and we have a decomposition
\begin{equation}
\label{eq:onebody/compact root space decomposition}
  \mathfrak k = \mathfrak t \oplus \bigoplus_{\{\pm\alpha\} \subseteq R_G} \mathfrak k_\alpha.
\end{equation}
The ``Pauli matrices'' $X_\alpha := E_\alpha + E_{-\alpha}$ and $Y_\alpha := i (E_{-\alpha} - E_\alpha)$ are a basis of $i \mathfrak k_\alpha$; they satisfy the commutation relations%
\nomenclature[RX alpha]{$X_\alpha, Y_\alpha, Z_\alpha$}{Pauli matrices corresponding to root $\alpha$}%
\nomenclature[RX]{$X, Y, Z, \dots$}{elements of $\mathfrak g$\nomnorefpage}%
\nomenclature[RH]{$H$}{element of $i \mathfrak t$\nomnorefpage}
\begin{equation}
\label{eq:onebody/pauli}
  [X_\alpha, Y_\alpha] = 2 i Z_\alpha, \quad
  [Y_\alpha, Z_\alpha] = 2 i X_\alpha \quad\text{and}\quad
  [Z_\alpha, X_\alpha] = 2 i Y_\alpha.
\end{equation}
Now choose a decomposition $R_G = R_{G,+} \sqcup R_{G,-}$ of the set of roots into \emph{positive and negative roots}\index{roots!positive and negative}\nomenclature[RR_Gpm]{$R_{G,\pm}$}{positive and negative roots}.
That is, $R_{G,-} = -R_{G,+}$ and each subset is strictly contained in a half-space of the (real) span of the roots (which is equal to $i\mathfrak t^*$ if $G$ is semisimple).
Then we have a decomposition
\begin{equation*}
  \mathfrak g
  = \mathfrak h \oplus \bigoplus_{\alpha \in R_{G,+}} \mathfrak g_\alpha \oplus \bigoplus_{\alpha \in R_{G,-}} \mathfrak g_\alpha
  =: \mathfrak h \oplus \mathfrak n_+ \oplus \mathfrak n_-,
\end{equation*}
where the $\mathfrak n_\pm$ are nilpotent Lie algebras; the corresponding Lie groups $N_\pm \subseteq G$ are called \emphindex{maximal unipotent subgroups}\nomenclature[RN^pm]{$N_\pm, \mathfrak n_\pm$}{maximal unipotent subgroups of $G$ and its Lie algebras}.

\bigskip

Another consequence of the choice of positive roots is the following.
Consider the dual of the adjoint representation of $G$ on $\mathfrak g$, given by
$(\Ad^*(g) \varphi)(X) = \varphi(\Ad(g^{-1})X)$
for all $g \in G$, $\varphi \in \mathfrak g^*$ and $X$ in $\mathfrak g$.
It is not hard to see that its restriction to $K$ preserves the real subspace%
\nomenclature[Rikstar]{$i \mathfrak k^*$}{functionals $\varphi \in \mathfrak g^*$ with $\varphi(i\mathfrak k) \subseteq \RR$}%
\begin{equation*}
  i \mathfrak k^* = \{ \varphi \in \mathfrak g^* : \varphi(i\mathfrak k) \subseteq \RR \}
\end{equation*}
and we shall call it the \emphindex{coadjoint representation}\nomenclature[Radstar]{$\Ad^*$}{coadjoint representation of $K$ on $i \mathfrak k^*$} of $K$ on $i \mathfrak k^*$ (our choice of factor $i$ is somewhat idiosyncratic but will be rather convenient in the sequel).
We may consider $\mathfrak h^* \subseteq \mathfrak g^*$ and $i \mathfrak t^* \subseteq i \mathfrak k^*$ by extending each functional by zero on the root spaces $\mathfrak g_\alpha$.
Then the \emphindex{positive Weyl chamber}%
\nomenclature[Ritstarplus]{$i \mathfrak t^*_+, i \mathfrak t^*_{>0}$}{positive Weyl chamber and its (relative) interior}%
\nomenclature[Rlambda, mu, nu]{$\lambda, \mu, \alpha, \beta, \gamma, \dots$}{elements of $i \mathfrak t^*_+$}
\begin{equation*}
  i\mathfrak t^*_+ = \{ \lambda \in i \mathfrak t^* : \lambda(Z_\alpha) \geq 0 \quad (\forall \alpha \in R_{G,+}) \}
\end{equation*}
is a cross-section for the coadjoint action of $K$.
In other words, each \emphindex{coadjoint orbit} $\Ad^*(K) \varphi$ intersects the positive Weyl chamber in a single point $\lambda \in i \mathfrak t^*_+$.
We write $\calO_{K,\lambda} := \Ad^*(K) \lambda$ for the coadjoint orbit through $\lambda$.%
\nomenclature[RO_Klambda]{$\calO_{K,\lambda}$}{coadjoint $K$-orbit through $\lambda$}
The positive Weyl chamber is a convex cone (pointed if $G$ is semisimple).
Its \emph{(relative) interior}\index{positive Weyl chamber!interior} is
\[i\mathfrak t^*_{>0} = \{ \lambda \in i \mathfrak t^* : \lambda(Z_\alpha) > 0 \quad (\forall \alpha \in R_{G,+}) \}.\]
For any $\lambda \in i \mathfrak t^*_{>0}$, the $K$-stabilizer is the maximal torus $T$, so that $\calO_{K,\lambda} \cong K/T$.

The last piece of structure is the \emphindex{Weyl group} $W_K = N_K(T) / T$, where $N_K(T) = \{ k \in K : kT = Tk \}$ denotes the normalizer of the maximal torus $T \subseteq K$.%
\nomenclature[RW_K]{$W_K$}{Weyl group of $K$}
It is a finite group that acts on $i \mathfrak t^*$.
For any representation $V$, the action of the Weyl group leaves the set of weights invariant.
In particular, the set of roots is left invariant.
The Weyl group acts simply transitively on the set of Weyl chambers obtained from different choices of positive roots.
In particular, every $W_K$-orbit in $i \mathfrak t^*$ has a unique point of intersection with the positive Weyl chamber $i \mathfrak t^*_+$, and there exists a Weyl group element, known as the \emphindex{longest Weyl group element} $w_0$, that exchanges the positive and negative roots and hence sends the ``negative Weyl chamber'' $- i \mathfrak t^*_+$ to $i \mathfrak t^*_+$.
\nomenclature[Rw_0]{$w_0$}{longest Weyl group element}%
More generally, one can define the \emph{length}\index{length of Weyl group element} $l(w)$ of a Weyl group element as the minimal number of certain standard generators required to write $w$, but we will not need this level of generality.\nomenclature[Rl(w)]{$l(w)$}{length of a Weyl group element}
\subsection*{Representation Theory}

Let $\Pi \colon G \rightarrow \GL(V)$ be a finite-dimensional representation of $G$, with infinitesimal representation $\pi \colon \mathfrak g \rightarrow \mathfrak{gl}(V)$ and weight space decomposition $V = \bigoplus_\omega V_\omega$.
A weight vector $\ket\psi \in V_\omega$ is called a \emphindex{highest weight vector} if
$\Pi(N_+) \ket\psi \equiv \ket\psi$, or, equivalently, if $\pi(\mathfrak n_+) \ket\psi = 0$.
The corresponding \emphindex{highest weight} $\omega$ is necessarily \emph{dominant}\index{weight!dominant}, i.e., an element of $\Lambda^*_{G,+} := \Lambda^*_G \cap i \mathfrak t^*_+$.\nomenclature[RLambdaG+]{$\Lambda^*_{G,+}$}{dominant weights}

The fundamental theorem of the representation theory of compact connected Lie groups $K$ asserts that the \emphindex{irreducible representations} of $G$ can be labeled by their \emphindex{highest weight}:
Any irreducible representation contains a highest weight vector $\ket\lambda$, unique up to multiplication by a scalar.
Conversely, for every $\lambda \in \Lambda^*_{G,+}$ there exists a unique irreducible representation $V_{K,\lambda}$ with $\lambda$ as the highest weight.\nomenclature[RV_Glambda]{$V_{G,\lambda}$}{irreducible $G$-representation with highest weight $\lambda$}
The \emph{dual representation}\index{irreducible representations!dual} $V^*_{K,\lambda}$ of an irreducible representation is again irreducible, and its highest weight is $\lambda^* := -w_0 \lambda$, where $w_0$ is the longest Weyl group element as defined above.%
\nomenclature[Rlambdastar]{$\lambda^*$}{highest weight of the dual representation $V_{G,\lambda}^*$}
Like any finite-dimensional representation of $K$, $V_{K,\lambda}$ extends to a \emph{rational} representation $V_{G,\lambda}$ of the algebraic group $G$, i.e., a representation whose matrix elements are given by rational functions on the algebraic group $G$ (that is, by morphisms of algebraic varieties, which is the appropriate notion in this context).
All irreducible rational representations of $G$ can be obtained in this way.\index{irreducible representations!rational} Therefore, the representation theory of $K$ and of $G$ are essentially equivalent.

An arbitrary $G$-representation $V$ can always be equipped with a $K$-invariant inner product (choose an arbitrary inner product and average). In this case, $\Pi(K) \subseteq \U(V)$, and so $\pi(\mathfrak k)$ consist of anti-Hermitian and $\pi(i \mathfrak k)$ of Hermitian operators.
Moreover, $V$ can always be decomposed into irreducible representations and the irreducible representations that occur in $V$
are in one-to-one correspondence with the highest weight vectors in $V$ (up to rescaling).
In particular, the subspace $V^G = \{ v \in V : \Pi(g) v = v \; (\forall g \in G) \}$ of \emph{invariant vectors}\index{invariant vector}\nomenclature[RV^G]{$V^G$}{subspace of invariant vectors in representation $V$} is the sum of all trivial representations that occur in $V$.

\subsection*{The General Linear and Unitary Groups}

\begin{table}
  \centering
  \begin{align*}
    \toprule
    G &= \GL(d) = \{ g \in \CC^{d \times d} \text{ invertible} \} \\
    \mathfrak g &= \mathfrak{gl}(d) = \CC^{d \times d} \\
    K &= \U(d) = \{ U \in \CC^{d \times d} : U U^\dagger = U^\dagger U = \Id \} \\
    \mathfrak k &= \mathfrak u(d) = \{ X \in \CC^{d \times d} : X^\dagger = -X \} \\
    i \mathfrak k &= \{ X = X^\dagger \in \CC^{d \times d} \} \\
    \Ad(g)X &= gXg^{-1}, \quad \ad(X)Y = [X,Y] \\
    \midrule
    T_\CC &= \{ \matrix{g_{1,1} && \\ &\ddots& \\ &&g_{d,d}} : g_{1,1}, \dots, g_{d,d} \neq 0 \} \cong (\CC^*)^d \\
    \mathfrak h &= \{ \matrix{z_{1,1} && \\ &\ddots& \\ &&z_{d,d}} : z_{1,1}, \dots, z_{d,d} \in \CC \} \cong \CC^d \\
    T &= \{ \matrix{e^{i x_{1,1}} && \\ &\ddots& \\ && e^{i x_{d,d}}} : x_{1,1}, \dots, x_{d,d} \in \RR \} \cong \U(1)^d \\
    \mathfrak t &= \{ \matrix{i x_{1,1} && \\ &\ddots& \\ && i x_{d,d}} : x_{1,1}, \dots, x_{d,d} \in \RR \} \cong i\RR^d \\
    i \mathfrak t &\cong \RR^d \\
    i \mathfrak t^* &= \{ \omega : H \mapsto \sum_{j=1}^d \omega_j H_{j,j} : \omega_1, \dots, \omega_j \in \RR \} \cong \RR^d \\
    \Lambda^*_G &= \{ \omega : H \mapsto \sum_{j=1}^d \omega_j H_{j,j} : \omega_1, \dots, \omega_d \in \ZZ \} \cong \ZZ^d \\
    \midrule
    R_G &= \{ \alpha_{ij} : i \neq j \} \text{ where } \alpha_{ij}(H) = H_{i,i} - H_{j,j} \\
    \mathfrak g_{ij} &= \CC E_{ij} \text{ where } E_{ij} = \ketbra i j \\
    X_{ij} &= \ketbra i j + \ketbra j i \\
    Y_{ij} &= i (\ketbra j i - \ketbra i j) \\
    Z_{ij} &= \proj i - \proj j \\
    \bottomrule
  \end{align*}
  \caption[The general linear and unitary groups]{Summary of the structure theory of the general linear and unitary group, and of some important representations.}
  \label{tab:onebody/unitary group}
\end{table}

\begin{table}
  \captionsetup{list=no}
  \ContinuedFloat
  \begin{align*}
    \toprule
    R_{G,+} &= \{ \alpha_{ij} : i < j \} \\
    R_{G,-} &= \{ \alpha_{ij} : i > j \} \\
    \mathfrak n_+ &= \{ \matrix{0 & * & * \\ & \ddots & * \\ & & 0} \} \\
    N_+ &= \{ \matrix{1 & * & * \\ & \ddots & * \\ & & 1} \} \\
    \mathfrak n_- &= \{ \matrix{0 &  &  \\ * & \ddots & \\ *&* & 0} \} \\
    N_- &= \{ \matrix{1 &  &  \\ * & \ddots & \\ *&* & 1} \} \\
    \midrule
    i \mathfrak k^* \! &= \{ \tr X (-) : X = X^\dagger \} \cong i \mathfrak k \\
    \calO_{K,\lambda} &\cong \{ X = X^\dagger \in \CC^{d \times d} : \spec X = \lambda \} \\
    i\mathfrak t^*_+ &\cong \{ \lambda \in \RR^d : \lambda_1 \geq \dots \geq \lambda_d \} \\
    i\mathfrak t^*_{>0} &\cong \{ \lambda \in \RR^d : \lambda_1 > \dots > \lambda_d \} \\
    W_K &= S_d \text{ (acts by permuting diagonal entries)} \\
    w_0 &= (1 \; d) (2 \; d\!-\!1) \dots = \matrix{1 & \dots & d \\ d & \dots & 1} \\
    (-1)^{l(w)} &= \sign w \\
    \Lambda^*_{G,+} &\cong \{ \lambda \in \ZZ^d : \lambda_1 \geq \dots \geq \lambda_d \} \\
    &\ni {\yng(1), \yng(2), \yng(1,1), \dots} \text{ (Young diagrams)} \\
    \lambda^* &= (-\lambda_d, \dots, -\lambda_1) \\
    \midrule
    V^d_\lambda &\hphantom{=} \text{irreducible representation with highest weight $\lambda$} \\  %
    V^d_{(k,0,\dots,0}) &= V^d_{\Yboxdim{5pt}\yng(1)\dots\yng(1)\Yboxdim{8pt}} = \Sym^k(\CC^d) \text{ symmetric subspace}\\
    V^d_{(\underbrace{\scriptstyle 1,\dots,1}_{k},0,\dots,0}) &= V^d_{\left.\Yboxdim{5pt}\begin{array}{c}\yng(1) \\[-0.5em] \scriptscriptstyle\vdots \\[-0.5em] \vspace{-0.12cm}\yng(1)\end{array}\Yboxdim{8pt}\right\}k} = \Alt^k(\CC^d) \text{ antisymmetric subspace} \\
    \bottomrule
  \end{align*}
  \caption{Summary of the structure theory of the general linear and unitary group, and of some important representations (cont.).}
\end{table}

The \emphindex{general linear group} $G = \GL(d)$ of invertible $d \times d$-matrices is a connected reductive algebraic group, with Lie algebra $\mathfrak g = \mathfrak{gl}(d)$ the space of complex $d \times d$-matrices. The Lie bracket is the usual commutator, $[X,Y] := XY-YX$, and the exponential map is the usual matrix exponential.%
\nomenclature[RGL(d)]{$\GL(d)$, $\mathfrak{gl}(d)$}{general linear group and its Lie algebra}
The \emphindex{unitary group} $K = \U(d)$, whose elements are unitary $d \times d$-matrices, is a maximal compact subgroup, with Lie algebra $\mathfrak k =  \mathfrak u(d)$ the space of \emph{anti-}Hermitian matrices.%
\nomenclature[RU(d)]{$\U(d)$, $\mathfrak{u}(d)$}{unitary group and its Lie algebra}
Thus $i \mathfrak k$ is the set of Hermitian matrices, which we may identify with $i \mathfrak k^*$ by using the Hilbert--Schmidt inner product.

The subgroup of diagonal unitary matrices is a maximal torus $T \subseteq \U(d)$ and can be identified with $\U(1)^d$.
Likewise, its Lie algebra $i \mathfrak t$ can be identified with $i \RR^d$.
Thus its complexification $\mathfrak h \cong \CC$ consists of general diagonal matrices and the corresponding Lie group $T_\CC \cong (\CC^*)^d$ is the subgroup of diagonal invertible matrices in $\GL(d)$.
Its irreducible representations are all of the form
\begin{equation*}
  \Pi \colon T_\CC \cong (\CC^*)^d \rightarrow \CC^*, \quad
\matrix{g_{1,1} && \\ & \ddots & \\ && g_{d,d}} \mapsto g_{1,1}^{\omega_1} \dotsm g_{d,d}^{\omega_d}
\end{equation*}
for integers $\omega_1, \dots, \omega_d \in \ZZ$.
The corresponding weight is $\omega(H) = \sum_{j=1}^d w_j H_{j,j}$.
In this way, the weight lattice $\Lambda^*_G \subseteq i \mathfrak t^*$ can be identified with the lattice $\ZZ^d \subseteq \RR^d$.

The roots of $\GL(d)$ are the functionals $\alpha_{ij}(H) = H_{i,i} - H_{j,j}$\nomenclature[Ralpha_ij]{$\alpha_{i,j}$}{root of $\GL(d)$ and $\SL(d)$}, and the corresponding root spaces $\mathfrak g_{ij}$ are spanned by the elementary matrices $E_{ij} = \ketbra i j$ that have a single non-zero entry in the $i$-th row and $j$-th column. Indeed, we have that $[H, \ketbra i j] = (H_{i,i} - H_{j,j}) \ketbra i j$ for any diagonal matrix $H = \sum_j H_{j,j} \proj j \in \mathfrak h$.
A choice of positive roots is given by those roots $\alpha_{ij}$ with $i < j$.
Thus the nilpotent Lie algebras $\mathfrak n_{\pm}$ consist of the strictly upper and lower triangular matrices, respectively.
The corresponding unipotent subgroups $N_\pm$ are upper and lower triangular with ones on the diagonal.

The adjoint action is by conjugation.
If we identify $i \mathfrak k \cong i \mathfrak k^*$ then the coadjoint orbits of $K = \U(d)$ consist of Hermitian matrices with fixed spectrum.
Then the positive Weyl chamber $i \mathfrak t^*_+$ can be identified with the set of Hermitian diagonal matrices whose entries are weakly decreasing, or with their spectra $\lambda_1 \geq \dots \geq \lambda_d$.
The assertion that $i \mathfrak t^*_+$ is a cross-section for the coadjoint action of $K$ on $i \mathfrak k^*$ corresponds to the plain fact that any Hermitian matrix can be diagonalized by a unitary.
Its interior $i \mathfrak t^*_{>0}$ then corresponds to the set of non-degenerate spectra $\lambda_1 > \dots > \lambda_d$.
The claim that the $K$-stabilizer of any $\lambda \in i \mathfrak t^*_{>0}$ is $T$ amounts to the fact that the only unitaries that commute with a diagonal matrix with non-degenerate spectrum are the diagonal unitary matrices.

Finally, the Weyl group can be identified with the symmetric group $S_d$; it acts on $i \mathfrak t$ by permuting diagonal entries.
The length $l(w)$ of a permutation $w$ is the number of transpositions $(i \; i\!+\!1)$ required to write the permutation $w \in S_d$, and the longest Weyl group element $w_0$ is the  ``order-reversing permutation'' which sends any $-\lambda \in -i \mathfrak t^*_+$ to $\lambda^* := (-\lambda_d, \dots, -\lambda_1) \in i \mathfrak t^*_+$.

\bigskip

We now turn to the representation theory of the general linear and unitary groups.
The dominant weights in $\Lambda^*_{G,+} = \Lambda^*_G \cap i \mathfrak t^*_+$ can be identified with integers vectors in $\lambda \in \ZZ^d$ which have weakly decreasing entries, $\lambda_1 \geq \dots \geq \lambda_d$.
They label the \emph{irreducible representations of $\U(d)$ and of $\GL(d)$}\index{irreducible representations!$\U(d)$ and $\GL(d)$},
which we abbreviate by $V^d_\lambda$.%
\nomenclature[RV^d_lambda]{$V^d_\lambda$}{irreducible representations of $\GL(d)$, $\U(d)$, $\SL(d)$, $\SU(d)$}
Of particular interest for us will be those dominant weights for which all entries are non-negative (i.e., $\lambda_d \geq 0$).
We will think of them as \emphindex{Young diagrams}, i.e., arrangements of $\abs\lambda := \sum_j \lambda_j$ \emph{boxes}\index{Young diagram!boxes} into $\max \{ j : \lambda_j > 0 \}$ \emph{rows}\index{Young diagram!rows}, where we put $\lambda_j$ boxes into the $j$-th row.
\nomenclature[Rlambda, mu, nuY]{$\lambda, \mu, \alpha, \beta, \gamma, \dots$}{Young diagrams}%
\nomenclature[Rlambda abs]{$\abs\lambda$}{number of boxes of a Young diagram}%
For example,
\begin{equation*}
  (3,1,0,0), (3,1,0), (3,1) \equiv \yng(3,1)
  \quad\text{and}\quad
  (1,1,1) \equiv \yng(1,1,1).
\end{equation*}
The first diagram in the example has two rows and four boxes, while the second diagram has three boxes as well as rows.
The number of rows is also called the \emph{height}\index{Young diagram!height} of a Young diagram $\lambda$.
We write $\lambda \vdash_d k$ for a Young diagram with $k$ boxes and at most $d$ rows.%
\nomenclature[Rlambda vdash_d k]{$\lambda \vdash_d k$}{Young diagram with $k$ boxes and at most $d$ rows}

Mathematically, the irreducible representations corresponding to Young diagrams $\lambda \vdash_d k$ are the \emph{polynomial} irreducible representations of $\GL(d)$\index{irreducible representations!polynomial}, i.e., those representations whose matrix elements are given by polynomial functions on $\GL(d)$.
The number of boxes $k$ corresponds to the degree of the polynomials, or, equivalently, to the power by which scalar multiples of the identity matrix in $\GL(d)$ act: $\Pi(\lambda \Id) = \lambda^k \Id$.
For example, the determinant representation $g \mapsto \det(g)$ is a polynomial representation with Young diagram $(1,\dots,1)$ and degree $d$, while its inverse $g \mapsto 1/\det(g)$ is only a rational representation.
For any irreducible representation $V^d_\lambda$ and $k \in \ZZ$, $V^d_\lambda \otimes {\det}^m$ is an irreducible representation with highest weight $\lambda + (m,\dots,m)$.
It follows that any rational representation can be made polynomial be tensoring with a sufficiently high power of the determinant; conversely, any rational representation can be obtained from a polynomial one by tensoring with a sufficiently high inverse power of the determinant.
In \autoref{tab:onebody/unitary group} we summarize the structure theory of the general linear and unitary groups and we also list some important irreducible representations that we will use in the sequel.

\bigskip

We briefly discuss the \emphindex{special linear group} $\SL(d) = \{ g \in \GL(d) : \det g = 1 \}$ and its maximal compact subgroup $\SU(d) = \SL(d) \cap \U(d)$.%
\nomenclature[RSL(d)]{$\SL(d)$, $\mathfrak{sl}(d)$}{special linear group and its Lie algebra}%
\nomenclature[RSU(d)]{$\SU(d)$, $\mathfrak{su}(d)$}{special unitary group and its Lie algebra}
Since $\SL(d)$ and $\GL(d)$ only differ in their center, the basic structure theory is essentially unchanged apart from the fact that the Lie algebras and their duals are obtained by projecting to the traceless part.
In particular, the set of roots is unchanged.
Moreover, any irreducible $\SL(d)$-representation is the restriction of an irreducible $\GL(d)$-representation $V^d_\lambda$; this restriction depends only on the traceless part of the highest weight, i.e., on the differences of rows $\lambda_j - \lambda_{j+1}$, since the determinant representation is trivial for matrices in $\SL(d)$\index{irreducible representations!$\SU(d)$ and $\SL(d)$}.
For example, the \emph{irreducible representations of $\SU(2)$ and of $\SL(2)$}\index{irreducible representations!$\SU(2)$ and $\SL(2)$} can all be obtained by restricting an irreducible $\GL(2)$-representation $V^2_{(k,0)} = \Sym^{k} \CC^2$; the half-integer $j = k/2$ is known as the \emphindex{spin} of the irreducible representation of $\SL(2)$.
The same representation is obtained by restricting $V^2_{(k+1,1)}$, $V^2_{(k+2,2)}$, etc.

\bigskip

We finally consider the general linear group $\GL(\calH)$ and the unitary group $\U(\calH)$ of an arbitrary finite-dimensional Hilbert space $\calH$.
By choosing an orthonormal basis, we may always identify $\calH$ with $\CC^d$, where $d$ is the dimension of $\calH$.
Then $\GL(\calH)$ is identified with $\GL(d)$, $\U(\calH)$ is identified with $\U(d)$, and the above theory is applicable.%
\nomenclature[RGL(H)]{$\GL(\calH)$, $\mathfrak{gl}(\calH)$}{general linear group of Hilbert space $\calH$, and its Lie algebra}%
\nomenclature[RSL(H)]{$\SL(\calH)$, $\mathfrak{sl}(\calH)$}{special linear group of Hilbert space $\calH$, and its Lie algebra}%
\nomenclature[RU(H)]{$\U(\calH)$, $\mathfrak u(\calH)$}{unitary group of Hilbert space $\calH$, and its Lie algebra}%
\nomenclature[RSU(H)]{$\SU(\calH)$, $\mathfrak su(\calH)$}{special unitary group of Hilbert space $\calH$, and its Lie algebra}

\section{Geometric Invariant Theory}
\label{sec:onebody/git}

In this section, we introduce some \emphindex{geometric invariant theory}, which is a powerful mathematical framework for studying the one-body quantum marginal problem and its variants. We refer to \cite{Kirwan84, MumfordFogartyKirwan94, Brion10, Woodward10, VergneBerline11, GeorgoulasRobbinSalamon13} for further material.

We start with a basic observation that motivates the general approach:
Consider the representation $\Pi$ of the group $G = \SL(\calH_1) \times \dots \times \SL(\calH_n)$ on the Hilbert space $\calH = \bigotimes_{k=1}^n \calH_k$ by tensor products.
The Lie algebra $\mathfrak g = \mathfrak{sl}(\calH_1) \oplus \dots \oplus \mathfrak{sl}(\calH_n)$ acts by $\pi(X_1, \dots, X_n) = \sum_{k=1}^n \Id^{\otimes k-1} \otimes X_k \otimes \Id^{\otimes n-k}$.
Suppose that $P$ is a non-constant $G$-invariant homogeneous polynomial on $\calH$ such that $P(\ket\psi) \neq 0$ for some vector $\ket\psi \in \calH$.
Let $\ket{\psi'} \in \overline{\Pi(G) \ket\psi}$ be a vector of minimal length in the closure of the $G$-orbit through $\ket\psi$ (cf.\ \autoref{fig:onebody/kempfness}).
This vector is non-zero, since $P(\ket{\psi'}) = P(\ket\psi) \neq 0$ by $G$-invariance, while $P(0) = 0$ by homogeneity.
Since $\ket{\psi'}$ is a vector of minimal length, its norm squared does not change in first order if we move infinitesimally into an arbitrary tangent direction of its orbit.
The same is of course true for the unit vector $\ket{\psi''} := \ket{\psi'} / \norm{\psi'}$.
But all tangent vectors are generated by the action of the Lie algebra $\mathfrak g$.
That is, for all tuples of traceless Hermitian operators $X = (X_1, \dots, X_n) \in i \mathfrak k \subseteq \mathfrak g$ and using \eqref{eq:onebody/rdm} we find that
\begin{equation}
\label{eq:onebody/hands-on kempf-ness}
  0 = \partial_{t=0} \norm{\exp(\pi(X) t) \ket{\psi''}}^2 = 2 \braket{\psi'' | \pi(X) | \psi''} = \sum_{k=1}^n \tr \rho''_k X_k,
\end{equation}
where $\rho'' = \proj{\psi''}$.
In other words, the traceless part of each one-body reduced density matrix $\rho''_k$ vanishes.
We conclude that each one-body reduced density matrix $\rho''_k$ is proportional to the identity matrix.
In the language of quantum information theory, the quantum state $\rho''$ is \emphindex{locally maximally mixed}.
This way of reasoning establishes a first link between the existence of invariants and of pure states with prescribed marginals.
In the following we will see that the above argument can be generalized to arbitrary one-body marginals and turned into an equivalence that completely characterizes the one-body quantum marginal problem and its variants.
We follow along the lines of the exposition in \cite{VergneBerline11} and take some ideas from \cite{NessMumford84, Brion87}.

\begin{figure}
  \begin{center}
    \includegraphics[height=3cm]{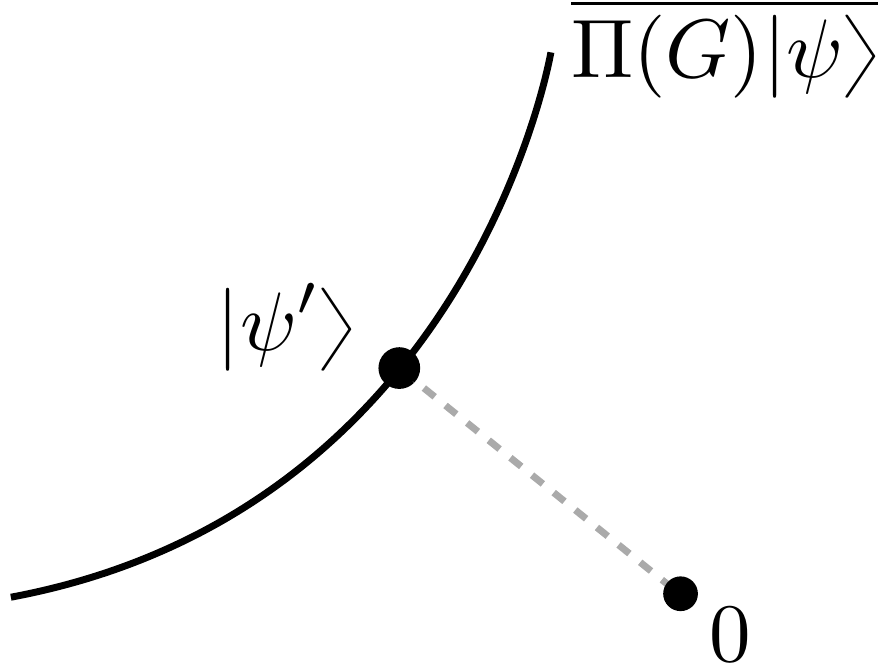}
  \end{center}
  \caption[Illustration of \autoref{lem:onebody/kempfness}]{In the closure of an orbit $\Pi(G) \ket\psi$ for the complexified group $G$, any non-zero vector of minimal norm $\ket{\psi'}$ gives rise to a pure state with maximally mixed marginals. \autoref{lem:onebody/kempfness} shows that a converse is also true.}
  \label{fig:onebody/kempfness}
\end{figure}

\subsection*{Projective Space}

Mathematically, the set of pure states on a Hilbert space $\calH$,
\[\PP(\calH) = \{ \rho = \proj\psi : \braket{\psi | \psi} = 1 \},\]
is known as a \emphindex{complex projective space}.%
\nomenclature[TX]{$\PP(\calH)$}{complex projective space of pure states on $\calH$}
It is a smooth submanifold of the real vector space of Hermitian operators on $\calH$.
The unitary group $\U(\calH)$ acts transitively by conjugation, $U \proj\psi U^\dagger$, so that the tangent space at a point $\rho = \proj\psi$ is spanned by the tangent vectors $X_\rho = [X, \rho]$ for all $X \in \mathfrak u(\calH)$.%
\nomenclature[TX_rho]{$X_\rho$}{tangent vector at $\rho$ generated by infinitesimal action of $X$}
Since the $X$ are anti-Hermitian, it is easy to verify that the tangent space can be equivalently written as%
\nomenclature[TT_rhoP(H)]{$T_\rho \PP(\calH)$}{tangent space of $\PP(\calH)$ at $\rho$}%
\nomenclature[TV]{$V, W, \dots$}{tangent vectors}
\begin{equation}
\label{eq:onebody/tangent space}
  T_\rho \PP(\calH) = \{ V = \ketbra\phi\psi + \ketbra\psi\phi : \ket\phi \in \calH, \braket{\phi | \psi} = 0 \} \cong \ket\psi^{\perp} \subseteq \calH.
\end{equation}
In this way, the tangent space acquires a complex structure, which can be written as\nomenclature[TJ]{$J[V]$}{complex structure of projective space}
\[J[V] = i [V, \rho] = i (\ketbra\phi\psi - \ketbra\psi\phi),\]
as well as a Hermitian inner product.
The real part of the inner product is a Riemannian metric, $g(V, W) = \tr V W$,%
\nomenclature[Tg]{$g(V, W)$}{Riemannian metric of projective space}
and its imaginary part is the \emphindex{Fubini--Study symplectic form}\index{symplectic form!Fubini--Study}%
\nomenclature[Tomega]{$\omega(V, W)$}{Fubini--Study symplectic form of projective space}
\begin{equation}
\label{eq:onebody/fubini-study}
  \omega(V, W) = g(J[V], W) = \tr J[V] W = -i \tr \rho [V, W].
\end{equation}
For tangent vectors generated by elements of the Lie algebra $\mathfrak u(\calH)$, this becomes
\begin{equation}
\label{eq:onebody/symplectic form}
  \omega(X_\rho, Y_\rho) =
  -i \tr \rho [X_\rho, Y_\rho] =
  i \tr \rho [X, Y]
  \quad (\forall X, Y \in \mathfrak u(\calH)).
\end{equation}
The action of $\U(\calH)$ can be extended to its complexification, the general linear group $\GL(\cal H)$ by the formula%
\nomenclature[T\cdot]{$g \cdot \rho$}{action of $\GL(\calH)$ on projective space $\PP(\calH)$}
\begin{equation}
\label{eq:onebody/complexified action}
  g \cdot \proj\psi := g \proj\psi g^\dagger / \braket{\psi | g^\dagger g | \psi}.
\end{equation}
The tangent vector generated by a Hermitian matrix $iX \in i \mathfrak u(\calH)$ is then given by
$(iX)_\rho := \{i X, \rho\} - 2 (\tr iX \rho) \rho$, with $\{A,B\} := AB+BA$ the \emphindex{anti-commutator}.
It is easily verified that $(iX)_\rho = J[X_\rho]$.
Thus the complex structures of projective space and of the Lie group are compatible with each other.

\subsection*{The Moment Map}

Let $K$ be a compact connected Lie group, $G$ its complexification, and $\Pi \colon G \rightarrow \GL(\calH)$ a representation on a Hilbert space $\calH$ with a $K$-invariant inner product; we denote the infinitesimal representation by $\pi \colon \mathfrak g \rightarrow \mathfrak{gl}(\calH)$.
Then $K$ and $G$ also act on the projective space $\PP(\calH)$, and as in \eqref{eq:onebody/complexified action} we will denote this action by $g \cdot \rho$.\nomenclature[T\cdot]{$g \cdot \rho$}{action of $G$ on projective space $\PP(\calH)$ of $G$-representation $\calH$}
The following is our basic object of interest:

\begin{dfn}
  The \emphindex{moment map} for the action of $K$ on $\PP(\calH)$ is defined by%
  \nomenclature[Tmu_K]{$\mu_K$}{moment map for the $K$-action}
  \begin{equation}
    \label{eq:onebody/moment map}
    \mu_K \colon \PP(\calH) \rightarrow i \mathfrak k^*,
    \quad
    (\mu_K(\rho), X) := \tr \rho \, \pi(X)
  \end{equation}
  for all $\rho \in \PP(\calH)$ and $X \in \mathfrak g$.
  Here and in the following, we write $(\varphi, X) = \varphi(X)$ for the pairing between $\mathfrak g^*$ and $\mathfrak g$.%
  \nomenclature[R(varphi,X)]{$(\varphi,X)$}{dual pairing of $\mathfrak g^*$ and $\mathfrak g$}
\end{dfn}

Unfortunately, there are as many conventions for the moment map as there are textbooks on the subject.
For the representation $\calH$ that we considered at the beginning of this section, the moment map maps pure states onto the functionals evaluating (traceless) local observables; cf.\ \eqref{eq:onebody/hands-on kempf-ness}.
Thus the one-body quantum marginal problem is equivalent to characterizing the image of a moment map.
In \autoref{sec:onebody/consequences} we will explain this connection in more detail.

A crucial property of the moment map is the following relation between the differential of its components and the tangent vector $X_\rho = [\pi(X), \rho]$ generated by the infinitesimal action of the Lie algebra of the compact group:%
\nomenclature[Tdf]{$df\big|_x$}{differential of smooth map $f$ at point $x$\nomnorefpage}
\begin{equation}
\label{eq:onebody/moment map property}
    d(\mu_K, iX)\big|_\rho %
  = \omega(X_\rho, -)
  \quad (\forall X \in \mathfrak k)
\end{equation}
This follows readily from \eqref{eq:onebody/symplectic form}.
Since the Fubini--Study form is non-degenerate, an immediate consequence is that the component \eqref{eq:onebody/moment map property} of the differential vanishes if and only if $X_\rho = 0$. Dually, we find that the range of the differential of the moment map at any point $\rho$ is given by the annihilator of $\mathfrak k_\rho := \{ X \in \mathfrak k : X_\rho = [\pi(X), \rho] = 0 \}$%
\nomenclature[TK_rho]{$K_\rho$, $\mathfrak k_\rho$}{$K$-stabilizer of $\rho$, and its Lie algebra}%
, the Lie algebra of the $K$-stabilizer of $\rho$ \cite{GuilleminSternberg82}:
\begin{equation}
\label{eq:onebody/moment map range}
    d\mu_K\big(T_\rho \PP(\calH))
  = \{ \lambda \in i \mathfrak k^* : \lambda\big|_{i \mathfrak k_\rho} \equiv 0 \}
\end{equation}
This holds even if we restrict to the tangent space of the $G$-orbit through $\rho$, which is a complex submanifold and in particular symplectic (i.e., the Fubini--Study form $\omega$ remains non-degenerate).
A second important property is that the moment map is $K$-equivariant: Indeed, for all $\rho \in \PP(\calH)$, $g \in K$ and $X \in i\mathfrak k$ we have that
\begin{align*}
  &(\mu_K(g \cdot \rho), X)
  = \tr \pi(g) \rho \, \pi(g^{-1}) \pi(X)
  = \tr \rho \, \pi(g^{-1}) \pi(X) \pi(g) \\
  = &\tr \rho \, \pi(\Ad(g^{-1})X)
  = (\mu_K(\rho), \Ad(g^{-1})X)
  = (\Ad^*(g)\mu_K(\rho), X).
\end{align*}
Thus its image consists of a union of coadjoint orbits, and so is characterized by its intersection with the positive Weyl chamber $i \mathfrak t^*_+$.
In the context of the quantum marginal problem, this amounts to our previous observation that its solution depends only on the eigenvalues of the one-body reduced density matrices.

The notion of a moment map can be defined more generally in symplectic geometry, and originates in Hamiltonian mechanics. For example, the function sending a point in the classical phase space $\RR^{6n}$ to the total linear and angular moment\emph{um} is a moment map in this more general sense for the canonical action of the Euclidean group (see, e.g., \cite{GuilleminSternberg84, Cannas08}).

\bigskip

Our basic argument above showed that any non-zero vector of minimal length in an orbit closure has expectation value zero with respect to all local traceless observables (i.e., the image under the moment map is zero).
The following result by Kempf and Ness shows a converse \cite{KempfNess79} (cf.\ \cite{Kempf78, NessMumford84} and \autoref{fig:onebody/kempfness} for an illustration).

\begin{lem}[Kempf--Ness]
\label{lem:onebody/kempfness}
  Let $\rho = \proj\psi \in \PP(\calH)$ with $\mu_K(\rho) = 0$.
  Then the $G$-orbit through $\ket\psi$ is closed (so in particular contains a non-zero vector of minimal length).\footnote{In fact, the $K$-orbit through $\ket\psi$ consists of the vectors of minimal length \cite{KempfNess79} (cf.\ the discussion in \autoref{sec:slocc/distillation}).}
\end{lem}
\begin{proof}
  Suppose for the sake of finding a contradiction that the $G$-orbit through $\ket\psi$ is not closed.
  Then the \emphindex{Hilbert--Mumford criterion} asserts that we can reach a point in the ``boundary'' $\overline{\Pi(G) \ket\psi} \setminus \Pi(G) \ket\psi$ by using a single one-parameter subgroup (e.g., \cite[p.\ 171]{Kraft85}).
  More formally, it states that there exists $X \in i\mathfrak k$ such that
  \begin{equation}
  \label{eq:onebody/kempf-ness limit}
    \lim_{t \rightarrow \infty} \exp(\pi(X) t) \ket\psi = \ket{\psi'} \not\in \Pi(G) \ket\psi.
  \end{equation}
  Let $\ket\psi = \sum_k \ket{\psi_k}$ be the decomposition of $\ket\psi$ into eigenvectors of the Hermitian operator $\pi(X)$, with $\ket{\psi_k}$ an eigenvector with eigenvalue $k$.
  Clearly, $\ket{\psi_k} = 0$ for all $k > 0$, since otherwise the limit \eqref{eq:onebody/kempf-ness limit} cannot exist.
  But then
  \begin{equation*}
    \sum_{k \leq 0} k \braket{\psi_k | \psi_k} = \braket{\psi | \pi(X) | \psi} = (\mu_K(\rho), X) = 0,
  \end{equation*}
  so that $\ket{\psi_k} = 0$ also for all $k < 0$.
  It follows that $\ket{\psi} = \ket{\psi_0}$, i.e.\ $\pi(X) \ket{\psi} = 0$.
  Thus the one-parameter subgroup in fact leaves the vector invariant,
  \begin{equation*}
    \exp(\pi(X) t) \ket\psi \equiv \ket\psi \quad (\forall t).
  \end{equation*}
  This is the desired contradiction to \eqref{eq:onebody/kempf-ness limit}.
\end{proof}

In the following we need to study the image of the moment map not only for the set of all pure states but also for certain subsets of projective space.
In the context of algebraic geometry, it is natural to consider $G$-invariant projective subvarieties, which we define in the following way (e.g., \cite{Hartshorne77}):

\begin{dfn}
\label{dfn:onebody/projective subvariety}
  An \emphindex{affine cone} $\calC \subseteq \calH$ is the common zero set of a family of homogeneous polynomials.\footnote{The terminology is standard and motivated by $\calC$ being closed under multiplication by $\CC$.}%
\nomenclature[TC]{$\calC$}{affine cone}
  The ring of \emphindex{regular functions} $R(\calC) := \bigoplus_{k=0}^\infty R_k(\calC)$ is defined as the ring of polynomials on $\calH$, graded by the degree $k$, where we identify any two polynomials if their difference vanishes on $\calC$.%
\nomenclature[TRC,RkC]{$R(\calC)$, $R_k(\calC)$}{regular functions (of degree $k$) on affine cone $\calC$}
  An affine cone is called \emph{irreducible}\index{affine cone!irreducible} if the ring of regular functions does not have any zero divisors (i.e.\ if for any two $P, Q \in R(\calC)$, $PQ = 0$ implies that $P = 0$ or $Q = 0$).
  Finally, a \emphindex{projective subvariety} of $\PP(\calH)$ is a set of the form
  \begin{equation*}
    \XX = \PP(\calC) := \{ \proj\psi : \ket\psi \in \calC, \braket{\psi | \psi} = 1 \} \subseteq \PP(\calH),
  \end{equation*}
  where $\calC \subseteq \calH$ is an irreducible affine cone.%
\nomenclature[TX,Y]{$\XX, \YY, \dots$}{projective subvarieties}%
\nomenclature[TX,Y]{$\PP(\calC)$}{projective subvariety corresponding to an affine cone $\calC$}
  In this case, we also write $R(\XX) := R(\calC)$.%
\nomenclature[TRX,RkX]{$R(\XX)$, $R_k(\XX)$}{regular functions (of degree $k$) on affine cone of projective subvariety $\XX$}
  We say that $\XX$ is \emph{$G$-invariant}\index{projective subvariety!$G$-invariant} if $\Pi(G) \calC \subseteq \calC$ or, equivalently, if $G \cdot \XX \subseteq \XX$.
\end{dfn}

By definition, projective subvarieties are always closed in the \emphindex{Zariski topology} of projective space, whose closed sets are common zero sets of families of homogeneous polynomials.
Zariski-closed sets are also closed in the usual topology of projective space as a manifold, but the converse is not necessarily true.
However, orbits $G \cdot \rho = \{ g \cdot \rho : g \in G \}$ of algebraic group actions are constructible and hence the usual closure and the Zariski closure coincide (e.g., \cite[Proposition 8.3]{Humphreys98} and \cite[Lemma 12.5.3]{SommeseWampler05}).
In fact, any \emphindex{orbit closure} $\overline{G \cdot \rho}$ is a $G$-invariant projective subvariety of $\PP(\calH)$ in the sense just defined (since our groups $G$ are connected).%
\nomenclature[TG \cdot rho1]{$G \cdot \rho$}{$G$-orbit through $\rho$}%
\nomenclature[TG \cdot rho2]{$\overline{G \cdot \rho}$}{closure of $G$-orbit through $\rho$}

\bigskip

For any $G$-invariant projective subvariety $\XX$, the ring of regular functions $R(\XX)$ becomes a $G$-representation in a natural way: For each $g \in G$ and $P \in R(\XX)$, we may define $\Pi(g)P$ by $(\Pi(g) P)(\ket\psi) := P(\Pi(g^{-1}) \ket\psi)$.
In the language that we have just introduced, the basic argument given at the beginning of the section can be succinctly summarized as follows (together with its converse):

\begin{lem}
\label{lem:onebody/mumford origin}
  For any $G$-invariant projective subvariety $\XX \subset \PP(\calH)$, we have
  \begin{equation*}
    R(\XX)^G \neq \CC \quad\Leftrightarrow\quad 0 \in \mu_K(\XX),
  \end{equation*}
  where $\CC$ stands for the constant functions.
\end{lem}
\begin{proof}
  $(\Rightarrow)$ is proved just as at the beginning of this section:
  Let $\XX = \PP(\calC)$ and $P \in R_k(\XX)^G$ a non-constant $G$-invariant homogeneous polynomial with $P(\ket\psi) \neq 0$ for some $\ket\psi \in \calC$.
  Since $P(0) = 0$, $\overline{\Pi(G) \ket\psi} \not\ni 0$, hence we can find a non-zero vector $\ket{\psi'} \in \overline{\Pi(G) \ket\psi} \subseteq \calC$ of minimal length.
  In particular, $\ket{\psi'}$ is of minimal length in its $G$-orbit, so that for all $X \in i \mathfrak k$ and using that $\pi(X) = \pi(X)^\dagger$ we find that
  \begin{equation*}
    0 = \partial_{t=0} \norm{\exp(\pi(X) t) \psi'}^2 = 2 \braket{\psi' | \pi(X) | \psi'}.
  \end{equation*}
  Thus $\rho := \proj{\psi'} / \braket{\psi' | \psi'} \in \XX$ is a pure state with $\mu_K(\rho) = 0$.

  $(\Leftarrow)$
  For any $\rho = \proj\psi$ with $\mu_K(\rho) = 0$, \autoref{lem:onebody/kempfness} asserts that the $G$-orbit through $\ket\psi$ is closed, so in particular it is disjoint from $\{0\}$.
  Since orbits are constructible in the sense of algebraic topology, their closure in the Hilbert space topology coincides with their closure in the Zariski topology, which is the topology whose closed sets are common zero sets of families of polynomials.
  But any two disjoint Zariski-closed sets can be separated by a polynomial:
  Hilbert's Nullstellensatz implies that we may find a polynomial $P$ that separates $\Pi(G) \ket\psi$ and $\{0\}$ such that $P(\Pi(G) \ket\psi) \equiv 1$ while $P(0) = 0$ \cite{Hartshorne77}.
  Without loss of generality, $P$ is homogeneous, and we may also average over $K$ to make it $G$-invariant.
  Then $P$ is a non-constant regular function in $R(\XX)^G$.
\end{proof}

To generalize \autoref{lem:onebody/mumford origin} to arbitrary points in the image of the moment map, we need as the last ingredient the \emphindex{Borel--Weil theorem} (see, e.g., \cite[Lemma 94]{VergneBerline11}).

\begin{lem}[Borel--Weil]
  \label{lem:onebody/borel weil}
  Let $V_{G,\lambda}$ be the irreducible $G$-representation with highest weight $\lambda \in \Lambda^*_{G,+}$ and highest weight vector $\ket{\lambda}$.
  Let $\XX_{G,\lambda} := K \cdot \proj\lambda$.\nomenclature[TX_G,lambda]{$\XX_{G,\lambda}$}{$G$-orbit through projector onto highest weight vector (projective subvariety isomorphic to the coadjoint orbit $\calO_{K,\lambda}$)}
  Then $\XX_{G,\lambda} = G \cdot \proj\lambda = \overline{G \cdot \proj\lambda}$ is a $G$-invariant projective subvariety of $\PP(V_{G,\lambda})$ for which
  \begin{equation}
  \label{eq:onebody/borel weil}
    R_k(\XX_{G,\lambda}) \cong V_{G,k\lambda}^*
    \quad\text{and}\quad
    \mu_{K,\lambda}(\proj\lambda) = \lambda,
  \end{equation}
  where $\mu_{K,\lambda}$ denotes the moment map for $K$-action on $\PP(V_{G,\lambda})$.
  In fact, the moment map is a bijection between the projective subvariety $\XX_{G,\lambda}$ and its image, which is the coadjoint orbit $\calO_{K,\lambda}$.
\end{lem}

\subsection*{The Moment Polytope}

The following proposition formalizes the fundamental link between the image of the moment map and the decomposition of the ring of regular functions into irreducible representations \cite{GuilleminSternberg82a, NessMumford84, Brion87}.

\begin{prp}
\label{prp:onebody/mumford}
  Let $\XX$ be a $G$-invariant projective subvariety of $\PP(\calH)$. Then:
  \begin{equation*}
    \mu_K(\XX) \cap \QQ \Lambda^*_{G,+} = \left\{ \frac \lambda k : V^*_{G,\lambda} \subseteq R_k(\XX) \right\}
  \end{equation*}
\end{prp}
\begin{proof}
  Fix $\lambda \in \Lambda^*_{G,+}$ and $k > 0$.
  Let $\widetilde\calH := \Sym^k(\calH) \otimes V_{G,\lambda^*}$. Then
  \begin{equation*}
    \widetilde\XX := \{ \proj\psi^{\otimes k} \otimes \proj\phi : \proj\psi \in \XX, \proj\phi \in \XX_{G,\lambda^*} \}
  \end{equation*}
  is a projective subvariety of $\PP(\widetilde\calH)$ -- known as the image of the product of the $k$-th \emphindex{Veronese embedding} of $\XX$ and $\XX_{G,\lambda^*}$ under the \emphindex{Segr\'{e} embedding} -- and it is not hard to see that its ring of regular functions is given by
  \begin{equation*}
    R_l(\widetilde\XX)
    \cong R_{kl}(\XX) \otimes R_l(\XX_{G,\lambda^*})
    \cong R_{kl}(\XX) \otimes V_{G,\lambda l},
  \end{equation*}
  where we have used the first assertion in \eqref{eq:onebody/borel weil}.
  Thus, $R_l(\widetilde\XX)^G \neq 0$ if and only if $V_{G,\lambda l}^* \subseteq R_{kl}(\XX)$.

  On the other hand, the infinitesimal action of $\mathfrak g$ on $\widetilde\calH$ is $\widetilde\pi(X) = k \, \pi(X) \otimes \Id_{V_{G,\lambda}} + \Id_{\calH} \otimes \pi_{\lambda^*}(X)$, with $\pi_{\lambda^*}$ the irreducible representation of $\mathfrak g$ on $V_{G,\lambda}$.
  Hence the moment map on $\widetilde\XX$ is given by the formula
  \begin{equation*}
    \widetilde\mu_K(\proj\psi^{\otimes k} \otimes \proj\phi) = k \, \mu_K(\proj\psi) + \mu_{K,\lambda^*}(\proj\phi).
  \end{equation*}
  It follows by using the second assertion in \eqref{eq:onebody/borel weil} and $\lambda^* = -w_0 \lambda$ that
  \begin{equation}
  \label{eq:onebody/shifted moment map images}
    \widetilde\mu_K(\widetilde\XX)
    = k \, \mu_K(\XX) + \calO_{K,\lambda^*}
    = k \, \mu_K(\XX) - \calO_{K,\lambda},
  \end{equation}
  hence $0 \in \widetilde\mu_K(\widetilde\XX)$ if and only if $\lambda/k \in \mu_K(\XX))$.
  The theorem follows from these two observations and \autoref{lem:onebody/mumford origin}.
\end{proof}

It is instructive to apply \autoref{prp:onebody/mumford} to the situation of \autoref{lem:onebody/borel weil}.

\begin{prp}
\label{prp:onebody/mumford convex polytope}
  Let $\XX$ be a $G$-invariant projective subvariety of $\PP(\calH)$.
  Then $\{ (\lambda,k) : V^*_{G,\lambda} \subseteq R_k(\XX) \}$ is a finitely generated semigroup.
  It follows that $\{ \lambda / k : V^*_{G,\lambda} \subseteq R_k(\XX) \}$ is a convex polytope over $\QQ$.
\end{prp}
\begin{proof}
  Set $\calS := \{ (\lambda,k) : V^*_{G,\lambda} \subseteq R_k(\XX) \}$.
  We first show that $\calS$ is a semigroup, i.e., closed under addition.
  Recall from \autoref{sec:onebody/lie} that the highest weight vectors are in one-to-one correspondence with the irreducible representations that occur in a given representation.
  Thus let $P \in R_k(\XX)$ and $Q \in R_l(\XX)$ be highest weight vectors of weight $\lambda$ and $\mu \in \Lambda^*_{G,+}$ corresponding to two points $(\lambda, k)$ and $(\mu, l) \in \calS$.
  We claim that their product $PQ$ is a highest weight vector of weight $\lambda + \mu$ in $R_{k+l}(\XX)$.
  Indeed, $PQ \neq 0$ since $\XX$ is irreducible by assumption; it has weight $\lambda + \mu$ since
  \begin{equation*}
    \pi(H) (PQ) = (\pi(H)P) Q + P (\pi(H)Q) = (\lambda(H)+\mu(H))PQ \quad (\forall H \in \mathfrak h),
  \end{equation*}
  and an analogous calculation shows that it is annihilated by $\mathfrak n_+$ and hence a highest weight vector.
  Thus we obtain the point $(\lambda+\mu,k+l) \in \calS$.
  To see that $\calS$ is finitely generated, we use the (highly non-trivial) fact that the algebra of $N_+$-invariants $R(\XX)^{N_+}$, whose elements are linear combinations of highest weight vectors, is finitely generated \cite{Grosshans73}.
  Let $P^{(1)}, \dots, P^{(m)}$ be a finite set of generators; we may assume that each generator is homogeneous, say $P^{(j)} \in R_{k^{(j)}}(\XX)^{N_+}$, and a weight vector, say of weight $\lambda^{(j)}$.
  Then $(\lambda^{(j)}, k^{(j)}) \in \calS$ for all $j=1, \dots, m$.
  On the other hand, we find just as above that any monomial $P^{(j_1)} \dotsm P^{(j_p)}$ has degree $\sum_i k^{(j_i)}$ and weight $\sum_i \lambda^{(j_i)}$, and therefore determines the point $\sum_i (\lambda^{(j_i)},k^{(j_i)}) \in \calS$.
  Since any highest weight vector can then be written as a sum of monomials in the $P^{(j)}$, we conclude that $\calS$ is a semigroup that is generated by the finitely many points $(\lambda^{(j)},k^{(j)})$ for $j=1,\dots,m$.

  It follows as an immediate consequence that
  \[\calP := \{ \lambda / k : V^*_{G,\lambda} \subseteq R_k(\XX) \} = \{ \lambda / k : (\lambda, k) \in \calS \}\]
  is a convex polytope over $\QQ$.
  Indeed, let $(\lambda,k)$ and $(\mu,l) \in \calS$ corresponding to two points $\lambda/k$ and $\mu/l \in \calP$.
  Then $pl (\lambda,k) + qk (\mu,l) \in \calS$ by the semigroup property, so that we obtain the point
  \begin{equation*}
      \frac {pl \lambda + qk \mu} {plk + qkl}
    = \frac p {p+q} \frac \lambda k + \frac q {p+q} \frac \mu l \in \calP.
  \end{equation*}
  Thus $\calP$ is closed under convex combinations with rational coefficients.
  Likewise, the fact that $\calS$ is generated by the finitely many generators $(\lambda^{(j)}, k^{(j)})$ implies that $\calP$ is the convex hull over $\QQ$ of the finitely many points $\lambda^{(j)} / k^{(j)}$ ($j=1,\dots,m$).
\end{proof}

For any $(\lambda,k)$ in the semigroup, its integral multiples $\ZZ_{>0} (\lambda,k)$ are also contained in the semigroup, and they determine the same point $\lambda/k$ in the polytope.
Therefore, the polytope only depends on the representation theory of $R_k(\XX)$ in the ``semiclassical limit''\index{semiclassical limit} of large $k$:
\begin{equation}
\label{eq:onebody/semiclassical limit}
  \left\{ \frac \lambda k : V^*_{G,\lambda} \subseteq R_k(\XX) \right\}
  =
  \left\{ \frac \lambda k : V^*_{G,\lambda} \subseteq R_k(\XX), k \geq k_0 \right\}
  \quad
  (\forall k_0 > 0)
\end{equation}

It is well-known that $\mu_K(\XX) \cap \QQ\Lambda^*_{G,+}$ is a dense subset of $\mu_K(\XX) \cap i\mathfrak t^*_+$ (see, e.g., \cite{GuilleminSternberg82, NessMumford84}; this also follows from the local model for symplectic group actions, cf.\ the discussion below \autoref{lem:onebody/curve lemma}). %
By combining \autoref{prp:onebody/mumford} and \autoref{prp:onebody/mumford convex polytope}, we thus arrive at the following fundamental theorem:

\begin{thm}[Mumford]
\label{thm:onebody/kirwan}
  Let $\XX$ be a $G$-invariant projective subvariety of $\PP(\calH)$.
  Then
  \begin{equation*}
    \Delta_K(\XX) = \mu_K(\XX) \cap i\mathfrak t^*_+ = \{ \lambda \in i \mathfrak t^*_+ : \calO_{K,\lambda} \subseteq \mu_K(\XX) \}
  \end{equation*}
  is a convex polytope with rational vertices, whose rational points are given by
  \begin{equation*}
    \mu_K(\XX) \cap \QQ \Lambda^*_{G,+}
    = \left\{ \frac \lambda k : V^*_{G,\lambda} \subseteq R_k(\XX) \right\}.
  \end{equation*}
  It is called the \emphindex{moment polytope} for the $K$-action on $\XX$.%
  \nomenclature[TDelta_K(X)]{$\Delta_K(\XX)$}{moment polytope for the $K$-action on $\XX$}
\end{thm}

In the proof of \autoref{prp:onebody/mumford convex polytope} we had identified the irreducible representations that occur in $R(\XX)$ with their highest weight vectors, which are precisely the weight vectors in $R(\XX)^{N_+}$.
From a conceptual point of view, the passage from highest weight vectors in $R(\XX)$ to weight vectors in $R(\XX)^{N_+}$ is a first instance of the reduction of a non-Abelian problem to an Abelian problem as alluded to in the abstract of this thesis.
It also leads to a basic way of computing the moment polytope $\Delta_K(\XX)$ (cf.\ the discussion at the end of \autoref{sec:slocc/entanglement polytopes}).
In the next chapters we will meet more refined variants of this reduction.

\medskip

The study of moment maps and their convexity properties has a long history in mathematics.
Among the well-known special cases are:
The Schur--Horn theorem concerning the diagonal entries of Hermitian matrices with fixed spectrum \cite{Schur23, Horn54}; Kostant's convexity theorem, which is the generalization to general coadjoint orbits \cite{Kostant73}; the Atiyah--Guillemin--Sternberg convexity theorem for torus actions \cite{Atiyah82, GuilleminSternberg82}; Heckman's convexity theorem, which considers projections of coadjoint orbits \cite{Heckman82}; and Kirwan's convexity theorem, which is the symplectic analogue of \autoref{thm:onebody/kirwan} \cite{Kirwan84a} (cf.\ \cite{Sjamaar98, Brion99} and the recent monograph \cite{GuilleminSjamaar05}).

\section{Consequences for the Quantum Marginal Problem}
\label{sec:onebody/consequences}

We now describe the precise connection between the geometry of the moment map and the one-body quantum marginal problem and draw some general consequences.

For distinguishable particles, consider the Hilbert space $\calH = \bigotimes_{k=1}^n \calH_k$ equipped with the representation of $G = \GL(\calH_1) \times \dots \times \GL(\calH_n)$ by tensor products, $\Pi(g_1, \dots, g_n) = g_1 \otimes \dots \otimes g_n$.
Its Lie algebra $\mathfrak g = \mathfrak{gl}(\calH_1) \oplus \dots \oplus \mathfrak{gl}(\calH_n)$ acts by $\pi(X_1, \dots, X_n) = \sum_{k=1}^n \Id^{\otimes k-1} \otimes X_k \otimes \Id^{\otimes n-k}$.
The group $K = \U(\calH_1) \times \dots \times \U(\calH_n)$ is a maximal compact subgroup of $G$ and it acts unitarily on $\calH$.
For all tuples of Hermitian matrices $X = (X_1, \dots, X_n) \in i \mathfrak k$, the moment map \eqref{eq:onebody/moment map} is given by
\begin{equation*}
  (\mu_K(\rho), X) %
  = \sum_{k=1}^n \tr \rho (\Id^{\otimes k-1} \otimes X_k \otimes \Id^{\otimes n-k})
  = \sum_{k=1}^n \tr \rho_k X_k,
\end{equation*}
where we have used the definition \eqref{eq:onebody/rdm} of the one-body reduced density matrices $\rho_k$.
Thus if we identify each $\rho_k$ with its dual, $\tr \rho_k (-)$, then the one-body quantum marginal problem, \autoref{pro:onebody/qmp}, is precisely equivalent to characterizing the image of the moment map.
Furthermore, the moment polytope $\Delta_K(\PP(\calH)) = \mu_K(\PP(\calH)) \cap i \mathfrak t^*_+$ can be identified with the collection of eigenvalues $(\vec\lambda_1, \dots, \vec\lambda_n)$ of the one-body reduced density matrices $\rho_1, \dots, \rho_n$ that are compatible with a global pure state $\rho \in \PP(\calH)$:
\begin{equation*}
  \Delta_K(\PP(\calH)) = \left\{ (\vec\lambda_1, \dots, \vec\lambda_n) : \vec\lambda_k = \spec \rho_k, \rho \in \PP(\calH) \right\},
\end{equation*}
where we identify diagonal matrices with non-increasing entries with their spectrum (as in \autoref{sec:onebody/lie} and \autoref{tab:onebody/unitary group}).

In practice, the above modeling of the quantum marginal problem has the disadvantage that the moment polytope is always of positive codimension: since $\tr \rho_k \equiv \tr \rho = 1$, $\Delta_K$ is contained in the affine subspace $\sum_j \lambda_{k,j} = 1$ for all $k=1,\dots,n$.
It will usually be more convenient to instead use the special linear and unitary groups, $G = \SL(\calH_1) \times \dots \times \SL(\calH_n)$ and $K = \SU(\calH_1) \times \dots \times \SU(\calH_n)$, so that $i \mathfrak k$ consists of tuples of \emph{traceless} Hermitian matrices.
Then the moment map preserves only the traceless part of the one-body reduced density matrices, which avoids the above degeneracy.

For fermions, we similarly choose $\calH = \Alt^n \calH_1$ and $G = \SL(\calH_1)$ acting by $\Pi(g_1) = g_1^{\otimes n}$.
With $K = \SU(\calH_1)$ and using \eqref{eq:onebody/rdm fermions} we obtain that
\begin{equation*}
  (\mu_K(\rho), X_1) = \sum_{k=1}^n \tr \rho (\Id^{\otimes k-1} \otimes X_k \otimes \Id^{\otimes n-k}) = n \tr \rho_1 X_1 = \tr \gamma_1 X_1,
\end{equation*}
where $\gamma_1 := n \rho_1$ is the first-order density matrix from quantum chemistry with trace $n$ that we had defined below \eqref{eq:onebody/rdm fermions}.
Again we find that the one-body $n$-representability problem, \autoref{pro:onebody/nrep}, is precisely equivalent to determining the moment polytope.

We can similarly model the other variants of the one-body quantum marginal problem alluded to at the end of the introduction of this chapter by considering different representations of unitary groups or their composition.
For example, if $\calH_0$ is a $K_0$-representation describing a pure-state problem then the corresponding mixed-state problem can be studied by taking $\calH = \calH_0 \otimes \calH_0$ and $K = K_0 \times \SU(\calH_0)$; e.g., the mixed-state problem for fermions amounts to the moment polytope for the $\SU(\calH_1) \otimes \SU(\Alt^n \calH_1)$-representation $\Alt^n \calH_1 \otimes \Alt^n \calH_1$.
In \autoref{sec:kirwan/examples} we discuss another example that involves the marginal problem for the spin and orbital degrees of freedom of a fermionic system.
In \autoref{tab:onebody/summary} we summarize the mathematical modeling of the scenarios of main physical interest.

\begin{table}
  \begin{center}
  \begin{tabular}{lcc}
    \toprule
    Setting  & Group $K$ & Representation $\calH$ \\
    \midrule
    Distinguishable particles & $\bigtimes_{k=1}^n \SU(\calH_k)$ & $\bigotimes_{k=1}^n \calH_k$ \\
    Fermions & $\SU(\calH_1)$ & $\Alt^n \calH_1$\\
    Bosons & $\SU(\calH_1)$ & $\Sym^n \calH_1$ \\
    \midrule
    Mixed-state version of $(\calH_0,K_0)$ & $K_0 \times \SU(\calH_0)$ & $\calH_0 \otimes \calH_0$ \\
    \bottomrule
  \end{tabular}
  \end{center}
  \caption[Mathematical modeling of the one-body quantum marginal problem]{The different variants of the one-body quantum marginal problem can all be modeled by studying the moment polytope $\Delta_K(\PP(\calH))$ for the action of a compact Lie group $K$ on the complex projective space of a $K$-representation $\calH$. The table lists some important scenarios.}
  \label{tab:onebody/summary}
\end{table}

\subsection*{Convexity}

For all these variants of the one-body quantum marginal problem, \autoref{thm:onebody/kirwan} immediately implies that the solution is given by a convex polytope, i.e., by linear inequalities on the eigenvalues of the one-body reduced density matrices.
For example, Pauli's original exclusion principle is one such inequality---but in general there are many further constraints.
As we will see in several concrete examples in \autoref{ch:kirwan}, there is a rich variety of subtle kinematic constraints on the one-body marginals of a multipartite quantum state.

To compute the actual linear inequalities for a given number of particles, statistics and local dimensions is in general a difficult problem that we will study in the next chapter. All known general solutions rely in one way or the other on the invariant-theoretic description of the moment polytope given by \autoref{thm:onebody/kirwan}, including the original solution by Klyachko \cite{Klyachko04} and the solution that we present in \autoref{ch:kirwan}.
In the remainder of this section we discuss the physical significance of the facets of the moment polytope, and we then describe more explicitly the representation-theoretic content of \autoref{thm:onebody/kirwan}.

\subsection*{Pinning}

An important consequence of the general theory is that the facets of the polytope have a rather particular structure. Before we show this, we record the following useful lemma for future reference.

\begin{lem}
\label{lem:onebody/curve lemma}
  Let $\rho \in \PP(\calH)$ such that $\mu_K(\rho) \in i \mathfrak t^*_{>0}$ and $V \in d\mu_K^{-1}\big|_\rho(i \mathfrak t^*)$. %
  Then there exists a smooth curve $\rho_t \in \PP(\calH)$ for $t \in (-\varepsilon,\varepsilon)$ such that
  \begin{equation*}
    \rho_0 = \rho,
    \quad
    \dot\rho_0 = V,
    \quad\text{and}\quad
    \mu_K(\rho_t) \in i\mathfrak t^*_{>0} \; (\forall t).
  \end{equation*}
\end{lem}
\begin{proof}
  Consider the \emph{symplectic cross section}\index{symplectic cross section|textbf} $Y := \mu_K^{-1}(i \mathfrak t^*_{>0})$.
  By $K$-equivariance of the moment map, $\mu_K$ meets $i \mathfrak t^*_{>0}$ transversally, so that $Y$ is a smooth manifold with tangent space $T_\rho Y = d\mu_K^{-1}\big|_\rho(i \mathfrak t^*)$ \cite[Theorem 26.7]{GuilleminSternberg84}.
  Thus we may choose any curve in $Y$ that starts with $\rho_0 = 0$ and $\dot\rho_0 = V$.
\end{proof}

In the context of the marginal problem, \autoref{lem:onebody/curve lemma} is in essence a reformulation of first-order perturbation theory for the one-body reduced density matrices.
We remark that its conclusions can be strengthened; it is in fact possible to walk ``finitesimally'' into any direction in $d\mu_K(T_\rho \PP(\calH)) \cap i \mathfrak t^*$, as can be seen by using the local model for symplectic group actions \cite{GuilleminSternberg82, GuilleminSternberg84a, Marle85}.

\begin{lem}[Selection Rule]
\label{lem:onebody/selection rule}
  Let $\rho = \proj\psi$ be a pure state such that $\mu_K(\rho) \in i \mathfrak t^*_{>0}$ is a point on a facet of the moment polytope corresponding to the inequality $(-,H) \geq c$.
  Then
  \begin{equation*}
    d(\mu_K, H)\big|_\rho = 0
  \end{equation*}
   and
  \begin{equation}
  \label{eq:onebody/selection rule}
    \pi(H) \ket\psi = c \ket\psi.
  \end{equation}
\end{lem}
\begin{proof}
  For any $\omega \in d\mu_K(T_\rho\PP(\calH)) \cap i \mathfrak t^*$, there exists a curve $\rho_t$ through $\mu_K(\rho)$ with $d\mu_K(\dot\rho_0) = \omega$ and $\mu_K(\rho_t) \in \Delta_K$ for all $t$ (\autoref{lem:onebody/curve lemma}).
  Therefore,
  \begin{equation*}
    (\mu_K(\rho_t), H) = c + t (\omega, H) + O(t^2).
  \end{equation*}
  Since $(\mu_K(\rho_t), H) \geq c$ is an inequality for the moment polytope, it follows that $(\omega, H) = 0$---for otherwise we could walk through the facet!
  On the other hand, \eqref{eq:onebody/moment map range} shows that
  \begin{equation*}
    d\mu_K\big(T_\rho \PP(\calH)) \cap i \mathfrak t^*
    = \{ \omega \in i \mathfrak t^* : \omega \big|_{i \mathfrak k_\rho} = 0 \}
    = \{ \omega \in i \mathfrak t^* : \omega \big|_{i \mathfrak t_\rho} = 0 \}.
  \end{equation*}
  Therefore, $H \in i \mathfrak t$ is necessarily an element of the Lie algebra of the $T$-stabilizer of $\rho$, i.e., $H_\rho = [\pi(H), \rho] = 0$.
  This implies that $d(\mu_K,H)\big|_\rho = 0$  by \eqref{eq:onebody/moment map property}, but also that $\ket\psi$ is an eigenvector of $\pi(H)$, with corresponding eigenvalue $\braket{\psi | \pi(H) | \psi} = (\mu_K(\rho), H) = c$.
\end{proof}

In the language of Klyachko, the eigenvalue equation \eqref{eq:onebody/selection rule} is called the \emphindex{selection rule} which is satisfied by a quantum state that is \emph{pinned}\index{pinning} to a facet of the moment polytope \cite{Klyachko09}.
Thus pinned states live on a potentially much lower-dimensional subspace of the Hilbert space, with potential implications on the physics.

It is an interesting question if and under which circumstances states in concrete systems are pinned.
For example, it is an empirical fact that many molecules are well-explained by assuming that the natural occupation numbers are close to 0 and 1 (pinning), so that the global state can be well-approximated by a Slater determinant (the corresponding selection rule), which is a first step to Hartree--Fock theory and the Aufbau principle. %
Thus it is not be unreasonable to wonder if approximate pinning might hold for some of the other defining inequalities of the moment polytope.
See \cite{Klyachko09, Klyachko13} for preliminary investigations in the context of small molecules and magnetism and \cite{SchillingGrossChristandl13} for a study of pinning in a model with small harmonic interactions.

Crucially, the selection rule is stable at least in an elementary sense.
We phrase the following result in terms of the \emphindex{trace norm}\nomenclature[Q<X_1]{$\norm{X}_1$}{trace norm, or Schatten-1 norm of operator $X$} $\norm{X}_1 := \tr \abs X$, which has a useful operational meaning (but this choice is completely arbitrary since the proof is based on a purely topological argument):

\begin{lem}[Stability]
  Let $\norm{-}$ denote an arbitrary norm on $i \mathfrak t^*$.
  Let $F = \{ \lambda \in \Delta_K(\PP(\calH)) : (\lambda, H) = c \}$ be a facet of the moment polytope and
  $\calH(H=c) = \{ \ket\psi : \pi(H) \ket\psi = \ket\psi c \}$ the corresponding subspace of states that satisfy the selection rule.
  For any closed subset $F_0 \subseteq F \cap i \mathfrak t^*_{>0}$ there exists a function $\delta(\varepsilon)$ such that
  \begin{equation*}
    \min_{\lambda \in F_0} \norm{\mu_K(\rho) - \lambda} \leq \varepsilon
    \Rightarrow
    \min_{\rho' \in \PP(\calH(H=c))} \norm{\rho - \rho'}_1 \leq \delta(\varepsilon)
    \qquad
    (\forall \rho)
  \end{equation*}
  and $\delta(\varepsilon) \searrow 0$ as $\varepsilon \searrow 0$.
\end{lem}
\begin{proof}
  Set $F_\varepsilon := \{ \mu \in i \mathfrak t^*_+ : \min_{\lambda \in F_0} \norm{\mu - \lambda} \leq \varepsilon \}$.
  Each $F_\varepsilon$ is a closed set, hence $\mu_K^{-1}(F_\varepsilon)$ is a compact subset of $\PP(\calH)$.
  Since moreover $d(\rho) := \min_{\rho' \in \PP(\calH(H=c))} \norm{\rho - \rho'}_1$ is continuous, it follows that the function
  \[\delta(\varepsilon) := \max d(\mu_K^{-1}(F_\varepsilon)) = \max \{ \min_{\rho' \in \PP(\calH(H=c))} \norm{\rho - \rho'}_1 : \rho \in \PP(\calH), \mu_K(\rho) \in F_\varepsilon \}\]
  is well-defined. Clearly, $\delta(\varepsilon)$ is monotonic in $\varepsilon$, and $\delta(0) = 0$ by \autoref{lem:onebody/selection rule}.

  It remains to show that $\delta(\varepsilon) \rightarrow 0$ as $\varepsilon \rightarrow 0$.
  For sake of contradiction, suppose that this is not the case.
  Then there exists a sequence $(\rho_k) \subseteq \PP(\calH)$ and $C > 0$ such that $\min_{\lambda \in F_0} \norm{\mu_K(\rho_k) - \lambda} \rightarrow 0$ while $\delta(\rho_k) \geq C$ for all $k$.
  By compactness, we may pass to a convergent subsequence; let us denote its limit by $\rho_\infty$.
  Then $\mu_K(\rho_\infty) \in F_0$, while $\delta(\rho_\infty) \geq C$.
  This is a contradiction to \autoref{lem:onebody/selection rule}.
\end{proof}

For concrete applications, it might be interesting to obtain explicit bounds of the form $\delta(\varepsilon) \leq L \varepsilon$.
So far this has only been achieved in rather special situations \cite{SchillingGrossChristandl13, Benavides-RiverosGracia-BondiaSpringborg13}.
It might be possible to obtain a general solution by carefully analyzing the local model for symplectic group actions \cite{GuilleminSternberg82, GuilleminSternberg84a, Marle85}. %

\section{Kronecker coefficients, Schur--Weyl duality, and Plethysms}
\label{sec:onebody/kronecker}

In this section we describe more explicitly the representation-theoretic content of \autoref{thm:onebody/kirwan} for Problems~\ref{pro:onebody/qmp} and \ref{pro:onebody/nrep}.

We first consider the case of three distinguishable particles and choose coordinates, so that $\calH = \CC^a \otimes \CC^b \otimes \CC^c$ and $G = \SL(a) \times \SL(b) \times \SL(c)$.
For $\XX = \PP(\calH)$, the ring of regular functions is equal to the ring of all polynomials on $\calH$, so that $R_k(\PP(\calH)) = \Sym^k(\calH)^*$.
Thus we obtain from \autoref{thm:onebody/kirwan} the following description:
\begin{align*}
  \Delta_K(\PP(\calH)) \cap \QQ\Lambda^*_{G,+}
  &= \{ (\alpha, \beta, \gamma) / k : V^a_\alpha \otimes V^b_\beta \otimes V^c_\gamma \subseteq \Sym^k(\CC^{abc}) \} \\
  &= \{ (\alpha, \beta, \gamma) / k : g_{\alpha,\beta,\gamma} > 0 \}.
\end{align*}
where we denote by $g_{\alpha,\beta,\gamma}$ the multiplicity of $V^a_\alpha \otimes V^b_\beta \otimes V^c_\gamma$ in $\Sym^k(\CC^{abc})$.
These multiplicities are known as the \emph{Kronecker coefficients}\index{Kronecker coefficients|textbf}.%
\nomenclature[Rg_alpha,beta,gamma]{$g_{\alpha,\beta,\gamma}$}{Kronecker coefficients}

There is another way of defining the Kronecker coefficients in terms of the symmetric groups $S_k$ that will be useful later.%
\nomenclature[RS_k]{$S_k$}{symmetric group}
For this, we consider the space $(\CC^d)^{\otimes k}$.
The general linear group $\GL(d)$ acts diagonally by tensor powers and the symmetric group $S_k$ acts by permuting the tensor factors; both actions commute.
Therefore, if we decompose $V$ into irreducible representations of $\GL(d)$ then the multiplicity spaces -- which we shall denote by $[\lambda]$ -- are representations of $S_k$.
It can be shown that the irreducible $\GL(d)$-representations that appear are precisely those whose Young diagram $\lambda$ has $k$ boxes and at most $d$ rows.
What is more, the corresponding representations $[\lambda]$ of the symmetric group are in fact irreducible.
For each Young diagram, one obtains a different irreducible representation (which does not depend on the concrete value chosen for $d$), and all the \emph{irreducible representations of the symmetric group}\index{irreducible representations!$S_k$} can be obtained in this way (if $d$ is large enough).%
\nomenclature[Rlambda]{$[\lambda], [\alpha], [\beta], \dots$}{irreducible representations of the symmetric group}
This result is known as \emph{Schur--Weyl duality}\index{Schur--Weyl duality|textbf} (e.g., \cite{CarterSegalMacDonald95}) and it can be compactly stated in the form:
\begin{equation}
\label{eq:onebody/schur-weyl}
  \left( \CC^d \right)^{\otimes k} \cong
  \; \smashoperator{\bigoplus_{\lambda \vdash_d k}} \; V^d_\lambda \otimes [\lambda]
\end{equation}
Now consider the tripartite case, where $d = abc$. Then the symmetric subspace $\Sym^k(\CC^{abc}) \subseteq (\CC^{abc})^{\otimes k}$ corresponds to the trivial representation of $S_k$.
On the other hand we may apply Schur--Weyl duality to each of the subsystems' tensor powers,
\begin{equation*}
  (\CC^{abc})^{\otimes k}
  = (\CC^a)^{\otimes k} \otimes (\CC^b)^{\otimes k} \otimes (\CC^c)^{\otimes k}
  = \; \smashoperator{\bigoplus_{\alpha \vdash_a k, \, \beta \vdash_b k, \, \gamma \vdash_c k}} \;
    V^a_\alpha \otimes V^b_\beta \otimes V^c_\gamma \otimes [\alpha] \otimes [\beta] \otimes [\gamma]
\end{equation*}
It follows that
\begin{equation}
\label{eq:onebody/schur-weyl tripartite}
  \Sym^k(\CC^{abc})
  = \; \smashoperator{\bigoplus_{\alpha \vdash_a k, \, \beta \vdash_b k, \, \gamma \vdash_c k}} \;
    V^a_\alpha \otimes V^b_\beta \otimes V^c_\gamma \otimes \left( [\alpha] \otimes [\beta] \otimes [\gamma] \right)^{S_k}
\end{equation}
Thus the Kronecker coefficients can be equivalently defined as the dimension of the invariant subspace in a triple tensor product of irreducible representations of the symmetric group:
\begin{equation*}
  g_{\alpha,\beta,\gamma} = \dim \, \left( [\alpha] \otimes [\beta] \otimes [\gamma] \right)^{S_k}
\end{equation*}
In particular, we find that each Kronecker coefficient only depends on the triple of Young diagrams rather than the concrete values chosen for $a$, $b$ and $c$ (but of course $a$, $b$ and $c$ have to be chosen at least as large as the number of rows of the Young diagrams).
The role of the Kronecker coefficients for the one-body quantum marginal problem has first been observed in \cite{ChristandlMitchison06} by using the spectrum estimation theorem (cf.\ \cite{Klyachko04, ChristandlHarrowMitchison07} and the proof of \autoref{thm:strong6j/main theorem}).
They also play a fundamental role in representation theory \cite{Fulton97} and in Mulmuley and Sohoni's geometric complexity theory approach to the $\Pclass$ vs.\ $\NP$ problem in computer science \cite{MulmuleySohoni01, MulmuleySohoni08, Mulmuley07, BuergisserLandsbergManivelEtAl11} (see \autoref{sec:mul/summary}), and they occur in the ``quantum method of types'' \cite{Harrow05}.
In \autoref{ch:multiplicities} we will give an efficient algorithm for their computation.

\bigskip

There is a different, asymmetric way of defining the Kronecker coefficients that is also quite useful.
For this, we recall that the irreducible representations of the symmetric group are self-dual, i.e., $[\lambda] \cong [\lambda]^*$ \cite[\S{}2.1]{JamesKerber81}. Therefore,
\begin{equation*}
    ([\alpha] \otimes [\beta] \otimes [\gamma])^{S_k}
  = ([\alpha]^* \otimes [\beta] \otimes [\gamma])^{S_k}
  = \Hom_{S_k}([\alpha], [\beta] \otimes [\gamma]),
\end{equation*}
where we have used the general notation
\begin{equation}
\label{eq:onebody/linear maps}
  \Hom_G(V, W) := \{ \phi \colon V \rightarrow W \text{ linear} : \phi(g v) = g \phi(v) \quad (\forall v \in V, g \in G) \}
\end{equation}
for the space of $G$-equivariant linear maps, or \emph{$G$-linear maps}\index{$G$-linear map} between two representations $V$ and $W$.%
\nomenclature[RHom_G(V, W)]{$\Hom_G(V, W)$}{space of $G$-linear maps between two representations $V$ and $W$}
For irreducible $V$ and $W$, \emphindex{Schur's lemma} asserts that
\begin{equation}
\label{eq:onebody/schurs lemma}
  \dim \Hom_{S_k}(V, W) = \begin{cases}
    1 & \text{if } V \cong W \\
    0 & \text{otherwise}.
  \end{cases}
\end{equation}
It follows that $g_{\alpha,\beta,\gamma}$ can also be defined as the multiplicity of $[\alpha]$ in the tensor product $[\beta] \otimes [\gamma]$ of two irreducible representations of the symmetric group:
\begin{equation}
\label{eq:onebody/kronecker asymmetric}
  [\beta] \otimes [\gamma] = \bigoplus_{\alpha \vdash k} g_{\alpha,\beta,\gamma} [\alpha]
\end{equation}
In particular, for $[\gamma] = \mathbf 1$ the trivial representation of $S_k$ we obtain that $g_{\alpha,\beta,\mathbf 1} = \delta_{\alpha,\beta}$.
In view of Schur--Weyl duality, this implies that
\begin{equation}
\label{eq:onebody/schur-weyl bipartite}
  \Sym^k(\CC^a \otimes \CC^b)
  = \; \smashoperator{\bigoplus_{\alpha \vdash_a k, \, \beta \vdash_b k}} \; \delta_{\alpha,\beta} \, V^a_\alpha \otimes V^b_\beta
  = \; \smashoperator{\bigoplus_{\alpha \vdash_{\min \{a,b\}} k}} \; V^a_\alpha \otimes V^b_\alpha.
\end{equation}
This equation corresponds to the one-body quantum marginal problem for $n=2$ particles and is therefore the bipartite counterpart of \eqref{eq:onebody/schur-weyl tripartite}.
The fact that the irreducible representations of the two factors are perfectly paired is the representation-theoretic version of the fact that the marginals of a bipartite pure state are isospectral (\autoref{lem:onebody/two})---indeed, the latter is a direct consequence of \eqref{eq:onebody/schur-weyl bipartite} and \autoref{thm:onebody/kirwan}.

\bigskip

In the case of fermions, we are similarly led to study the decomposition of $R_k^*(\calH) = \Sym^k \left( \Alt^n \CC^d \right)$ into irreducible representations.
This is an instance of a \emph{plethysm}, which is\ more generally defined as the composition of Schur functors $\mathcal H \mapsto V_{\GL(\mathcal H),\lambda}$ (see, e.g., \cite{MacDonald95}).

\subsection*{Sums of Matrices}
\index{sums of matrices}

We conclude this section by mentioning an interesting related problem that is amenable to similar methods.
In \cite{Weyl12}\index{Weyl's problem}, Weyl considered the relation between the eigenvalues of Hermitian matrices $A$ and $B$ and their sum $A+B$ and derived first non-trivial constraints.
The general solution was famously conjectured by \cite{Horn62} in terms of certain linear inequalities.
Weyl's question can be phrased in the general framework of geometric invariant theory:
We want to find triples of coadjoint orbits for $K = \U(d)$ such that $\calO_{K,\alpha} + \calO_{K,\beta} \ni \calO_{K,\lambda}$.
For fixed integral $\alpha$ and $\beta$, the solution can be obtained as the moment polytope for the diagonal $K$-action on $\calO_{K,\alpha} \times \calO_{K,\beta}$, which can be considered as a projective subvariety of $\PP(V^d_\alpha \otimes V^d_\beta)$ (\autoref{lem:onebody/borel weil}) \cite{Klyachko98, Knutson00, Klyachko04}.
From the perspective of representation theory, this is related to the decomposition of tensor products $V^d_{k\alpha} \otimes V^d_{k\beta}$ of irreducible $\GL(d)$-representations.
The corresponding multiplicities are known as the \emphindex{Littlewood--Richardson coefficients} for $\SU(d)$ \cite{Lidskii82}.
For $d = 2$, the decomposition is multiplicity-free; in physics it is known as the \emphindex{Clebsch--Gordan series}.
A necessary and sufficient set of inequalities was first obtained in \cite{Klyachko98} by the same algebraic-geometric methods that can be used to solve the one-body quantum marginal problem.
Shortly after, \emphindex{Horn's conjecture} was established in \cite{KnutsonTao99} (cf.\ \cite{Fulton00, KnutsonTao01, KnutsonTaoWoodward03}).
Strikingly, the one-body quantum marginal problem subsumes the problem of characterizing the eigenvalues of sums of Hermitian matrices in a precise technical sense \cite{Klyachko04}, and an analogue statement holds on the level of representation theory \cite{Littlewood58, Murnaghan55}.
We will later generalize this result and in particular show that the problem of determining the relation between the eigenvalues of three matrices $A$, $B$, $C$ and their partial sums can similarly be seen as a special case of a more general quantum marginal problem with overlapping marginals (\autoref{sec:strong6j/sums of matrices}).

\subsection*{Geometric Quantization}
\index{geometric quantization|textbf}\index{quantization}

Before we proceed, we offer a word of caution for people acquainted with the theory of geometric quantization \cite{GuilleminSternberg77, GuilleminSternberg84, Woodhouse92}.
Although we formally use a similar mathematical framework as in geometric quantization, the physical interpretation is markedly different.
Unlike in geometric quantization, our quantum states do \emph{not} arise via some quantization procedure from a classical symplectic phase space.
On the contrary, in the mathematical modeling of the quantum marginal problem the projective space of pure states corresponds to the classical phase space, while its description in terms of the representations that occur in the ring of regular functions can be seen as its ``quantization''.
The ``semiclassical limit''\index{semiclassical limit} $k \rightarrow \infty$ in which we recover the description of the moment polytope plays a purely purely mathematical role (cf.\ \autoref{sec:mul/asymptotics}).

\chapter{Solving The One-Body Quantum Marginal Problem}
\label{ch:kirwan}

In this chapter we review some of the history of the one-body quantum marginal problem that culminated in Klyachko's general solution and give some concrete examples.
We then present a different approach to the problem of computing moment polytopes for projective space, which we have seen subsumes the one-body quantum marginal problem and its variants.
Significantly, our geometric approach completely avoids many technicalities that have appeared in previous solutions to the problem, such as Schubert calculus, and it can be readily implemented algorithmically.
We illustrate our method with a number of illustrative examples.

The results in this chapter are based on unpublished joint work in progress with Mich\`{e}le Vergne.

\subsection*{Prior Work and Examples}

The history of the quantum marginal problem goes back at least to the late 1950s, where it had been observed that the ground state energy of a two-body Hamiltonian is a function of the two-body reduced density matrices only \cite{Loewdin55, Mayer55}.
The main focus was therefore on the \emph{two-body $n$-representability problem}\index{n-representability problem@$n$-representability problem!two-body}---given a two-body density matrix, is it compatible with a state of $n$ fermions \cite{Coleman63, Ruskai69, ColemanYukalov00, Coleman01}?
Some results had also been obtained for the one-body marginals.
For instance, Coleman proved that a first-order density matrix $\gamma_1$ is compatible with a (not necessarily pure) state of $n$ fermions if and only if the natural occupation numbers do not exceed $1$---that is, if and only if the Pauli principle is satisfied \cite[Theorem 9.3]{Coleman63}.

In the 1970s, it was shown by Borland and Dennis that for three fermions with six-dimensional single-particle Hilbert space, $\calH = \Alt^3 \CC^6$, the following conditions on the eigenvalues $\lambda_1 \geq \dots \geq \lambda_6$ of the first-order density matrix $\gamma_1 = 3 \rho_1$ are both necessary and sufficient for the existence of a global pure state \cite{BorlandDennis72},
\begin{equation}
\label{eq:kirwan/borland-dennis}
  \begin{aligned}
    \lambda_1 + \lambda_6 &= \lambda_2 + \lambda_5 = \lambda_3 + \lambda_4 = 1 \\
    \lambda_5 + \lambda_6 &\geq \lambda_4
  \end{aligned}
\end{equation}
(see \autoref{fig:kirwan/borland-dennis}).
This was perhaps the first non-trivial solution of \autoref{pro:onebody/nrep}.
Remarkably, the resulting polytope is only three-dimensional; this coincides with the fact that any pure state can be written as a linear combination of only 8 Slater determinants as was proved by Ruskai and Kingsley (while $\dim i \mathfrak t^* = 5$ and $\dim \calH = 20$).
It is interesting to observe that the equality $\lambda_1 + \lambda_6 = 1$ strengthens the Pauli principle $\lambda_1 \leq 1$.

In the context of quantum information theory, Higuchi, Sudbery and Szulc have first considered the one-body quantum marginal problem, \autoref{pro:onebody/qmp}, for $n$ qubits, $\calH = (\CC^2)^{\otimes n}$ \cite{HiguchiSudberySzulc03}. They showed that the following \emphindex{polygonal inequalities} on the maximal eigenvalues $\lambda_{k,1}$ are both necessary and sufficient for the compatibility with a global pure state,
\begin{equation}
\label{eq:kirwan/polygonal}
  \sum_{k \neq l} \lambda_{k,1} \leq (n-2) + \lambda_{l,1}
  \quad (\forall l=1,\dots,n)
\end{equation}
(see \autoref{fig:kirwan/three-qubits}).
In fact, these inequalities hold for any multipartite quantum state, but they are in general not sufficient for compatibility (see \autoref{prp:strong6j/embedding converse} for an elementary proof based on the variational principle).
In the meanwhile, the three-qutrit polytope had already been computed by Franz \cite{Franz02}, as was only later recognized.

Subsequently, Bravyi solved the case of mixed states of two qubits by a remarkable explicit argument \cite{Bravyi04}. Here, the necessary and sufficient conditions are given by
\begin{equation}
\label{eq:kirwan/bravyi}
  \begin{aligned}
    \max \{ \lambda_{A,1}, \lambda_{B,1} \} &\leq \lambda_{AB,1} + \lambda_{AB,2} \\
    \lambda_{A,1} + \lambda_{B,1} &\leq 1 + \lambda_{AB,1} - \lambda_{AB,4}  \\
    \abs{\lambda_{A,1} - \lambda_{B,1}} &\leq \min \{ \lambda_{AB,1} - \lambda_{AB,3}, \lambda_{AB,2} - \lambda_{AB,4} \}
  \end{aligned}
\end{equation}
(see \autoref{fig:kirwan/bravyi}). We remark that this scenario is equivalent to the pure-state problem for $\CC^2 \otimes \CC^2 \otimes \CC^4$ as was explained in \autoref{ch:onebody} (cf.\ \autoref{tab:onebody/summary}).

\begin{figure}[p]
  \centering
  \includegraphics[height=4cm]{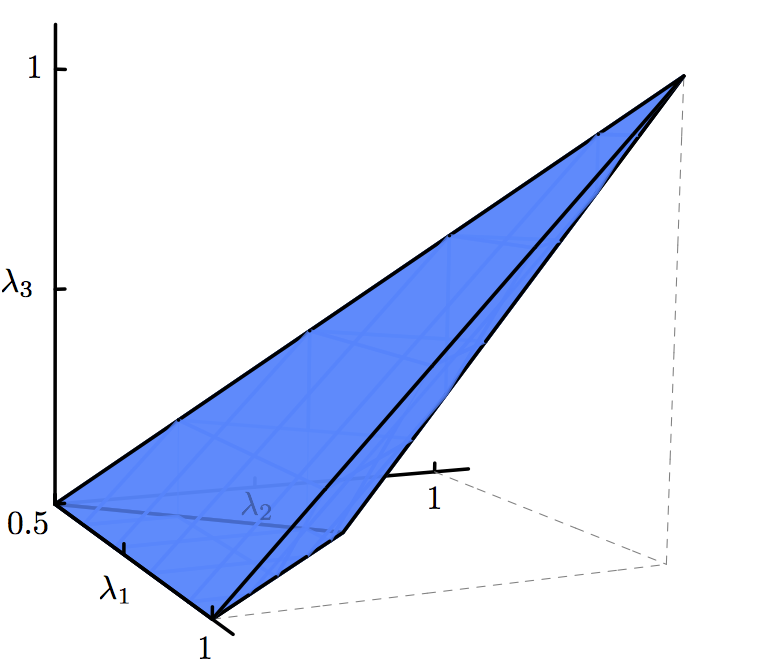}
  \caption[Borland--Dennis polytope]{\emph{Borland--Dennis polytope.}
    The solution of the one-body $n$-representability problem for three fermions with local dimension six, as given by the Borland--Dennis inequalities \eqref{eq:kirwan/borland-dennis}. The vertex $(1,1,1)$ corresponds to a single Slater determinant.}
  \label{fig:kirwan/borland-dennis}
  \includegraphics[height=4cm]{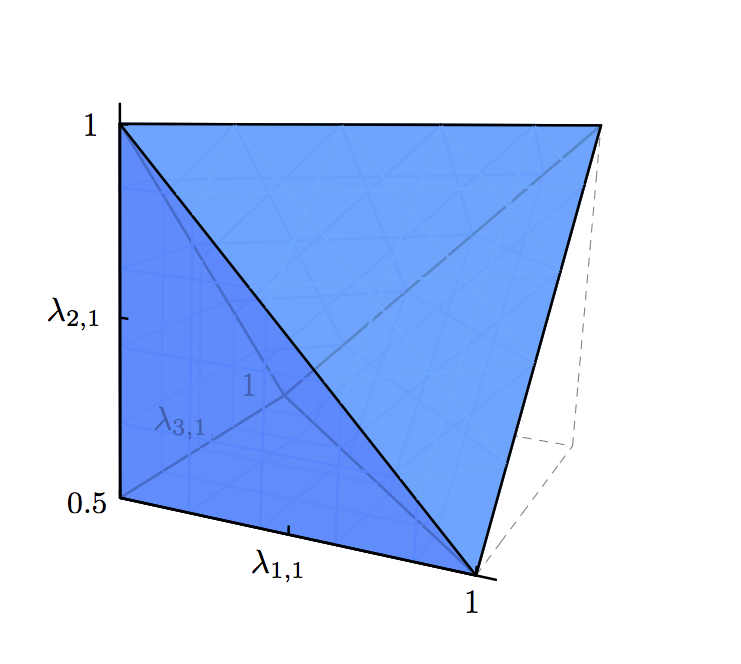}
  \caption[Three-qubit polytope]{\emph{Three-qubit polytope.}
    The solution of the one-body quantum marginal problem for pure states of three qubits, as given by the polygonal inequalities \eqref{eq:kirwan/polygonal} for $n=3$.}
  \label{fig:kirwan/three-qubits}
  \includegraphics[height=4cm]{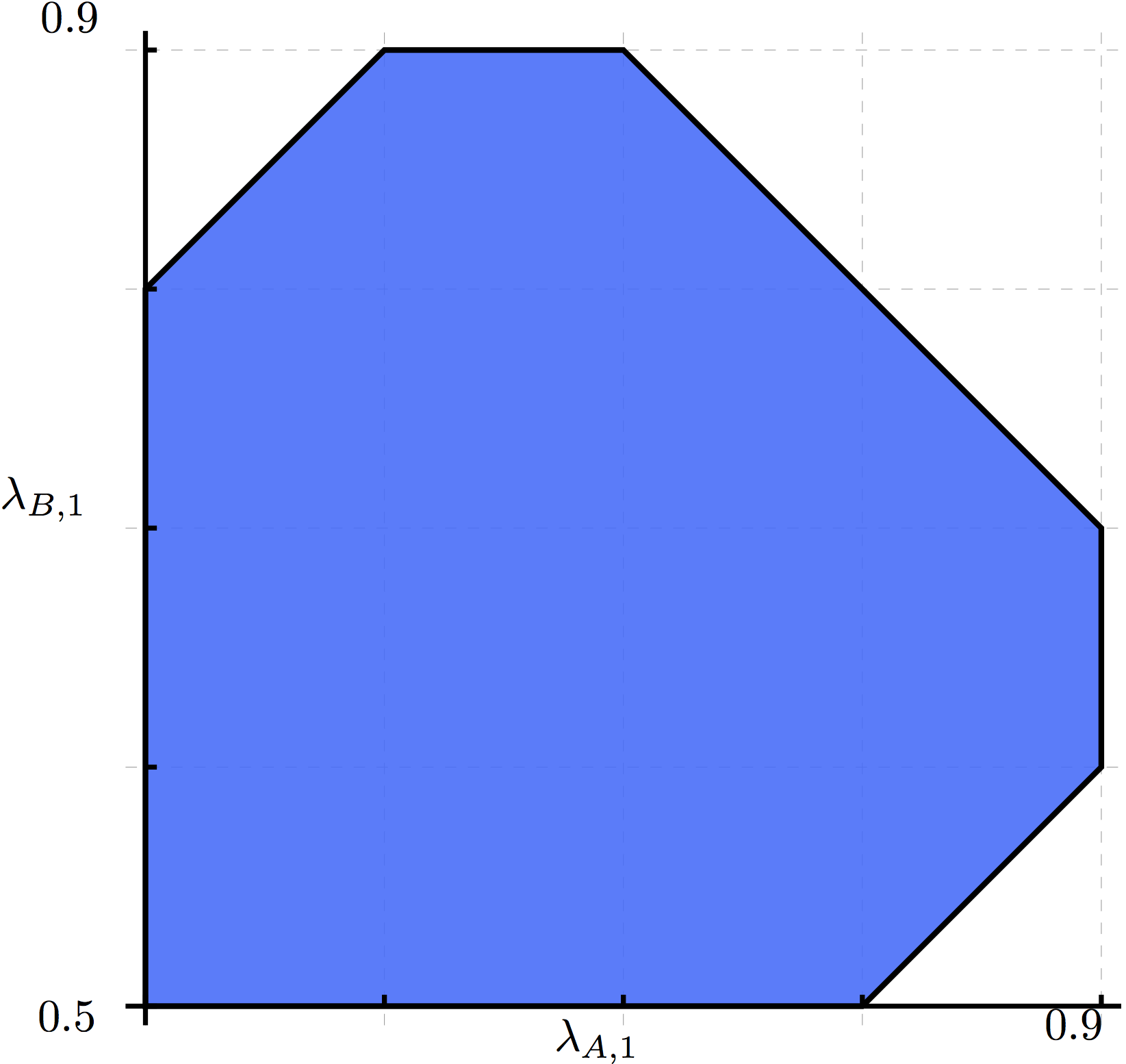}
  \caption[Bravyi polytope]{\emph{Bravyi's polytope},
    corresponding to his solution \eqref{eq:kirwan/bravyi} of the one-body quantum marginal problem for two qubits and global spectrum $\lambda_{AB} = (0.6, 0.3, 0.1, 0)$.}
  \label{fig:kirwan/bravyi}
\end{figure}

The connection of the one-body quantum marginal problem to representation theory was first observed in \cite{ChristandlMitchison06} by using quantum information methods rather than the theory of \autoref{sec:onebody/git} (cf.\ \cite{ChristandlHarrowMitchison07}).
Shortly after, a completely general solution was given by Klyachko both for distinguishable particles \cite{Klyachko04} and for fermions \cite{KlyachkoAltunbulak08, Altunbulak08}.
Almost simultaneously, Daftuar and Hayden had published a solution to the ``one-sided'' problem that concerns the constraints between the eigenvalues of $\rho_{AB}$ and $\rho_A$ \cite{DaftuarHayden04}.
Both results build on previous work by Berenstein and Sjamaar \cite{BerensteinSjamaar00}, who used geometric invariant theory to study the moment polytope for projections of coadjoint orbits; this latter work in turn generalizes techniques from Klyachko's seminal paper on Weyl's problem \cite{Klyachko98} (see the discussion at the end of the preceding chapter).
We refer to \cite{Klyachko04, Knutson09} for eloquent expositions of the method.
More recently, Ressayre has refined the result of Berenstein and Sjamaar to give an irredundant set of necessary and sufficient inequalities in a very general mathematical setup \cite{Ressayre10, Ressayre10a}.
We remark that a variant of the quantum marginal problem for Gaussian states has been considered in \cite{EisertGross08} (mathematically, this scenario is covered by a more general convexity theorem for non-compact manifolds with proper moment maps).

\section{Summary of Results}

Throughout this chapter we will assume for simplicity that the moment polytope $\Delta_K := \Delta_K(\PP(\calH))$ is of maximal dimension.%
\nomenclature[TDelta_K]{$\Delta_K$}{moment polytope for the $K$-action on projective space}
For the quantum marginal problem, this assumption is usually satisfied except for two distinguishable particles (which we have already solved in \autoref{ch:onebody}) and in the fermionic Borland--Dennis scenario $\Alt^3 \CC^6$ (see \eqref{eq:kirwan/borland-dennis} above), and it is easily checked in practice. %
Since the moment polytope is a convex polytope, it can be defined by a finite list of linear inequalities
\begin{equation*}
  \Delta_K = \{ \lambda \in i\mathfrak t^* : (\lambda, H_1) \geq c_1, \dots, (\lambda, H_f) \geq c_f \},
\end{equation*}
where $H_1, \dots, H_f$ are the \emph{normal vectors}\index{moment polytope!normal vector} of the finitely many \emph{facets}\index{moment polytope!facet} of $\Delta_K$.
All our normal vectors will always be pointing inwards.
Since the moment polytope is obtained by intersecting $\mu_K(\PP(\calH))$ with the positive Weyl chamber $i \mathfrak t^*_+$, which is a maximal-dimensional polyhedral cone, some of the facets of $\Delta_K$ might be subsets of facets of $i \mathfrak t^*_+$, and we call those the \emph{trivial facets}\index{moment polytope!trivial facet} of the moment polytope.

Our approach in the following is based on analyzing the non-trivial facets in terms of the local differential geometry of their preimages up to second order. Conceptually, such an analysis should be sufficient since the moment map is locally quadratic.

In \autoref{sec:kirwan/torus} we start by studying the moment map to first order.
It is well-known that interior points of non-trivial facets are critical values for $(\mu_K, H)$, where $H$ is the normal vector of the facet. Indeed, this is equivalent to the selection rule from \autoref{sec:onebody/consequences}.
We show that as a consequence any non-trivial facet of $\Delta_K$ is necessarily contained in a hyperplane spanned by weights of the representation $\calH$. %
This already reduces the problem to a finite set of candidates.

In \autoref{sec:kirwan/facets} we then consider the Hessian of the moment map.
If a state is mapped into the interior of a non-trivial facet of the polytope then this implies a positive semidefiniteness of the corresponding Hessian in certain tangent directions.
We exploit this fact to obtain another necessary condition that is satisfied by non-trivial facets $(-, H) \geq c$ of the moment polytope.
To state it, let $\mathfrak n_-(H < 0)$ denote the direct sum of negative root spaces $\mathfrak g_\alpha$ with $(\alpha, H) < 0$ and $\calH(H < c)$ the sum of eigenspaces of $\pi(H)$ with eigenvalue smaller than $c$.
Then we show that there necessarily exists an eigenvector $\ket\psi \in \calH$ of $\pi(H)$ with eigenvalue $c$ such that the map
\begin{equation*}
  \mathfrak n_-(H < 0) \rightarrow \calH(H < c), X \mapsto \pi(X) \ket\psi
\end{equation*}
is an isomorphism. %
Conversely, we prove that any inequality $(-,H) \geq c$ that satisfies the above two necessary conditions is a valid inequality of the moment polytope. %
We thus obtain a complete description of the moment polytope $\Delta_K$ in terms of what we call \emph{inequalities of Ressayre type} (\autoref{dfn:kirwan/ressayre} and \autoref{thm:kirwan/main}).

Our notion of a Ressayre-type inequality is closely related to Ressayre's notion of a \emph{dominant pair}, and our approach is inspired by his ideas \cite{Ressayre10, Ressayre10a}.
Our description of the moment polytope is also related to a result by Brion \cite{Brion99}, as we explain further below.
However, while the more refined results of \cite{Brion99, Ressayre10, Ressayre10a} are established using high-powered algebraic geometry, we proceed in essence by a straightforward differential-geometric analysis, combined with \autoref{thm:onebody/kirwan}.

In \autoref{sec:kirwan/examples}, we illustrate the method with some examples.
We remark that, crucially, our description of the moment polytope obtained in \autoref{thm:kirwan/main} can be completely automatized:
It is straightforward to determine all inequalities of Ressayre type in a mechanical fashion, and hence also on a computer.

\section{The Torus Action}
\label{sec:kirwan/torus}

We use the same notation and conventions as in \autoref{ch:onebody}.
Thus let $G$ be a connected reductive algebraic group, $K \subseteq G$ a maximal compact subgroup, and $\Pi \colon G \rightarrow \GL(\calH)$ a representation on a Hilbert space $\calH$ with a $K$-invariant inner product; we denote the infinitesimal representation by $\pi \colon \mathfrak g \rightarrow \mathfrak{gl}(\calH)$.
Our object of interest is the moment polytope $\Delta_K := \Delta_K(\PP(\calH))$ of the projective space associated with the representation $\calH$.
As stated above, we assume throughout this chapter that $\Delta_K$ is of maximal dimension, i.e., $\dim \Delta_K = \dim i \mathfrak t^*$.
This is the case whenever there exists a pure state with finite stabilizer, as follows from \eqref{eq:onebody/moment map range} and the local submersion theorem.
The facets of $\Delta_K$ are of codimension one in $i \mathfrak t^*$ and may be identified with the defining linear inequalities $(-,H) \geq c$ of the moment polytope.

\begin{dfn}
  A facet of the moment polytope is \emph{trivial}\index{moment polytope!trivial facet|textbf} if it is of the form $(-, Z_\alpha) \geq 0$ for some positive root $\alpha \in R_{G,+}$.
  Otherwise, the facet is called \emph{non-trivial}\index{moment polytope!non-trivial facet}.
\end{dfn}

Non-trivial facets have also been called ``general'' in the literature \cite{Brion99}.
We record the following straightforward observation:

\begin{lem}
\label{lem:kirwan/non-trivial facets}
  Any non-trivial facet of $\Delta_K$ meets the relative interior $i \mathfrak t^*_{>0}$ of the positive Weyl chamber.
\end{lem}
\begin{proof}
  Any facet of $\Delta_K$ that does not intersect the relative interior of the positive Weyl chamber is fully contained in
  \begin{equation*}
    i \mathfrak t^*_+ \setminus i \mathfrak t^*_{>0} = \bigcup_{\alpha \in R_{G,+}} \{ \xi \in i\mathfrak t^* : (\xi, Z_\alpha) = 0 \},
  \end{equation*}
  which is a finite arrangement of hyperplanes.
  Since by assumption the facet is of codimension one, it has to be contained in a single one of these hyperplanes.
  Thus its normal vector is either $\pm Z_\alpha$ for some positive root $\alpha$.
  If it was $-Z_\alpha$ then the moment polytope would be strictly contained in the hyperplane $(-, Z_\alpha) = 0$, and therefore not of maximal dimension, in contradiction with our assumption.
  We conclude that the facet is of the form $(-, Z_\alpha) \geq 0$, and therefore trivial.
\end{proof}

We now consider the moment map for the action of the maximal torus $T \subseteq K$\index{moment map!maximal torus},%
\nomenclature[Tmu_K]{$\mu_T$}{moment map for the $T$-action}
\begin{equation}
\label{eq:kirwan/abelian moment map}
  \mu_T \colon \PP(\calH) \rightarrow i \mathfrak t^*,
  \quad
  (\mu_T(\rho), X) := \tr \rho \, \pi(X).
\end{equation}
For the quantum marginal problem, this amounts to considering diagonal entries rather than eigenvalues.
Let $\calH = \bigoplus_{\omega \in \Omega} \calH_\omega$ be the decomposition of $\calH$ into weight spaces, and $\rho = \proj\psi$ a pure state with $\ket\psi = \sum_\omega \psi_\omega \ket\omega$ decomposed accordingly.
Then $\mu_T$ has the following concrete description:
\begin{equation}
  \label{eq:kirwan/abelian convex combination}
  \mu_T(\rho) = \sum_\omega \abs{\psi_\omega}^2 \omega
\end{equation}
(cf.\ the proof of \autoref{thm:onebody/kirwan}).
Observe that $\mu_T(\rho)$ is a convex combination of weights.
It follows that the ``Abelian'' moment polytope $\Delta_T := \Delta_T(\PP(\calH))$ of $\PP(\calH)$ is precisely equal to the convex hull of the set of weights; it is maximal-dimensional since it contains $\Delta_K$.%
\nomenclature[TDelta_T]{$\Delta_T$}{moment polytope for the $T$-action on projective space}
More generally, if $\Omega' \subseteq \Omega$ is a subset of weights and $\calH_{\Omega'} := \bigoplus_{\omega \in \Omega'} \calH_\omega$ then $\Delta_T(\PP(\calH_{\Omega'})) = \conv \Omega'$.
For the next lemma recall that a \emphindex{critical point} of a smooth map $f \colon M \rightarrow M'$ is a point $m \in M$ where the differential $df\big|_m$ is not surjective; a \emphindex{critical value} is the image $f(m)$ of a critical point \cite{Lee13}.

\begin{lem}
  \label{lem:kirwan/abelian critical}
  The set of critical values of $\mu_T$ is equal to the union of the codimension-one convex hulls of subsets of weights.
\end{lem}
\begin{proof}
  Let $\rho = \proj\psi \in \PP(\calH)$ with weight decomposition $\ket\psi = \sum_\omega \psi_\omega \ket\omega$.
  By \eqref{eq:kirwan/abelian convex combination}, $\mu_T(\rho)$ is a convex combination of weights in $\Omega_\rho := \{ \omega : \psi_\omega \neq 0 \}$.
  By \eqref{eq:onebody/moment map property} and non-degeneracy of the symplectic form, $\rho$ is a critical point if and only if there exists $0 \neq X \in \mathfrak t$ such that $X_\rho = [\pi(X), \rho] = 0$, i.e., if and only if $(\omega - \omega')(X) = 0$ for any two $\omega \neq \omega' \in \Omega_\rho$.
  It follows that $\rho$ is a critical point if and only if the convex hull of $\Omega_\rho$ is of positive codimension.

  In particular, any critical value is contained in a convex hull of weights of codimension one, since we may always add additional weights.
  Conversely, if $\Omega' \subseteq \Omega$ is a subset of weights that spans a convex hull of codimension one then $\Delta_T(\PP(\calH_{\Omega'})) = \conv \Omega'$ consists of critical values.
\end{proof}

We now derive a basic necessary condition that cuts down the defining inequalities of the moment polytope to a finite set of candidates. %

\begin{lem}
\label{lem:kirwan/admissible}
  Any non-trivial facet of $\Delta_K$ is contained in an affine hyperplane spanned by a subset of weights.
\end{lem}
\begin{proof}
  By \autoref{lem:kirwan/non-trivial facets}, the intersection of any non-trivial facet with the interior of the positive Weyl chamber $i \mathfrak t^*_{>0}$ is non-empty.
  Each point in this intersection is a critical value for $(\mu_K,H) = (\mu_T,H)$ by selection rule (\autoref{lem:onebody/selection rule}), hence of $\mu_T$, and therefore contained in an affine hyperplane spanned by a subset of weights (\autoref{lem:kirwan/abelian critical}).
  Since this is true for all points in the intersection, which contains the relative interior of the facet, it follows that the facet is in fact contained in a single such hyperplane.
\end{proof}

We remark that if $(\vec a,\vec b) \in \RR^{a+b}$ is an \emphindex{extremal edge} in the sense of \cite{Klyachko04} and $\vec c \in \RR^{ab}$ any vector with components $a_i + b_j$ then $H := (\vec a, \vec b,-\vec c)$ is the normal vector of a linear hyperplane spanned by weights of $\CC^a \otimes \CC^b \otimes \CC^{ab}$.
Thus the inequalities produced in \cite{Klyachko04} for the mixed-state problem satisfy the conclusion of the lemma.

\section{Facets of the Moment Polytope}
\label{sec:kirwan/facets}

Throughout this section, let $\rho$ be the preimage of a point $\mu_K(\rho) \in i \mathfrak t^*_{>0}$ on the facet $(-,H) \geq c$ of the moment polytope.
We have seen that the selection rule \autoref{lem:onebody/selection rule} shows that $\rho$ is a critical point of the component $(\mu_K, H)$, and we have used this in \autoref{lem:kirwan/admissible} to gain information on the set of possible facets.

\subsection*{The Hessian}

It is thus natural to continue by studying the \emph{Hessian}\index{moment map!Hessian} of $(\mu_K, H)$, which is a quadratic form $Q(-,-)$ on the tangent space at $\rho$.%
\nomenclature[TQ]{$Q(-,-)$}{Hessian of moment map component at critical point}
Since $\rho$ is a critical point, we can compute it by
\begin{align*}
  Q(X_\rho, X_\rho)
  &= \partial_t^2 \, (\mu_K(e^{Xt} \rho e^{-Xt}), H) \big|_{t=0} \\
  &= \partial_t^2 \, \tr \rho \, e^{-Xt} \pi(H) e^{Xt} \big|_{t=0} \\
  &= \tr \rho \, \left( X^2 \pi(H) - 2 X \pi(H) X + \pi(H) X^2 \right) \\
  &= \tr \rho \, [[\pi(H), X], X]
\end{align*}
for all anti-Hermitian operators $X \in \mathfrak u(\calH)$.
It follows that \cite[(32.8)]{GuilleminSternberg84}
\begin{equation}
\label{eq:kirwan/hessian}
  Q(X_\rho, Y_\rho) = \tr \rho \, [[\pi(H), X], Y] = \omega([\pi(-iH), X]_\rho, Y_\rho).
\end{equation}
for all $X, Y \in \mathfrak u(\calH)$ (which is indeed symmetric in $X$ and $Y$).
We now decompose
\begin{equation}
\label{eq:kirwan/weight space decomposition}
  \calH = \calH(H < c) \oplus \calH(H = c) \oplus \calH(H > c),
\end{equation}
where $\calH(H < c) = \bigoplus_{\omega : (\omega,H) < c} \calH_\omega$ is the sum of the eigenspaces of the Hermitian operator $\pi(H)$ with eigenvalue less than $c$, etc.%
\nomenclature[RH(H<0)]{$\calH(H < c)$, \dots}{sum of $\pi(H)$-eigenspaces with eigenvalue less than $c$, \dots}
Then we have the following interpretation of the \emph{index}\index{index of quadratic form} of the Hessian (the dimension of a maximal subspace on which the quadratic form $Q$ is negative definite):

\begin{lem}
  \label{lem:kirwan/index hessian}
  The index of the Hessian at $\rho$ is equal to the real dimension of $\calH(H<c)$.
\end{lem}
\begin{proof}
  Recall from \eqref{eq:onebody/tangent space} that the tangent space at $\rho=\proj\psi$ can be identified with the orthogonal complement of $\ket\psi$ in $\calH$.
  For any $\ket\phi \in \ket\psi^\perp$, the corresponding tangent vector is $V = \ketbra\phi\psi + \ketbra\psi\phi \in T_\rho \PP(\calH)$, which we can write as $X_\rho = [X, \rho]$ for $X = \ketbra\phi\psi - \ketbra\psi\phi \in \mathfrak u(\calH)$.
  By \eqref{eq:kirwan/hessian},
  \begin{align*}
    Q(V, V) &= Q(X_\rho, X_\rho) = \tr \rho \, [[\pi(H), X], X] = \tr X_\rho \, [\pi(H), X]\\
    &= 2 \left( \braket{\phi | \pi(H) | \phi} - \braket{\phi | \phi} \braket{\psi | \pi(H) | \psi} \right) \\
    &= 2 \left( \braket{\phi | \pi(H) | \phi} - \braket{\phi | \phi} (\mu_K(\rho), H) \right) \\
    &= 2 \braket{\phi | \pi(H) - c | \phi}.
  \end{align*}
  Since $\psi$ itself is in $\calH(H=c)$ by the selection rule, the claim follows.
\end{proof}

Since $(\mu_K(\rho), H) = c$ and $d(\mu_K,H)\big|_\rho = 0$, the Hessian is necessarily positive semidefinite on the subspace of those tangent vectors that get mapped to $i \mathfrak t^*$.
Indeed, let $V \in d\mu_K^{-1}\big|_\rho(i \mathfrak t^*)$ and consider the curve $\rho_t$ from \autoref{lem:onebody/curve lemma}. Then,
\begin{equation}
  (\mu_K(\rho_t), H) = c + \frac {t^2} 2 Q(V, V) + O(t^3),
\end{equation}
which shows that, indeed, $Q(V, V) \geq 0$.
The subspace of all such $V$ can be computed in a different way.
For this, let $\mathfrak r = \bigoplus_{\alpha \in R_{G,+}} \mathfrak k_\alpha$ denote the sum of the root spaces of the compact Lie algebra as in \eqref{eq:onebody/compact root space decomposition} and
\[M := \{ R_\rho = [\pi(R), \rho] : R \in \mathfrak r \} \subseteq T_\rho \PP(\calH)\]
the corresponding subspace of tangent vectors.
Then,
\begin{equation*}
\begin{aligned}
    &\{ V \in T_\rho \PP(\calH) : d\mu_K(V) \in i \mathfrak t^* \} \\
  =\,&\{ V \in T_\rho \PP(\calH) : d(\mu_K, iR)(V) = 0 \quad (\forall R \in \mathfrak r) \} \\
  =\, &\{ V \in T_\rho \PP(\calH) : \omega(V, R_\rho) = 0 \quad (\forall R \in \mathfrak r) \}
\end{aligned}
\end{equation*}
which is by definition equal to the \emphindex{symplectic complement}\nomenclature[TMomega]{$M^\omega$}{symplectic complement of subspace $M \subseteq T_\rho \PP(\calH)$} $M^\omega$ of $M$ in $T_\rho \PP(\calH)$.
We summarize our discussion in the following lemma.

\begin{lem}
\label{lem:kirwan/cross section hessian positive semidefinite}
  The Hessian at $\rho$ is positive semidefinite on the subspace
  \begin{equation}
  \label{eq:kirwan/cross section tangent space}
  \begin{aligned}
    M^\omega
    :=\;&\{ V \in T_\rho \PP(\calH) : \omega(V, W) = 0 \quad (\forall W \in M) \} \\
    =\;&\{ V \in T_\rho \PP(\calH) : d\mu_K(V) \in i \mathfrak t^* \}
  \end{aligned}
  \end{equation}
\end{lem}

\begin{lem}
\label{lem:kirwan/torus complement injective}
  For any $\lambda \in i \mathfrak t^*_{>0}$,  %
  the coadjoint action
  $\mathfrak r \rightarrow i\mathfrak k^*$, $R \mapsto \ad^*(R)\lambda$
  is injective.
  Therefore, the tangent map
  $\mathfrak r \rightarrow T_\rho \PP(\calH)$, $R \mapsto R_\rho$
  is injective.
\end{lem}
\begin{proof}
  The first claim is a reformulation of the fact that the $K$-stabilizer of any $\lambda \in i\mathfrak t^*_{>0}$ is $T$, while $\mathfrak k = \mathfrak t \oplus \mathfrak r$.
  The second claim follows from the first, since $\mu_K(\rho) \in i \mathfrak t^*_{>0}$ and
  $d\mu_K(R_\rho) = \ad^*(R) \mu_K(\rho)$
  by equivariance of the moment map.
\end{proof}

\begin{lem}
\label{lem:kirwan/tangent space decomposition}
  The tangent space at $\rho$ decomposes as a direct sum
  \begin{equation*}
    T_\rho \PP(\calH) = M \oplus M^\omega
  \end{equation*}
  and the Hessian $Q$ is block-diagonal with respect to this decomposition.
\end{lem}
\begin{proof}
  For the first claim we only need to show that $M \cap M^\omega = 0$, which is standard (see, e.g., \cite[Lemma~6.7]{GuilleminSternberg82}):
  Let $R \in \mathfrak r$.
  Suppose that $R_\rho \in M^\omega$. That is,
  \begin{equation*}
    \ad^*(R) \mu_K(\rho) = d\mu_K(R_\rho) \in i \mathfrak t^*
  \end{equation*}
  according to \eqref{eq:kirwan/cross section tangent space} and equivariance of the moment map.
  But $[\mathfrak r, i\mathfrak t] \subseteq i\mathfrak r$, so it follows that $\ad^*(R) \mu_K(\rho) = 0$ and so $R = 0$ and $R_\rho = 0$, since we have just seen that the coadjoint action of $\mathfrak r$ is injective at $\lambda = \mu_K(\rho) \in i \mathfrak t^*_{>0}$ (\autoref{lem:kirwan/torus complement injective}).
  Thus $M \cap M^\omega = 0$ and $M \oplus M^\omega = T_\rho \PP(\calH)$ by dimension counting.

  For the second claim, observe that \eqref{eq:kirwan/hessian} immediately implies that
  \begin{equation*}
    Q(R_\rho, V) = \omega([\pi(-iH), \pi(R)]_\rho, V) = \omega([-iH, R]_\rho, V) = 0
  \end{equation*}
  for all $R_\rho \in M$ and $V \in M^\omega$, since $[-iH, R] \in \mathfrak r$ and hence $[-iH, R]_\rho \in M$.
\end{proof}

\begin{lem}
  \label{lem:kirwan/index hessian roots}
  The index of the Hessian at $\rho$ is equal to twice the number of positive roots $\alpha \in R_{G,+}$ such that $(\alpha, H) > 0$.
\end{lem}
\begin{proof}
  Since $Q$ is positive semidefinite on $M^\omega$ (\autoref{lem:kirwan/cross section hessian positive semidefinite}) and block-diagonal with respect to the decomposition $T_\rho \PP(\calH) = M \oplus M^\omega$ (\autoref{lem:kirwan/tangent space decomposition}), it suffices to compute the index of $Q$ on $M = \{ R_\rho : R \in \mathfrak r \}$.
  By \eqref{eq:kirwan/hessian},
  \begin{equation*}
    Q(R_\rho, S_\rho)
    = \tr \rho [[\pi(H), \pi(R)], \pi(S)]
    = (\mu_K(\rho), [[H, R], S]).
  \end{equation*}
  Since the tangent map $R \mapsto R_\rho$ is injective (\autoref{lem:kirwan/torus complement injective}), we may instead consider the form
  \begin{equation*}
    \widetilde Q(R, S) := (\mu_K(\rho), [[H, R], S])
  \end{equation*}
  on $\mathfrak r$.
  For this, observe that $\widetilde Q$ is block-diagonal with respect to $\mathfrak r = \bigoplus_{\alpha \in R_{G,+}} \mathfrak k_\alpha$, since for all $R \in \mathfrak k_\alpha$ and $S \in \mathfrak k_\beta$, $[[H, R], S] \in i\mathfrak k_{\alpha\pm\beta}$, while $\mu_K(\rho) \in i \mathfrak t^*$.
  And for each root space $\mathfrak k_\alpha$, we have that
  \begin{align*}
    &\widetilde Q(iX_\alpha, iX_\alpha)
    = -(\mu_K(\rho), [[H, X_\alpha], X_\alpha]) \\
    = &-i (\alpha, H) (\mu_K(\rho), [Y_\alpha, X_\alpha])
    = -2 (\alpha, H) (\mu_K(\rho), Z_\alpha)
  \end{align*}
  where we have used the ``Pauli matrices'' $X_\alpha, Y_\alpha, Z_\alpha$ and their commutation relations \eqref{eq:onebody/pauli}.
  Likewise, $\widetilde Q(iY_\alpha, iY_\alpha) = -2 (\alpha, H) (\mu_K(\rho), Z_\alpha)$, while $\widetilde Q(iX_\alpha, iY_\alpha) = 0$.
  Since $(\mu_K(\rho), Z_\alpha) > 0$, we conclude that the index of $Q$ is equal to twice the number of positive roots $\alpha$ with $(\alpha, H) > 0$.
\end{proof}

\subsection*{Ressayre Elements}

So far we have used the Lie algebra $\mathfrak k$ of the compact Lie group in our analysis.
We will now translate the preceding to the complexified setting.
To this end, we consider the Lie algebra $\mathfrak n_- = \bigoplus_{\alpha \in R_{G,-}} \mathfrak g_\alpha$ of the negative unipotent subgroup, which plays a role analogous to $\mathfrak r$ for states $\rho$ that are mapped into the positive Weyl chamber (compare the following with \autoref{lem:kirwan/torus complement injective}).

\begin{lem}
  \label{lem:kirwan/negative unipotent injective}
  The tangent map $\mathfrak n_- \rightarrow \calH, X \to \pi(X) \ket\psi$ is injective.
\end{lem}
\begin{proof}
  Let $E_- := \sum_{\alpha \in R_{G,+}} z_\alpha E_{-\alpha}$ be an arbitrary element in $\mathfrak n_-$.
  Since $\pi(E_{\pm\alpha})^\dagger = \pi(E_{\mp\alpha})$, we find that $\pi(E_-)^\dagger = \pi(E_+)$ with
  $E_+ := \sum_{\alpha \in R_{G,+}} \bar{z}_\alpha E_\alpha$ (the Cartan involution of $E_-$).
  Therefore,
  \begin{align*}
    \norm{\pi(E_-) \ket\psi}^2
    = &\tr \rho \, \pi(E_+) \pi(E_-)
    = \tr \rho \,\pi([E_+, E_-])
    + \norm{\pi(E_+) \ket\psi}^2 \\
    \geq &\tr \rho \, \pi([E_+, E_-]).
  \end{align*}
  But \eqref{eq:onebody/sl2} and \eqref{eq:onebody/pauli} imply that
  \begin{equation*}
    [E_+, E_-] - \sum_{\alpha \in R_{G,+}} \abs{z_\alpha}^2 Z_\alpha \in i \mathfrak r,
  \end{equation*}
  so that by using $\mu_K(\rho) \in i \mathfrak t^*_{>0}$ we find that
  \begin{equation*}
    \tr \rho \, \pi([E_+, E_-]) = \sum_{\alpha \in R_{G,+}} \abs{z_\alpha}^2 (\mu_K(\rho), Z_\alpha) > 0.
    \qedhere
  \end{equation*}
\end{proof}

In contrast to \autoref{lem:kirwan/torus complement injective}, which continues to hold true if $\mu_K(\rho)$ is mapped to the relative interior of a different Weyl chamber (since $K$ is invariant under the Weyl group), it is important in \autoref{lem:kirwan/negative unipotent injective} to choose the negative unipotent subgroup (relative to the choice of positive Weyl chamber).
For example, consider an irreducible $G$-representation $\calH = V_{G,\lambda}$ with highest weight $\lambda \in i\mathfrak t^*_{>0}$ and highest weight vector $\ket\lambda$.
Then $\mu_K(\proj\lambda) = \lambda \in i \mathfrak t^*_{>0}$ by \autoref{lem:onebody/borel weil}, and the ``lowering operators'' in $\mathfrak n_-$ act indeed injectively
(all $\pi(E_{-\alpha}) \ket\lambda$ live in different weight spaces and are non-zero, since $\pi(E_\alpha) \pi(E_{-\alpha}) \ket\lambda = \pi(Z_\alpha) \ket\lambda = (\lambda, Z_\alpha) \ket\lambda \neq 0$).
On the other hand, the ``raising operators'' in the positive nilpotent Lie algebra $\mathfrak n_+$ annihilate the highest weight vector (by definition).

\bigskip

We now decompose the Lie algebra $\mathfrak n_-$ similarly to \eqref{eq:kirwan/weight space decomposition},
\begin{equation*}
  \mathfrak n_- = \mathfrak n_-(H < 0) \oplus \mathfrak n_-(H = 0) \oplus \mathfrak n_-(H > 0),
\end{equation*}
where $\mathfrak n_-(H < 0) = \bigoplus_{\alpha \in R_{G,-} : (\alpha,H) < 0} \mathfrak g_\alpha$ is the sum of the complex root spaces with negative $H$-weight $(\alpha, H) < 0$, etc.%
\nomenclature[RN_-(H<0)]{$\mathfrak n_-(H < 0)$, \dots}{sum of negative root spaces with $(\alpha,H) < 0$, \dots}
By combining Lemmas~\ref{lem:kirwan/index hessian} and \ref{lem:kirwan/index hessian roots} and using $R_{G,-} = -R_{G,+}$, we observe that
\begin{equation}
\label{eq:kirwan/dimensions match}
  \dim_\CC \calH(H < c) = \dim_\CC \mathfrak n_-(H < 0).
\end{equation}
Note that $\pi(\mathfrak n_-(H < 0)) \calH(H = c) \subseteq \calH(H < c)$, since for any $\ket\psi \in \calH(H=c)$ and $X \in \mathfrak g_\alpha$ we have that
\begin{equation*}
  \pi(H) \pi(X) \ket\psi
  = \pi([H,X]) \ket\psi + \pi(X) \pi(H) \ket\psi
  = ((\alpha, H) + c) \pi(X) \ket\psi.
\end{equation*}
Thus we obtain the following important result:

\begin{prp}
  \label{prp:kirwan/non-trivial walls give isomorphism}
  Let $\rho = \proj\psi \in \PP(\calH)$ such that $\mu_K(\rho) \in i \mathfrak t^*_{>0}$ is a point on the facet of the moment polytope corresponding to the inequality $(-,H) \geq c$.
  Then $\ket\psi \in \calH(H = c)$ and the tangent map restricts to an isomorphism
  \begin{equation*}
    \mathfrak n_-(H < 0) \rightarrow \calH(H < c), \quad
    X \mapsto \pi(X) \ket\psi.
  \end{equation*}
\end{prp}
\begin{proof}
  The fact that $\ket\psi \in \calH(H = c)$ is just a reformulation of the selection rule.
  By the preceding discussion, the tangent map is well-defined as a map from $\mathfrak n_-(H < 0)$ to $\calH(H < c)$;
  it is injective by \autoref{lem:kirwan/negative unipotent injective} and surjective since the dimensions agree according to \eqref{eq:kirwan/dimensions match}.
\end{proof}

We now prove a partial converse to \autoref{prp:kirwan/non-trivial walls give isomorphism}.
Our proof is inspired by the argument of Ressayre \cite{Ressayre10}.

\begin{prp}
  \label{prp:kirwan/partial converse}
  Suppose there exists $\ket\psi \in \calH(H=c)$ such that the tangent map
  \begin{equation*}
    \mathfrak n_-(H < 0) \rightarrow \calH(H < c), \quad
    X \mapsto \pi(X) \ket\psi.
  \end{equation*}
  is surjective. Then $(-, H) \geq c$ is a valid inequality for the moment polytope.
\end{prp}
\begin{proof}
  Consider the smooth map
  \[N_- \times \calH(H \geq c) \rightarrow \calH, \quad (g, \ket\phi) \mapsto \Pi(g)\ket\phi.\]
  Its differential at $(1,\ket\psi)$ is the linear map
  \[\mathfrak n_- \oplus \calH(H \geq c) \rightarrow \calH, \quad (X,\ket\phi) \mapsto \pi(X) \ket\psi + \ket\phi.\]
  The assumption implies that this map is surjective.
  It follows that $\Pi(N_-) \calH(H \geq c) \subseteq \calH$ contains a small Euclidean ball around $\ket\psi$.
  In particular, any polynomial that is zero on $\Pi(N_-) \calH(H \geq c)$ is zero everywhere on $\calH$.

  We now prove the inequality.
  By the description of the moment polytope of \autoref{thm:onebody/kirwan}, it suffices to show that
  $(\lambda/k, H) \geq c$ for all highest weights $\lambda$ such that $V^*_{G,\lambda} \subseteq R_k(\calH)$.
  Recall that the highest weight of $V^*_{G,\lambda}$ is $\lambda^* = -w_0 \lambda$, where $w_0$ is the longest Weyl group element
  that flips the positive and negative roots.
  Consider a \emphindex{lowest weight vector}, i.e., a homogeneous polynomial $P \in R_k(\calH)$ that is a weight vector of weight $-\lambda$ and stabilized by $N_-$.
  Then $\pi(H) P = -(\lambda, H) P$ and the restriction of $P$ to $\calH(H \geq c)$ is non-zero by our discussion above.
  But this restriction is an element of $R_k(\calH(H \geq c)) = \Sym^k(\calH(H \geq c))^*$, the space of homogeneous polynomials of degree $k$ on $\calH(H \geq c)$. Thus all $H$-weights in $R_k(\calH(H \geq c))$ are less or equal to $-kc$.
  It follows that $-(\lambda, H) \leq -kc$, hence $(\lambda/k, H) \geq c$, as we set out to prove.
\end{proof}

We remark that \autoref{prp:kirwan/partial converse} holds unconditionally without any assumption on the dimension of the moment polytope $\Delta_K$.
We summarize our findings in the following definition and theorem:

\begin{dfn}
\label{dfn:kirwan/ressayre}
  An inequality $(-, H) \geq c$ is said to be of \emph{Ressayre type}\index{Ressayre-type inequality} if
  \begin{enumerate}
  \item $(-, H) = c$ is an affine hyperplane spanned by a subset of weights.
  \item There exists $\ket\psi \in \calH(H = c)$ such that the map
    \[\mathfrak n_-(H < 0) \rightarrow \calH(H < c), \quad  X \mapsto \pi(X) \ket\psi\]
    is an isomorphism.
\end{enumerate}
We note that the first condition is invariant under the action of the Weyl group.
\end{dfn}

\begin{thm}
\label{thm:kirwan/main}
  The moment polytope is given by
  \[\Delta_K = \{ \lambda \in \mathfrak t^*_+ : (\lambda, H) \geq c \text{ for all Ressayre-type inequalities} \}.\]
\end{thm}
\begin{proof}
  This follows directly from \autoref{lem:kirwan/admissible}, \autoref{prp:kirwan/non-trivial walls give isomorphism} and \autoref{prp:kirwan/partial converse}.
\end{proof}

\autoref{thm:kirwan/main} gives a complete description of the moment polytope of the projective space of an arbitrary $G$-representation $\calH$ (under the assumption that $\Delta_K$ is of maximal dimension).
The set of inequalities thus obtained may still be redundant (i.e., not all inequalities necessarily correspond to facets of the moment polytope).
In contrast, Ressayre's \emph{well-covering pairs} \cite{Ressayre10, Ressayre10a} characterize the facets of the moment polytope precisely.
Our characterization is also related to \cite[Theorem 2]{Brion99}, which uses algebraic geometry to characterize non-trivial faces of arbitrary codimension. Unlike \autoref{prp:kirwan/partial converse}, it relies on an assumption about lower-dimensional moment polytopes, which can in principle be obtained recursively. %

It is straightforward to enumerate all inequalities of Ressayre type.
Since there are only finitely many weights, the first condition in \autoref{dfn:kirwan/ressayre} cuts down the number of possible inequalities down to a finite list of candidates, and for each such candidate $(-,H) \geq c$, the second condition can be easily checked:
Indeed, we only need to verify that $\dim \mathfrak n_-(H < 0) = \dim \calH(H < c)$ and that the \emphindex{determinant polynomial}
\begin{equation}
\label{eq:kirwan/determinant}
  \calH(H = c) \rightarrow \CC,
  \ket\psi \mapsto \det \Big( \mathfrak n_-(H < 0) \ni X \mapsto \pi(X) \ket\psi \in \calH(H < c) \Big)
\end{equation}
is non-zero (take the determinant with respect to any fixed pair of bases).
Both steps can easily be implemented in a short computer program.

We remark that in view of \autoref{thm:onebody/kirwan} the non-vanishing of the determinant \eqref{eq:kirwan/determinant} can be understood as a statement about a moment polytope of $\PP(\calH(H=c))$. In this way, further necessary conditions can be extracted that cut down the list of candidate inequalities. We will comment on this aspect and related results in a forthcoming paper.

\section{Examples}
\label{sec:kirwan/examples}

We now illustrate the method by considering some of the examples discussed at the beginning of the chapter.

\subsection*{Qubits}

We will start by showing that the solution of the one-body quantum marginal problem for $n$ qubits is indeed given by the polygonal inequalities \eqref{eq:kirwan/polygonal}, which we recall are given by
\begin{equation}
\label{eq:kirwan/polygonal first}
  \lambda_{1,1} + \dots + \lambda_{n-1,1} \leq (n-2) + \lambda_{n,1},
\end{equation}
and its permutations.

Let $K = \SU(2)^{\times n}$, $G = \SL(2)^{\times n}$ its complexification, and $\pi \colon \mathfrak g \rightarrow \mathfrak{gl}(\calH)$ the infinitesimal representation of the Lie algebra on $\calH = (\CC^2)^{\otimes n}$, the Hilbert space of $n$ qubits.
The Lie algebra $i \mathfrak t$ is spanned by the Pauli $\sigma_z$-matrices $Z_1, \dots, Z_n$ of the individual $\mathfrak{sl}(2)$ summands, and each $Z_k$ acts on the respective factor of the Hilbert space by the Pauli $\sigma_z$-matrix $\matrix{1 & \\ & -1}$.
Likewise, there are $n$ positive roots $\alpha_1, \dots, \alpha_n$ and they satisfy $(\alpha_k, Z_l) = 2 \delta_{kl}$.

By using that $\lambda_{k,1} + \lambda_{k,2} = 1$ for each qubit, we find that \eqref{eq:kirwan/polygonal first} is equivalent to the following inequality for the moment polytope $\Delta_K$,
\begin{equation*}
  (-, \underbrace{- Z_1 - \dots - Z_{n-1} + Z_n}_{=: H}) \geq 2-n,
\end{equation*}
which we will now prove by using the criterion of \autoref{prp:kirwan/partial converse}.

For this, let $\ket{0}$ and $\ket{1}$ denote the standard basis vectors of $\CC^2$ and $\ket{\vec x}$ the corresponding product basis vectors of $\calH$, where $\vec x \in \{0, 1\}^{\times n}$.
Then $\pi(Z_k) \ket{\vec x} = (1-2 x_k) \ket{\vec x}$, so that each basis vector $\ket{\vec x}$ is a weight vector whose weight is given by $(\omega, Z_k) = (1 - 2 x_k)$ for all $k=1,\dots,n$.
In particular, the eigenvalue of $\pi(H)$ is $(\omega, H) = 2-n + 2(x_1 + \dots + x_{n-1} - x_n)$, so that
\begin{align}
  \calH(H = 2-n)
  &= \Span_\CC \{ \ket{\vec x} : x_1 + \dots + x_{n-1} - x_n = 0 \} \nonumber \\
  &= \Span_\CC \{ \ket{10\dots{}01}, \dots, \ket{0\dots011}, \ket{0\dots{}0} \} \label{eq:kirwan/qubits Hnull} \\
  \calH(H < 2-n) &= \CC \ket{0\dots{}01}. \nonumber
\end{align}
On the other hand, there is precisely one negative root with $(-,H) < 0$, namely $-\alpha_n$. Therefore,
\begin{equation*}
  \mathfrak n_-(H < 0) = \mathfrak g_{-\alpha_n} = \CC E_{-n},
\end{equation*}
where the generator $E_{-n}$ acts as the ``lowering operator'' $\ketbra 1 0$ on the $n$-th tensor factor of $\calH$.
But then $\pi(E_{-n}) \ket{0\dots{}0} = \ket{0\dots{}01}$, so that the map
\[\mathfrak n_-(H < 0) \rightarrow \calH(H < 2-n), \quad X \mapsto \pi(X) \ket\psi\]
is indeed an isomorphism for $\ket\psi = \ket{0\dots{}0} \in \calH(H = 2-n)$.
Thus \autoref{prp:kirwan/partial converse} shows that the polygonal inequality \eqref{eq:kirwan/polygonal first} is indeed valid.
We remark that the $n$ weights corresponding to \eqref{eq:kirwan/qubits Hnull} span the hyperplane $(-, H) = 2-n$ and hence the inequality is of Ressayre type---but we did not need this to apply the proposition.
The other polygonal inequalities follow from the above since the solution to the one-body quantum marginal problem is symmetric under permutation of the qubits.

To see that there are no further constraints, we could now verify that the polygonal inequalities imply all other Ressayre-type inequalities (which, according to \autoref{thm:kirwan/main}, characterize the moment polytope completely).
In the case at hand, we observe instead
that for each $k \leq n-2$ as well as for $k=n$ the quantum states
\[\ket 0^{\otimes k} \otimes \left( \ket 0^{\otimes n-k} + \ket 1^{\otimes n-k} \right) / \sqrt 2\]
have local eigenvalues
\[\lambda_{1,1} = \dots = \lambda_{k,1} = 1, \quad \lambda_{k+1,1} = \dots = \lambda_{n,1} = 0.5,\]
and likewise for their permutations.
These points are just the vertices of the convex polytope cut out by the polygonal inequalities \eqref{eq:kirwan/polygonal}, which is therefore equal to the moment polytope.

This idea of producing an \emph{inner approximation}\index{moment polytope!inner approximation} from the local eigenvalues of suitably constructed quantum states -- or from the normalized highest weights $\lambda/k$ of concretely exhibited representations $V_{G,\lambda}^* \subseteq R_k(\XX)$ \cite{KlyachkoAltunbulak08, Altunbulak08} -- and then comparing with an \emph{outer approximation}\index{moment polytope!outer approximation} defined by a subset of valid inequalities (obtained, e.g., from \autoref{prp:kirwan/partial converse}) is rather useful to compute moment polytopes in practice.

\subsection*{Mixed States of Two Qubits}

We now consider the mixed-state one-body quantum marginal problem for two qubits.
We will focus on the last constraint in \eqref{eq:kirwan/bravyi},
\begin{equation*}
  \abs{\lambda_{A,1} - \lambda_{B,1}} \leq \min \{ \lambda_{AB,1} - \lambda_{AB,3}, \lambda_{AB,2} - \lambda_{AB,4} \},
\end{equation*}
which by symmetry can be reduced to the following two linear inequalities
\begin{equation}
\label{eq:kirwan/bravyi difficult up to symmetry}
\begin{aligned}
  \lambda_{A,1} - \lambda_{B,1} &\leq \lambda_{AB,1} - \lambda_{AB,3}, \\
  \lambda_{A,1} - \lambda_{B,1} &\leq \lambda_{AB,2} - \lambda_{AB,4}.
\end{aligned}
\end{equation}
In \cite{Bravyi04}, the inequalities are proved by a ``formidable'' two-page calculation that is tailored towards the two-qubit scenario.
In contrast we will obtain \eqref{eq:kirwan/bravyi difficult up to symmetry} completely mechanically by using the general machinery developed in \autoref{sec:kirwan/facets}.

For this, we consider its equivalent formulation in terms of pure states on $\calH = \CC^2 \otimes \CC^2 \otimes \CC^4$.
Let $K = \SU(2) \times \SU(2) \times \SU(4)$ and $G = \SL(2) \times \SL(2) \times \SL(4)$.
The roots are the union of the roots of the individual factors, and we will denote them by $\alpha_{A,ij}$, $\alpha_{B,ij}$ and $\alpha_{C,ij}$ according to the notation and conventions in \autoref{tab:onebody/unitary group}.
Thus a root is positive if $i < j$ and negative if $i > j$.
We write $E_{A,ij}$, etc., for the generators of the root spaces $\mathfrak g_{\alpha_{A,ij}}$; they act ``raising and lowering operators'' $\ketbra i j$ on the respective tensor factor.
Similarly, the corresponding generators of $i \mathfrak t$ are denoted by $Z_{A,ij}$, etc.; they act as the diagonal matrices $\proj i - \proj j$ on the respective tensor factor of the Hilbert space.
Using these definitions, the first inequality in \eqref{eq:kirwan/bravyi difficult up to symmetry} can be written as the following linear inequality for the moment polytope $\Delta_K$:
\begin{equation*}
  (\lambda, \underbrace{-Z_{A,12} + Z_{B,12} + 2 Z_{C,13}}_{=: H}) \geq 0
\end{equation*}
Finally, let $\ket{ijk}$ denote the standard product basis vectors of $\calH = \CC^2 \otimes \CC^2 \otimes \CC^4$ ($i,j=1,2$, $k=1,\dots,4$).
Then,
\begin{align}
  \calH(H=0) &= \Span_\CC \{ \ket{112}, \ket{114}, \ket{121}, \ket{213}, \ket{222}, \ket{224} \}, \nonumber \\
  \calH(H<0) &= \Span_\CC \{ \ket{113}, \ket{122}, \ket{124}, \ket{223}, \ket{123} \}, \label{eq:kirwan/bravyi Hneg}
\end{align}
while
\begin{equation}
  \mathfrak n_-(H<0) = \Span_\CC \{ E_{B,21}, E_{C,21}, E_{C,41}, E_{C,32}, E_{C,31} \}. \label{eq:kirwan/bravyi Nneg}
\end{equation}
It remains to verify that the map
\begin{equation}
  \label{eq:kirwan/bravyi isomorphism}
  \mathfrak n_-(H<0) \rightarrow \mathfrak \calH(H<0), X \mapsto \pi(X) \ket\psi
\end{equation}
is an isomorphism for some $\ket\psi \in \calH(H=0)$.
We can do so completely mechanically by computing the determinant polynomial \eqref{eq:kirwan/determinant} with respect to the bases in \eqref{eq:kirwan/bravyi Hneg} and \eqref{eq:kirwan/bravyi Nneg}.
Using that each $E_{X,ij}$ acts by $\ketbra i j$ on the corresponding tensor factor, we readily find that it is given by
\begin{equation*}
  \det \Matrix{
    0 & 0 & 0 & \psi_{112} & 0 \\
    \psi_{112} & \psi_{121} & 0 & 0 & 0 \\
    \psi_{114} & 0 & \psi_{121} & 0 & 0 \\
    \psi_{213} & 0 & 0 & \psi_{222} & 0 \\
    0 & 0 & 0 & 0 & \psi_{121}
  }
  = - \psi_{121}^3 \psi_{112} \psi_{213} \neq 0
\end{equation*}
and thus is indeed non-zero.
For example, $\ket\psi = \ket{121} + \ket{112} + \ket{213} \in \calH(H=0)$ is a choice for which \eqref{eq:kirwan/bravyi isomorphism} is an isomorphism.
Thus \autoref{prp:kirwan/partial converse} shows that the first inequality in \eqref{eq:kirwan/bravyi difficult up to symmetry} is indeed a valid inequality for the one-body marginal problem for mixed states of two qubits.
The second inequality can be established in a completely identical fashion.

\subsection*{Three Qutrits}

For our next example, we consider the one-body quantum marginal problem for pure states of three qubits, $\calH = \CC^3 \otimes \CC^3 \otimes \CC^3$.
We will focus on the inequality
\begin{equation}
\label{eq:kirwan/franz h7}
  (\lambda_{A,1} + 2 \lambda_{A,2}) +
  (2 \lambda_{B,1} + \lambda_{B,2})
  - (4 \lambda_{C,1} + 3 \lambda_{C,2} + 2 \lambda_{C,3}) \leq 0
\end{equation}
corresponding to the facet $h_7$ in \cite[Proposition~5.1]{Franz02}.

The significance of this facet is that it is neither trivial nor does it contain the vertex $\lambda_{A,1} = \lambda_{B,1} = \lambda_{C,1} = 1$ of the moment polytope that corresponds to product states $\ket{000}$ -- unlike in the case of $n$ qubits, where the polygonal inequalities \eqref{eq:kirwan/polygonal} are saturated for product states (cf.\ \autoref{fig:kirwan/three-qubits}).
As was pointed out in \cite[Remark~5.3]{Franz02}, this shows that in general the moment polytope is \emph{not} determined by the trivial facets together with the local cone at the point corresponding to the product state (more generally, the local cone at the highest weight of an irreducible representation).

Let $K = \SU(3) \times \SU(3) \times \SU(3)$ and $G = \SL(3) \times \SL(3) \times \SL(3)$ with their tensor product action on the Hilbert space $\calH$. Using the same notation as in the preceding example, we find that \eqref{eq:kirwan/franz h7} can be written in the form
\begin{equation*}
  (\lambda, \underbrace{-Z_{A,23} - Z_{B,13} + Z_{C,13}}_{=: H}) \geq -1.
\end{equation*}
It can be readily computed that
\begin{align*}
  \calH(H=-1) &= \Span_\CC \{ \ket{112}, \ket{123}, \ket{211}, \ket{222}, \ket{233}, \ket{313} \}, \\
  \calH(H<-1) &= \Span_\CC \{ \ket{113}, \ket{212}, \ket{213}, \ket{223} \}, \\
  \mathfrak n_-(H<0) &= \Span_\CC \{ E_{A,21}, E_{C,21}, E_{C,31}, E_{C,32} \}.
\end{align*}
The determinant polynomial is
\begin{equation*}
  \det \Matrix{
    0 & 0 & 0 & \psi_{112} \\
    \psi_{112} & \psi_{211} & 0 & 0 \\
    0 & 0 & \psi_{211} & 0 \\
    \psi_{123} & 0 & 0 & \psi_{222}
  }
  = -\psi_{112} \psi_{123} \psi_{211}^2 \neq 0
\end{equation*}
and thus \eqref{eq:kirwan/franz h7} is indeed a valid inequality.

\subsection*{Three Fermions with Total Spin $1/2$}
\index{fermions}

Our last example is motivated by quantum chemistry.
We consider three fermions with $d$-dimensional orbital degree of freedom and spin $1/2$, as described by the Hilbert space $\Alt^3(\CC^d \otimes \CC^2)$ \cite[\S{}3.1]{KlyachkoAltunbulak08}.
Let $G = \SL(d) \times \SL(2)$ and $K = \SU(d) \times \SU(2)$.
Then the moment map sends a pure state $\rho$ of three fermions to the pair $(\gamma_O, \gamma_S)$, where $\gamma_O$ is the first-order density matrix of the orbital degrees of freedom and $\gamma_S$ likewise for the spin. Thus the moment polytope encodes the constraints between the \emph{orbital occupation numbers}\index{occupation numbers!orbital} and \emph{spin occupation numbers}\index{occupation numbers!spin}.
The anti-symmetric subspace is not an irreducible $G$-representation but rather decomposes as follows:
\begin{align*}
  \Alt^3(\CC^d \otimes \CC^2)
  = &V^d_{\yng(2,1)} \otimes V^2_{\yng(2,1)} \;\oplus\; V^d_{\yng(1,1,1)} \otimes V^2_{\yng(3)}
\end{align*}
(In analogy to \eqref{eq:onebody/schur-weyl bipartite}, each Young diagram is paired with its transpose.)
The different summands in the decomposition correspond to different sectors of total spin $J=1/2, 3/2$.
Now suppose that we know in addition that the global state is a pure state in the sector for total spin $J=1/2$:
\begin{equation*}
   \calH = V^d_{\yng(2,1)} \otimes V^2_{\yng(2,1)} \cong V^d_{\yng(2,1)} \otimes \CC^2.
\end{equation*}
Then the problem of determining the relation between the orbital and spin occupation numbers amounts to computing the moment polytope for the $G$-representation $\calH$.
This illustrates how representations other than those in \autoref{tab:onebody/summary} may enter the picture due to physical considerations.
In \cite[\S{}6.1]{KlyachkoAltunbulak08} it was observed that there are five inequalities which are ``apparently independent'' of the orbital dimension $d$, as they verified for $d=4,5$. Their first equation is
\begin{equation}
\label{eq:kirwan/klyachko altunbulak first}
  \lambda_1 - \lambda_2 \leq 1 + \mu_2
\end{equation}
where $\lambda_1 \geq \dots \geq \lambda_d$ denote the natural occupation numbers of the first-order orbital density matrix $\gamma_O$, normalized to trace $3$, and $\mu_1 \geq \mu_2$ the eigenvalues of the spin density matrix $\rho_S = \gamma_S/3$, normalized to trace $1$.

We will now prove \eqref{eq:kirwan/klyachko altunbulak first} for arbitrary $d$.
For this, we will use the following description of the irreducible $\GL(d)$-representation $V^d_{\Ytwoone}$ (see, e.g., \cite{Fulton97}):
Let $V$ denote the vector space with one basis vector $\yket{ab,c}$ for each \emph{filling}\index{Young diagram!filling} of the Young diagram $\Ytwoone$ with numbers $a,b,c \in \{1,\dots,d\}$.
Then $V^d_{\Ytwoone}$ can be identified with the quotient of $V$ by the relations
\begin{equation}
\label{eq:kirwan/21 relations}
  \yket{ab,a} = 0,
  \quad
  \yket{ab,c} = -\yket{cb,a}
  \quad\text{and}\quad
  \yket{ab,c} = \yket{ba,c} + \yket{ac,b}.
\end{equation}
It is not hard to see that a basis of $V^d_{\Ytwoone}$ is given by those $\yket{ab,c}$ with $a \leq b$, $a < c$ (in the literature, these fillings are known as the \emph{semistandard Young tableaux}\index{Young tableau!semistandard} of shape $\Ytwoone$).
Each vector $\yket{ab,c}$ is a weight vector of weight $\omega(H) = H_{a,a} + H_{b,b} + H_{c,c}$, and the generators $E_{ij}$ of $\mathfrak g_{ij}$ send $\yket{ab,c}$ to the sum of all vectors that arise by replacing a single $j$ by $i$. For example,
\begin{align}
  E_{i2} \yket{11,2} &= \yket{11,i}, \label{eq:kirwan/21 root action} \\
  E_{21} \yket{11,2} &= \yket{21,2} + \yket{12,2} = \yket{12,2}, \nonumber
\end{align}
where we have used the first relation in \eqref{eq:kirwan/21 relations}.

Now consider the group $G = \SL(d) \times \SL(2)$ and its representation on the Hilbert space $\calH = V^d_{\Ytwoone} \otimes \CC^2$.
Using the same notation as before, \eqref{eq:kirwan/klyachko altunbulak first} amounts to the following inequality for the moment polytope:
\begin{equation*}
  (-, \underbrace{2 Z_{A,21} + Z_{B,21}}_{=: H}) \geq -3,
\end{equation*}
where we have used that $\mu_1 - \mu_2 = 1 - 2 \mu_2$. Furthermore,
\begin{align*}
  \calH(H = -3) &\ni \yket{11,2} \otimes \ket 1 \\
  \calH(H < -3) &= \Span_\CC \{ \yket{11,3} \otimes \ket 1, \dots, \yket{11,d} \otimes \ket 1 \} \\
  \mathfrak n_-(H < 0) &= \Span_\CC \{ E_{A,32}, \dots, E_{A,d2} \}
\end{align*}
But then \eqref{eq:kirwan/21 root action} shows that the tangent map $\mathfrak n_-(H < 0) \rightarrow \calH(H < -3)$  is an isomorphism for $\ket\psi = \yket{11,2} \otimes \ket 1$.
By \autoref{prp:kirwan/partial converse}, it follows that the inequality \eqref{eq:kirwan/klyachko altunbulak first} is true for all $d$.
In the same way the other four inequalities asserted in \cite[\S{}6.1]{KlyachkoAltunbulak08} can be proved for arbitrary $d$; but this will be presented elsewhere.

\bigskip

We have developed a small computer program using \textsc{Sage} \cite{Steinothers13} that uses the methods of this chapter to automatically enumerate and verify inequalities for moment polytopes of projective spaces by using the method developed in this chapter.

\section{Discussion}
\label{sec:kirwan/discussion}

In this chapter we have described a solution of the one-body quantum marginal problem that is based on computing the moment polytope for the projective space of global pure states, to which we can always reduce by purification (\autoref{ch:onebody}).
Our geometric approach avoids various technicalities and instead reduces the computation of the moment polytopes to a largely combinatorial question about roots and weights.
This is contrast to the method of Klyachko and Berenstein--Sjamaar \cite{Klyachko04, BerensteinSjamaar00}, which relies on working with complete flag varieties (corresponding to global states with non-degenerate spectrum), where intersections can be detected cohomologically via Schubert calculus.
For example, in the latter approach the pure-state solution for three qutrits, $\calH = \CC^3 \otimes \CC^3 \otimes \CC^3$, that we have discussed in the examples had to be obtained from the solution of the mixed-state problem for $\CC^3 \otimes \CC^3$, which in turn is equivalent to the pure-state problem for $\CC^3 \otimes \CC^3 \otimes \CC^9$!
Even though we directly work with the lower-dimensional models, it is still a computational challenge to compute the marginal inequalities in some scenarios of physical interest---in particular, in fermionic scenarios.
It is part of ongoing work in progress to obtain additional simplifications that should make the solution of the marginal problem for quantum systems of moderate dimensions computationally feasible.
Such efforts might also shine further light on the character of the marginal constraints in the limit of large dimensions, where surprisingly little is known.

\bigskip

In \autoref{sec:slocc/distillation} we describe a rather different approach to the computation of moment polytopes which is based on Kirwan's gradient flow \cite{Kirwan84}.
The resulting numerical algorithm is probabilistic in nature but has the advantage of being applicable to more general projective subvarieties than projective space, such as the SLOCC entanglement classes introduced in \autoref{ch:slocc}.
\chapter{Multipartite Entanglement}
\label{ch:slocc}

In this chapter we consider entanglement in multipartite quantum systems. %
We show that in the case of pure states, features of the global entanglement can already be extracted from local information alone.
This is achieved by associating with any given class of entanglement an entanglement polytope---a geometric object which characterizes the
one-body marginals compatible with that class.
In this way we obtain local witnesses for the multipartite entanglement of a global pure state.
Our approach is applicable to systems of arbitrary size and statistics, and it can be generalized to states affected by low levels of noise.
We also describe a gradient flow technique that can be used for entanglement distillation and the computation of moment polytopes.

The results in this chapter have been obtained in collaboration with Matthias Christandl, Brent Doran, David Gross and Konstantin Wernli, and they have appeared in \cite{WalterDoranGrossEtAl13}.

\section{Summary of Results}

Entanglement is a uniquely quantum mechanical feature. It is responsible for
fundamentally new effects -- such as quantum non-locality -- and
constitutes the basic resource for concrete tasks such as quantum
computing \cite{Vidal03} and interferometry beyond the standard limit
\cite{LeibfriedBarrettSchaetzEtAl04, GiovannettiLloydMaccone04}.
Considerable efforts have been directed at obtaining a systematic
characterization of multi-particle entanglement; however, our understanding remains limited as the complexity of entanglement scales exponentially with the number of particles \cite{DuerVidalCirac00}.

In this chapter, we show that, for pure quantum states, single-particle information alone can serve as a powerful witness to multipartite entanglement.
In fact, we find that a finite list of linear inequalities characterizes the eigenvalues of the one-body reduced density matrices in any given class of entanglement.
Their violation provides a criterion for witnessing multipartite entanglement that (i) only requires access to a linear number of degrees of freedom, (ii) applies universally to quantum systems of arbitrary size and statistics, and (iii) distinguishes among many
important classes of entanglement, including genuine multipartite entanglement.
Geometrically, these inequalities cut out a hierarchy of polytopes, which captures all information about the global pure-state entanglement deducible from local information alone.
Our methods are sufficiently robust to be applicable to situations where the state is affected by low levels of noise.

\bigskip

Formally, a pure state $\rho = \proj\psi$ is said to be \emph{entangled}\index{entanglement!pure-state} if it cannot be written as a tensor product $\ket\psi \neq \ket{\psi_1} \otimes \dots \otimes \ket{\psi_n}$ \cite{NielsenChuang04}.
Two states can be considered to belong to the same \emph{entanglement class}\index{entanglement class!SLOCC} if they can be converted into each other with finite probability of success using local operations and classical communication (stochastic LOCC, or SLOCC) \cite{BennettPopescuRohrlichEtAl00, DuerVidalCirac00}.
In physical terms, this corresponds to performing arbitrary quantum operations on each of the individual particles, where each operation may depend on outcomes of previous measurements on different particles and where we may post-select on measurement outcomes.
We give a precise definition in \autoref{sec:slocc/classification}.
For small systems, these entanglement classes are well-understood.
In the simplest scenario of three qubits, there exist two classes of genuinely entangled states of strikingly different nature:
the first contains the famous \emph{Greenberger--Horne--Zeilinger (GHZ) state}\index{GHZ state}
$(\ket{000} + \ket{111}) / \sqrt 2$,
which exhibits a particularly strong form of quantum correlations \cite{GreenbergerHorneZeilinger89};
the second contains the \emphindex{W state}
$(\ket{100} + \ket{010} + \ket{001}) / \sqrt 3$
\cite{DuerVidalCirac00}.
Whereas states in the W class can be approximated to arbitrary precision by states from the GHZ class, the converse is not
true---implying stronger entanglement of the GHZ class \cite{DuerVidalCirac00}.
Already for four particles there exist infinitely many entanglement classes \cite{VerstraeteDehaeneDeEtAl02}, and the number of parameters required to determine the class grows exponentially with the particle number \cite{DuerVidalCirac00}.
As a result, only sporadic results have been obtained for larger systems, despite the enormous amount of literature dedicated to the problem.
Mathematically, the characterization of SLOCC entanglement classes can be formulated in terms of invariant theory, studied since the 19th century---Cayley's hyperdeterminant, e.g., appears as the 3-tangle \cite{CoffmanKunduWooters00}.
Similar techniques underpin modern developments, such as the geometric complexity theory approach to the $\Pclass$ vs.~$\NP$ problem \cite{BuergisserLandsbergManivelEtAl11} (see \autoref{sec:mul/summary}).

\bigskip

\begin{figure}
  \centering
  \includegraphics[width=0.9\linewidth]{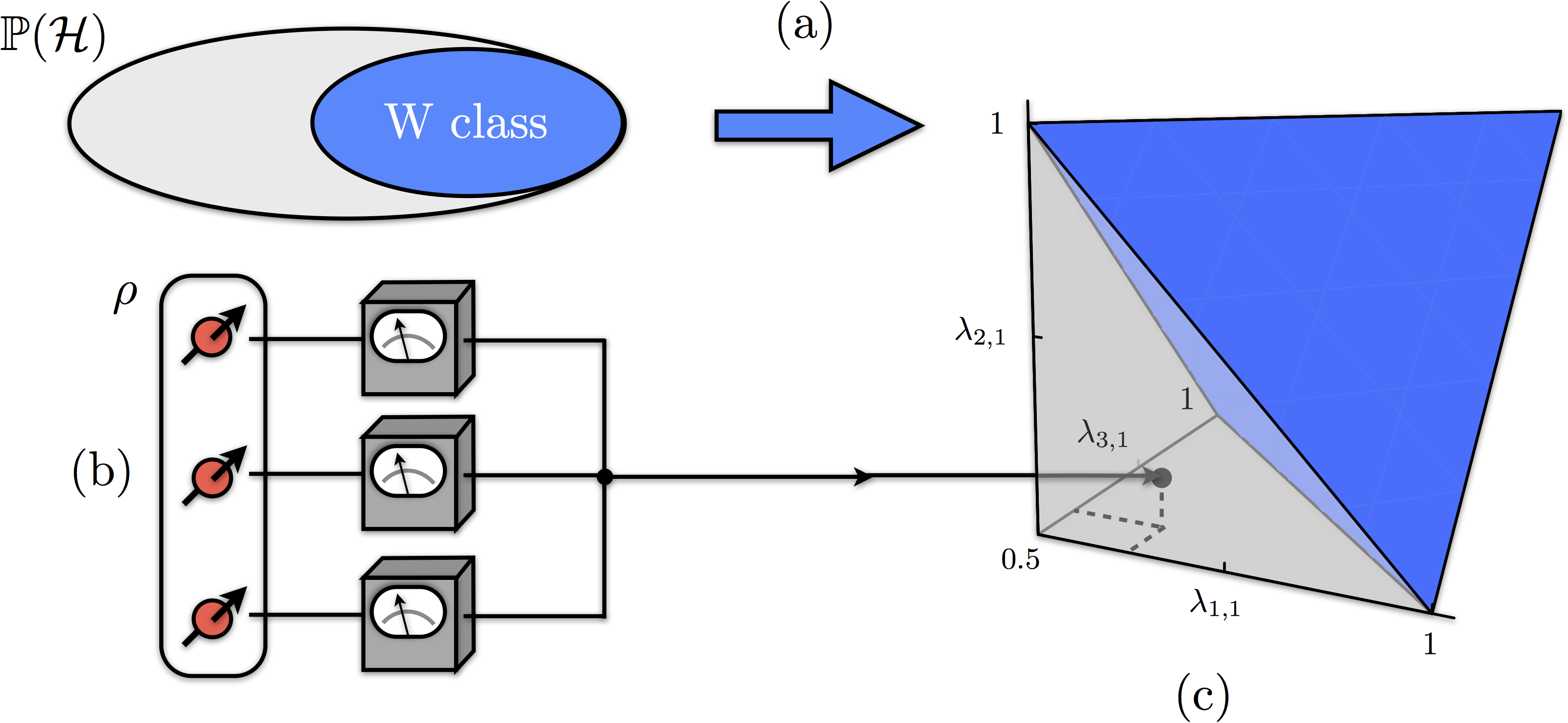}
  \caption[Entanglement polytopes as witnesses]{%
    \emph{Entanglement polytopes as witnesses (illustrated for three qubits).}
    (a) An entanglement polytope contains all possible local eigenvalues of states in the entanglement class and its closure (the W class and its polytope are shown in blue).
    (b) For a sufficiently pure quantum state $\rho$, local tomography is performed to determine its local eigenvalues.
    (c) The indicated eigenvalues are not compatible with the W polytope, hence $\rho$ must have GHZ-type entanglement.}
  \label{fig:slocc/witness}
\end{figure}

Our approach to multipartite entanglement is based on establishing a connection to the one-body quantum marginal problem %
(\autoref{sec:slocc/entanglement polytopes}).
The crucial observation is that the one-body reduced density matrices $\rho_1, \dots, \rho_n$ or, equivalently, their eigenvalues $\vec\lambda_1, \dots, \vec\lambda_n$ alone can already give considerable information about the entanglement of the global state, provided that it is pure.
To make this precise, we consider the set of local eigenvalues $\vec\lambda = (\vec\lambda_1, \dots, \vec\lambda_n)$ of the states in the closure $\overline{\calX}$ of a given entanglement class $\calX$.
Surprisingly, this set also forms a convex polytope (i.e., it is the convex hull of finitely many such vectors), and we call it the \emphindex{entanglement polytope} $\Delta_{\calX}$ of the class.
Entanglement polytopes immediately lead to a local criterion for witnessing multipartite entanglement:
If the collection of eigenvalues $\vec\lambda$ of the one-body reduced density matrices of a pure quantum state $\rho = \proj\psi$ does not lie in an entanglement polytope, then the given state cannot belong to the closure of the corresponding entanglement class $\calX$ (\autoref{fig:slocc/witness}):
\begin{equation}
  \label{eq:slocc/criterion}
  \vec\lambda = (\vec\lambda_1, \dots, \vec\lambda_n) \notin \Delta_\calX
  \,\Longrightarrow\,
  \rho \notin \overline{\calX}.
\end{equation}
In other words, the criterion allows us to witness the presence of a highly entangled state by showing that its local eigenvalues are incompatible with all less-entangled classes.
This way of reasoning is similar to the exclusion of a local hidden variable model by witnessing the violation of a Bell inequality.
Strikingly, there are always only finitely many entanglement polytopes, and they naturally form a hierarchy:
if a state in some class $\calX$ can be approximated arbitrarily well by states from $\calY$ then $\Delta_\calX \subseteq \Delta_\calY$.
This reflects geometrically the fact that states in the second class are more powerful for quantum information processing.

To describe $\Delta_\calX$ mathematically and establish these claims, we work in the framework of \autoref{ch:onebody}.
Using the characterization of SLOCC operations as invertible local operators $g_1 \otimes \dots \otimes g_n$ \cite{DuerVidalCirac00}, we find that the closure of an entanglement class is a projective subvariety of the space of pure states.
Thus $\Delta_\calX$ can be identified with its moment polytope, and the above-mentioned properties follow from the general theory.
We explain how $\Delta_\calX$ can in principle be computed using computational invariant theory \cite{DerksenKemper02}.

\bigskip

We now illustrate the method with some examples taken from \autoref{sec:slocc/examples}.
For qubit systems, each one-body reduced density matrix $\rho_k$  has two eigenvalues, which are non-negative and sum to one; hence its spectrum is completely characterized by the maximal eigenvalue $\lambda_{k,1}$, which can take values in the interval $[0.5, 1]$.
In the case of three qubits, we may therefore regard the entanglement polytopes as subsets of three-dimensional space.
There are two full-dimensional polytopes, as has already been observed in \cite{HanZhangGuo04}: one for the W class (the upper pyramid in \autoref{fig:slocc/witness}) and the other for the GHZ class (the entire polytope, i.e., the union of both pyramids).
The tip of the upper pyramid constitutes a polytope by itself, indicating a product state.
Three further one-dimensional polytopes are given by the edges emanating from this vertex.
They correspond to the three possibilities of embedding an EPR pair $(\ket{00} + \ket{11})/\sqrt{2}$ into three qubits.
Thus, eigenvalues in the interior of the polytope are compatible only with the W and GHZ classes, i.e., genuine three-partite entanglement.
If the eigenvalues lie in the lower pyramid,
$\lambda_{1,1} + \lambda_{2,1} + \lambda_{3,1} < 2$,
then by \eqref{eq:slocc/criterion} the state cannot be contained in the closure of the $W$ class---we have witnessed GHZ-type entanglement.

In systems of 4 qubits, there exist 9 infinite families of entanglement classes, each described by up to %
three complex parameters \cite{VerstraeteDehaeneDeEtAl02} that are not directly accessible; arguably, the complete classification is too detailed to be practical.
In contrast, entanglement polytopes strike an attractive balance between coarse-graining and preserving structure (\autoref{fig:slocc/four qubits extract}):
Up to permutations, there are 12 %
entanglement polytopes, 7 of which are full-dimensional and correspond to distinct types of genuine four-partite entanglement. One example is the 4-qubit W class: in complete analogy to the previous case, its polytope is an ``upper pyramid'' of eigenvalues that fulfill $\lambda_{1,1} + \lambda_{2,1} + \lambda_{3,1} + \lambda_{4,1} \geq 3$.

\begin{figure}
  \centering
  \fourqubitcutsheader

  \includegraphics[width=0.9\linewidth]{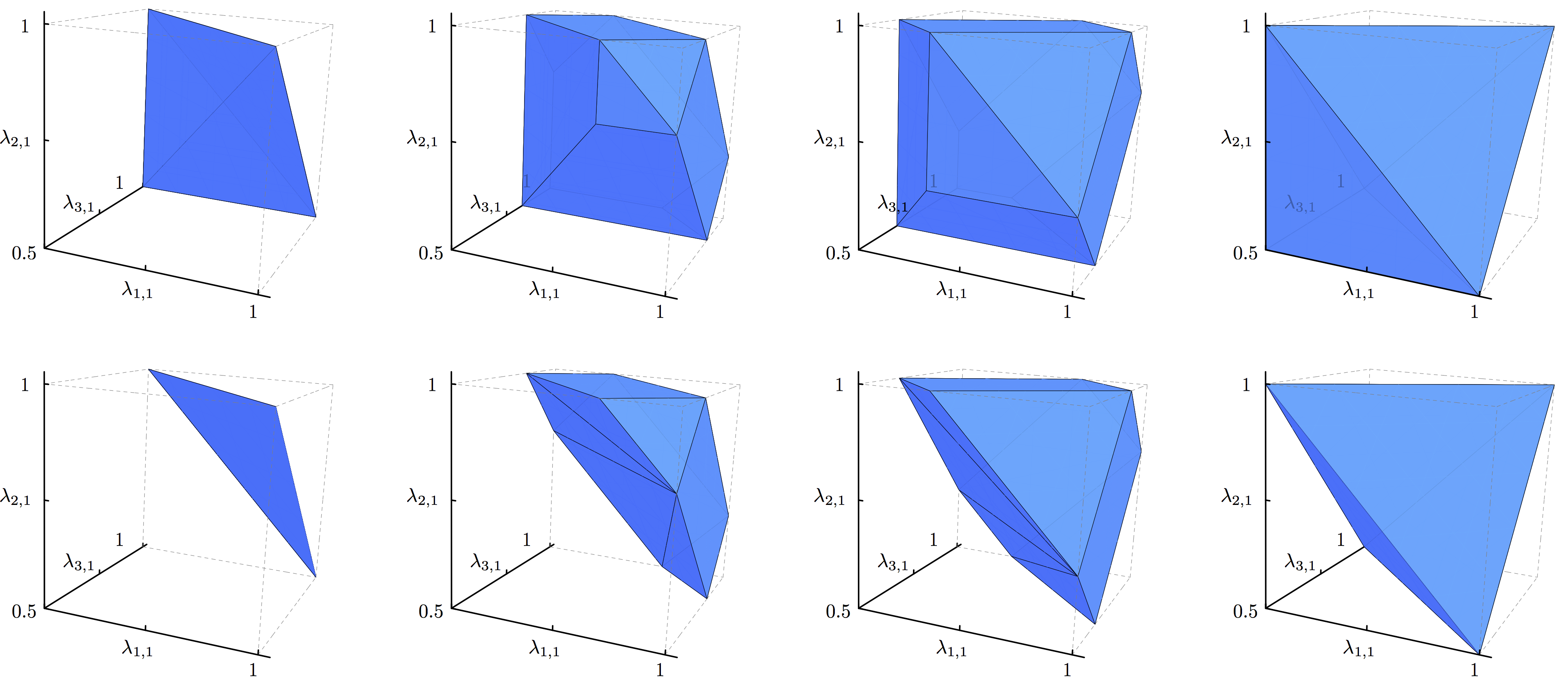}
  \caption[Cross-sections of two entanglement polytopes for four qubits]{%
    \emph{Cross-sections of two entanglement polytopes for four qubits.}
    Each row shows cross-sections of a four-dim\-en\-sio\-nal entanglement polytope for four fixed values of $\lambda_{4,1}$.
    The first row corresponds to the entanglement class $L_{a_4}$ for $a=0$ in the characterization of \cite{VerstraeteDehaeneDeEtAl02}, represented by the state $(i \ket{0001} + \ket{0110} - i \ket{1011})/\sqrt 3$.
    The second row corresponds to the four-qubit W class, $(\ket{1000} + \dots + \ket{0001})/2$.
    One can clearly identify the ``upper-pyramid form'' explained in the text.
    Properties that had previously been computed algebraically %
    can be read off directly:
    E.g., the final column corresponds to the situation when the fourth qubit has been projected onto a pure state;
    as apparent from the polytopes, the state of the remaining three sites is generically of GHZ type in the first row, and of W type in the second row.
    See \cite{Walter12} for an interactive visualization of all four-qubit entanglement polytopes.}
  \label{fig:slocc/four qubits extract}
\end{figure}

\bigskip

In \autoref{sec:slocc/examples} we give further details on these computations.
We also discuss the notion of \emphindex{genuinely multipartite entangled} states, which are of particular interest \cite{GuehneTothBriegel05}.
These are pure states which do not factorize with respect to any partition of the system into two sets of subsystems.
We show for arbitrary systems that the entanglement polytopes of the \emphindex{biseparable} states (i.e., the states that do factorize) do not account for all possible eigenvalues.
Therefore, the presence of genuine multipartite entanglement in a pure quantum state can be witnessed by checking that the local eigenvalues do not lie in any biseparable polytope.
We can also obtain more quantitative information about the multipartite entanglement of a quantum state, e.g., by witnessing genuine $k$-partite entanglement \cite{GuehneTothBriegel05} by using a generalization of the method sketched above.
Entanglement polytopes can also be constructed for quantum systems composed of bosons or fermions \cite{WalterDoranGrossEtAl13}.

\bigskip

In \autoref{sec:slocc/distillation} we then consider the \emphindex{linear entropy of entanglement}
$E(\rho) = 1 - \frac{1}{n}\sum_{j=1}^n \tr \rho_j^2$
\cite{ZurekHabibPaz93, BarnumKnillOrtizEtAl04, BoixoMonras08}, used, e.g., in metrology \cite{FuruyaNemesPellegrino98}.
Entanglement polytopes allow us to bound the maximal linear entropy of entanglement distillable by SLOCC operations:
Since $1-E(\rho)$ corresponds to the Euclidean length of the vector $\vec\lambda$ of local eigenvalues, shorter vectors imply more entanglement.
In particular, quantum states of maximal entropy of entanglement in a class $\calX$ map to the point of minimal distance to the origin in the entanglement polytope $\Delta_{\calX}$.
Therefore, if the local eigenvalues of a given state lie only in polytopes with small distance to the origin, a high amount of entanglement can be distilled.
We explain how to turn this observation into a quantitative statement and describe a distillation procedure based on Kirwan and Ness' gradient flow \cite{Kirwan84, NessMumford84}.
Intriguingly, we find that a generalization of this technique gives rise to a probabilistic \emph{classical} algorithm for the computation of entanglement polytopes and, in fact, general moment polytopes of orbit closures (cf.\ \cite{Wernli13}).

\bigskip

Quantum states prepared in the laboratory are always subject to noise and therefore never perfectly pure.
In \autoref{sec:slocc/noise} we show how our method for witnessing entanglement using \eqref{eq:slocc/criterion} can be adapted to this situation as long as the noise is not too large.
To make this statement precise, we assume that a lower bound $1-\varepsilon$ on the purity $\mathrm{tr}\,\rho^2$ of a quantum state $\rho$ is available. This implies that $\rho$ has fidelity $\langle \psi | \rho | \psi \rangle \geq 1-\varepsilon$ with a pure state $\lvert\psi\rangle$ whose local eigenvalues deviate from the measured ones by no more than a small amount $\delta(\varepsilon)$.
In the case of $n$ qubits one has that $\delta(\varepsilon) \approx n \varepsilon / 2$ for small impurities.
Therefore, as long as the distance of the measured eigenvalues $\vec\lambda$ to an entanglement polytope $\Delta_\calX$ is at least $\delta(\varepsilon)$, the experimentally prepared state $\rho$ has high fidelity with a pure state that is more entangled than the class $\calX$.
These ideas can be further extended to show that $\rho$ itself cannot be written as a convex combination of quantum states in a given entanglement class.
Unlike the local eigenvalues, the purity $\mathrm{tr}\,\rho^2$ cannot be estimated by single-particle tomography alone.
However, it is in general not necessary to perform full tomography of the global state in order to determine the purity \cite{BuhrmanCleveWatrousEtAl01}.
Whereas it is an experimental challenge to achieve the levels of purity necessary for the application of our method, we believe that they are in the reach of current technology \cite{RauschenbeutelNoguesOsnaghiEtAl00, PanDaniellGasparoniEtAl01, MonzSchindlerBarreiroEtAl11}.

\bigskip

After completion of the work described in this chapter, we have learned about independent related work by Sawicki, Oszmaniec and Ku\'{s} \cite{SawickiOszmaniecKus12, SawickiOszmaniecKus12a}.

\section{Classification of Entanglement}
\label{sec:slocc/classification}

A pure quantum state $\rho = \ketbra \psi \psi$ of a multi-particle system with Hilbert space $\calH = \calH_1 \otimes \dots \otimes \calH_n$ is called \emph{entangled} if it cannot be written as a tensor product \cite{NielsenChuang04}
\begin{equation*}
  \ket\psi \neq \ket{\psi_1} \otimes \dots \otimes \ket{\psi_n}.
\end{equation*}
It is easy to see that $\rho$ is unentangled, or \emph{separable}, if and only if each its one-body reduced density matrices $\rho_k$ is itself a pure state.
More generally, mixed states are called separable if they are convex combinations of product states, but in this chapter we are primarily concerned with the entanglement of pure states.

In order to classify the entanglement present in a given multipartite quantum state $\rho$, one fruitful approach has been to compare the capability of $\rho$ for quantum information processing tasks with that of other quantum states \cite{BennettPopescuRohrlichEtAl00}.
Specifically, suppose that $\rho'$ is a state that can be obtained from $\rho$ by some suitable class of operations.
Then $\rho$ can be used as a replacement for $\rho'$ in any quantum information processing scenario where these operations are considered to be ``free''.
If the operations themselves cannot create entanglement then we may think of $\rho$ to be at least as entangled as $\rho'$.
If, conversely, $\rho$ can also be obtained from $\rho'$ then we may regard the two states to possess the same kind of multipartite entanglement.
In this way the set of quantum states is partitioned into equivalence classes.

For example, any \emph{local operation} that only acts on one of the $n$ subsystems can never create entanglement.
Local operations, including measurements, can be described by trace-preserving, completely positive maps that act on the space of density operators on $\calH_k$. Conversely, Stinespring's dilation theorem shows that such maps can always be implemented by locally adding an ancilla system, performing a unitary operation, and tracing out.
It is also useful to allow for \emph{classical communication} between the subsystems, which can create classical correlations but no entanglement.
This allows to condition each subsequent operation on the outcomes of previous measurements, even if those were performed at different subsystems.
To model this formally, one needs an additional, classical register that stores the measurement outcome (as opposed to the post-measurement state); this leads to the definition of a quantum instrument (see, e.g., the exposition in \cite{ChitambarLeungMancinskaEtAl14}).
We thus obtain a class of operations known as \emph{local operations and classical communication}, or \emphindex{LOCC}.
For mixed states, the resulting theory is rather rich and has been extensively studied in the literature (also asymptotically in the limit of many copies of a given state).
For pure states, however, the resulting notion of LOCC equivalence is quite restrictive: Two pure states can be interconverted by LOCC if and only if they are related by local unitaries \cite{Nielsen99, BennettPopescuRohrlichEtAl00}.

\bigskip

In the context of this work it will thus be convenient to consider a stochastic version of LOCC, known as \emphindex{SLOCC} \cite{BennettPopescuRohrlichEtAl00, DuerVidalCirac00}.
Here, the conversion of one state into another is only required to succeed with non-zero probability.
Operationally, this amounts to allowing \emphindex{post-selection} on measurement outcomes that occur with non-zero probability.

In \cite{DuerVidalCirac00}, it has been shown that two pure states $\rho = \proj\psi$ and $\rho' = \proj{\psi'}$ are interconvertible using SLOCC if and only if there exist invertible operators $g_k$ acting on $\calH_k$ such that $\ket{\psi'} = (g_1 \otimes \dots \otimes g_n) \ket\psi$.
Here, the non-trivial part is to show that invertible operators $g_k$ can be constructed by following a successful branch of an SLOCC conversion protocol.
For the converse, we may simply successively perform local POVM measurements
with Kraus operators
\begin{equation}
  \label{eq:slocc/povm}
  S_k = \sqrt{p_k} g_k, ~ F_k = \sqrt{\Id - p_k (g_k)^\dagger g_k},
\end{equation}
where $p_k > 0$ is a suitable normalization constant such that $p_k g_k^\dagger g_k \leq \Id$.
If all measurements succeed (corresponding to outcomes $S_1, \dots, S_n$) then $\rho'$ has been successfully distilled from $\rho$ \cite{DuerVidalCirac00}.

Mathematically, the result of \cite{DuerVidalCirac00} can be phrased in terms of the action \eqref{eq:onebody/complexified action} of the Lie group $G = \SL(\calH_1) \times \dots \times \SL(\calH_n)$ on projective space $\PP(\calH)$:
Two pure states $\rho$ and $\rho'$ are interconvertible by SLOCC if and only if one state is contained in the $G$-orbit of the other, or, equivalently, if and only if the two orbits are the same: $G \cdot \rho = G \cdot \rho'$.
We take this as the definition of an entanglement class:

\begin{dfn}
  Let $\calH = \calH_1 \otimes \dots \otimes \calH_n$ be a tensor-product Hilbert space.
  The \emph{(SLOCC) entanglement class}\index{entanglement class!SLOCC|textbf} of a pure state $\rho = \proj\psi$ is defined as%
  \nomenclature[QX, Y]{$\calX, \calY, \dots$}{entanglement classes}%
  \nomenclature[QX_rho1]{$\calX_\rho$}{entanglement class of a pure state $\rho$}
  \begin{equation*}
    \calX_\rho := G \cdot \rho =
    \left\{ \frac {\Pi(g) \proj\psi \Pi(g)^\dagger} {\braket{\psi | \Pi(g)^\dagger \Pi(g) | \psi}} : g \in G \right\}
    \subseteq
    \PP(\calH).
  \end{equation*}
\end{dfn}

In this way, the group $G$ that had already appeared in the preceding chapters acquires a physical interpretation in terms of SLOCC operations.
It is clear that the product states form a single entanglement class.
The closure of an entanglement class contains in addition to the class itself also those quantum states which can be arbitrarily well approximated by states in the class.
It can thus be given a similar operational interpretation as the class itself.
While the entanglement classes partition the set of multipartite quantum states, their closures naturally form a hierarchy.
Indeed, it is immediate that $\rho' \in \overline{\calX_\rho}$ implies $\overline{\calX_{\rho'}} \subseteq \overline{\calX_\rho}$.
Every entanglement class contains in its closure the class of unentangled states.

\bigskip

Stochastic local operations and classical communication provide a systematic framework for studying multi-particle entanglement.
However, it is immediate from the fact that the dimension of $G$ grows only linearly with the particle number $n$ that there is generically an infinite number of distinct SLOCC entanglement classes, labeled by an exponential number of continuous parameters (cf.~\autoref{sec:slocc/examples}).
It is therefore necessary to coarsen the classification in a systematic way in order to arrive at a tractable way of witnessing multi-particle entanglement.
This is one motivation for the notion of entanglement polytopes that we will define in the next section.

We conclude this section by noting that an extraordinary amount of research has been devoted to the classification of entanglement and its experimental identification.
The field is far too large to allow for an exhaustive bibliography; we refer to \cite{HorodeckiHorodeckiHorodeckiEtAl09, EisertGross08} for reviews of the general theory and to \cite{GuehneToth09} for a review focusing on detection.
Methods from algebraic geometry and classical invariant theory have long been used to analyze entanglement classes, see, e.g., \cite{
VerstraeteDehaeneDeEtAl02,
Klyachko02,
BriandLuqueThibon03,
VerstraeteDehaeneDe03,
Miyake03,
LeiferLindenWinter04,
HanZhangGuo04,
OsterlohSiewert05,
OsterlohSiewert06,
Klyachko07,
JokovicOsterloh09,
BastinKrinsMathonetEtAl09,
Osterloh10,
GourWallach11,
ViehmannEltschkaSiewert11,
EltschkaBastinOsterlohEtAl12}
and references therein.

\section{Entanglement Polytopes}
\label{sec:slocc/entanglement polytopes}

Let $\calH = \calH_1 \otimes \dots \otimes \calH_n$ and $G = \SL(\calH_1) \times \dots \times \SL(\calH_n)$.
We are interested in studying the one-body reduced density matrices $\rho_1, \dots, \rho_n$ for all pure states $\rho$ in a given entanglement class $\calX \subseteq \PP(\calH)$ and its closure $\overline\calX$.
This is a refinement of the one-body quantum marginal problem, \autoref{pro:onebody/qmp}, where we had instead considered all pure states on $\calH$.
Since $G$ contains the maximal compact subgroup $K = \SU(\calH_1) \times \dots \times \SU(\calH_n)$ of local unitaries, we may -- just as in the case of \autoref{pro:onebody/qmp} -- equivalently study the local eigenvalues $\vec\lambda_k$, as captured by the following definition:

\begin{dfn}
  The \emph{entanglement polytope}\index{entanglement polytope|textbf} of an entanglement class $\calX$ is%
  \nomenclature[QDelta_X]{$\Delta_\calX$}{entanglement polytope of a class $\calX$}
  \begin{equation*}
    \Delta_\calX =
    \left\{
      (\vec\lambda_1, \dots, \vec\lambda_n)
      :
      \vec\lambda_k = \spec \rho_k,
      \rho \in \overline\calX
    \right\},
  \end{equation*}
  the set of local eigenvalues of the one-body reduced density matrices of all quantum states in the closure of the entanglement class.
\end{dfn}

As $\overline\calX$ is the closure of a $G$-orbit, it is a $G$-invariant projective subvariety of $\PP(\calH)$.
Its moment polytope $\Delta_K(\overline\calX) = \mu_X(\overline\calX) \cap i \mathfrak t^*_+$ can be identified with $\Delta_\calX$ in the same way as we did in \autoref{sec:onebody/consequences} for the one-body quantum marginal problem. We summarize:

\begin{lem}
\label{lem:slocc/convexity}
  The set $\Delta_\calX$ is a convex polytope that can be identified with the moment polytope $\Delta_K(\overline\calX)$ of the projective subvariety $\overline\calX \subseteq \PP(\calH)$.
\end{lem}

If the vector $\vec\lambda = (\vec\lambda_1, \dots, \vec\lambda_n)$ of local eigenvalues of a state $\rho$ is not contained in a given entanglement polytope $\Delta_\calX$ then by definition $\rho$ cannot be contained in the closure of the class $\calX$.
This establishes \eqref{eq:slocc/criterion}, our criterion for witnessing multi-particle entanglement.
The entanglement polytopes form a hierarchy which coarsens the hierarchy of the closures of entanglement classes:
\begin{equation*}
  \overline\calX \subseteq \overline\calY
  \,\Longrightarrow\,
  \Delta_\calX \subseteq \Delta_\calY
\end{equation*}
We remark that $\vec\lambda$ is a natural generalization of the Schmidt coefficients or entanglement spectrum for bipartite states.

\bigskip

As we saw in \autoref{ch:onebody}, the convexity of moment polytopes can be established for arbitrary $G$-invariant projective subvarieties of $\PP(\calH)$.
For example, we may also consider the set $\{ \rho = \rho_I \otimes \rho_{I^c} \}$ of pure states that are \emphindex{biseparable} with respect to a fixed bipartition $I:I^c$ of the $n$ subsystems, known as a \emphindex{Segr\'{e} variety} in algebraic geometry.
We will return to this construction in our discussion of genuine multipartite entanglement in \autoref{sec:dhmeasure/examples}.

We remark that subsequent works have proposed a similar analysis of entanglement based on entropies rather than eigenvalues, which leads to a coarser notion than our entanglement polytopes \cite{HuberVicente13, HuberPerarnau-LlobetVicente13} (cf.\ \autoref{ch:entropy}).

\subsection*{Invariant-Theoretic Description}

The convexity of moment polytopes and their description relies on the decomposition of the ring of regular functions into irreducible representations (\autoref{thm:onebody/kirwan}).
In the case of the closure of an entanglement class $\calX = G \cdot \rho$, this description can be slightly simplified.
For this, we consider the following definition from classical invariant theory \cite{KraftProcesi96}.

\begin{dfn}
  A \emphindex{covariant}\nomenclature[RPhi]{$\Phi$}{covariant} of degree $k$ and weight $\lambda = (\lambda_1, \dots, \lambda_n)$ is a $G$-equivariant map
  \begin{equation*}
    \Phi \colon \calH = \calH_1 \otimes \dots \otimes \calH_n \longrightarrow V_{G,\lambda} = V^{d_1}_{\lambda_1} \otimes \dots \otimes V^{d_n}_{\lambda_n}
  \end{equation*}
  whose components are homogeneous polynomials of degree $k$ and where $d_i = \dim \calH_i$.
\end{dfn}

Since $\Phi$ is homogeneous of degree $k$, we may think of each $\lambda_i$ as a Young diagram with $k$ boxes and at most $d_i$ rows, i.e., $\lambda_i \in \ZZ_{\geq 0}^{d_i}$ and $\sum_j \lambda_{i,j} = k$ for all $i=1,\dots,n$.
Indeed, this would be the only choice for the number of boxes if we were to work with $\GL(d_1) \times \dots \times \GL(d_n)$ rather than $G = \SL(d_1) \times \dots \times \SL(d_n)$.
In the present context it is slightly more convenient to stay with the usual conventions and work with the latter group.
A covariant for which $V_{G,\lambda}$ is the trivial representation is nothing but an \emphindex{invariant polynomial}.
This is the case if and only if each representation is a power of the determinant representation, i.e., if $\lambda_i = (k/d_i,\dots,k/d_i)$ for all $i=1,\dots,n$ (cf.\ \autoref{sec:onebody/lie}).
While covariants are functions defined on the Hilbert space $\calH$, they are homogeneous and so their non-vanishing can be well-defined on points of projective space by taking any representative vector; we will denote this by $\Phi(\rho) \neq 0$.

The covariants of degree $k$ are in bijection with the highest weight vectors in $R_k(\calH)$, the space of homogeneous polynomials of degree $k$ on $\calH$.
Indeed, if $\Phi$ is a covariant of degree $k$ and weight $\lambda$ then its component
\begin{equation}
\label{eq:slocc/highest weight vector}
  P \colon \calH \rightarrow \CC, \ket\psi \mapsto \bra{w_0 \lambda} \Phi(\ket\psi)
\end{equation}
is a highest weight vector in $R_k(\calH)$ of highest weight $\lambda^*$, where $\ket{w_0 \lambda} \in V_{G,\lambda}$ is a fixed ``lowest weight vector'' of weight $w_0 \lambda = -\lambda^*$.
What is more,
\begin{equation}
\label{eq:slocc/vanishing equivalence}
  \Phi(\rho) \neq 0
  \;\Leftrightarrow\;
  P\big|_{\overline{G \cdot \rho}} \neq 0,
\end{equation}
since the $G$-orbit through $\ket{w_0\lambda}$ spans the entire irreducible representation $V_{G,\lambda}$,
while any polynomial that vanishes on $G \cdot \rho$ must also vanish on the closure.
The following lemma is already implicit in \cite{Brion87}:

\begin{lem}
  \label{lem:slocc/covariants vs irreps}
  For an entanglement class $\calX = G \cdot \rho$, the following are equivalent:
  \begin{enumerate}
    \item[(1)] There exists an irreducible representation $V_{G,\lambda^*} \subseteq R_k(\overline{\calX})$.
    \item[(2)] There exists a covariant $\Phi$ of degree $k$ and weight $\lambda$ such that $\Phi(\rho) \neq 0$.
  \end{enumerate}
\end{lem}
\begin{proof}
  $(1) \Rightarrow (2)$:
  Let $P$ denote a highest weight vector of the irreducible representation $V_{G,\lambda^*} \subseteq R_k(\overline\calX)$.
  Then we may consider $P$ as a polynomial in $R_k(\calH)$ that does not vanish on $\overline\calX$, and \eqref{eq:slocc/vanishing equivalence} implies that the corresponding covariant $\Phi$ is non-zero at $\rho$.

  $(2) \Rightarrow (1)$:
  Conversely, any covariant $\Phi$ with $\Phi(\rho) \neq 0$ corresponds to a highest weight vector $P$ of highest weight $\lambda^*$ in $R_k(\calH)$.
  By virtue of \eqref{eq:slocc/vanishing equivalence}, $P$ does not fully vanish on $\overline\calX$.
  Therefore, $P$ determines a highest weight vector in $R_k(\overline\calX)$ and there exists a corresponding irreducible representation $V_{G,\lambda^*} \subseteq R_k(\overline\calX)$.
\end{proof}

To show that moment polytopes are not only convex but indeed polytopes, we had used the crucial fact that the algebra of $N_+$-invariant polynomials $R(\calH)^{N_+}$ -- whose elements are linear combinations of highest weight vectors -- is finitely generated \cite{Grosshans73} (see proof of \autoref{prp:onebody/mumford convex polytope}).
We saw that there exist finitely many highest weight vectors $P^{(1)}, \dots, P^{(m)}$ such that all other highest weight vectors can be obtained as linear combinations of monomials in the $P^{(j)}$.
Let us call the corresponding covariants $\Phi^{(1)}, \dots, \Phi^{(m)}$ a \emph{generating set of covariants}\index{covariants!generating set}.

\begin{prp}
  \label{prp:slocc/description}
  The entanglement polytope of an entanglement class $\calX = G \cdot \rho$ is given by
  \begin{equation*}
    \Delta_\calX = \conv \left\{
    \frac {\lambda^{(j)}} {k^{(j)}} = \frac {(\lambda^{(j)}_1, \dots, \lambda^{(j)}_n)} {k^{(j)}} : \Phi^{(j)}(\rho) \neq 0
    \right\}.
  \end{equation*}
  where $\Phi^{(1)}, \dots, \Phi^{(m)}$ denotes a generating set of covariants with degrees $k^{(j)}$ and weights $\lambda^{(j)}$.
\end{prp}
\begin{proof}
  Recall from \autoref{lem:slocc/convexity} that $\Delta_\calX$ can be identified with the moment polytope of the $K$-action on $\overline\calX$. By \autoref{thm:onebody/kirwan} and \autoref{lem:slocc/covariants vs irreps}, the rational points of the entanglement polytope are thus given by the normalized weights $\lambda/k$ corresponding to covariants that do not vanish at $\rho$.
  In particular, the inclusion $(\supseteq)$ is immediate.

  For $(\subseteq)$, we closely follow the proof of \autoref{prp:onebody/mumford convex polytope}.
  Let $\lambda/k \in \Delta_\calX$ with corresponding covariant $\Phi$ of degree $k$ and weight $\lambda$.
  We denote by $P \in R_k(\calH)$ the highest weight vector \eqref{eq:slocc/highest weight vector} corresponding to the covariant $\Phi$.
  Since the $P^{(j)}$ are generators of the algebra of highest weight vectors, we may write $P$ as a linear combination of monomials in the $P^{(j)}$.
  If the linear combination is chosen minimally then the degrees $k^{(j)}$ and weights $\lambda^{(j)}$ of each monomial $P^{(j_1)} \dotsm P^{(j_p)}$ add up to the degree $k$ and weight $\lambda$ of $P$. Thus,
  \begin{equation}
  \label{eq:slocc/convex combination}
    \frac \lambda k
    = \sum_i \frac {k^{(j_i)}} k \frac {\lambda^{(j_i)}} {k^{(j_i)}}
    \in \conv \left\{ \frac {\lambda^{(j_1)}} {k^{(j_1)}}, \dots, \frac {\lambda^{(j_p)}} {k^{(j_p)}} \right\}
  \end{equation}
  We know from \eqref{eq:slocc/vanishing equivalence} that $P$ does not completely vanish on $\calX$.
  Thus same must be true for all factors in at least one of the monomials in the linear combination.
  Again using \eqref{eq:slocc/vanishing equivalence}, the corresponding covariants $\Phi^{(j_1)}, \dots, \Phi^{(j_p)}$ are all non-zero at $\rho$.
  But then \eqref{eq:slocc/convex combination} shows that $\lambda/k$ is indeed a convex combination of normalized weights $\lambda^{(j)}/k^{(j)}$ of non-vanishing covariants.
\end{proof}

Finite generation also implies other desirable properties, which hold more generally for moment polytopes of projective subvarieties \cite{Brion87}.
For example, there are only \emph{finitely many} entanglement polytopes, since by \autoref{prp:slocc/description} any entanglement polytope is the convex hull of some subset of the finite list of normalized weights $\lambda^{(j)} / k^{(j)}$.
What is more, the set of quantum states for which all generators are non-zero,
\begin{equation*}
  \bigcap_{j=1}^m \{ \rho \in \PP(\calH) : \Phi^{(j)}(\rho) \neq 0 \},
\end{equation*}
is a finite intersection of Zariski-open sets, hence itself Zariski-open.
In particular, its complement has positive codimension.
It follows that the entanglement polytope of a generic quantum state is maximal and equal to
\begin{align*}
    &\conv \{ \lambda^{(j)} / k^{(j)} : j=1,\dots, m \} \\
  = &\overline{ \{ \lambda / k : \text{ $\exists$ covariant of degree $k$ and weight $\lambda$} \} } \\
  = &\overline{ \{ \lambda / k : \text{ $\exists$ }V_{G,\lambda^*} \subseteq R_k(\calH) \} } \\
  = &\Delta_K(\PP(\calH)),
\end{align*}
which we recognize as the solution of the one-body quantum marginal problem for pure-states on $\calH$.

We stress that this observation does not imply that the method is trivial.
Even if mathematically a given entanglement class has measure zero in projective space, it can still be an important task to show that a quantum state prepared in the laboratory is not contained in the class---this is perhaps most obvious if the class is the set of separable pure states!
In our approach, this can be done using criterion \eqref{eq:slocc/criterion} by showing that the local eigenvalues of the state are sufficiently far away from the entanglement polytope of the class.
As we explain in \autoref{sec:slocc/noise}, in the presence of small noise our method can be adapted to exclude convex combinations of classes (which always have positive measure).
We remark that in classical statistics, testing for non-independence of random variables is similarly concerned with rejecting a measure-zero property (with respect to the natural measure on the probability simplex of joint distributions).

\subsection*{Computation}

By virtue of \autoref{prp:slocc/description}, the \emph{computation} of entanglement polytopes is a finite problem that can in principle be completely algorithmized.
By using the relation
\begin{equation*}
  R(\calH)^{N_+} \cong
  \left( R(\calH) \otimes R(G)^{N_+} \right)^G \cong
  R(\calH \times G/N_+)^G,
\end{equation*}
which is in fact at the heart of the proof of finite generation \cite{Grosshans73}, the problem of computing a set of generating covariants is transformed into a problem of computing invariants for a complex reductive group
(see \cite[\S{}4.2]{Dolgachev03} for details; cf.~\cite{DoranKirwan07}).
For the latter problem there exists a Gr\"obner-basis algorithm in computational invariant theory \cite{DerksenKemper02} that has been implemented, e.g., in the \textsc{Magma} computer algebra system \cite{BosmaCannonPlayoust97} (but see \autoref{sec:slocc/discussion}).
Once a set of generating covariants has been found, \autoref{prp:slocc/description} can be used to compute the entanglement polytope both for specific states as well as for families of states.
We demonstrate this method when computing examples in the next section (relying on generating sets of covariants that had been previously computed).
In \autoref{sec:slocc/distillation} we describe an alternative approach that does not rely on computational invariant theory.

\section{Examples}
\label{sec:slocc/examples}

In the preceding section we have seen that by associating to every entanglement class its entanglement polytopes we obtain a finite yet systematic classification of multipartite entanglement.
We will now discuss some illustrative examples for qubit systems, where the Hilbert space is of the form $\calH = \CC^{\otimes 2} \otimes \dots \otimes \CC^{\otimes 2}$.
Mathematically, the covariants of a multi-qubit system are in one-to-one correspondence with the covariants of binary multilinear forms, which have been intensely studied in classical invariant theory \cite{Olver99, Luque07}.

Before we proceed, we introduce some notational simplifications.
It will be convenient to represent the local eigenvalues of a pure state of $n$ qubits by the tuple $(\lambda_{1,1}, \dots, \lambda_{n,1}) \in [0.5,1]^n$ of \emph{maximal} local eigenvalues; this is without loss of information since the eigenvalues of the one-body reduced density matrix of each qubit sum to one.
Thus a covariant of degree $k$ and weight $(\lambda_1, \dots, \lambda_n)$ determines the point $(\lambda_{1,1}, \dots, \lambda_{1,n}) / k$ in the polytope of maximal local eigenvalues.
\subsection*{Two Qubits}

We briefly discuss the rather trivial case of two qubits, where $\calH = \CC^2 \otimes \CC^2$.
Here there are two entanglement classes:
One class is represented by the maximally entangled \emph{Einstein--Podolsky--Rosen (EPR) pair}\index{EPR pair}\nomenclature[QEPR]{$\ket\EPR$}{EPR state}
\[\ket\EPR = (\ket{00} + \ket{11})/\sqrt 2,\]
the other is the class of product states, represented by $\ket{00}$.
Using \autoref{lem:onebody/two}, it is not hard to see that the entanglement polytope of the former class is given the convex hull of $(0.5,0.5)$ and $(1,1)$, while the entanglement polytope of the product states is a single point $(1,1)$.
We can also see this formally by using the method of covariants.
There are two generating covariants:
One is the identity map
\[\CC^2 \otimes \CC^2 \rightarrow V^2_{\Yboxdim{5pt}\yng(1)\Yboxdim{8pt}} \otimes V^2_{\Yboxdim{5pt}\yng(1)\Yboxdim{8pt}},\]
which never vanishes and therefore shows that the point $(1,1)$ is contained in both entanglement polytopes.
The other is the map
\begin{equation*}
  \CC^2 \otimes \CC^2 \rightarrow V^2_{\Yboxdim{5pt}\yng(1,1)\Yboxdim{8pt}} \otimes V^2_{\Yboxdim{5pt}\yng(1,1)\Yboxdim{8pt}} \cong \CC,
  \quad
  \ket\psi = \sum_{i,j} \psi_{ij} \ket{ij} \mapsto \psi_{00} \psi_{11} - \psi_{10} \psi_{01}
\end{equation*}
that sends a quantum state to the determinant of its coefficients.
It is of degree 2 and therefore corresponds to the point $(0.5,0.5)$ in an entanglement polytope.
Since the determinant is zero on the class of product states but not on the EPR pair, we obtain the two entanglement polytopes from above by using \autoref{prp:slocc/description}.

\subsection*{Three Qubits}

\begin{table}
  \centering
\Yboxdim{5pt}
  \begin{tabular}{ccccccccc}
    \toprule
    Covariant & Degree & Weight & $\ket\GHZ$ & $\ket W$ & $\ket{B1}$ & $\ket{B2}$ & $\ket{B3}$ & $\ket\SEP$ \\
    \midrule
    $f$ & $1$ & $(\yng(1),\yng(1),\yng(1))$ & $\times$ & $\times$ & $\times$ & $\times$ & $\times$ & $\times$ \\
    $H_x$ & $2$ & $(\yng(2),\yng(1,1),\yng(1,1))$ & $\times$ & $\times$ & $\times$ &  &  &  \\
    $H_y$ & $2$ & $(\yng(1,1),\yng(2),\yng(1,1))$ & $\times$ & $\times$ &  & $\times$ &  &  \\
    $H_z$ & $2$ & $(\yng(1,1),\yng(1,1),\yng(2))$ & $\times$ & $\times$ &  &  & $\times$ &  \\
    $T$ & $3$ & $(\yng(2,1),\yng(2,1),\yng(2,1))$ & $\times$ & $\times$ &  &  &  &  \\
    $\Delta$ & $4$ & $(\yng(2,2),\yng(2,2),\yng(2,2))$ & $\times$ &  &  &  &  &  \\
    \bottomrule
  \end{tabular}
\Yboxdim{8pt}
  \caption[Three-qubit covariants]{\emph{Three-qubit covariants \cite{Le81, Luque07}.} Degrees and weights of a generating set of three-qubit covariants labeled as in \cite[p.~31]{Luque07}, as well as their vanishing behavior on the quantum states representing the six entanglement classes described in the text ($\times$ denotes non-vanishing).
  The invariant $\Delta$ is Cayley's hyperdeterminant, which is closely related to the 3-tangle \cite{CoffmanKunduWooters00, Miyake03}.}
  \label{tab:slocc/three qubit covariants}
\end{table}

In the case of three qubits, where $\calH = \CC^2 \otimes \CC^2 \otimes \CC^2$, there are six distinct entanglement classes \cite{DuerVidalCirac00}:
The classes of the \emph{Greenberger--Horne--Zeilinger (GHZ) state}\index{GHZ state|textbf}\nomenclature[QGHZ]{$\ket\GHZ$}{GHZ state} \cite{GreenbergerHorneZeilinger89} and of the \emph{W state}\index{W state|textbf}\nomenclature[QW]{$\ket W$}{W state} \cite{DuerVidalCirac00}, respectively,
\begin{align*}
  \ket\GHZ &= (\ket{000} + \ket{111}) / \sqrt 2, \\
  \ket W &= (\ket{100} +\ket{010} + \ket{001}) / \sqrt 3,
\end{align*}
three classes that correspond to EPR pairs shared between any two of the three subsystems,
\begin{align*}
  \ket{B1} &= (\ket{010} + \ket{001}) / \sqrt 2, \\
  \ket{B2} &= (\ket{100} + \ket{001}) / \sqrt 2, \\
  \ket{B3} &= (\ket{100} + \ket{010}) / \sqrt 2,
\end{align*}
and the class of product states, represented by
\begin{equation*}
  \ket\SEP = \ket{000}.
\end{equation*}
We shall now compute the corresponding entanglement polytopes by following the general method of covariants described in \autoref{sec:slocc/entanglement polytopes}.
A minimal generating set of six covariants has been determined in late 19th century invariant theory \cite{Le81} in the context of the classification of binary three-linear forms, as explained in \cite{Luque07}.
Recall that each irreducible representation $V^2_\mu$ of $\SU(2)$ can be realized as the space of homogeneous polynomials of degree $\mu_1 - \mu_2$ in two formal variables.
Therefore, any covariant for three qubits can be written as a homogeneous polynomial in formal variables $x_i$, $y_j$ and $z_k$, whose coefficients are themselves homogeneous polynomials in the components $\psi_{ijk}$ of the quantum state (the degree in the $\psi_{ijk}$ is equal to the degree of the covariant, whereas the degrees in the formal variables determine its weight).
For example, the identity map can be written as the multilinear polynomial
$f = \sum_{i,j,k} \psi_{ijk} x_i y_j z_k$;
it is a generating covariant of degree $1$ and weight $({\Yboxdim{5pt}\yng(1)\Yboxdim{8pt}},{\Yboxdim{5pt}\yng(1)\Yboxdim{8pt}},{\Yboxdim{5pt}\yng(1)\Yboxdim{8pt}})$.
A more interesting generator is the covariant
\begin{align*}
  H_x &= x_0^2 \left(\psi_{000} \psi_{011}-\psi_{001} \psi_{010}\right) \\
      &+ x_1 x_0 \left(\psi_{011} \psi_{100}-\psi_{010} \psi_{101}-\psi_{001} \psi_{110}+\psi_{000} \psi_{111}\right) \\
      &+ x_1^2 \left(\psi_{100} \psi_{111}-\psi_{101} \psi_{110}\right)
\end{align*}
which is of degree $2$ and weight $({\Yboxdim{5pt}\yng(2)\Yboxdim{8pt}},{\Yboxdim{5pt}\yng(1,1)\Yboxdim{8pt}},{\Yboxdim{5pt}\yng(1,1)\Yboxdim{8pt}})$.
We refer to \cite[p.\ 31]{Luque07} for the complete list.
By \autoref{prp:slocc/description}, computing the entanglement polytopes is now a mechanical task:
for any quantum state representing the entanglement class, we merely need to collect those covariants which do not vanish on the state, and take the convex hull of their normalized weights.
In \autoref{tab:slocc/three qubit covariants} we have summarized the covariants' properties and their vanishing behavior on the six entanglement classes.
The resulting entanglement polytopes are illustrated in \autoref{fig:slocc/three qubits}.
They are in one-to-one correspondence with the six entanglement classes described above.
Since only the GHZ-class polytope is maximal, it follows that all states can be approximated arbitrarily well by states of GHZ type.
In mathematical terms, the GHZ class is dense; this is of course well-known \cite{DuerVidalCirac00}.
Similarly, the polytope of the W class (upper pyramid) contains all entanglement polytopes except for the GHZ one;
and it is easy to see that by states in the W class one can approximate all states except those of GHZ class.
Thus in this low-dimensional scenario the hierarchy of closures of entanglement classes is faithfully mapped onto the hierarchy of entanglement polytopes---there is no coarse-graining.

\begin{figure}
  \centering
  \includegraphics[width=0.9\linewidth]{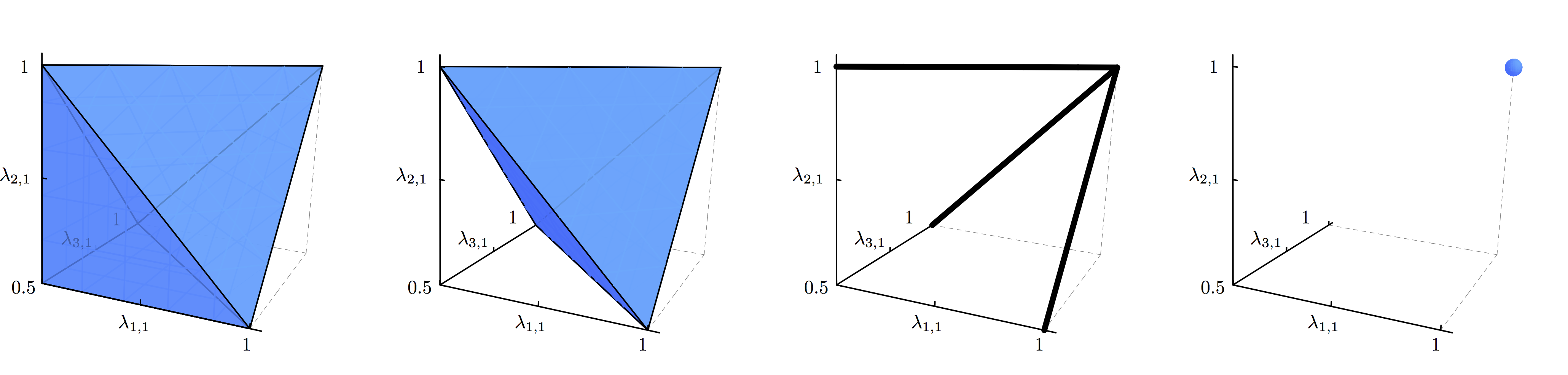}
  \begin{tabularx}{0.9\linewidth}{*4{>{\centering\arraybackslash}X}}
  {\footnotesize (a)\hphantom{xxxxx}} &
  {\footnotesize (b)\hphantom{xxxx}} &
  {\footnotesize (c)\hphantom{xxx}} &
  {\footnotesize (d)\hphantom{xx}}\\
  \end{tabularx}
  \caption[Entanglement polytopes for three qubits]{\emph{Entanglement polytopes for three qubits.}
  (a) GHZ polytope (maximal polytope, i.e., upper and lower pyramid),
  (b) W polytope (upper pyramid),
  (c) three polytopes corresponding to EPR pairs shared between any two of the three parties (three solid edges in the interior),
  (d) polytope of the unentangled states (interior vertex).}
  \label{fig:slocc/three qubits}
\end{figure}

We remark that using techniques crafted towards the special situation of three qubits, these same polytopes had already been previously computed in \cite{HanZhangGuo04, SawickiWalterKus13}.
The one-body quantum marginal problem, which as we have explained amounts to computing the maximal entanglement polytope, has been solved in \cite{HiguchiSudberySzulc03}, and we saw a different derivation in \autoref{ch:kirwan}.

We now illustrate the method of entanglement witnessing via \eqref{eq:slocc/criterion}.
Let $\rho$ be a quantum state with maximal local eigenvalues $\vec\lambda = (\lambda_{1,1}, \lambda_{2,1}, \lambda_{3,1})$.
\begin{itemize}
\item If the point $\vec\lambda$ is contained in the lower part of the entanglement polytope of the GHZ class (lower pyramid in \autoref{fig:slocc/three qubits}, (a)),
  \begin{equation*}
    \lambda_{1,1} + \lambda_{2,1} + \lambda_{3,1} < 2,
  \end{equation*}
  then it is not contained in any other entanglement polytope.
  Therefore the quantum state $\rho$ must be entangled of GHZ type.

\item More generally, if $\vec\lambda$ is \emph{not} contained in any of the lower-dimensional polytopes corresponding to EPR pairs (\autoref{fig:slocc/three qubits}, (c), which includes (d)), i.e., if
  \begin{equation*}
    \lambda_{k,1} < 1
    \quad
    (\forall k=1,2,3),
  \end{equation*}
  then the quantum state $\rho$ must be entangled of either GHZ or W class.
  These classes of states are the ones that possess genuine three-qubit entanglement (see discussion below).
\end{itemize}
As a final example, we consider the quantum state $\rho = \proj\psi$, where
\begin{equation}
\label{eq:slocc/distilled state}
  \ket\psi = (\ket{000} + \ket{001} + \ket{100} + 2 \ket{110}) / \sqrt 7.
\end{equation}
It is easy to verify by the method of covariants that the entanglement polytope of $\rho$ is full-dimensional.
Thus it follows from the above classification that $\rho$ is of GHZ type.
However, its collection of local eigenvalues $(\lambda_{1,1},\lambda_{2,1},\lambda_{3,1}) \approx (0.76, 0.79, 0.88)$
is contained in the interior of the upper pyramid.
As we just discussed, the entanglement criterion in this case only allows us to conclude that $\rho$ is either of GHZ or of W type.
In \autoref{sec:slocc/distillation} we will describe a distillation procedure that allows us to transform $\rho$ by SLOCC operations into another state whose local eigenvalues are arbitrarily close to the ``origin'' $(0.5,0.5,0.5)$ (cf.\ \autoref{fig:slocc/three qubit distillation}).
In this way, we may arrive at a quantum state which is both more entangled and for which the entanglement criterion is maximally informative.

\subsection*{Four Qubits}

In contrast to case of three qubits, where there are finitely many entanglement classes represented faithfully by the hierarchy of entanglement polytopes, the situation for four qubits is the generic one:
There are infinitely many entanglement classes.
According to the classification of \cite{VerstraeteDehaeneDeEtAl02}, up to permutation of the qubits they can be partitioned into nine families with up to three complex continuous parameters each (cf.\ \cite{ChterentalDjokovic07}).
Neither the family itself nor the complex parameters within a family are directly experimentally accessible.

We have determined all entanglement polytopes of four qubits using the general method of \autoref{sec:slocc/entanglement polytopes} applied to a minimal generating set of 170 covariants found in \cite{BriandLuqueThibon03}.
More precisely, for every family in \cite{VerstraeteDehaeneDeEtAl02}, we consider the covariants as a function of the parameters $a$, $b$, etc.\ of the family.
Deciding whether a normalized weight $\lambda^{(j)} / k^{(j)}$ is included in the entanglement polytope $\Delta_\calX$ of a state in the family then amounts to solving the explicit polynomial equation $P_j \neq 0$ in the parameters $a$, $b$, etc.; this can be automatized by using a computer algebra system.

The maximal entanglement polytope is equal to the solution of the one-body quantum marginal problem for four qubits as given by the polygonal inequalities \eqref{eq:kirwan/polygonal} for $n=4$.
It is a convex hull of 12 vertices, which can be easily described as follows:
One vertex, $(1, 1, 1, 1)$ corresponds to the class of product states;
$(0.5, 0.5, 1, 1)$ and its permutations correspond to the six possibilities of embedding an EPR pair into four qubits;
$(0.5, 0.5, 0.5, 1)$ and its permutations correspond to the four possibilities of embedding a GHZ state of three qubits;
and the vertex $(0.5, 0.5, 0.5, 0.5)$ is the image of, e.g., a four-partite GHZ state.

As in the case of three qubits, there are several lower-dimensional entanglement polytopes, corresponding to the different ways of embedding entanglement classes of systems of fewer qubits into four qubits.
These are precisely the biseparable entanglement classes, i.e., the classes whose elements are tensor products $\rho = \rho_I \otimes \rho_{I^c}$ with respect to some proper bipartition $I:I^c$ of the four qubits.
For example, the state $\ket\GHZ \otimes \ket 0$ generates an entanglement class whose elements are tensor products of the three-qubit GHZ class and a one-qubit pure state, and its entanglement polytope is the Cartesian product of the three-qubit GHZ polytope
with the point $\{1\}$.
All possibilities are listed in \autoref{tab:slocc/four qubit lower dimensional}.

\begin{table}
  \centering
  \begin{tabular}{lcc}
    \toprule
    Representative State & Number of Embeddings & Dimension \\
    \midrule
    $\ket\GHZ \otimes \ket0$ & 4 & 3 \\
    $\ket W \otimes \ket0$ & 4 & 3\\
    $\ket\EPR \otimes \ket\EPR$ & 3 & 2 \\
    $\ket\EPR \otimes \ket0 \otimes \ket0$ & 6 & 1 \\
    $\ket0 \otimes \ket0 \otimes \ket0 \otimes \ket0$ \qquad & 1 & 0 \\
    \bottomrule
  \end{tabular}
  \caption[Biseparable entanglement classes of four-qubit states]{\emph{Biseparable entanglement classes of four-qubit states.}
    Each biseparable entanglement class is obtained by embedding entanglement classes of fewer qubits.
    In the left column, we give a representative state for each class;
    the middle column lists the number of permutations of qubits that lead to distinct embeddings, and hence to distinct polytopes,
    and the right column lists the dimension of the polytope.
  }
  \label{tab:slocc/four qubit lower dimensional}
\end{table}

We now turn to the full-dimensional polytopes.
There are seven such polytopes, listed together with some of their properties in \autoref{tab:slocc/four qubits} and \autoref{fig:slocc/four qubits}.
For example, the entanglement classes of the four-qubit GHZ state and of the cluster states \cite{BriegelRaussendorf01} are both associated with the maximal polytope (last row in \autoref{tab:slocc/four qubits} and \autoref{fig:slocc/four qubits}).
The four-qubit $W$-state, $\ket{W_4} = ( \ket{0001} + \ket{0010} + \ket{0100} + \ket{1000}) / 2$, corresponds to polytope no.~5 in \autoref{tab:slocc/four qubits} and \autoref{fig:slocc/four qubits}.
In analogy with the $W$-state for three qubits, its polytope is an ``upper pyramid'' given by the intersection of the maximal polytope with the half-space
\begin{equation*}
  \lambda_{1,1} + \lambda_{2,1} + \lambda_{3,1} + \lambda_{4,1} \geq 3.
\end{equation*}
Any violation of this inequality may be taken as an indication of ``high entanglement''.
One way to make this precise is to read off \autoref{tab:slocc/four qubits} that violations imply that the state $\rho$ can be converted into one whose linear entropy of entanglement is at least $0.45$, which might be much higher than $E(\rho)$ itself.
We will explain this more carefully in \autoref{sec:slocc/distillation} below.

\begin{table}[ht]
  \centering
  \begin{tabular}{cclccccl}
    \toprule
    No & Family & Parameters & Vertices & Facets & Perms & $E(\Delta_{\calX})$ \\
    \midrule
    1 & $L_{0_{7\oplus \bar{1}}}$ & & 12 & 13 & 4 & 0.482143 \\
    2 & $L_{0_{5\oplus 3}}$ & & 10 & 13 & 4 & 0.458333 \\
    3 & $L_{a_4}$ & $a=0$ & 9 & 14 & 6 & 0.45 \\
    4 & $L_{abc_2}$ & $a=c=0, b \neq 0$ & 8 & 16 & 1 & 0.5 \\
    5 & $L_{ab_3}$ & $b=a=0$ & 7 & 9 & 1 & 0.375 \\
    6 & $L_{a_2b_2}$ &  $b=-a\neq0$ & 10 &14 & 6 & 0.5 \\
    7 & $G_{abcd}$ & generic & 12 & 12 & 1 & 0.5 \\
    \bottomrule
  \end{tabular}

  \bigskip

  \begin{tabular}{cclccccl}
    \toprule
    No & Vertices compared to maximal polytope \\
    \midrule
    1
    & $(0.5, 0.5, 0.5, 0.5)$ replaced by $(0.75, 0.5, 0.5, 0.5)$\\
    2
    & $(0.5, 0.5, 0.5, 0.5), (1, 0.5, 0.5, 0.5)$ missing \\
    3
    & $(0.5, 0.5, 0.5, 0.5), (0.5, 1, 0.5, 0.5), (0.5, 0.5, 0.5, 1)$ missing  \\
    4
    & $(0.5,0.5,0.5,1)$ and permutations missing \\
    5
    & $(0.5,0.5,0.5,0.5)$, $(0.5,0.5,0.5,1)$ and permutations missing \\
    6
    & $(0.5,1,0.5,0.5), (0.5,0.5,0.5,1)$ missing \\
    7
    & (this is the maximal polytope) \\
    \bottomrule
  \end{tabular}
  \caption[Entanglement polytopes for genuinely four-partite entangled states]{%
    \emph{Entanglement polytopes for genuinely four-partite entangled states.}
    For each polytope, the second and third column specifies one choice of family and parameters according to the classification of \cite{VerstraeteDehaeneDeEtAl02} (there might be further choices, in part because their parametrization is not always unique).
    ``Perms'' is the number of distinct polytopes one obtains when permuting the qubits.
    $E(\Delta_{\calX})$ is the maximal linear entropy of entanglement of any state in the class.
    The bottom table describes how to obtain each polytope from the maximal entanglement polytope (the solution of the one-body quantum marginal problem).}
  \label{tab:slocc/four qubits}
\end{table}

We now comment on properties that can be read off graphically from \autoref{fig:slocc/four qubits}.
From the first column, one can see that only three of the entanglement polytopes (no.~4, 6 and 7) include the ``origin'' $(0.5, 0.5, 0.5, 0.5)$.
These reach the maximal value for the linear entropy of entanglement $E(\rho) = 1 - \frac{1}{n}\sum_{j=1}^n \tr \rho_j^2$.
The last column exhibits the behavior of the class when the fourth particle is projected onto a generic pure state.
We then obtain a pure state of three qubits on the remaining subsystems.
Polytopes no.~1, 3, 6, and 7 give the full three-qubit polytope, implying that in general a GHZ-type state is generated, while polytopes no.~2, 4, and 5 collapse to the upper pyramid of \autoref{fig:slocc/witness}.
Hence states in the latter classes can never product a state of GHZ-type when the fourth qubit is projected onto a pure state.
It follows that the mixed 3-tangle (and any other convex-roof extension of a polynomial monotone) vanishes on the mixed state generated by tracing out the last particle of any state in these classes.
This observation allows us to graphically recover some properties calculated algebraically in \cite{VerstraeteDehaeneDeEtAl02}, such as the vanishing 3-tangle for the class $L_{{ab}_3}, a=b=0$ (corresponding to polytope no.~5).

\begin{figure}
  \centering
  \fourqubitcutsheadertwo

  \includegraphics[width=0.9\textwidth]{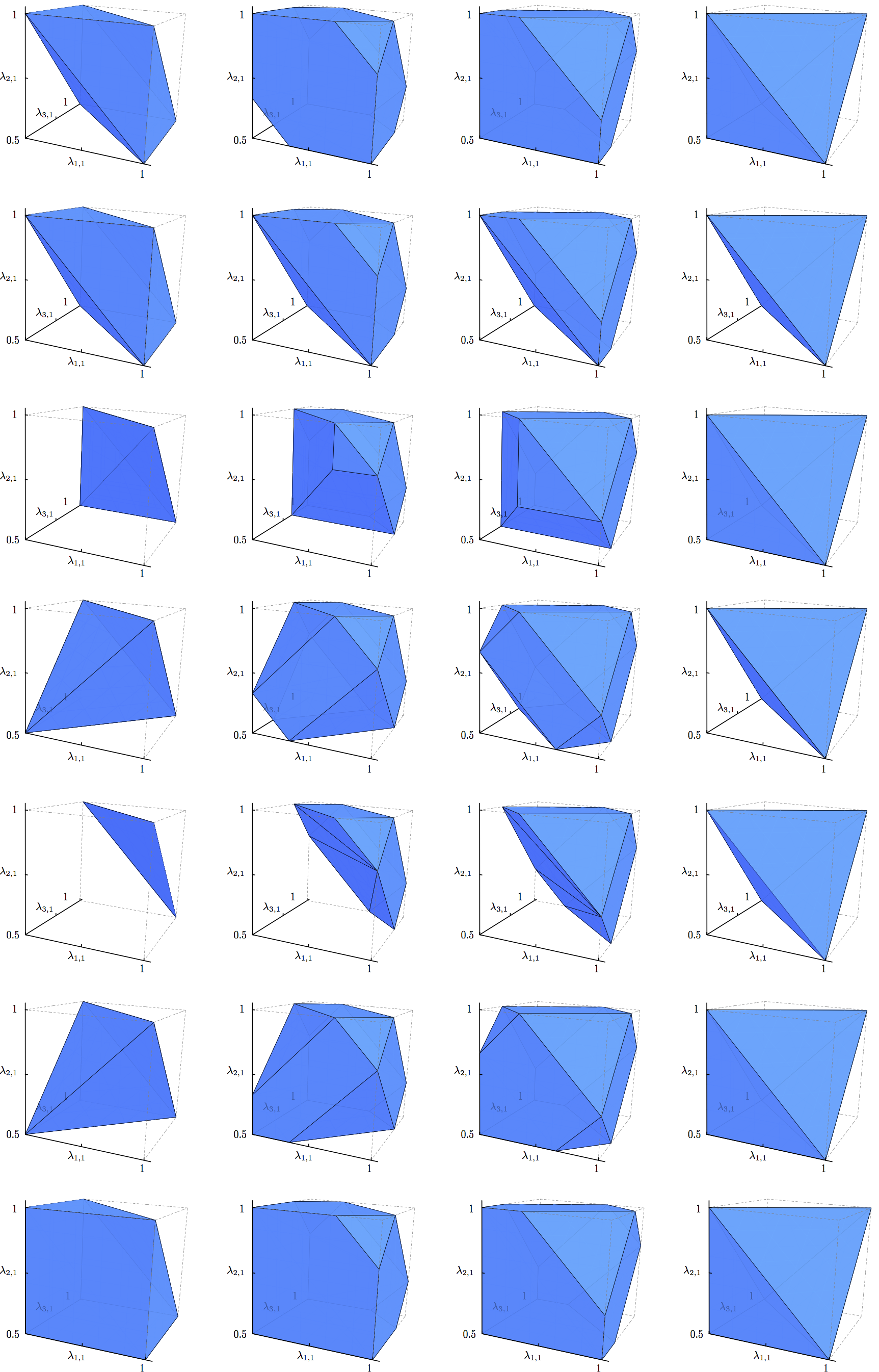}
  \caption[Illustration of the full-dimensional four-qubit entanglement polytopes]{
  \emph{Illustration of the full-dimensional entanglement polytopes for four qubits.}
    The seven rows correspond to the four-dimensional entanglement polytopes in \autoref{tab:slocc/four qubits}.
    The four columns show cross-sections for four fixed values of $\lambda_{4,1}$.
  }
  \label{fig:slocc/four qubits}
\end{figure}

In summary, up to permutations, there are 12 %
entanglement polytopes for four qubits, 7 of which are full-dimensional and belong to genuinely four-partite entangled states.
These numbers increase to 41 and 22, respectively, if distinct permutations are counted separately.
Slightly abusing the tradition of \cite{DuerVidalCirac00, VerstraeteDehaeneDeEtAl02}, we might say that from the perspective of entanglement polytopes, four qubits can be entangled in seven different ways.

\bigskip

We remark that there is a numerical coincidence between our findings and the ones in \cite{JokovicOsterloh09}, where also seven non-biseparable entanglement classes have been identified on four qubits. The two classifications are, however, not identical.
Indeed, \cite{JokovicOsterloh09} is based purely on \emph{invariants}, as opposed to the more general \emph{covariant}-theoretic description of our polytopes.
It follows from \autoref{prp:slocc/description} -- as we discuss in \autoref{sec:slocc/distillation} below -- that all polynomial invariants vanish identically on any entanglement class whose polytope does not contain the ``origin'' $(0.5, 0.5, 0.5, 0.5)$.
Since this is the case for our polytopes no.~1, 2, 3 and 5, the corresponding classes cannot be distinguished from each other by polynomial invariants (in fact, not even from the class of product states!)
Thus our classification differs from the one in \cite{JokovicOsterloh09}.
From a mathematical perspective, our methods are complementary (entanglement polytopes can also distinguish among unstable vectors, whereas \cite{JokovicOsterloh09} provides a better resolution in the semistable case).

\subsection*{Genuine Multipartite Entanglement}

In the preceding examples, entanglement classes corresponding to \emph{biseparable}\index{biseparable|textbf} states, i.e., pure states that can be factorized in the form $\rho = \rho_I \otimes \rho_{I^c}$, had polytopes of lower dimension.
This is no longer true for $n \geq 6$ particles.
For example, the entanglement polytope $\Delta_{\GHZ} \times \Delta_{\GHZ} \subseteq [0.5,1]^6$ associated with the biseparable state $\ket\GHZ \otimes \ket\GHZ \in \CC^{\otimes 6}$ of six qubits is certainly six-dimensional.

Perhaps surprisingly, it still remains true that spectral information alone can be used to show that a state is not biseparable.
To make this precise, we consider the following definition from \cite{GuehneTothBriegel05}; see also \cite{HorodeckiHorodeckiHorodecki01, HorodeckiHorodeckiHorodeckiEtAl09, GuehneToth09, LeviMintert13, SeevinckUffink01, HuberMintertGabrielEtAl10} and references therein.

\begin{dfn}
  Let $\calH = \calH_1 \otimes \dots \otimes \calH_n$.
  A pure state $\rho = \proj\psi$ on $\calH$ is called \emphindex{producible using $k$-partite entanglement} if it is of the form
  \begin{equation*}
    \rho = \rho_{I_1} \otimes \dots \otimes \rho_{I_m},
  \end{equation*}
  where $I_1 \cup \dots \cup I_m = \{1, \dots, n\}$ and each subset $I_j$ is of size at most $k$.

  Otherwise, $\rho$ is called \emphindex{genuinely $(k+1)$-partite entangled}\index{genuinely multipartite entangled|textbf}.
  In particular, the \emph{genuinely $n$-partite entangled} states are precisely those which are not biseparable.
\end{dfn}

We remark that some authors use the term ``genuine $n$-partite entanglement'' in a different sense (e.g., \cite{OsterlohSiewert05, OsterlohSiewert06} where it is also required that there exists a non-vanishing invariant polynomial).

\bigskip

To state our result, recall that the maximal entanglement polytope for $n$ qubits is given by the polygonal inequalities \eqref{eq:kirwan/polygonal}.
Therefore, the constraints on the local eigenvalues of states that factorize with respect to a fixed partition $I_1 \cup \dots \cup I_m = \{1, \dots, n\}$ are given by
\begin{equation}
\label{eq:slocc/polygonal for partition}
  \sum_{l \neq k \in I_j} \lambda_{k,1} \leq \left(\abs{I_j} - 2 \right) + \lambda_{l,1}
  \quad (\forall l \in I_j, j=1, \dots, m)
\end{equation}
Mathematically, while this set of states does not form a single entanglement class, it is still a $G$-invariant projective subvariety -- known as a Segr\'{e} variety in algebraic geometry --, and so has a corresponding moment polytope, namely the one cut out by the inequalities \eqref{eq:slocc/polygonal for partition}.
The inequalities \eqref{eq:slocc/polygonal for partition} hold for quantum states of arbitrary local dimension, since the same is true for the polygonal inequalities, but in general there are additional constraints.

\begin{prp}
\label{prp:slocc/genuine}
  Let $n \geq 3$ and $k \in \{ \lceil n/2 \rceil, \dots, n \}$.
  For any $\varepsilon \in (0,1/(2k-2)]$, the local eigenvalues
  \begin{equation*}
    \lambda_{1,1} = 1 - (k-1)\varepsilon, \quad
    \lambda_{2,1} = \dots = \lambda_{n,1} = 1 - \varepsilon
  \end{equation*}
  can only originate from a genuinely $k$-partite entangled state.

  Conversely, if $k \neq n-1$ then there exists a corresponding pure state of $n$ qubits that is producible using $k$-partite entanglement (i.e., that is not genuinely $k+1$-partite entangled).
\end{prp}
\begin{proof}
  For the first claim, suppose that $\rho = \rho_{I_1} \otimes \dots \otimes \rho_{I_m}$ is a state with the displayed local eigenvalues.
  Without loss of generality, suppose that $1 \in I_1$.
  Then \eqref{eq:slocc/polygonal for partition} for $l=1$ reads
  \begin{equation}
  \label{eq:slocc/genuinely interesting inequality}
    \big( \abs{I_1} - 1 \big) \big( 1 - \varepsilon \big)
    \leq \big( \abs{I_1} - 2 \big) + \big( 1 - (k-1) \varepsilon \big)
    \quad\Leftrightarrow\quad
    \abs{I_1} \geq k.
  \end{equation}
  Since this holds for all partitions, we conclude that $\rho$ is genuinely $k$-partite entangled.

  For the second claim, consider the bipartition $I_1 = \{1, \dots, k\}$, $I_2 = \{k+1, \dots, n\}$.
  Then \eqref{eq:slocc/genuinely interesting inequality} is satisfied and all other inequalities in \eqref{eq:slocc/polygonal for partition} are satisfied for $I_1$ since $k \geq \lceil n/2 \rceil \geq 2$.
  For $I_2$, which we need to consider only if $k < n$, \eqref{eq:slocc/polygonal for partition} is equivalent to $k \leq n-2$.
\end{proof}

\begin{figure}
  \centering
  \includegraphics[width=0.4\linewidth]{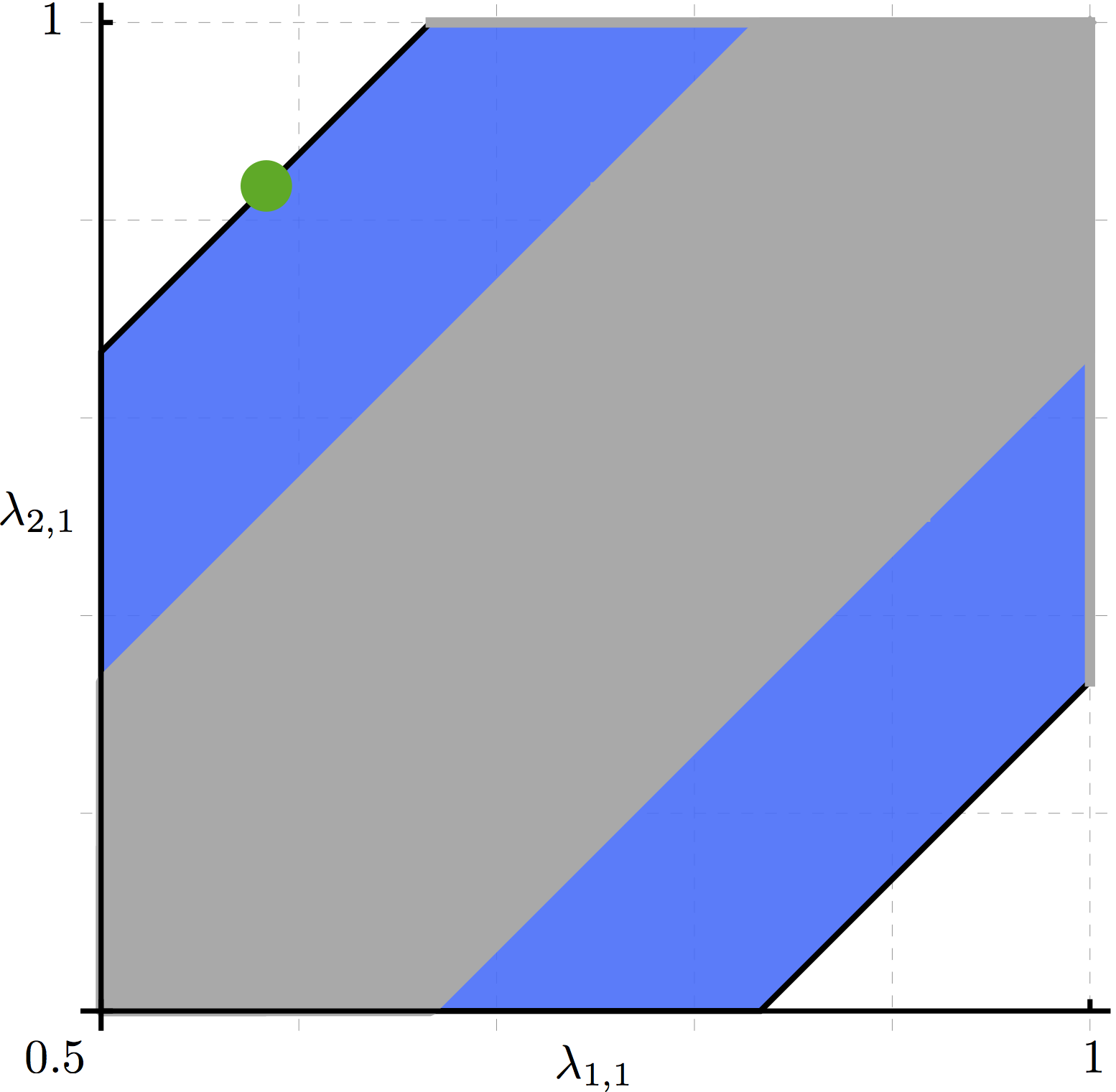}
  \caption[Witnessing genuine multipartite entanglement with entanglement polytopes]{%
    \emph{Witnessing genuine multipartite entanglement.}
    The figure displays the two-dimensional cross-section through the six-qubit eigenvalue polytope (blue) and the biseparable subpolytopes (gray), where we fix $\lambda_{3,1} = \dots \lambda_{6,1} = 11/12$.
    Any global pure state for which the local eigenvalues do not belong to the latter is necessarily genuinely six-qubit entangled.
    The spectrum produced by \autoref{prp:slocc/genuine} for $k=n=6, \varepsilon=1/12$ is one such example (green point).}
  \label{fig:slocc/genuine}
\end{figure}

\autoref{prp:slocc/genuine} shows that some correlations between the one-body reduced density matrices of a global pure state can only be explained by the presence of genuine $n$-partite entanglement.
Since the complement of the union of biseparable entanglement polytopes is open, the first claim in the proposition is in fact true for a small ball around the eigenvalues (intersected with the overall polytope).
See \autoref{fig:slocc/genuine} for an illustration.

\section{Gradient Flow}
\label{sec:slocc/distillation}
\subsection*{The Role of the Origin}

In \autoref{sec:onebody/git} we had seen that the pure states $\rho$ whose one-body reduced density matrices are maximally mixed, $\rho_1, \dots, \rho_n \propto \Id$, play a special mathematical role in geometric invariant theory.
Some of their properties can be reinterpreted from the perspective of entanglement theory.
For example, let $P$ be a $G$-invariant homogeneous polynomial and
\begin{equation}
\label{eq:slocc/monotone}
  M \colon \PP(\calH) \rightarrow \RR_{\geq 0},
  \rho = \proj \psi \mapsto \abs{P(\ket\psi)}^2
\end{equation}
the corresponding entanglement monotone \cite{VerstraeteDehaeneDe03}.
Then $P$ attains at $\rho$ its maximal value over all states in the entanglement class $\calX_\rho = G \cdot \rho$ \cite{Klyachko07}.
To see this, recall from \eqref{eq:onebody/hands-on kempf-ness} that for a locally maximally mixed state $\rho = \proj\psi$ the norm square of the corresponding vector $\ket\psi$ does not change to first order as we move along the $G$-orbit in Hilbert space; that is $\ket\psi$ is a critical point of $\norm{-}^2$ on its $G$-orbit.
Kempf and Ness have shown that in fact any such vector has \emph{minimal} length in its $G$-orbit \cite{KempfNess79}, so that $\norm{\Pi(g) \ket\psi} \geq \norm{\ket\psi} = 1$ for all $g \in G$. Thus,
\begin{equation*}
  M(\rho)
  = \norm{\Pi(g)\ket\psi}^k \abs{P\left(\frac {\Pi(g)\ket\psi} {\norm{\Pi(g)\ket\psi}}\right)}
  \geq \abs{P\left(\frac {\Pi(g)\ket\psi} {\norm{\Pi(g)\ket\psi}}\right)}
  = M(g \cdot \rho),
\end{equation*}
where $k$ is the degree of $P$.

From the perspective of entanglement polytopes, locally maximally mixed quantum states correspond to the point
\begin{equation*}
  \vec O = (\underbrace{\frac 1 {d_1}, \dots, \frac 1 {d_1}}_{d_1}, \dots, \underbrace{\frac 1 {d_n}, \dots \frac 1 {d_n}}_{d_n}),
\end{equation*}
which we will call the \emphindex{origin}.
We have used $\vec O$ as the origin of the coordinate systems in most of the figures in this chapter.
Indeed, if we identify an entanglement polytope $\Delta_\calX$ with the corresponding moment polytope $\Delta_K(\overline\calX)$ then the origin $\vec O$ corresponds to the point $0 \in i \mathfrak k^*$, which justifies the terminology (recall that we work with $G = \SL(\calH_1) \times \dots \times \SL(\calH_n)$).
It follows from the basic \autoref{lem:onebody/mumford origin} that
\begin{equation}
\label{eq:slocc/origin vs invariants}
  \vec O \in \Delta_\calX
  \Leftrightarrow
  \text{$\exists$ $G$-invariant homogeneous polynomial $P$ with $P(\rho) \neq 0$}
\end{equation}
In particular, any entanglement monotone of the form \eqref{eq:slocc/monotone} vanishes on quantum states whose entanglement polytope does not include the origin; such states are also called \emph{unstable}\index{unstable state} in geometric invariant theory \cite{MumfordFogartyKirwan94}.
This observation has lead to the suggestion that unstable states should be considered ``unentangled'' \cite{Klyachko07} or ``not genuinely multipartite entangled'' \cite{OsterlohSiewert05, OsterlohSiewert06}.
However they are certainly considered entangled according to the standard definition that we have adopted in this work (\autoref{sec:slocc/classification}).
There is also an interesting connection between the theory of entanglement polytopes and the selection rule of \autoref{eq:onebody/selection rule} that should be well-known mathematically:

\begin{lem}
\label{lem:slocc/pinning implies instability}
  Let $\calX = G \cdot \rho$ be the entanglement class of a quantum state $\rho$ with non-degenerate local eigenvalues that is pinned to a facet (of the maximal entanglement polytope) that does not contain the origin.
  Then $\Delta_\calX$ does not contain the origin.
\end{lem}
\begin{proof}
  Without loss of generality, we may assume that $\rho = \proj\psi$ is locally diagonal with diagonal entries ordered non-increasingly (since local unitaries change neither the local eigenvalues nor the entanglement class).
  Thus $\mu_K(\rho) \in i \mathfrak t^*_{>0}$ in the language of the moment map.
  Let $(-,H) \geq c$ denote a facet of $\Delta_K(\PP(\calH))$ that does not contain the origin, i.e., $c \neq 0$, and assume that $\rho$ is pinned to the facet, i.e., $(\mu_K(\rho),H) = c$.
  Then \autoref{lem:onebody/selection rule} asserts that $\pi(H) \ket\psi = c \ket\psi$.
  It follows that $(\omega,H) = c$ for all weights $\omega$ that appear in the decomposition $\ket\psi = \sum_\omega \ket\omega$ of $\ket\psi$ into weight vectors.
  But then
  \begin{equation*}
      \pi(\exp(tH)) \ket\psi
    = \sum_\omega e^{(\omega,H)t} \ket\omega
    = \sum_\omega e^{ct} \ket\omega
    \rightarrow 0
  \end{equation*}
  for $t \rightarrow \pm\infty$ depending on the sign of $c \neq 0$.
  Therefore, $0 \in \overline{\Pi(G) \ket\psi}$.
  But then any $G$-invariant homogeneous polynomial ought to vanish at $\rho$.
  We conclude from \eqref{eq:slocc/origin vs invariants} that $\Delta_\calX$ does not contain the origin.
\end{proof}

\autoref{lem:slocc/pinning implies instability} can be strengthened by considering facets of the entanglement polytope $\Delta_\calX$ rather than of the maximal entanglement polytope.
\subsection*{The Linear Entropy of Entanglement}

The geometric picture provided by entanglement polytopes suggests another way of quantifying entanglement.
For this, we consider the multipartite version of the \emph{linear entropy of entanglement}\index{linear entropy of entanglement|textbf}\nomenclature[QE(rho)]{$E(\rho)$}{linear entropy of entanglement} \cite{ZurekHabibPaz93, FuruyaNemesPellegrino98, BarnumKnillOrtizEtAl04},
\begin{equation}
\label{eq:slocc/linear entropy of entanglement}
  E(\rho)
  = 1 - \frac 1 n \sum_{k=1}^n \normHS{\rho_k}^2
  = 1 - \frac 1 n \sum_{k=1}^n \norm{\vec\lambda_k}_2^2,
\end{equation}
where $\normHS{X}^2 := \tr X^\dagger X$ defines the \emphindex{Hilbert--Schmidt norm} and $\norm{x}^2_2 := \sum_j \abs{x_j^2}$ the \emph{$\ell_2$-norm}\index{l_2-norm@$\ell_2$-norm}\nomenclature[Q<x _2]{$\norm{x}_2$}{$\ell_2$-norm of vector $x$}.
For qubit systems, $E(\rho)$ reduces to the Meyer--Wallach measure of entanglement \cite{MeyerWallach02} (cf.\ \cite{Brennen03, BoixoMonras08}).
We remark that $E(\rho)$ is non-zero for all entangled pure states.
By convexity, any entanglement polytope $\Delta_\calX$ contains a unique point of minimal Euclidean norm. %
The corresponding quantum states $\rho$ are those states that maximize $E(\rho)$ over all states in the closure $\overline\calX$ of the class.
Thus the maximal value of the linear entropy of entanglement can be readily computed from the entanglement polytope.

\subsection*{Gradient Flow and Distillation}

Given the linear entropy of entanglement as a means of quantifying multipartite entanglement, it is natural to ask for a corresponding distillation procedure, i.e., for a protocol that transforms a given quantum state by SLOCC operations to a state with maximal linear entropy of entanglement.
This transformation might only be possible asymptotically, as the maximum might only be attained by points in the closure proper.
Our approach will be based on maximizing $E(\rho)$ by following its gradient flow.

The gradient flow for $E(\rho)$ is closely related to the gradient flow for the norm-square of the moment map as studied by Kirwan and Ness \cite{Kirwan84, NessMumford84}.
To see this, recall that $i \mathfrak k$ consists of tuples of traceless Hermitian matrices.
We may equip $i \mathfrak k$ with the $K$-invariant inner product corresponding to the norm $\norm{(X_1,\dots,X_n)}^2 = \sum_k \normHS{X_k}^2$ and use this to identify $i \mathfrak k \cong i \mathfrak k^*$.
This amounts to identifying $\mu_K(\rho) \in i \mathfrak k^*$ with the traceless part of its one-body reduced density matrices,
\begin{equation}
\label{eq:slocc/identification traceless}
  \mu_K(\rho) \in i \mathfrak k^*
  \;\longleftrightarrow\;
   (\rho_1 - \frac {\Id} {d_1}, \dots, \rho_n - \frac {\Id} {d_n}) \in i \mathfrak k,
\end{equation}
so that
\begin{equation*}
  \norm{\mu_K(\rho)}^2
  = \sum_k \normHS{\rho_k - \frac \Id {d_k}}^2
  = \norm{(\vec\lambda_1, \dots, \vec\lambda_n) - \vec O}^2_2
  = \left(n - \sum_{k=1}^n \frac 1 {d_k} \right) - n E(\rho).
\end{equation*}
Thus the norm-square of the moment map and the linear entropy of entanglement are directly related by an affine transformation---maximizing the linear entropy of entanglement is equivalent to minimizing the norm-square of the moment map, and the gradients are proportional.

We will now review some known results on the norm-square of the moment map and its gradient flow.
All these results hold for general moment maps on projective space and we will thus use the general language of \autoref{sec:onebody/git}; see, e.g., \cite{GeorgoulasRobbinSalamon13} for a comprehensive recent exposition from the differential-geometric point of view.
The first observation is that the gradient of $\norm{\mu_K}^2$ is given by \cite{Kirwan84, NessMumford84}
\begin{equation}
\label{eq:slocc/gradient}
  \grad \norm{\mu_K(\rho)}^2 = 2 \mu_K(\rho)_\rho,
\end{equation}
where $\mu_K(\rho)_\rho$ denotes as in \autoref{sec:onebody/git} the tangent vector at $\rho$ generated by the infinitesimal action of $\mu_K(\rho) \in i \mathfrak k^*$, considered as an element of $i \mathfrak k$ by using the $K$-invariant inner product.
This follows from the calculation
\begin{equation*}
  g(\grad \norm{\mu_K}^2, -) =
  d \norm{\mu_K}^2 =
  2 \omega(J[\mu_K(\rho)_\rho], -) =
  2 g(\mu_K(\rho)_\rho, -),
\end{equation*}
where we have used \eqref{eq:onebody/moment map property} and the relation \eqref{eq:onebody/fubini-study} between the Riemannian metric $g$ and the Fubini--Study form $\omega$.
An important consequence of \eqref{eq:slocc/gradient} is that the gradient flow
\begin{equation}
\label{eq:slocc/gradient flow}
  \begin{cases}
    \dot\rho_t &= - \grad \norm{\mu_K(\rho)}^2 = - 2 \mu_K(\rho)_\rho \\
    \rho_0 &= \rho
  \end{cases}
\end{equation}
automatically stays in the $G$-orbit of $\rho$, i.e., in its entanglement class $\calX = G \cdot \rho$ at all times $t \geq 0$.
What is more, the gradient flow converges to a unique limit point $\rho_\infty = \lim_{t \rightarrow \infty} \rho_t \in \overline{\calX}$ \cite{Lerman05, GeorgoulasRobbinSalamon13}.
Since any critical point is a minimum \cite{Kirwan84, NessMumford84}, this limit point is a state that minimizes the norm-square of the moment map over all states in the orbit closure $\overline{\calX} = \overline{G \cdot \rho}$ \cite{GeorgoulasRobbinSalamon13}.
We remark that such states are unique up to the $K$-action \cite{NessMumford84}.

\bigskip

We now specialize these results to the scenario of entanglement polytopes.
By the above discussion, the gradient flow \eqref{eq:slocc/gradient flow} converges to the global maximum of the linear entropy of entanglement of all states in $\overline\calX = \overline{G \cdot \rho}$.
Remarkably, at each point this flow is given by the infinitesimal action of (the traceless part of) the one-body reduced density matrices \eqref{eq:slocc/identification traceless}.
In practice, the gradient flow needs to be implemented with finite time steps $\Delta t$, which amounts to the SLOCC transformation
\begin{equation}
\label{eq:slocc/gradient flow time step}
  \rho_t
  \mapsto \left( e^{-\Delta t \rho_{t,1}} \otimes \dots \otimes e^{-\Delta t \rho_{t,n}} \right) \cdot \rho_t
  = e^{-\Delta t \, \pi(\mu_K(\rho_t))} \cdot \rho_t
  \approx \rho_{t + \Delta t},
\end{equation}
where $\pi \colon \mathfrak g \rightarrow \mathfrak{gl}(\calH)$ denotes the infinitesimal action of the Lie algebra and where we have used that scalar multiples of the identity act trivially on projective space according to \eqref{eq:onebody/complexified action}.
To realize this scheme in the laboratory, one would start by preparing the quantum state $\rho$, measuring its one-body reduced density matrices, re-preparing, and implementing the SLOCC transformation \eqref{eq:slocc/gradient flow time step} by using local POVM measurements with Kraus operators as in \eqref{eq:slocc/povm}.
If the transformation succeeded then entanglement has been distilled.
By successively repeating this procedure with the concatenated SLOCC operations, one asymptotically arrives at a quantum state with maximal linear entropy of entanglement.
Notably, this method of entanglement distillation only requires local tomography and works on a single copy of the state at a time.
See \autoref{fig:slocc/three qubit distillation} for a numerical simulation.

\begin{figure}
  \centering
  \includegraphics[width=0.9\linewidth]{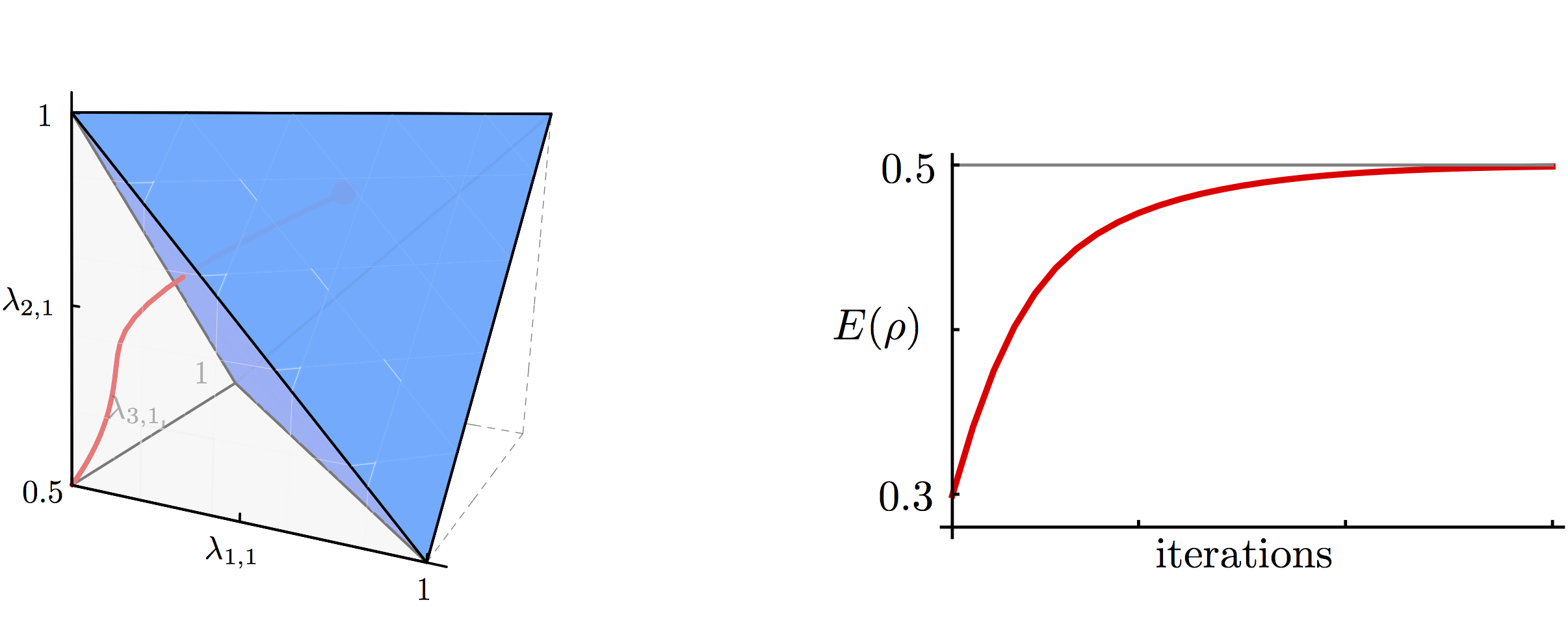}
  \begin{tabularx}{0.9\linewidth}{*2{>{\centering\arraybackslash}X}}
  {\footnotesize (a)} &
  {\footnotesize (b)}
  \end{tabularx}
  \caption[Illustration of entanglement distillation]{\emph{Illustration of entanglement distillation.}
   (a) For the three-qubit quantum state $\rho = \ketbra \psi \psi$ from \eqref{eq:slocc/distilled state}, the gradient flow (red trajectory) reaches the origin $O=(0.5,0.5,0.5)$ asymptotically.
   (b) The linear entropy of entanglement $E(\rho)$ increases monotonically until it reaches its maximal value.}
  \label{fig:slocc/three qubit distillation}
\end{figure}

\bigskip

\subsection*{Towards a Probabilistic Algorithm for Computing Moment Polytopes of Orbit Closures}

From a theoretical perspective, the limit point $\rho_\infty$ of the gradient flow can also be seen as a normal form of the state $\rho$ in its entanglement class, as it is unique up to local unitaries. This is the point of view taken in \cite{VerstraeteDehaeneDe03}, where a similar algorithm has been proposed for the case when the entanglement polytope contains the origin.
In contrast, the gradient flow works in the general case and flows towards the point in the moment polytope of minimal Euclidean norm.
We will now sketch how this idea leads towards a probabilistic algorithm for computing moment polytopes of arbitrary orbit closures.
As above, let $\norm{-}$ denote the norm corresponding to a $K$-invariant inner product $(-,-)$ on $i \mathfrak k^*$.
We start with a simple observation:

\begin{lem}
\label{lem:slocc/kostant}
  Let $\lambda, \mu \in i \mathfrak t^*_+$. Then:
  \begin{equation*}
    \min_{\varphi \in \calO_{K,\mu}} \norm{\varphi - \lambda}^2 = \norm{\mu - \lambda}^2
  \end{equation*}
\end{lem}
\begin{proof}
  By $K$-invariance of the inner product,
  \begin{equation*}
      \min_{\varphi \in \calO_{K,\mu}} \norm{\varphi - \lambda}^2
    = \norm{\mu}^2 + \norm{\lambda}^2 - 2 \max_{\varphi \in \calO_{K,\mu}} (\varphi,\lambda)
  \end{equation*}
  The right-hand side maximization can be understood as an optimization of the linear functional $(-,\lambda)$ over the Abelian moment polytope of the coadjoint orbit $\calO_{K,\mu}$.
  By Kostant's convexity theorem \cite{Kostant73}, the latter is equal to the convex hull of the orbit of $\mu$ under the Weyl group, which in turn is a subset of $\mu$ plus the cone spanned by the negative roots.
  We conclude that
  \begin{equation*}
      \max_{\varphi \in \calO_{K,\mu}} (\varphi,\lambda)
    = \max_{w \in W_K} (w \mu, \lambda)
    = (\mu, \lambda). \qedhere
  \end{equation*}
\end{proof}

Now let $\XX = \overline{G \cdot \rho} \subseteq \PP(\calH)$ be an orbit closure.
As in the proof of \autoref{prp:onebody/mumford}, fix $\lambda \in \Lambda^*_+$ and $k > 0$.
Let $\widetilde\calH := \Sym^k(\calH) \otimes V_{G,\lambda^*}$ and consider the projective subvariety
$\widetilde\XX := \{ \proj\psi^{\otimes k} \otimes \proj\phi : \proj\psi \in \XX, \proj\phi \in \XX_{G,\lambda^*} \}$ of $\PP(\widetilde\calH)$.
As a set, $\widetilde\XX$ can be identified with the Cartesian product $\XX \times \calO_{K,\lambda^*}$ equipped with the diagonal $K$-action \autoref{lem:onebody/borel weil}), and the moment map $\widetilde\mu_K$ takes the form $\widetilde\mu_K(\rho, \varphi) = k \, \mu_K(\rho) + \varphi$.
The following lemma is another instance of the ``shifting trick'' we used in the proof of \autoref{prp:onebody/mumford}.

\begin{lem}
\label{lem:slocc/shifting trick}
  Let $(\sigma,\varphi) \in \XX \times \calO_{K,\lambda^*} \cong \widetilde\XX$ be such that $\widetilde\mu_K(\sigma,\varphi)$ is a point of minimal norm in $\widetilde\mu(\widetilde\XX)$.
  If $\mu_K(\sigma) \in i \mathfrak t^*_+$ then it is the point of minimal distance to $\lambda/k$ in $\Delta_K(\XX)$.
\end{lem}
\begin{proof}
  By assumption and using that $\calO_{K,\lambda^*} = -\calO_{K,\lambda}$,
  \begin{equation*}
      \norm{\mu_K(\sigma) + \varphi/k}
    = \min_{\varphi' \in \calO_{K,\lambda^*}} \norm{\mu_K(\sigma) + \varphi'/k}
    = \min_{\varphi' \in \calO_{K,\lambda}} \norm{\mu_K(\sigma) - \varphi'/k}
    =  \norm{\mu_K(\sigma) - \lambda/k},
  \end{equation*}
  where the last equality is due to \autoref{lem:slocc/kostant}.
  On the other hand,
  \begin{align*}
      \norm{\mu_K(\sigma) + \varphi/k}
    = &\min_{\sigma' \in \XX, \varphi' \in \calO_{K,\lambda^*}} \norm{\mu_K(\sigma') + \varphi'/k}
    = \min_{\sigma' \in \XX, \varphi' \in \calO_{K,\lambda}} \norm{\mu_K(\sigma') - \varphi'/k} \\
    = &\min_{\sigma' \in \XX} \norm{\mu_K(\sigma') - \lambda/k}
    = \min_{\mu \in \Delta_K(\XX)} \norm{\mu - \lambda/k}.
  \end{align*}
  The third equality is due to the $K$-invariance of the norm; for the last equality we have used \autoref{lem:slocc/kostant} once more, but this time applied to the coadjoint orbit through the optimal $\mu_K(\sigma')$.
  Since the convex polytope $\Delta_K(\XX)$ contains a unique point of minimal Euclidean distance to $\lambda/k$, this point ought to be $\mu_K(\sigma) \in \Delta_K(\XX)$.
\end{proof}

Note that $\widetilde\XX$ is \emph{not} an orbit closure.
However, just as we found in \autoref{sec:slocc/entanglement polytopes} for projective space, the moment polytope of the orbit closure of a generic point in $\widetilde\XX$ is always equal to the moment polytope $\Delta_K(\widetilde\XX)$ \cite{Brion87}.
Since the ``boundary'' $(\XX \setminus G \cdot \rho) \times \calO_{K,\lambda^*}$ is of positive codimension in $\widetilde\XX$ and by $G$-equivariance, we may in fact start with a point of the form $(\rho, \Ad^*(g)\lambda^*)$.
If we choose $g \in K$ according to, e.g., the Haar measure on $K$ then with probability one the gradient flow will asymptotically reach the point of minimal Euclidean norm in the moment polytope of $\widetilde\XX$.
In view of \autoref{lem:slocc/shifting trick} we thus obtain a ``probabilistic oracle'' for obtaining from any $\nu = \lambda/k \in \QQ \Lambda^*_+$ an approximation to the point $\mu \in \Delta_K(\XX)$ that is closest to $\nu$.
Note that $\mu$ is always rational if the inner product takes rational values on the integral lattice \cite{Kirwan84}.
Let us assume for simplicity that the oracle does in fact never fail and always returns the exact result.
Then we obtain the following algorithm:

\begin{alg}
  \label{alg:slocc/gradient flow}
  Let $\calH$ be a complex $G$-representation and $\XX = \overline{G \cdot \rho} \subseteq \PP(\calH)$ an orbit closure.
  Let $\calV, \calW \subseteq \QQ \Lambda^*_+$ be the vertices of an inner and an outer approximation of the moment polytope, i.e.,
  \begin{equation}
  \label{eq:slocc/alg invariant}
    \conv \calV \subseteq \Delta_K(\XX) \subseteq \conv \calW.
  \end{equation}
  The following algorithm returns the moment polytope of the orbit closure (if it terminates): %
  \begin{algorithmic}
    \While{$\calV \subsetneq \calW$}
      \State $\nu \gets \calW \setminus \calV$.
      \State $\mu \gets $ the closest point in $\Delta_K(\XX)$ to $\nu$ \Comment{Oracle}
      \If{$\mu = \nu$}
      \State $\calV \gets $ vertices of $\conv \left( \calV \cup \{\nu\} \right)$
      \Else
      \State $\calW \gets $ vertices of $\conv(\calW) \cap \{ \lambda \in i \mathfrak t^* : (\lambda - \mu, \nu - \mu) \leq 0 \}$
      \EndIf
    \EndWhile
    \State \textbf{return} $\conv \calV = \conv \calW$
  \end{algorithmic}
\end{alg}
\begin{proof}
  To show that the algorithm is correct it suffices to observe that \eqref{eq:slocc/alg invariant} is a loop invariant.
  In the case $\mu = \nu$ this is immediate.
  If $\mu \neq \nu$ then the convexity of the moment polytope implies that it is contained in the half-space $\{ \lambda \in i \mathfrak t^* : (\lambda - \mu, \nu - \mu) \leq 0 \}$.
\end{proof}

\begin{figure}
  \centering
  \includegraphics[width=0.4\linewidth]{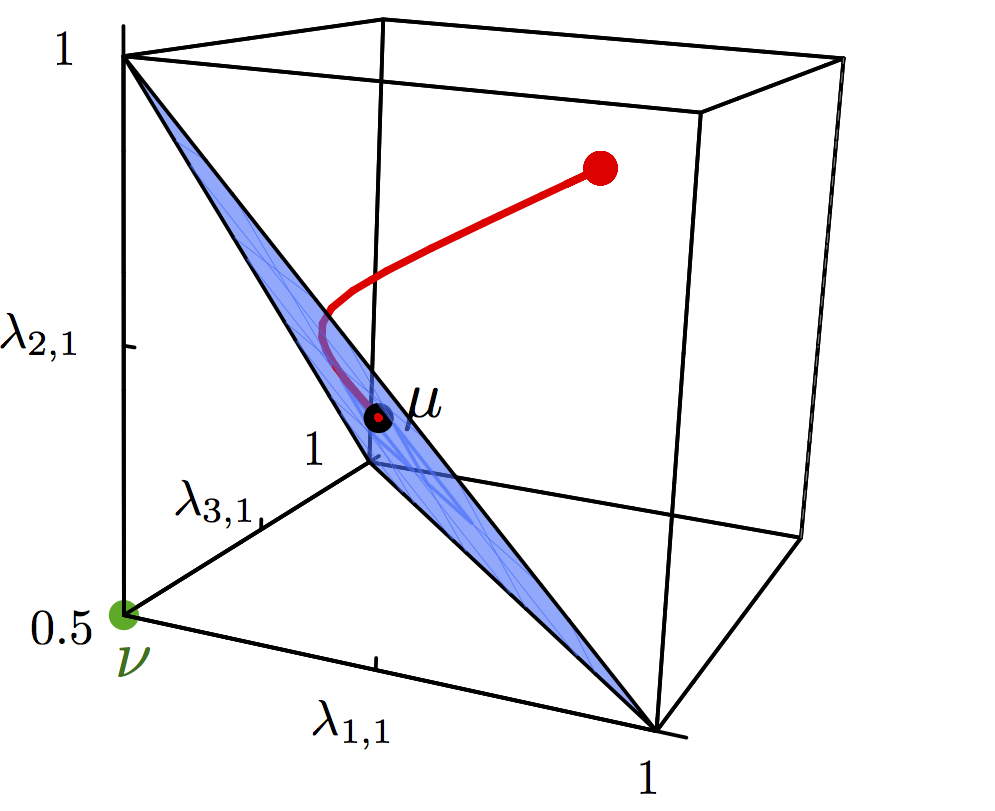}
  \caption[Illustration of Algorithm~\ref*{alg:slocc/gradient flow}]{\emph{Illustration of \autoref{alg:slocc/gradient flow}} for a state $\rho$ of W class.
    In the first step, the gradient flow tries to flow from the spectrum of $\rho$ (red point) to the origin $\nu=(0.5,0.5,0.5)$ (green point), which is a vertex of the initial outer approximation (cube).
    The gradient flow stops at $\mu = (2/3,2/3,2/3)$, which is the closest point to $\nu$ in the entanglement polytope of $\rho$.
    We thus obtain a bounding hyperplane (blue) which gives rise to a smaller outer approximation.
    By iterating this procedure we will eventually recover the polytope of the W class (\autoref{fig:slocc/three qubits}, (b)).
  }
  \label{fig:slocc/algo}
\end{figure}

One choice of outer approximation is given by the moment polytope for the action of the maximal torus $T \subseteq K$, which is equal to the convex hull of the weights of the representation $\calH$ (\autoref{sec:kirwan/torus}). See \autoref{fig:slocc/algo} for an illustration of the algorithm.
\autoref{alg:slocc/gradient flow}, its theoretical properties and implications still need be investigated more carefully; we refer to \cite{Wernli13} for an initial study and first applications of the algorithm to the computation of entanglement polytopes.
We remark that it can also be used to obtain a solution of the one-body quantum marginal problem when applied to a generic state $\rho$ (chosen, e.g., according to the unitarily invariant measure on projective space).
We remark that the gradient flow gives us in essence a \emphindex{separation oracle} \cite{GroetschelLovaszSchrijver93}.
There are geometric algorithms that work with a separation oracle, e.g., based on the ellipsoid method, which are well-studied in the optimization community (as was kindly pointed out to us by Peter B\"urgisser).
A combination of these techniques might lead to further progress towards solving the membership problem for moment polytopes.

\section{Experimental Noise}
\label{sec:slocc/noise}

A quantum state prepared in the laboratory will always be a mixed state $\rho$ and it is a priori unclear what statements can be inferred about its entanglement from its local eigenvalues. Here, we give two slightly different ways for applying the preceding results to the more realistic scenario of small noise.

For both approaches, we will assume that a lower bound $p > 1/2$ on the \emphindex{purity} $\tr\rho^2$ is available.
One natural way of obtaining such an estimate is the well-known \emphindex{swap test} \cite{BuhrmanCleveWatrousEtAl01}, which directly estimates $\tr\rho^2$ using a series of two-body measurements on two copies of $\rho$.
We sketch an alternative procedure which may be simpler to implement for some experimental platforms.
Suppose that $\rho$ has been prepared by acting on an initial product state with a quantum operation $\Lambda$ that approximates an entangling unitary gate $U$ (e.g., a spin squeezing operation).
Act on $\rho$ by a quantum operation $\Lambda'$ that approximates the corresponding ``disentangling unitary'' $U^{-1}$ and denote by $\rho' = \Lambda'(\rho)$ the state thus obtained.
Now assume that the noise mechanism never increases the purity---this holds, e.g., for dephasing and depolarizing noise, which are two noise models applicable to the majority of experiments.
Then we can use the following lower bound \cite{Audenaert07}
\begin{equation*}
  \tr \rho^2 \geq \tr (\rho')^2 \geq \sum_{k=1}^n \norm{\vec{\lambda'_k}}^2_2 - (n-1)
\end{equation*}
for the purity in terms of the local spectra $\vec{\lambda'_1}, \dots, \vec{\lambda'_n}$, which can be obtained from tomography of the single-particle reduced density matrices $\rho'_1, \dots, \rho'_n$.
For small noise, $\rho'$ is still approximately product and so this bound likely not too loose (as it is tight for product states).

\subsection*{Fidelity}

The first way of dealing with noise is to realize that in the vicinity of any mixed state $\rho$ there is a pure state whose local eigenvalues do not differ too much from those of the mixed state, provided that the purity of the mixed state is sufficiently high.
To make this precise, it is convenient to consider the \emphindex{fidelity}\nomenclature[QF]{$F(\rho, \sigma)$}{fidelity between $\rho$ and $\sigma$} between $\rho$ and an arbitrary pure state $\sigma = \proj{\psi}$, which is defined by $F(\rho, \sigma) := \braket{\psi | \rho | \psi}$.
Note that $F(\rho, \sigma) \leq 1$, with equality if and only if $\rho = \sigma$.

\begin{prp}
\label{prp:slocc/fidelity}
  Let $\rho$ be a mixed state with purity $\tr \rho^2 \geq p > 1/2$.
  Then there exists a pure state $\sigma = \proj{\psi}$ with fidelity $F(\rho, \sigma) \geq p$
  such that
  \begin{equation}
  \label{eq:slocc/fidelity bound}
    \sum_{k=1}^n \norm{\vec\lambda_k(\rho) - \vec\lambda_k(\sigma)}_1 \leq n (1 - \sqrt{2p -1}),
  \end{equation}
  where $\vec\lambda_k(\rho) := \spec \rho_k$ and where $\norm{x}_1 := \sum_j \abs{x_j}$ denotes the \emph{$\ell_1$-norm}\index{l_1-norm@$\ell_1$-norm}\nomenclature[Q<x _1]{$\norm{x}_1$}{$\ell_1$-norm of vector $x$}.
\end{prp}
\begin{proof}
  Consider the spectral decomposition $\rho = \sum_i r_i \proj {\psi_i}$ with eigenvalues ordered non-increasingly, $r_i \geq r_{i+1}$.
  Our assumption on the purity implies immediately that the maximal eigenvalue can also be lower-bounded by $1/2$:
  \begin{equation}
  \label{eq:slocc/fidelity in proof}
    r_1 = r_1 \sum_i r_i \geq \sum_i r_i^2 = \tr \rho^2 \geq p > \frac 1 2
  \end{equation}
  Thus if we set  $\sigma = \proj{\psi_1}$ then $F(\rho, \sigma) = \braket{\psi_1 | \rho | \psi_1} = r_1 \geq p$.

  On the other hand, using that $r_j \leq \sum_{i>1} r_i = 1 - r_1$ for all $j > 1$, we find that
  \begin{equation*}
    p
    \leq \tr \rho^2
    = r_1^2 + \sum_{j>1} r_j^2
    \leq r_1^2 + (1 - r_1)^2
    = 1 - 2 r_1 (1 - r_1).
  \end{equation*}
  We solve this quadratic relation and obtain two possible solutions,
  \begin{equation*}
    r_1 \leq \frac 1 2 \left( 1 - \sqrt{2p - 1} \right)
    \quad\text{and}\quad
    r_1 \geq \frac 1 2 \left( 1 + \sqrt{2p - 1} \right).
  \end{equation*}
  Only the right-hand side solution is compatible with $r_1 > 1/2$.
  By using Weyl's perturbation theorem relating the $\ell^1$-norm distance of eigenvalues with the trace-norm distance of the corresponding operators \cite[(11.46)]{NielsenChuang04} we thus obtain that
  \begin{align}
  \label{eq:slocc/weyl bound}
    &\sum_{k=1}^n \norm{\vec\lambda_k(\rho) - \vec\lambda_k(\sigma)}_1
    \leq \sum_{k=1}^n \norm{\rho_k - \sigma_k}_1 \\
  \label{eq:slocc/trace norm bound}
    \leq\;&n \norm{\rho - \proj{\psi_1}}_1
    = 2 n (1 - r_1)
    \leq n \left( 1 - \sqrt{2p - 1} \right).
  \end{align}
\end{proof}

In the case of qubits, the bound \eqref{eq:slocc/fidelity bound} can be equivalently written in terms of the maximal local eigenvalues,
\begin{equation}
\label{eq:slocc/fidelity bound qubits}
  \sum_{k=1}^n \abs{\lambda_{k,1}(\rho) - \lambda_{k,1}(\sigma)} \leq \frac n 2 (1 - \sqrt{2p - 1}).
\end{equation}
For small noise, the right-hand side is equal to $n \varepsilon / 2$ in first order in $\varepsilon = 1 - p \approx 0$.

\bigskip

We now illustrate the approach with a numerical example.
Suppose that $\rho$ is an experimentally prepared quantum state of four qubits with purity no less than $p = 0.9$.
Then by \autoref{prp:slocc/fidelity} above there exists a pure state $\sigma = \proj\psi$ with fidelity $F(\rho, \sigma) \geq 0.9$ for which \eqref{eq:slocc/fidelity bound qubits} reads
\begin{equation}
  \label{eq:slocc/fidelity bound example}
  \sum_{k=1}^4 \abs{\lambda_{k,1}(\rho) - \lambda_{k,1}(\sigma)} \leq \frac 4 2 \left( 1 - \sqrt{2 p - 1} \right) \approx 0.21.
\end{equation}
At this resolution, the differences between the various four-qubit entanglement polytopes are already well visible.
For example, suppose that we would like to use the inequality
\begin{equation*}
  \lambda_{1,1}(\sigma) + \lambda_{2,1}(\sigma) + \lambda_{3,1}(\sigma) + \lambda_{4,1}(\sigma) < 3
\end{equation*}
to deduce that $\sigma$ is not entangled of W-type (cf.\ the summary of results in this chapter).
For this, it suffices by \eqref{eq:slocc/fidelity bound example} to verify that the local eigenvalues of the experimentally realized state $\rho$ satisfy the relation
\begin{equation*}
  \lambda_{1,1}(\rho) + \lambda_{2,1}(\rho) + \lambda_{3,1}(\rho) + \lambda_{4,1}(\rho) < 3 - 0.21 = 2.79.
\end{equation*}
For comparison, the left-hand side of this inequality is equal to $2$ for a symmetric Dicke state
$\left( \ket{0011} + \ket{0101} + \ket{0110} + \ket{1001} + \ket{1010} + \ket{1100} \right) / \sqrt 6$.

\subsection*{Convex Extension}

A second, alternative approach for treating noise aims to show that the experimentally prepared mixed state $\rho$ cannot be written as a convex combination of pure states in the closure of an entanglement class $\calX$.
For this, we consider the distance between a spectrum $\vec\lambda$ and an entanglement polytope $\Delta_\calX$ defined as
\begin{equation*}
  d(\vec\lambda, \Delta_{\calX}) :=
  \min_{\vec\mu \in \Delta_{\calX}}
  \sum_{k=1}^n \norm{\vec\lambda - \vec\mu}_1.
\end{equation*}

\begin{prp}
  There exists a continuous function $\delta(p) \geq 0$ with $\delta(1) = 0$ such that
  \begin{equation*}
    d(\vec\lambda(\rho), \Delta_\calX) > \delta(p)
    \,\Longrightarrow\,
    \rho \not\in \conv \overline{\calX}
  \end{equation*}
  for all mixed states $\rho$ with purity $\tr \rho^2 \geq p > 1/2$.
\end{prp}
\begin{proof}
  Let $\delta(p) := n (1 + 2 \sqrt{1-p} - \sqrt{2 p - 1})$.
  Then $\delta$ is continuous and non-negative on $[1/2,1]$, and $\delta(1) = 0$.
  Now assume that $\rho$ is a mixed state with purity $\tr \rho^2 \geq p > 1/2$ and $d(\vec\lambda(\rho), \Delta_\calX) > \delta(p)$.
  As in the proof of \autoref{prp:slocc/fidelity}, let $\sigma = \proj{\psi_1}$ where $\ket{\psi_1}$ is an eigenvector corresponding to the maximal eigenvalue of $\rho$.
  For any $\chi = \proj\phi \in \overline\calX$, the triangle inequality gives
  \begin{equation*}
    \norm{\chi - \sigma}_1 \geq
    \norm{\chi - \rho}_1 - \norm{\rho - \sigma}_1.
  \end{equation*}
  We may lower-bound the first summand by reversing the argument of \eqref{eq:slocc/weyl bound} and using the assumption on $\rho$,
  \begin{equation*}
    \norm{\chi - \rho}_1
    \geq \frac 1 n \sum_{k=1}^n \norm{\chi_k - \rho_k}_1
    \geq \frac 1 n \sum_{k=1}^n \norm{\vec\lambda_k(\chi) - \vec\lambda_k(\rho)}_1
    \geq \frac {d(\vec\lambda(\rho), \Delta_\calX)} n
    > \frac {\delta(p)} n,
  \end{equation*}
  while the second summand can be upper-bounded as in \eqref{eq:slocc/trace norm bound},
  \begin{equation*}
    \norm{\rho - \sigma}_1 = \norm{\rho - \proj{\psi_1}}_1 \leq \left( 1 - \sqrt{2p - 1} \right).
  \end{equation*}
  Together this implies
  \begin{equation*}
    \norm{\chi - \sigma}_1 > \frac {\delta(p)} n - \left( 1 - \sqrt{2p - 1} \right) = 2 \sqrt{1-p},
  \end{equation*}
  so that by using a standard upper bound for the fidelity in terms of the trace norm \cite{FuchsGraaf99} we obtain the following estimate which holds for all $\chi = \proj\phi \in \overline\calX$:
  \begin{equation*}
    \abs{\braket{\psi_1 | \phi}}^2 \leq 1 - \frac {\norm{\chi - \sigma}_1^2} 4 < p.
  \end{equation*}
  Now assume for the sake of reaching a contradiction that $\rho$ can in fact be written as a convex combination $\rho = \sum_j p_j \proj{\phi_j}$ of pure states from $\overline\calX$. Then,
  \begin{equation*}
    \braket{\psi_1 | \rho | \psi_1} = \sum_j p_j \abs{\braket{\psi_1 | \phi_j}}^2 < p.
  \end{equation*}
  But on the other hand, $\braket{\psi_1 | \rho | \psi_1} \geq p$ by \eqref{eq:slocc/fidelity in proof}, which is the desired contradiction.
\end{proof}

\section{Discussion}
\label{sec:slocc/discussion}

\autoref{prp:slocc/description} provides a complete description of the entanglement polytopes based on a generating set of covariants.
While the latter can in principle be found using computational invariant theory, current algorithms based on Gr\"obner bases work best for low-dimensional scenarios.
It would therefore be highly desirable to find an alternative characterization that does not rely on an explicit knowledge of the covariants---for example, based on the differential-geometric approach of \autoref{ch:kirwan} or by using semistability computations in geometric invariant theory.
The fact that only the vanishing behavior of the covariants enters the description in \autoref{prp:slocc/description} indicates that such a characterization could indeed be achievable.

While the distillation method in \autoref{sec:slocc/distillation} is conceptually pleasing, it is not clear when its realization will become experimentally feasible.
In contrast, \autoref{alg:slocc/gradient flow} has already been successfully implemented and might provide a way of circumventing the combinatorial challenges faced by exact methods in higher dimensions.
Apart from its immediate applications to the marginal problem and entanglement witnessing, the computation of moment polytopes is also relevant in other disciplines such as in mathematics and in geometric complexity theory \cite{BuergisserLandsbergManivelEtAl11} (cf.\ \autoref{ch:multiplicities}), and it might be worthwhile to study our algorithm in this context.
\chapter{Random Marginals} %
\label{ch:dhmeasure}

In this chapter we consider a quantitative version of the one-body quantum marginal problem.
For a random pure state drawn from the unitarily invariant measure, we give an algorithm to compute the joint probability distribution of the eigenvalues of its one-body reduced density matrices.
We obtain the exact probability distribution by reducing to the corresponding distribution of diagonal entries, which corresponds to a quantitative version of a classical marginal problem.
This reduction is an instance of a more general principle that can be used to compute Duistermaat--Heckman measures in symplectic geometry.

The results in this chapter have been obtained in collaboration with Matthias Christandl, Brent Doran and Stavros Kousidis, and they have appeared in \cite{ChristandlDoranKousidisEtAl12}.

\section{Summary of Results}
\label{sec:dhmeasure/summary}

Let $\rho$ be a pure quantum state of $n$ particles, drawn at random according to the unitarily invariant probability measure on projective space.
We consider the problem of determining the joint distribution of its one-body reduced density matrices $\rho_1, \dots, \rho_n$, which are again random variables. This is a \emph{quantitative version} of the one-body quantum marginal problem, \autoref{pro:onebody/qmp}, and strictly generalizes the latter -- for a given collection of density matrices, we ask how likely it is to obtain them as the one-body marginals of a pure state rather than whether this is possible at all.

The starting point for our work is the observation that the joint distribution of the one-body reduced density matrices is invariant under the action of the local unitary group.
In particular, this implies that we may equivalently consider the joint distribution of their eigenvalues, $\Pspec$, which is a probability measure on the positive Weyl chamber $i \mathfrak t^*_{\geq 0}$ with support equal to the moment polytope (i.e., the solution to the one-body quantum marginal problem).%
\nomenclature[QPspec]{$\Pspec$}{joint distribution of local eigenvalues of random pure state}
The crucial fact is that in general $\Pspec$ can be recovered from the corresponding distribution of local diagonal entries $\Pdiag$ by taking a number of partial derivatives:%
\nomenclature[QPdiag]{$\Pdiag$}{joint distribution of local diagonal entries of random pure state}
\begin{equation}
  \label{eq:dhmeasure/derivative principle qmp}
  \Pspec = p(\vec\lambda_1, \dots, \vec\lambda_n) \prod_{k=1}^n \prod_{i<j}^{d_k} \left( \partial_{\lambda_{k,j}} - \partial_{\lambda_{k,i}} \right) \Pdiag \big|_{i \mathfrak t^*_{>0}},
\end{equation}
where $d_1, \dots, d_k$ are the local dimensions and where $p$ is an explicitly given polynomial (namely, a product of Vandermonde determinants).
This is an instance of a more general \emph{derivative principle} that relates invariant measures for the coadjoint action of a compact Lie group to their projections onto a Cartan subalgebra, and follows from a result by Harish-Chandra \cite{Harish-Chandr57} (\autoref{sec:dhmeasure/derivative principle}).
A similar reduction is not possible on the level of the moment polytopes. %
To compute the distribution of diagonal entries, we show in \autoref{sec:dhmeasure/abelian} that its density $f_{\diag}$ can be written as the push-forward of Lebesgue measure on a simplex along a linear map. Concretely:
\begin{equation}
\label{eq:dhmeasure/qmp cmp}
  \begin{aligned}
    f_{\diag}(\vec\lambda_1, \dots, \vec\lambda_n) = \vol \{
      (p_{i_1,\dots,i_n}) \in \RR^{d_1 \dotsm d_n}_{\geq 0} :
      \sum_{\mathclap{i_1,\dots,i_n}} p_{i_1,\dots,i_n} = 1, \\
      \sum_{\mathclap{i_2, \dots, i_n}} p_{i_1,\dots,i_n} = \vec\lambda_{1,i_1} \; (\forall i_1),
      \dots,
      \sum_{\mathclap{i_1, \dots, i_{n-1}}} p_{i_1,\dots,i_n} = \vec\lambda_{n,i_n} \; (\forall i_n)
    \}
  \end{aligned}
\end{equation}
It is amusing to note that this corresponds precisely to a quantitative marginal problem for ordinary random variables.
Any measure of the form \eqref{eq:dhmeasure/qmp cmp} is given by piecewise homogeneous polynomials on convex chambers that can be explicitly computed.
To do so algorithmically, we adapt a result of Boysal and Vergne that can be used to recursively compute closely related measures by evaluating residues \cite{BoysalVergne09}.
By putting together all ingredients, we obtain an effective algorithm for computing $\Pspec$ for an arbitrary number of particles and statistics (\autoref{sec:dhmeasure/derivative principle}).
See \autoref{fig:dhmeasure/reduction} for a summary of the method.
This generalizes previous results in the literature significantly, where exact results were only obtained for two distinguishable particles \cite{LloydPagels88, ZyczkowskiSommers01}.
In \cite{ChristandlDoranKousidisEtAl12} we also give a variant of the algorithm that can be directly applied to non-pure global spectrum (while the general case can always be reduced to the case of random pure states, such an algorithm can be useful for the manual computation of concrete examples).

From a mathematical perspective, the distributions that we compute are \emph{Duister\-maat--Heck\-man measures}, which are defined more generally using the push-forward of the Liouville measure on a symplectic manifold along the moment map \cite{Heckman82, GuilleminSternberg82a, GuilleminSternberg84, GuilleminLermanSternberg88, GuilleminLermanSternberg96, GuilleminPrato90} (\autoref{sec:dhmeasure/dhmeasure}).
For the purposes of this thesis, it will be convenient to restrict our attention to projective space; we refer to \cite{ChristandlDoranKousidisEtAl12} for an exposition from the symplectic point of view.

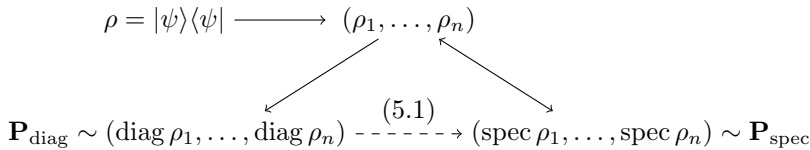
\begin{figure}
  \begin{center}
    \begin{tikzpicture}
    \draw [->] (-0.8,0) -- (0.4,0);
    \node [left] at (-0.8,0) {$\rho = \proj\psi$};
    \node at (1.5,0) {$(\rho_1, \dots, \rho_n)$};
    \draw [->] (1.1,-0.25) -- (-0.4,-1.2);
    \draw [<->] (1.9,-0.25) -- (3.4,-1.2);
    \node [left] at (0.8,-1.5) {$\Pdiag \sim (\diag \rho_1, \dots, \diag \rho_n)$};
    \node [above] at (1.5,-1.5) {\eqref{eq:dhmeasure/derivative principle qmp}};
    \draw [->, dashed] (0.8,-1.5) -- (2.2,-1.5);
    \node [right] at (2.2,-1.5) {$(\spec \rho_1, \dots, \spec \rho_n) \sim \Pspec$};
    \end{tikzpicture}
  \end{center}
  \caption[Illustration of the method for computing the distribution of random quantum marginals]{\emph{Illustration of the method.} A random pure state $\rho$ gives rise to random one-body reduced density matrices $\rho_1, \dots, \rho_n$.
  The derivative principle \eqref{eq:dhmeasure/derivative principle qmp} allows their distribution of eigenvalues $\Pspec$ to be recovered from the distribution of diagonal entries $\Pdiag$, which we compute algorithmically.}
  \label{fig:dhmeasure/reduction}
\end{figure}

\section{Duistermaat--Heckman Measures}
\label{sec:dhmeasure/dhmeasure}

We use the same notation and conventions as in \autoref{ch:onebody}.
Thus let $G$ be a connected reductive algebraic group, $K \subseteq G$ a maximal compact subgroup, and $\Pi \colon G \rightarrow \GL(\calH)$ be a representation on a Hilbert space $\calH$ with $K$-invariant inner product.
Throughout this chapter we will assume that the moment polytope $\Delta_K := \Delta_K(\PP(\calH))$ has non-empty intersection with the interior of the positive Weyl chamber, $\Delta_K \cap i \mathfrak t^*_{>0} \neq \emptyset$.
We explain below that this is without loss of generality for our application to eigenvalue distributions.

Let $\rho$ be a random pure state in $\PP(\calH)$, chosen according to the $\U(\calH)$-invariant probability measure on projective space.
Then $\mu_K(\rho)$ is contained in a unique coadjoint orbit $\calO_{K,\lambda}$, where $\lambda$ is a now random variable which takes values in the moment polytope $\Delta_K$.
We call the distribution of $\lambda$ the \emph{non-Abelian Duistermaat--Heckman probability measure}\index{Duistermaat--Heckman measure!non-Abelian}\nomenclature[TProb_K]{$\Prob_K$}{non-Abelian Duistermaat--Heckman measure} and denote it by $\Prob_K = \Prob_{K,\PP(\calH)}$.
In other words, $\Prob_K$ is defined as the push-forward of the invariant measure on projective space first along the moment map and then along the continuous map that sends $\calO_{K,\lambda} \mapsto \lambda$.
Likewise, we may also consider the image of $\rho$ under the moment map \eqref{eq:kirwan/abelian moment map} for the maximal torus $T \subseteq K$.
Then $\mu_T(\rho)$ is a random variable that takes values in the ``Abelian'' moment polytope $\Delta_T$, which we have seen is just the convex hull of weights.
We call the distribution of $\mu_T(\rho)$ the \emph{Abelian Duistermaat--Heckman probability measure}\index{Duistermaat--Heckman measure!Abelian}\nomenclature[TProb_T]{$\Prob_T$}{Abelian Duistermaat--Heckman measure} and denote it by $\Prob_T = \Prob_{T,\PP(\calH)}$.

More generally, we may consider a random quantum state $\rho$ of fixed spectrum $\mu$ and consider the corresponding Duistermaat--Heckman measures $\Prob_{K,\mu}$ and $\Prob_{T,\mu}$ constructed in the same way as before.
We will see in \eqref{eq:dhmeasure/purification} below that the computation of these measures can always be reduced to the case of random pure states.

We remark that there is a general definition of a Duistermaat--Heckman measure in symplectic geometry that goes as follows:
Let $(M, \omega)$ be a Hamiltonian $K$-manifold of dimension $2n$ and $\mu_K \colon M \rightarrow i \mathfrak k^*$ its moment map.
Then $\omega^n / (2\pi)^n n!$ is a volume form on $M$ that determines the Liouville measure of $M$.
By pushing forwarding along the moment map $\mu_K$ and further along the map $\calO_{K,\lambda} \mapsto \lambda$ we obtain the Duistermaat--Heckman measure for the $K$-action on $M$.
If $M$ is compact then we can renormalize to obtain a probability measure $\Prob_{K,M}$ on the moment polytope.
We remark that in our definition of the non-Abelian Duistermaat--Heckman measure in \cite{ChristandlDoranKousidisEtAl12} we had furthermore divided the push-forward measure at each point by the Liouville volume of the respective coadjoint orbit, given by%
\nomenclature[Tp_K(lambda)]{$p_K(\lambda)$}{Liouville volume of coadjoint orbit}
\begin{equation}
\label{eq:dhmeasure/kks volume}
  p_K(\lambda) = \prod_{\alpha \in R_{K,+}} \frac {(\lambda, H_\alpha)} {(\rho_K, H_\alpha)}
\end{equation}
where $\rho_K = \frac 1 2 \sum_{\alpha \in R_{K,+}} \alpha$ is the \emphindex{Weyl vector}\nomenclature[Rrho_K]{$\rho_K$}{Weyl vector} \cite[Proposition 7.26]{BerlineGetzlerVergne03}.
This is conceptually more appealing since it corresponds to ``intersecting'' with the positive Weyl chamber -- just as in the definition of the moment polytope! -- and it makes the formulas slightly cleaner.
However, it comes at the expense of working with measures that are not probability measures, and we have chosen not to adopt this convention herein.

\subsection*{Eigenvalues and Diagonal Entries}

Let us now consider the groups and representations that correspond to the one-body quantum marginal problem and its variants (cf.\ \autoref{tab:onebody/summary}).
In the case of distinguishable particles, $G = \SL(\calH_1) \times \dots \times \SL(\calH_n)$, $K = \SU(\calH_1) \times \dots \times \SU(\calH_n)$, and $\calH = \calH_1 \otimes \dots \otimes \calH_n$.
As in the preceding chapters, we may identify the collection of local eigenvalues with points in the positive Weyl chamber $i \mathfrak t^*_+$; likewise, the collection of local diagonal entries can be identified with points in $i \mathfrak t^*$ (\autoref{sec:onebody/consequences}).
In this way, the eigenvalue distribution $\Pspec$ and the distribution of local diagonal entries $\Pdiag$ introduced in \autoref{sec:dhmeasure/summary} are identified with the Duistermaat--Heckman probability measures $\Prob_K$ and $\Prob_T$, respectively.

\bigskip

For two distinguishable particles, $\calH = \CC^d \otimes \CC^d$, the solution is well-known \cite{LloydPagels88, ZyczkowskiSommers01}:
\begin{equation}
\label{eq:dhmeasure/two}
  \int d\Pspec \, f
  = \frac 1 Z \int_{\Delta_{d,+}} d\vec\lambda_A \; p^2_d(\vec\lambda_A) \, f(\vec\lambda_A, \vec\lambda_A)
\end{equation}
for any test function $f = f(\vec\lambda_A, \vec\lambda_B)$.
Here,
\begin{equation*}
  \Delta_{d,+} = \{ \vec\lambda \in \RR^d_{\geq 0} : \sum_j \lambda_j = 1, \, \lambda_1 \geq \dots \geq \lambda_d \}
\end{equation*}
denotes the non-increasingly ordered chamber in the $(d-1)$-dimensional standard simplex\nomenclature[TDelta_d,+]{$\Delta_{d,+}$}{non-increasingly ordered chamber in $\Delta_d$}, which can be identified with the set of possible spectra of a density operators on $\CC^d$;
\begin{equation}
  \label{eq:dhmeasure/kks volume unitary}
  p_d(\vec\lambda) = \prod_{i < j} \frac {\lambda_i - \lambda_j} {j - i}
\end{equation}
is the volume polynomial \eqref{eq:dhmeasure/kks volume} for $K = \SU(d)$,
and $Z > 0$ is a suitable normalization constant.
Note that the measure $\Pspec$ is supported on the ``diagonal'' $\vec\lambda_A = \vec\lambda_B$, in agreement with \autoref{lem:onebody/two}.
We will later give a simple proof of \eqref{eq:dhmeasure/two} using the methods of this chapter (see \eqref{eq:dhmeasure/lloyd-pagels both marginals}).
The corresponding distribution of the one-body reduced density matrix $\rho_A$ is known as the \emph{Hilbert--Schmidt probability measure}\index{Hilbert--Schmidt probability measure|textbf}.\nomenclature[Qdrho_A]{$d\rho_A$}{Hilbert--Schmidt probability measure}

\bigskip

Just as \autoref{lem:onebody/two} did for the one-body quantum marginal problem, its quantitative version \eqref{eq:dhmeasure/two} can be used to reduce the seemingly more general problem of computing the local eigenvalue distribution for random states with fixed global spectrum to the case of global pure states.
More generally, let $\calH_0$ be a $K_0$-representation (e.g., one from \autoref{tab:onebody/summary}) and $\calH = \calH_0 \otimes \calH_0$ the corresponding representation of $K = K_0 \times \SU(\calH_0)$ (its ``purification'').
Then \eqref{eq:dhmeasure/two} implies that the probability measure $P_K$ for pure states can be obtained as a direct integral of the measures $P_{K_0,\mu}$ for fixed global spectrum $\mu$,
\begin{equation}
\label{eq:dhmeasure/purification}
  \int d\Prob_K \, f
  = \frac 1 Z \int_{\Delta_{d,+}} d\vec\mu \; p^2_d(\vec\mu) \int d\Prob_{K_0,\mu}(\vec\lambda) \; f(\vec\lambda, \vec\mu),
\end{equation}
where $d = \dim \calH_0$.
Conversely, since the probability distributions $\Prob_{K_0,\mu}$ vary continuously with the global spectrum $\mu$, they can be reconstructed from $\Prob_K$ by taking limits.

\bigskip

Our assumption that the moment polytope $\Delta_K$ has non-empty intersection with the interior of the positive Weyl chamber amounts to showing that there exists a global pure state $\rho$ whose reduced density matrices all have non-degenerate eigenvalue spectrum.
We first give a criterion in the case of distinguishable particles:

\begin{lem}
\label{lem:dhmeasure/criterion distinguishable}
  Let $n \geq 1$ and $d_1 \leq \dots \leq d_n \leq d_{n+1}$.
  Then there exists a global pure state $\rho$ on $\calH = \CC^{d_1} \otimes \dots \otimes \CC^{d_{n+1}}$ such that $\rho_1, \dots, \rho_{n+1}$ all have non-degenerate spectrum if and only if $d_{n+1} \leq \left( \prod_{k=1}^n d_k \right) + 1.$
\end{lem}
\begin{proof}
  The condition is clearly necessary, since it follows from \autoref{lem:onebody/two} that at most $\prod_{k=1}^n d_k$ eigenvalues of $\rho_{n+1}$ are non-zero.
  To show that it is sufficient, consider the following mixed state on $\CC^{d_1} \otimes \dots \otimes \CC^{d_n}$,
  \begin{equation*}
    \rho_{1,\dots,n} \propto \sum_{j=1}^{d_n} 2^{-j} \proj{j}_1 \otimes \dots \otimes \proj{j}_n,
  \end{equation*}
  where we set $\proj j_k := \proj{d_k}$ if $j > d_k$.
  It is not hard to see that each $\rho_k$ has non-degenerate spectrum.
  The same remains true if we perturb $\rho_{1,\dots,n}$ slightly such that it has non-degenerate global spectrum with all eigenvalues positive.
  Let $\rho$ denote a purification of $\rho_{1,\dots,n}$ on $\calH$.
  Then \autoref{lem:onebody/two} shows $\rho_{n+1}$ and therefore all one-body reduced density matrices of $\rho$ have non-degenerate spectrum.
\end{proof}

Note that the conditions of \autoref{lem:dhmeasure/criterion distinguishable} are always satisfied for the purification, where $d_{n+1} = d_1 \dotsm d_n$.
The following lemma shows that the same is true for the one-body $n$-representability problem:

\begin{lem}
\label{lem:dhmeasure/criterion fermions}
  Let $d > n \geq 1$.
  Then there exists a mixed state $\rho$ on $\Alt^n \CC^d$ such that both $\rho$ and $\rho_1$ have non-degenerate spectrum.
\end{lem}
\begin{proof}
  The weights of $\Alt^n \CC^d$ can be identified with fermionic occupation numbers, i.e., binary strings $\vec\omega \in \{0,1\}^d$ with precisely $n$ ones.
  As long as $1 \leq n < d$, we can find $d$ linearly independent such weights.
  Thus their convex hull is maximal-dimensional and we can find a convex combination $\sum_j p_j \, \vec\omega_j \in i \mathfrak t^*_{>0}$.
  Then the one-body reduced density matrix $\rho_1$ of $\rho = \sum_j p_j \proj{\omega_j}$ has non-degenerate spectrum, and by slightly perturbing $\rho$ we can arrange that the same is true for the global state.
\end{proof}

For bosons, we have already seen that any one-body reduced density matrix is consistent with a global pure state.
\subsection*{The Post-Selection Bound}

We briefly digress to prove a useful property of the Hilbert--Schmidt probability measure that will be needed in \autoref{ch:strong6j}.
Let $P_{k,d}$ denote the projector onto the symmetric subspace $\Sym^k(\CC^d) \subseteq (\CC^d)^{\otimes k}$ and $d_{k,d}$ its dimension.
Then we have that
\begin{equation}
\label{eq:dhmeasure/integral projector}
  \frac {P_{k,d}} {d_{k,d}} = \int_{\PP(\CC^d)} d\rho \, \rho^{\otimes k},
\end{equation}
where $d\rho$ the probability distribution of a random pure state on $\CC^d$.
To see this, note that the right-hand side is an operator on the irreducible $\U(d)$-representation $\Sym^k(\CC^d)$ that commutes with the group action, and therefore proportional to $P_{k,d}$ by Schur's lemma \eqref{eq:onebody/schurs lemma}. Thus
\eqref{eq:dhmeasure/integral projector} follows from the observation that both the left and the right-hand side have trace one.

Now let $\sigma_{A^k}$ be a permutation-invariant density operator on $(\CC^a)^{\otimes k}$.
Then there exists a purification $\sigma_{(AA')^k}$ of $\sigma_{A^k}$ that is supported on the symmetric subspace $\Sym^k(\CC^a \otimes \CC^a) \subseteq (\CC^a \otimes \CC^a)^{\otimes k} = (\CC^a)^{\otimes k} \otimes (\CC^a)^{\otimes k}$ \cite{ChristandlKoenigMitchisonEtAl07, Renner05}.
Indeed, by Schur's lemma, $\sigma_{A^k}$ is of the form $\bigoplus_{\alpha \vdash_a k} \sigma_\alpha \otimes \Id_{[\alpha]}/\!\dim [\alpha]$ with respect to the Schur--Weyl decomposition $\bigoplus_{\alpha \vdash_a k} V^a_\alpha \otimes [\alpha]$ \eqref{eq:onebody/schur-weyl}. In view of \eqref{eq:onebody/schur-weyl bipartite}, the desired purification can thus be constructed by purifying each part $\sigma_\alpha$.
Together with \eqref{eq:dhmeasure/integral projector} for $d=a^2$, we find that
\begin{equation*}
  \sigma_{(AA')^k} \leq P_{k,a^2} = d_{k,a^2} \int_{\PP(\CC^a \otimes \CC^a)} d\rho_{AA'} \, \rho_{AA'}^{\otimes k}.
\end{equation*}
where ``$\leq$'' denotes the positive semidefinite order of Hermitian matrices. Since the latter is preserved by the partial trace, we obtain the following bound, which holds for any permutation-invariant density operator $\sigma_{A^k}$:
\begin{equation}
\label{eq:dhmeasure/postselection}
  \sigma_{A^k} \leq d_{k,a^2} \int d\rho_A \, \rho_A^{\otimes k},
\end{equation}
where $d\rho_A$ is the Hilbert--Schmidt probability measure\index{Hilbert--Schmidt probability measure} on the set of density matrices on $\CC^a$.
It is a remarkably useful tool to lift permutation-invariance to tensor-product form at the expense of a factor $d_{k,a^2} = \binom{k+a^2-1}{k}$ which grows only polynomially with $k$ (for fixed local dimension $a$).
A more refined version of this observation is known as the \emphindex{post-selection technique} \cite{ChristandlKoenigRenner09, Renner10} (cf.\ \cite{Hayashi10}).

\section{The Abelian Measure}
\label{sec:dhmeasure/abelian}

We now compute the Abelian Duistermaat--Heckman probability measure $\Prob_T$ for the projective space $\PP(\calH)$.
Let $\Omega$ denote the set of weights of the representation $\calH$.
Choose an orthonormal basis $\ket k$ of weight vectors with corresponding weights $\omega_k \in \Omega$.
According to \eqref{eq:kirwan/abelian convex combination}, the image of a pure state $\rho = \proj\psi$ with $\ket\psi = \sum_k \psi_k \ket k$ is given by the convex combination
\begin{equation}
\label{eq:dhmeasure/abelian convex combination}
  \mu_T(\rho) = \sum_k \abs{\psi_k}^2 \omega_k.
\end{equation}
Let \[\Delta_D := \{ (p_k) \in \RR^D_{\geq 0} : \sum_k p_k = 1 \}\] denote the $(D-1)$-dimensional \emphindex{standard simplex}\nomenclature[TDelta_d]{$\Delta_d$}{$(d-1)$-dimensional standard simplex in $\RR^d$}, where $D = \dim \calH$.
In the following, a \emphindex{Lebesgue measure} on a closed convex body is the restriction of a translation-invariant measure on its affine hull (such a measure is unique up to normalization).

\begin{lem}
\label{lem:dhmeasure/abelian}
  The Abelian Duistermaat--Heckman probability measure $\Prob_T$ is equal to the push-forward of Lebesgue measure $dp$ on the standard simplex $\Delta_D$, normalized to probability one, along the linear map $\Phi \colon (p_k) \mapsto \sum_k p_k \omega_k$.
\end{lem}
\begin{proof}
  The invariant probability measure on $\PP(\calH)$ is the push-forward of the usual round measure on the unit sphere $S^{2D-1} \cong \{ \ket\psi : \braket{\psi | \psi} = 1 \}$, normalized to total measure one, along the quotient map $\ket\psi \mapsto \proj\psi$.
  On the other hand, the round measure on the unit sphere also induces Lebesgue measure on the standard simplex $\Delta_D$ by pushing forward along the map $\ket\psi \mapsto (\abs{\psi_k}^2)$ (polar coordinates!).
  Thus the claim follows from \eqref{eq:dhmeasure/abelian convex combination}.
\end{proof}

Now fix a Lebesgue measure $d\lambda$\nomenclature[Tdlambda]{$d\lambda$}{Lebesgue measure on $\Delta_T$} on the affine hull of Abelian moment polytope $\Delta_T = \conv \Omega$.
Away from the boundary of the standard simplex, the map $(p_k) \mapsto \sum_k p_k \omega_k$ is a submersion onto this affine space, and therefore the measure $\Prob_T$ is absolutely continuous with respect to $d\lambda$.
In fact, its density function $f_T(\lambda)$ is given by
\begin{equation}
\label{eq:dhmeasure/abelian as volume}
\begin{aligned}
  f_T(\lambda)
  = &\vol~(\Phi^{-1}(\lambda) \cap \Delta_D) \\
  = &\vol~\{ (p_k) \in \RR^D_{\geq 0} : \sum_k p_k = 1, \sum_k p_k \omega_k = \lambda \},
\end{aligned}
\end{equation}
i.e., as the volume\nomenclature[Tvol]{$\vol$}{volume of parametrized polytope} of a \emphindex{parametrized polytope}, measured with respect to a Lebesgue measure $dp/d\lambda$ on the fiber $\Phi^{-1}(\lambda) \cap \Delta_D$ that is normalized such that $dp = dp/d\lambda \wedge d\lambda$ (we freely identify differential pseudo-forms and the measures induced by them).
For the one-body quantum marginal problem, \eqref{eq:dhmeasure/abelian as volume} reduces to \eqref{eq:dhmeasure/qmp cmp}, the quantitative version of a classical marginal problem that we discussed in the introduction.

\bigskip

In special case that $D - 1 = \dim \Delta_T$, the map $\Phi$ maps the standard simplex bijectively onto the Abelian moment polytope $\Delta_T$ (which in this case is itself a simplex). It follows that
\begin{equation}
\label{eq:dhmeasure/abelian simplex}
  \Prob_T = \frac {(D-1)!} {d\lambda(\omega_2 - \omega_1, \dots, \omega_D - \omega_1)} d\lambda,
\end{equation}
where the denominator denotes the volume of the parallelotope spanned by the $(\omega_j - \omega_1)$ as measured by $d\lambda$.
In particular, the Abelian Duistermaat--Heckman measure for the maximal torus of $\U(\calH)$ is simply Lebesgue probability measure on the standard simplex, which is in agreement with \autoref{lem:dhmeasure/abelian}.

In the general case, $D - 1 > \dim \Delta_T$.
Consider subsets of weights whose convex hull have codimension one in $\Delta_T$.
These convex hulls are precisely the critical points of the Abelian moment map, viewed as a map onto the affine hull of $\Delta_T$ (the proof of \autoref{lem:kirwan/abelian critical} works just as well if $\Delta_T$ is of positive codimension).
They cut $\Delta_T$ into convex polytopes that we will call the \emph{regular chambers}\index{regular chamber}.
We will also considered the complement of $\Delta_T$ as a chamber, called the \emph{unbounded chamber}\index{regular chamber!unbounded}, although it is not convex.
If the closures of two chambers have a common boundary of maximal dimension (i.e., of codimension one) then we shall say that the two chambers are \emph{adjacent}\index{regular chamber!adjacent}.
The affine hyperplane spanned by their common boundary will be called a \emphindex{critical wall}.
The significance of these notions is the following:
Inside each regular chamber, the vertices of the convex polytopes $\Phi^{-1}(\lambda) \cap \Delta_D$ can be parametrized as affine functions of $\lambda$ (see, e.g., \cite{ClaussLoechner98, VerdoolaegeSeghirBeylsEtAl07} for details).
It follows that the density function $f_T(\lambda)$ is given by a polynomial of degree at most $(D - 1) - r_T$ on (the interior of) each regular chamber, where we set $r_T = \dim \Delta_T$.
Thus we say that $f_T(\lambda)$ is a \emph{piecewise polynomial function}\index{piecewise polynomial}.
As we cross a critical wall, this polynomial will in general change, and we will now describe a way to compute these ``jumps''.

\subsection*{The Boysal--Vergne--Paradan Jump Formula}

In \cite{BoysalVergne09}, Boysal and Vergne have analyzed the push-forward of Lebesgue measure on the cone $\RR^D_{\geq 0}$ along a linear map. In view of \autoref{lem:dhmeasure/abelian}, such measures are closely related to the Abelian Duistermaat--Heckman measure $\Prob_T$, and it is straightforward to translate their \cite[Theorem 4.3]{BoysalVergne09} into the projective scenario.
We refer to \cite{ChristandlDoranKousidisEtAl12} for a detailed derivation and only present the result:

Consider two adjacent chambers separated by a critical wall $(-,H) = c$.
Let $\Omega_0 := \{ \omega \in \Omega : (\omega, H) = c \}$ denote the set of weights on the wall,
$\calH_0 = \bigoplus_{\omega \in \Omega_0} \calH_\omega$ the corresponding sum of weight spaces,
and $D_0 := \dim \calH_0$ its dimension.
Let $d\lambda_0$ be Lebesgue measure on the affine hyperplane $(-,H)=c$, normalized such that
\begin{equation}
\label{eq:dhmeasure/wall lebesgue}
  d\lambda = d\lambda_0 \wedge d(-,H).
\end{equation}
The common boundary between the two chambers is contained in a single regular chamber for $\PP(\calH_0)$.
Thus the Abelian Duistermaat--Heckman measure $\Prob_{T,\PP(\calH_0)}$ is there given by a polynomial density function with respect to $d\lambda_0$.
Let $f_{T,0}(\lambda)$ denote any polynomial that extends this density to all of $i \mathfrak t^*$ and consider its ``homogeneous extension'' $F_{T,0}(\lambda, t) := t^{D_0 - r_T} f_{T,0}(\lambda/t)$.
Finally, let $f_{T,\pm}$ denote the polynomial density function of $\Prob_{T,\PP(\calH)}$ on the chamber on the positive and negative side of the critical wall, respectively.
Then the jump over the wall is given by \cite[Theorem 4.3]{BoysalVergne09}
\begin{equation}
\label{eq:dhmeasure/abelian jump formula}
\begin{aligned}
  &f_{T,+}(\lambda) - f_{T,-}(\lambda) = \frac {(D-1)!} {(D_0-1)!} \\
  \times \Res_{z=0} &\left(
  F_{T,0}(\partial_X, \partial_t)
  \frac
    {e^{z ((\lambda, H) - c) + (\lambda, X) + t}}
    {\prod_{k : \omega_k \not\in \Omega_0} z ((\omega_k, H) - c) + (\lambda, X) + t)}
  \bigg|_{X=0, t=0}\right),
\end{aligned}
\end{equation}
where $\Res_{z=0} g = a_{-1}$ denotes the \emphindex{residue}\nomenclature[TRes_z]{$\Res_{z=0}$}{residue of formal Laurent series} of a formal Laurent series $g = \sum_k a_k z^k$ (it appears as part of an inversion formula for the Laplace transform).
In \autoref{sec:dhmeasure/examples} we will give many illustrations of how to use \eqref{eq:dhmeasure/abelian jump formula}.

\bigskip

The jump formula \eqref{eq:dhmeasure/abelian jump formula} can be simplified in case only a minimal number of weights lie on the wall (that is, if $D_0 = r_T$).
In this case, it follows from \eqref{eq:dhmeasure/abelian simplex} and \eqref{eq:dhmeasure/wall lebesgue} that the density on the wall is constant,
\begin{equation*}
  F_{T,0} = f_{T,0}
  = \frac {(D_0-1)!} {d\lambda(\tilde\omega_2-\tilde\omega_1, \dots, \tilde\omega_{D_0}-\tilde\omega_1, \xi)},
\end{equation*}
where $\Omega_0 = \{\tilde\omega_1, \dots, \tilde\omega_{D_0}\}$ and where $\xi \in i \mathfrak t^*$ is chosen such that $(\xi, H) = 1$.
Thus the jump \eqref{eq:dhmeasure/abelian jump formula} across the critical wall simplifies to (cf.\ \cite{GuilleminLermanSternberg88}):
\begin{equation}
\label{eq:dhmeasure/abelian jump formula minimal}
  \begin{aligned}
  &f_{T,+}(\lambda) - f_{T,-}(\lambda)
  = f_{T,0} \frac {(D-1)!} {(D_0-1)!}
  \Res_{z=0} \left(
  \frac
    {e^{z ((\lambda,H) - c)}}
    {\prod_{k : \omega_k \not\in \Omega_0} z ((\omega_k, H) - c)}
  \right) \\
  = &\frac 1 {d\lambda(\tilde\omega_2-\tilde\omega_1, \dots, \tilde\omega_{D_0}-\tilde\omega_1, \xi)} \frac {(D-1)!} {(D-D_0-1)!}
  \frac
    { ((\lambda,H) - c)^{D-D_0-1}}
    {\prod_{k : \omega_k \not\in \Omega_0} ((\omega_k, H) - c)}
\end{aligned}
\end{equation}

Similarly, we may also consider the case where the critical wall is zero-dimensional (that is, if $r_T - 1 = 0$).
In this case, $\Omega_0$ consists of a single point and the Abelian Duistermaat--Heckman measure for $\Prob_{T,\PP(\calH_0)}$ is the Dirac measure at this point. Thus our normalization \eqref{eq:dhmeasure/wall lebesgue} implies that
\begin{equation*}
  f_{T,0} \equiv \frac 1 {d\lambda(\xi)}
  \text{ and }
  F_{T,0}(t) = \frac {t^{D_0-1}} {d\lambda(\xi)},
\end{equation*}
where $\xi \in i \mathfrak t^*$ is again chosen such that $(\xi,H)=1$, and \eqref{eq:dhmeasure/abelian jump formula} simplifies to
\begin{equation}
\label{eq:dhmeasure/abelian jump formula zero-dimensional}
\begin{aligned}
  &f_{T,+}(\lambda) - f_{T,-}(\lambda) \\
   = &\frac 1 {d\lambda(\xi)} \frac {(D-1)!} {(D_0-1)!} \Res_{z=0} \left(
  \partial_t^{D_0-1}
  \frac
    {e^{z ((\lambda, H) - c) + t}}
    {\prod_{k : \omega_k \not\in \Omega_0} z ((\omega_k, H) - c) + t)}
  \bigg|_{t=0}\right).
\end{aligned}
\end{equation}

\subsection*{Algorithm}

Equations~\eqref{eq:dhmeasure/abelian jump formula}, \eqref{eq:dhmeasure/abelian jump formula minimal} and \eqref{eq:dhmeasure/abelian jump formula zero-dimensional} together can be used to recursively compute the Abelian Duistermaat--Heckman measure $\Prob_T$ for arbitrary projective spaces. We sketch the procedure in the following algorithm:

\begin{alg}
\label{alg:dhmeasure/abelian}
  The following recursive algorithm computes the piecewise polynomial density function $f_T$ of the Abelian Duistermaat--Heckman measure $\Prob_T = \Prob_{T,\PP(\calH)}$, where $\Omega$ denotes the set of weights of $\calH$:
  \begin{algorithmic}
    \Function{abelian}{$\calH$, $\Omega$}
    \State{compute the decomposition of $\Delta_T = \conv \Omega$ into regular chambers}
    \State{$r_T \gets \dim \Delta_T$}
    \State{$f_{T} \gets 0$ on the unbounded chamber}
    \While{$\exists$ regular chamber where the density is not known}
      \State{find adjacent chambers with known polynomial density $f_{T,-}$ on one chamber and unknown polynomial density $f_{T,+}$ on the other}
      \State{$\Omega_0 \gets $ weights on the critical wall separating the chambers}
      \State{$\calH_0 \gets \bigoplus_{\omega \in \Omega_0} \calH_\omega$, $D_0 \gets \dim \calH_0$}
      \If{$D_0 = r_T$} \Comment{Minimal wall}
        \State{$f_{T,+} \gets f_{T,-} + \eqref{eq:dhmeasure/abelian jump formula minimal}$}
      \ElsIf{$r_T = 1$} \Comment{Zero-dimensional wall}
        \State{$f_{T,+} \gets f_{T,-} + \eqref{eq:dhmeasure/abelian jump formula zero-dimensional}$}
      \Else
        \State{$f_{T,0} \gets \textsc{abelian}(\calH_0, \Omega_0)$} \Comment{Recursion}
        \State{$F_{T,0} \gets t^{D_0-r_T} f_{T,0}(\lambda/t)$}
        \State{$f_{T,+} \gets f_{T,-} + \eqref{eq:dhmeasure/abelian jump formula}$} \Comment{General wall}
      \EndIf
      \State \textbf{return} $f_T$
    \EndWhile
    \EndFunction
  \end{algorithmic}
\end{alg}

\bigskip

We remark that exist other algorithms that can be used to compute the volume of parametrized polytopes (see, e.g., \cite{Verdoolaege14, ClaussLoechner98, VerdoolaegeSeghirBeylsEtAl07} and references therein).
In this chapter we will not pursue this route any further.
However, in \autoref{ch:multiplicities} we will use Barvinok's algorithm to solve the corresponding discrete problem of counting the number of integral points to give an efficient algorithm for computing multiplicities in Lie group representations.

\section{The Derivative Principle}
\label{sec:dhmeasure/derivative principle}

In this section we describe a way to obtain the non-Abelian Duistermaat--Heckman measure $\Prob_K$ from the Abelian measure $\Prob_T$ that we have studied in the preceding section.

We start by considering the following model problem:
Let $\varphi$ be a random point in $\calO_{K,\lambda}$, chosen according to the unique $K$-invariant probability measure on the coadjoint orbit.
Then its restriction $\varphi|_{i \mathfrak t}$ is a random variable that takes values in $i \mathfrak t^*$, and we denote its distribution by $\Prob_{T,\calO_{K,\lambda}}$.%
\nomenclature[TProb_T,O_K,lambda]{$\Prob_{T,\calO_{K,\lambda}}$}{Abelian Duistermaat--Heckman measure for coadjoint orbit}
We remark that $\Prob_{T,\calO_{K,\lambda}}$ is a Duistermaat--Heckman measure in the more general sense sketched in \autoref{sec:dhmeasure/dhmeasure}.
For $K = \U(d)$, it can be identified with the distribution of diagonal entries of a random Hermitian matrix with spectrum $\lambda$.

In the case where $\lambda \in i \mathfrak t^*_{>0}$, Harish-Chandra has proved the following fundamental formula for the Fourier transform \cite[Theorem 2]{Harish-Chandr57}:
For all $X \in i \mathfrak t$ not orthogonal to a root,
\begin{equation}
\label{eq:dhmeasure/harish-chandra}
  \int d\Prob_{T,\calO_{K,\lambda}}(\mu) e^{i(\mu,X)}
  = \frac 1 {p_K(\lambda)} \sum_{w \in W_K} (-1)^{l(w)} \frac {e^{i(w \lambda, X)}} {\prod_{\alpha \in R_{G,+}} i(\alpha, X)}
\end{equation}
where the factor $p_K(\lambda)$ is defined in \eqref{eq:dhmeasure/kks volume}.
Since partial derivatives in real space correspond to multiplication in Fourier space, \eqref{eq:dhmeasure/harish-chandra} implies that \cite{Heckman82}
\begin{equation}
\label{eq:dhmeasure/heckman}
  p_K(\lambda) \left( \prod_{\alpha \in R_{G,+}} \partial_{-\alpha} \right) \Prob_{T,\calO_{K,\lambda}}
  =
  \sum_{w \in W_G} (-1)^{l(w)} \, \delta_{w \lambda}
\end{equation}
where $\delta_{w\lambda}$ is the \emphindex{Dirac measure} at a point $w\lambda$ and where $\partial_{-\alpha}$ denotes the \emph{partial derivative}\nomenclature[Rdelta_alpha]{$\partial_{-\alpha}$}{partial derivative in direction $-\alpha$} in direction $-\alpha$ in the sense of distributions.
In other words, for any smooth and compactly supported test function $g$ on $i \mathfrak t^*$ we have that
\begin{equation}
\label{eq:dhmeasure/heckman for test functions}
  p_K(\lambda) \int d\Prob_{T,\calO_{K,\lambda}} \left( \prod_{\alpha \in R_{G,+}} \partial_\alpha \right) g
  =
  \sum_{w \in W_G} (-1)^{l(w)} \, g(w \lambda).
\end{equation}
In fact, $\Prob_{T,\calO_{K,\lambda}}$ is the unique compactly supported solution to the differential equation \eqref{eq:dhmeasure/heckman} \cite{Heckman82}, and it can be explicitly written as an alternating sum of convolutions of Heaviside distributions in the directions of the positive roots \cite{GuilleminLermanSternberg96}.
To lift Harish-Chandra's result to general Duistermaat--Heckman measures, we need the following observation:

\begin{lem}
  \label{lem:dhmeasure/abelian via coadjoint}
  For all smooth, compactly supported test functions $f$ on $i \mathfrak t^*$,
  \begin{equation*}
    \int d\Prob_T \, f = \int d\Prob_K(\lambda) \int d\Prob_{T,\calO_{K,\lambda}} \, f.
  \end{equation*}
\end{lem}
\begin{proof}
  Let $\rho$ be a random pure state in $\PP(\calH)$ and denote by $\Qrob$ the probability distribution of its image under the non-Abelian moment map $\mu_K$.
  Since $\Qrob$ is a $K$-invariant measure, we have that
  \begin{equation*}
    \int d\Qrob \, h = \int d\Prob_K(\lambda) \int_K dg \, h(\Ad^*(g) \lambda)
  \end{equation*}
  for all test functions $h$ on $i \mathfrak k^*$, where $dg$ denotes the Haar probability measure on the compact group $K$.
  Now recall that $\Prob_T$ is the push-forward of $\Qrob$ along the restriction $i \mathfrak k^* \rightarrow i \mathfrak t^*$.
  It follows that %
  \begin{align*}
      \int d\Prob_T \, f
    = &\int d\Qrob \, f(-\big|_{i \mathfrak t})
    = \int d\Prob_K(\lambda) \int_K dg \, f(\Ad^*(g) \lambda \big|_{i \mathfrak t}) \\
    = &\int d\Prob_K(\lambda) \int d\Prob_{T,\calO_{K,\lambda}} \, f.
    \qedhere
  \end{align*}
\end{proof}

As similarly observed in \cite{Heckman82, GuilleminPrato90}, Harish-Chandra's formula implies the following fundamental \emphindex{derivative principle}:

\begin{lem}
\label{lem:dhmeasure/derivative principle}
  \begin{equation}
  \label{eq:dhmeasure/derivative principle}
    \Prob_K \big|_{i \mathfrak t^*_{>0}} = p_K \left( \prod_{\alpha \in R_{G,+}} \partial_{-\alpha} \right) \Prob_T \big|_{i \mathfrak t^*_{>0}},
  \end{equation}
  where the partial derivatives $\partial_{-\alpha}$, the multiplication by $p_K$ and the restrictions to $i \mathfrak t^*_{>0}$ are all in the sense of distributions.
  In other words, we have that
  \begin{equation}
  \label{eq:dhmeasure/derivative principle expanded}
    \int d\Prob_K \, f = \int d\Prob_T \left( \prod_{\alpha \in R_{G,+}} \partial_\alpha \right) \left( p_K f \right)
  \end{equation}
  for any smooth test function $f$ on $i \mathfrak t^*$ that is compactly supported in the interior of the positive Weyl chamber.
\end{lem}
\begin{proof}
  Let $f$ be a smooth test function that is compactly supported in the interior of the positive Weyl chamber.
  Then the same is true for $p_K f$, and we find that
  \begin{equation}
  \label{eq:dhmeasure/heckman for test function times kk volume}
    \int d\Prob_{T,\calO_{K,\lambda}} \left( \prod_{\alpha \in R_{G,+}} \partial_\alpha \right) (p_K f) = f(\lambda).
  \end{equation}
  for all $\lambda \in i \mathfrak t^*_{>0}$ by applying \eqref{eq:dhmeasure/heckman for test functions} to $p_K f$ and dividing by $p_K(\lambda) \neq 0$.
  In fact, \eqref{eq:dhmeasure/heckman for test function times kk volume} can be extended to all $\lambda \in i \mathfrak t_+$, since both the left-hand side and the right-hand side are continuous in $\lambda$.
  Thus,
  \begin{align*}
    &\int d\Prob_K(\lambda) f(\lambda)
    = \int d\Prob_K(\lambda) \int d\Prob_{T,\calO_{K,\lambda}} \left( \prod_{\alpha \in R_{G,+}} \partial_\alpha \right) (p_K f) \\
    =\;&\int d\Prob_T \left( \prod_{\alpha \in R_{G,+}} \partial_\alpha \right) (p_K f),
  \end{align*}
  where the second equality is \autoref{lem:dhmeasure/abelian via coadjoint}.
\end{proof}

In mathematics,
Heckman has used a variant of \autoref{lem:dhmeasure/derivative principle} together with the Harish-Chandra formula to study the asymptotics of multiplicities in the subgroup restriction problem \cite{Heckman82} (cf.\ \autoref{sec:mul/asymptotics} in the next chapter).
Guillemin and Prato have used the same idea to derive an alternating-sum formula for non-Abelian Duistermaat--Heckman measures in symplectic geometry, which however is not directly applicable to the pure-state problem \cite{GuilleminPrato90} (cf.\ the discussion in \cite{ChristandlDoranKousidisEtAl12}).
Woodward also mentions Paradan as a source \cite{Woodward05}.
There is also a version of Harish-Chandra's formula \eqref{eq:dhmeasure/harish-chandra} for lower-dimensional coadjoint orbits \cite[Theorem 7.24]{BerlineGetzlerVergne03}.

\subsection*{Consequences}

Our basic assumption that the moment polytope $\Delta_K$ intersects the interior of the positive Weyl chamber implies that a random pure state is mapped into the interior of the positive Weyl chamber with probability one \cite[p.~504]{GuilleminSternberg82}.
\emph{It follows that the non-Abelian Duistermaat--Heckman measure $\Prob_K$ can be fully recovered from the Abelian measure $\Prob_T$ by taking partial derivatives in the directions of the negative roots, multiplying by the polynomial $p_K(\lambda)$, and restricting to the positive Weyl chamber.}
In the case of the one-body quantum marginal problem, this is just \eqref{eq:dhmeasure/derivative principle qmp} in the introduction to this chapter, with $p(\vec\lambda_1, \dots, \vec\lambda_n) = p_{d_1}(\vec\lambda_1) \dotsm p_{d_n}(\vec\lambda_n)$ the product of the Vandermonde determinants \eqref{eq:dhmeasure/kks volume unitary}.
We note that for the computation of averages the non-Abelian measure does not necessarily have to be computed explicitly.
Instead, we may use \eqref{eq:dhmeasure/derivative principle expanded} to reduce to the calculation of an expectation value for the Abelian measure (see \autoref{lem:dhmeasure/average linear entropy of entanglement} for an example).

The derivative principle allows us to lift structural properties to the non-Abelian measure.
For example, recall from \autoref{sec:dhmeasure/abelian} that the density of the Abelian measure is on each regular chamber given by a polynomial. \autoref{lem:dhmeasure/derivative principle} implies that, on the interior of each chamber, the same is true for the non-Abelian measure---indeed, the polynomials for $\Prob_K$ can be obtained from the ones of $\Prob_T$ by taking partial derivatives and multiplying with $p_K(\lambda)$ according to \eqref{eq:dhmeasure/derivative principle}.
If the Abelian measure is $\#R_{G,+}$ times continuously differentiable then there cannot be any measure on the walls, so that $\Prob_K$ is absolutely continuous with piecewise polynomial density.
In this case it also follows that the support of the measure is equal to the union of those regular chambers where the non-Abelian polynomial is non-zero, i.e., that it is a finite union of convex polytopes.
It is instructive to compare this observation with the main result of \cite{GuilleminSternberg82} and with \autoref{lem:kirwan/admissible}, where we had already seen that the facets of $\Delta_K$ are always contained in critical walls.
In fact, the support of $\Prob_K$ is always equal to the non-Abelian moment polytope $\Delta_K$, and hence a \emph{single} convex polytope. Moreover, $\Prob_K$ is always absolutely continuous with respect to Lebesgue measure on $\Delta_K$.
This is folklore and can be established, e.g., by using the symplectic cross section\index{symplectic cross section} and the local submersion theorem.
It also follows from \cite{Okounkov96}, where Okounkov shows how to construct an abstract convex body from which the non-Abelian measure can be obtained by pushing forward in a similar way as in \autoref{lem:dhmeasure/abelian}.

\subsection*{Computation}

By combining \autoref{alg:dhmeasure/abelian} with the derivative principle, we obtain a general algorithm for computing the non-Abelian Duistermaat--Heckman measure $\Prob_K$ under our basic assumption that $\Delta_K$ intersects the interior of the positive Weyl chamber.
We note that this solves the problem of exactly computing the local eigenvalue distribution of random quantum states in complete generality, since we have shown in \autoref{sec:dhmeasure/dhmeasure} that the assumption can always be satisfied by considering the purified scenario.

We conclude this section by explicitly stating the non-Abelian jump formula for critical walls that contain a minimal number of weights ($D_0 = r_T$), which will be useful for the computation of examples:
\begin{align}
  &f_{K,+}(\lambda) - f_{K,-}(\lambda)
  = \frac {p_K(\lambda)} {d\lambda(\tilde\omega_2-\tilde\omega_1, \dots, \tilde\omega_{D_0}-\tilde\omega_1, \xi)} \frac {(D-1)!} {(D-D_0-\#R_{G,+}-1)!} \nonumber \\
  &\times \left( \prod_{\alpha \in R_{G,+}} - (\alpha, H) \right)
  \frac
    { ((\lambda,H) - c)^{D-D_0-\#R_{G,+}-1}}
    {\prod_{k : \omega_k \not\in \Omega_0} ((\omega_k, H) - c)}
  \label{eq:dhmeasure/non-abelian jump formula minimal}
\end{align}
Equation~\eqref{eq:dhmeasure/non-abelian jump formula minimal} can be immediately obtained from its Abelian counterpart \eqref{eq:dhmeasure/abelian jump formula minimal} and the derivative principle \eqref{eq:dhmeasure/derivative principle}.
It is applicable as long as $D - D_0 - 1 \geq \#R_{G,+}$, so that the non-Abelian Duistermaat--Heckman measure is absolutely continuous across the wall.

\section{Examples}
\label{sec:dhmeasure/examples}

We will now illustrate the general method in some low-dimensional examples, where the polytopes and measures can be easily visualized.

\subsection*{Qubits}

Let $K = \SU(2)^n$, $G = \SL(2)^n$ its complexification, and consider the representation of $G$ on $\calH = (\CC^2)^{\otimes n}$, the Hilbert space of $n$ qubits.
The Lie algebra $i \mathfrak t$ is spanned by the generators $Z_1, \dots, Z_n$, where $Z_k$ acts as the Pauli matrix $\matrix{1 & \\ & -1}$ on the $k$-th factor of the Hilbert space.
There are $n$ positive roots $\alpha_1, \dots, \alpha_n$ and they satisfy $(\alpha_k, Z_l) = 2 \delta_{kl}$ (cf.\ \autoref{sec:kirwan/examples}).
In the following, it will be useful to identify $i \mathfrak t^* \cong \RR^n$ by assigning to each $\lambda \in i \mathfrak t^*$ the vector $\vec\lambda = (\lambda_1, \dots, \lambda_n)$ with components $\lambda_k = (\lambda, Z_k)$.
Thus the positive roots correspond to the vectors $\vec\alpha_1 = (2,0,\dots,0)$, \dots, $\vec\alpha_n = (0,\dots,0,2)$ and the product basis vectors $\ket{\vec x} \in \calH$ for $\vec x \in \{0,1\}^{\times n}$ are weight vectors with weight $\vec\omega = (1-2 x_1, \dots, 1-2 x_n)$.
The Abelian Duistermaat--Heckman measure $\Prob_T$ is the joint distribution of the $\sigma_z$-expectation values of the one-body reduced density matrices a random pure state of $n$ qubits, while the non-Abelian measure $\Prob_K$ is the joint distribution of the differences of local eigenvalues $\lambda_k = \lambda_{k,1} - \lambda_{k,2}$.
We choose $d\vec\lambda$ to be the usual Lebesgue measure on $\RR^n$.
Finally, we identify $i \mathfrak t \cong \RR^n$ by using the generators $Z_k$. Then $\RR^n_{\geq 0}$ is the positive Weyl chamber and the pairing between $i \mathfrak t$ and $i \mathfrak t^*$ amounts to the usual inner product.

\bigskip

We start with the trivial example of two qubits ($n=2$).
It will be illustrative to see that the derivative principle still works even if the non-Abelian moment polytope is of lower-dimension than the Abelian one.
Using the conventions fixed at the beginning of this section, the four weights of $\CC^2 \otimes \CC^2$ are $(\pm 1, \pm 1)$, the vertices of a square.
In \autoref{fig:dhmeasure/two qubits},~(a) we show the moment polytope and its decomposition into regular chambers as cut out by the critical walls (which are spanned by any pair of weights).
We now compute the Abelian Duistermaat--Heckman measure $\Prob_T$ by following \autoref{alg:dhmeasure/abelian}.
We start in the unbounded chamber (\circled{0} in the figure), where the density is equal to zero, and cross the horizontal critical wall at the top, which is given by the equation $\vec\lambda \cdot (0, -1) = -1$.
Since only a minimal number of weights lie on this wall ($D_0 = r_T = 2$), we may use the jump formula \eqref{eq:dhmeasure/abelian jump formula minimal} with $\xi = (0,-1)$ to see that the density on the upper regular chamber (\circled{1} in the figure) is given by the polynomial
\begin{equation*}
    \frac 1 2 \frac {(4-1)!} {(4-2-1)!} \frac {(-\lambda_2 + 1)^{4-2-1}} {(1+1) (1+1)}
  = \frac 3 4 (1 - \lambda_2),
\end{equation*}
where we have ordered the terms on the left-hand side in the same way as in the jump formula.
Next, we cross the wall with equation $\vec\lambda \cdot (1, -1) = 0$ that separates the upper and the right-hand side regular chamber (\circled{1} and \circled{2} in the figure).
Using \eqref{eq:dhmeasure/abelian jump formula minimal} with $\xi = (1, -1)/2$ we find that the density changes by
\begin{equation*}
    \frac 1 2 \frac {(4-1)!} {(4-2-1)!} \frac {(\lambda_1 - \lambda_2)^{4-2-1}} {(-2) 2}
  = - \frac 3 4 (\lambda_1 - \lambda_2)
\end{equation*}
Therefore, $\Prob_T$ has the following piecewise polynomial density function on the positive Weyl chamber:
\begin{equation*}
  f_T(\lambda_1, \lambda_2) = \frac 3 4 \max \{ 1 - \max \{ \lambda_1, \lambda_2 \}, 0 \},
\end{equation*}
which can be extended to all of $i \mathfrak t^* \cong \RR^2$ by symmetry (see \autoref{fig:dhmeasure/two qubits},~(b)).

\begin{figure}
  \centering
  \includegraphics[width=0.9\linewidth]{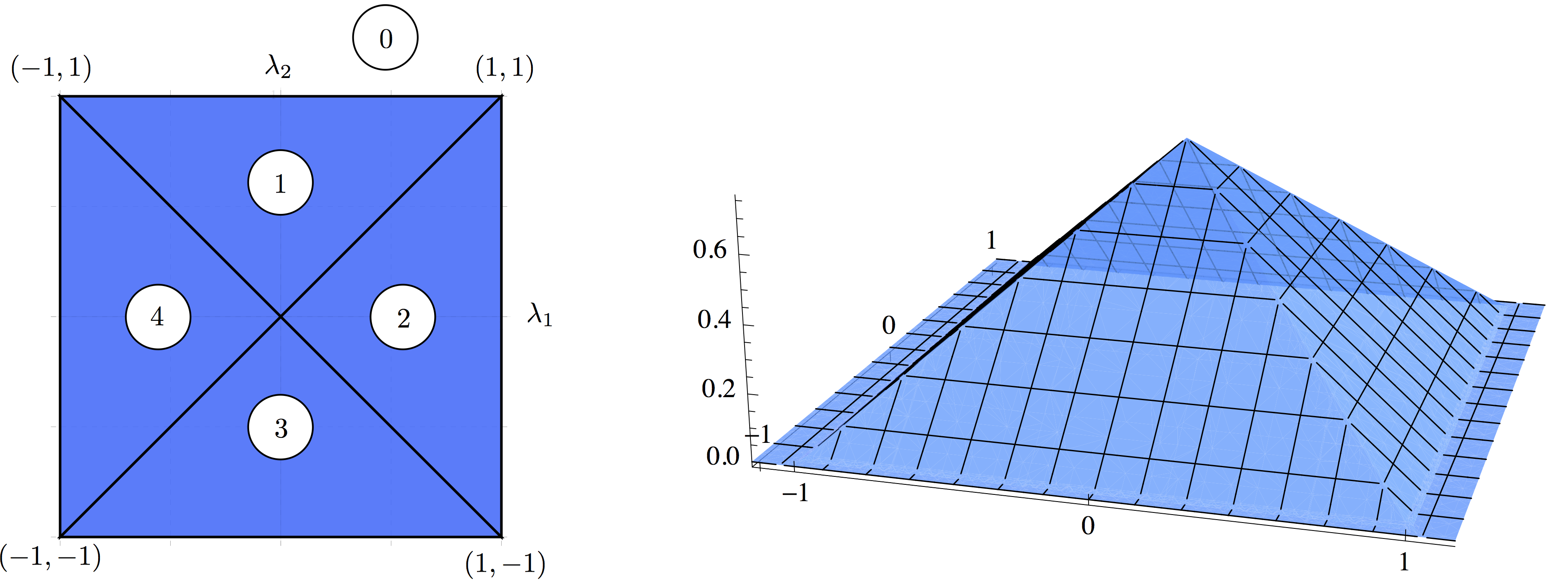}
  \caption[Abelian Duistermaat--Heckman measure for two qubits]{\emph{Two Qubits.}
    (a) Abelian moment polytope for two qubits and its decomposition into four bounded regular chambers separated by the critical walls (thick lines).
    (b) Density function of the Abelian Duistermaat--Heckman measure.}
  \label{fig:dhmeasure/two qubits}
\end{figure}

We now compute the non-Abelian Duistermaat--Heckman measure by using the derivative principle.
Observe that $\partial_{\alpha_1} \partial_{\alpha_2} f_T \equiv 0$ away from the critical walls.
On the walls we have to be more careful, since the density function is not twice differentiable there.
The horizontal and vertical walls do not carry any measure, since they are constant in one of the two coordinate directions.
However, this is not so on the diagonal; and we readily find by integrating against a test function that $\partial_{\alpha_1} \partial_{\alpha_2} \Prob_T |_{\RR^2_{>0}}$ is equal to Lebesgue measure on the non-Abelian moment polytope $\Delta_K = \{ \vec\lambda \in [0,1]^2 : \lambda_1 = \lambda_2 \}$.
According to the derivative principle \eqref{eq:dhmeasure/derivative principle} we obtain the non-Abelian Duistermaat--Heckman measure $\Prob_K$ by multiplying this result with $p_K(\vec\lambda) = \lambda_1 \lambda_2$. Thus,
\begin{equation*}
    \int d\Prob_K \, f
  = 3 \int_0^1 d\lambda_1 \, \lambda_1^2 \, f(\lambda_1, \lambda_1)
\end{equation*}
for all test functions $f$ on $\RR^2_{\geq 0}$.
By a simple transformation of variables, we conclude that the joint distribution of the maximal local eigenvalues of a random pure state of two qubits is given by
\begin{equation*}
  \int d\Pspec \, g
  = 24 \int_{0.5}^1 d\lambda_{1,1} \, (\lambda_{1,1} - 0.5)^2 \, g(\lambda_{1,1}, \lambda_{1,1})
\end{equation*}
for all test functions $g = g(\lambda_{1,1}, \lambda_{2,1})$.
This eigenvalue distribution has been computed more generally for bipartite pure states chosen at random from $\CC^a \otimes \CC^b$ \cite{LloydPagels88, ZyczkowskiSommers01}.
We will further below show how it can be obtained as a direct consequence of the derivative principle.

\bigskip

We now consider the case of three qubits ($n=3$).
The eight weights of $\CC^2 \otimes \CC^2 \otimes \CC^2$ are $(\pm 1, \pm 1, \pm 1)$, the vertices of a cube.
We will use some shortcuts to reduce the amount of computation required.
In \autoref{fig:dhmeasure/three qubits} we have visualized the \emph{non-Abelian} moment polytope $\Delta_K$, which itself is a union of regular chambers.
It is of maximal dimension and therefore the non-Abelian Duistermaat--Heckman measure $\Prob_K$ will be absolutely continuous with a piecewise polynomial density function.
Since this density necessarily vanishes outside of $\Delta_K$, we may start our computation right away by considering the jump over a facet of the non-Abelian moment polytope.
Let us thus consider the critical wall spanned by the weights $(1,1,1)$, $(-1,-1,1)$ and $(1,-1,-1)$ (the boundary of the blue chambers in \autoref{fig:dhmeasure/three qubits}).
It given by the equation $\vec\lambda \cdot (-1,1,-1) = -1$ and contains only a minimal number of weights.
Thus we may directly apply the non-Abelian jump formula \eqref{eq:dhmeasure/non-abelian jump formula minimal} with $\xi = (-1,1,-1)/3$.
We obtain that the polynomial density function of $\Prob_K$ on the blue chambers in \autoref{fig:dhmeasure/three qubits} is given by
\begin{align*}
    &\frac {\lambda_1 \lambda_2 \lambda_3} 4
    \frac {(8-1)!} {(8-3-3-1)!}
    \left( 2^2 (-2) \right)
    \frac {(-\lambda_1+\lambda_2-\lambda_3 + 1)^{8-3-3-1}} {(-2) 2^3 4} \\
  = &\frac {7!} {32} \lambda_1 \lambda_2 \lambda_3 \left(1 - \lambda_1 - \lambda_2 - \lambda_3 + 2 \lambda_2 \right)
\end{align*}
We now cross the critical wall $\vec\lambda \cdot (-1,-1,-1) = -1$ that separates the blue and the green chambers in \autoref{fig:dhmeasure/three qubits}.
It is spanned by the weights $(1,1,-1)$, $(1,-1,1)$ and $(-1,1,1)$ and is therefore again minimal.
By \eqref{eq:dhmeasure/non-abelian jump formula minimal} with $\xi = (-1,-1,-1)/3$, the density function of $\Prob_K$ changes by the polynomial
\begin{align*}
    &\frac {\lambda_1 \lambda_2 \lambda_3} 4
    \frac {(8-1)!} {(8-3-3-1)!}
    2^3
    \frac {(-\lambda_1-\lambda_2-\lambda_3+1)^{8-3-3-1}} {(-2) 2^3 4} \\
  = & -\frac{7!}{32} \lambda_1 \lambda_2 \lambda_3 (1-\lambda_1-\lambda_2-\lambda_3)
\end{align*}
as we cross the wall. It follows that the density function of $\Prob_K$ is on the two green chambers that face the viewer in \autoref{fig:dhmeasure/three qubits} equal to $\frac {7!} {16} \lambda_1 \lambda_2 \lambda_3 \lambda_2$.
Note that the blue and green chambers together form the part of the three-qubit polytope where $\lambda_2 = \min_k \lambda_k$.
By using permutation-symmetry to extend the formulas to all of $\Delta_K$, we conclude that the non-Abelian Duistermaat--Heckman measure is given by the following piecewise polynomial density function:
\begin{equation*}
  f_K(\lambda_1, \lambda_2, \lambda_3) = \frac {7!} {32} \lambda_1 \lambda_2 \lambda_3 \begin{cases}
    2 \min_k \lambda_k & \text{in the lower pyramid} \\
    1 - \sum_{k=1}^3 \lambda_k + 2 \min_k \lambda_k & \text{in the upper pyramid} \\
    0 & \text{otherwise}
  \end{cases}
\end{equation*}
As in the case of two qubits, it is straightforward to deduce from this the marginal eigenvalue distribution of a random pure state of three qubits.

\begin{figure}
  \centering
  \includegraphics[width=0.9\linewidth]{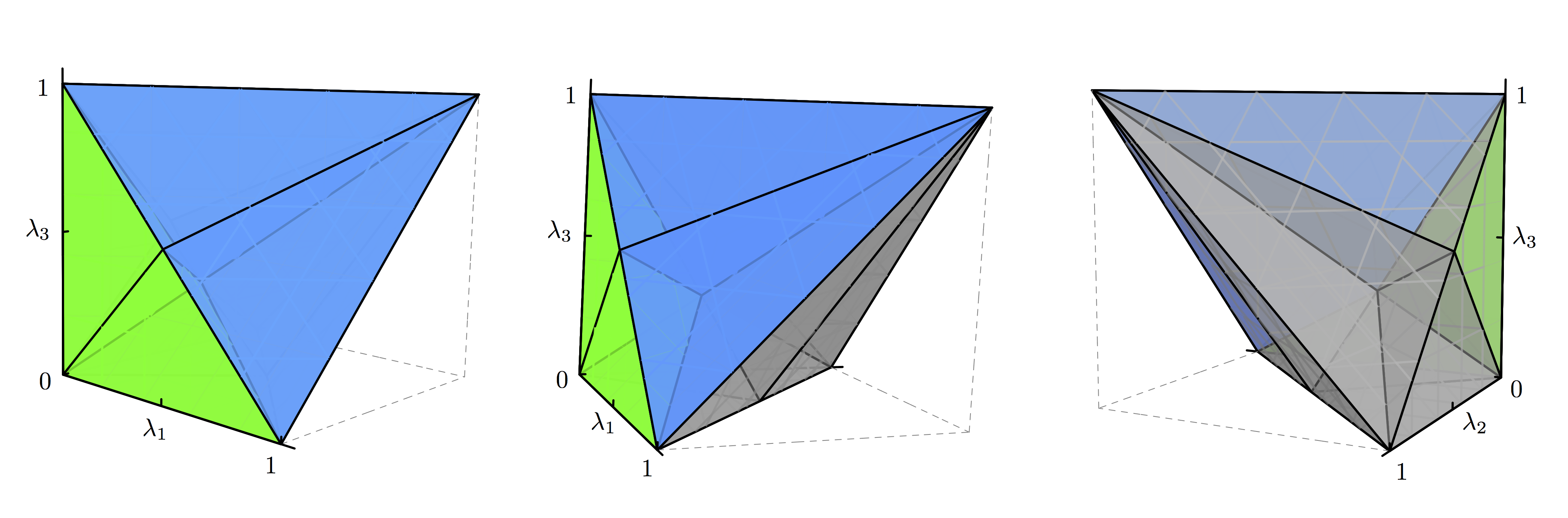}
  \caption[Decomposition of three-qubit moment polytope into regular chambers]{\emph{Three Qubits.} Non-Abelian moment polytope and its decomposition into twelve regular chambers (from three different perspectives).}
  \label{fig:dhmeasure/three qubits}
\end{figure}

\subsection*{Bosons}

We now turn to random pure states of $n$ bosonic qubits, with Hilbert space $\calH = \Sym^n(\CC^2)$.
Let $K = \SU(2)$. We identify $i \mathfrak t^* \cong \RR$ in the same way as above.
Thus the positive root is $2$, the occupation number basis vectors in $\Sym^n(\CC^2)$ are weight vectors of weight $-n, -n+2, \dots, n\in \ZZ$, and the Abelian moment polytope is $\Delta_T = [-n, n]$.

We apply \autoref{alg:dhmeasure/abelian} starting from the left-hand side unbounded chamber $\{ \lambda_1 \leq -n \}$ and successively cross the critical walls, which are points and correspond to a single weight each.
At each point $m = -n, -n+2, \dots, n$, the jump formula \eqref{eq:dhmeasure/abelian jump formula minimal} shows that the Abelian measure changes by
\begin{align*}
    &\frac 1 1 \frac {((n+1)-1)!} {((n+1)-1-1)!} \frac {(\lambda_1 - m)^{(n+1)-1-1}} {\prod_{m \neq m'=-n,-n+2,\dots,n} (m'-m)} \\
  =\; &n \frac {(\lambda_1 - m)^{n-1}} {\prod_{m \neq m'=-n,-n+2,\dots,n} (m'-m)} \\
  =\; &\frac 1 {2^n (n-1)!} (-1)^{(n+m)/2} \binom{n}{\frac {n+m} 2} (\lambda_1 - m)^{n-1}.
\end{align*}
It follows that the Abelian Duistermaat--Heckman density has Lebesgue density
\begin{equation}
\label{eq:dhmeasure/bosons abelian}
  f_T(\lambda_1) = \frac 1 {2^n (n-1)!} \sum_{m=-n,-n+2,\dots,n} (-1)^{(n+m)/2} \binom{n}{\frac {n+m} 2} (\lambda_1 - m)^{n-1}_+,
\end{equation}
where we set $(\lambda_1 - m)^{n-1}_+ := (\lambda_1 - m)^{n-1}$ for $\lambda_1 \geq m$ and $0$ otherwise.
Amusingly, \eqref{eq:dhmeasure/bosons abelian} is precisely the probability density of a sum of $n$ independent random variables that are each distributed uniformly in the interval $[-1,1]$ \cite[\S{}I.9, Theorem 1a]{Feller71}.
It follows from the derivative principle that the non-Abelian measure has Lebesgue density
\begin{equation*}
  f_K(\lambda_1) = \frac {-\lambda_1} {2^{n-1} (n-2)!} \sum_{m=-n,-n+2,\dots,n} (-1)^{(n+m)/2} \binom{n}{\frac {n+m} 2} (\lambda_1 - m)^{n-2}_+,
\end{equation*}
on $[0,\infty)$, which can be readily translated into, e.g., the distribution of the maximal local eigenvalue.
See \autoref{fig:dhmeasure/bosons} for an illustration.

\begin{figure}
\centering
\includegraphics[width=0.9\linewidth]{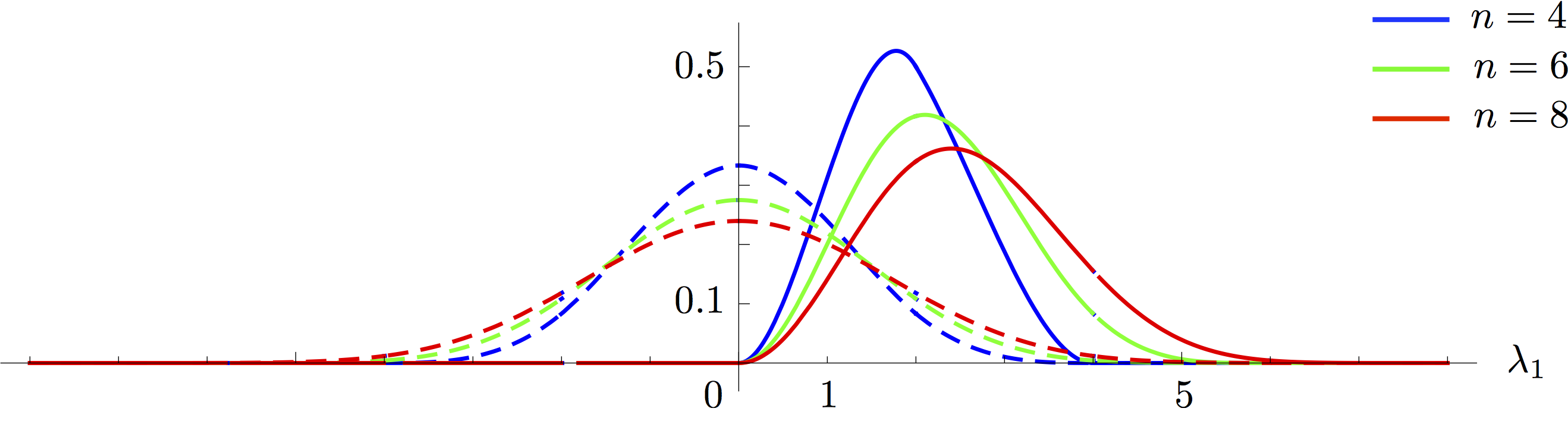}
\caption[Duistermaat--Heckman measures for bosonic qubits]{\emph{Bosonic qubits.} Densities of the Abelian (dashed) and of the non-Abelian (solid) Duistermaat--Heckman measure of $\Sym^n(\CC^2)$ for various particle numbers $n$.}
\label{fig:dhmeasure/bosons}
\end{figure}

As an application, we compute the average value of the linear entropy of entanglement \eqref{eq:slocc/linear entropy of entanglement} for a random pure state of $n$ bosonic qubits. For a bosonic state $\rho$ on $\calH$, it is given by
\begin{equation*}
  E(\rho) = 1 - \tr \rho_1^2,
\end{equation*}
where $\rho_1$ denotes the one-body reduced density matrix.

\begin{lem}
\label{lem:dhmeasure/average linear entropy of entanglement}
  The average linear entropy of entanglement of a random pure state of $n$ bosonic qubits is equal to $1/2 - 1/(2n)$.
\end{lem}
\begin{proof}
  The average linear entropy of entanglement is given by
  \begin{equation*}
      \int d\rho \, E(\rho)
    = \int d\Pspec(\vec\lambda_1) \left( 1 - (\lambda_{1,1}^2 + \lambda_{1,2}^2 ) \right)
    = \frac 1 2 - \frac 1 {2n^2} \int d\Prob_K(\lambda_1) \lambda_1^2.
  \end{equation*}
  To evaluate the right-hand side, we observe that $\lambda_1^2$ vanishes at $\lambda_1=0$, the boundary of the Weyl chamber;
  thus we may use \eqref{eq:dhmeasure/derivative principle expanded} and take limits to reduce to an integral over the Abelian measure. We obtain
  \begin{equation}
  \label{eq:dhmeasure/average lee abelian}
    \int d\Prob_K(\lambda_1) \lambda_1^2
    = \int_0^\infty d\Prob_T(\lambda_1) \partial_{2\lambda_1} (\lambda_1 \lambda_1^2)
    = 3 \int_{-\infty}^\infty d\Prob_T(\lambda_1) \lambda_1^2,
  \end{equation}
  where we have used that $\Prob_T$ is symmetric about the origin.
  The right-hand side integral is the variance of $\Prob_T$.
  But recall that $\Prob_T$ is the distribution of a sum of independent random variables that are uniformly distributed on $[-1,1]$.
  Since the variance of each summand is $1/3$ and it is additive for independent sums, we find that \eqref{eq:dhmeasure/average lee abelian} is precisely equal to $n$. By plugging this into the first formula, we conclude that, indeed,
  $\int d\rho \, E(\rho) = 1/2 - 1/(2n)$.
\end{proof}

In accordance with the concentration of measure phenomenon, $\rho_1 \rightarrow \Id/2$ in distribution as $n \rightarrow \infty$.
We remark that our result agrees with \cite[Theorem 34]{MuellerDahlstenVedral12} if one works out the quantities left uncalculated therein.
Note that the proof illustrates the power of the derivative principle:
Instead of explicitly computing the average over the eigenvalue distribution, we reduced to an average over the Abelian measure, whose structure we could then exploit.

\subsection*{Bipartite Systems}

We conclude this series of examples by giving a concise derivation of the marginal eigenvalue distribution of a random bipartite pure state $\rho$ on $\calH = \CC^a \otimes \CC^b$.
Without loss of generality, assume that $a \leq b$.
By \autoref{lem:onebody/two}, it suffices to derive the distribution of eigenvalues of $\rho_A$.
It will be convenient to work with $K = \U(a)$ rather than the special unitary group, so that we may identify $i \mathfrak t^* \cong \RR^a$.

We first compute the Abelian Duistermaat--Heckman measure $\Prob_T$.
For this, observe that the product basis vectors $\ket{ij} \in \CC^a \otimes \CC^b$ are weight vectors with weight the $i$-th standard basis vector of $\RR^a$.
Thus $\Delta_T$ is equal to the convex hull of these unit vectors, i.e., to the standard simplex $\Delta_a$.
By \eqref{eq:dhmeasure/abelian as volume}, the density function of $\Prob_T$ with respect to Lebesgue measure on $\Delta_T$ can thus be computed in the following way:
\begin{align*}
  f_T(\vec\lambda_A)
  =\;&\vol~\{ (p_{i,j}) \in \RR^{ab}_{\geq 0} : \sum_{j=1}^b p_{i,j} = \lambda_{A,i} \, (\forall i=1,\dots,a) \} \\
  =\;&\prod_{i=1}^a \vol~\{ (p_j) \in \RR^b_{\geq 0} : \sum_{j=1}^b p_j = \lambda_{A,i} \} \\
  \propto\;&\prod_{i=1}^a \lambda_{A,i}^{b-1},
\end{align*}
since each factor is the volume of a $(b-1)$-dimensional simplex.
We now consider the negative partial derivative in the direction in the positive roots: Since there are $\binom a 2$ positive roots,
\begin{equation*}
  \left( \prod_{i < j} \partial_{\lambda_{A,j} - \lambda_{A,i}} \right) \left( \prod_{i=1}^a \lambda_{A,i}^{b-1} \right)
\end{equation*}
is a polynomial of degree no more than $a(b-1) - \binom a 2$.
Since we differentiate each variable at most $a-1$ times, the result will be a multiple of the symmetric polynomial $\prod_{i=1}^a \lambda_{A,i}^{b-a}$.
On the other hand, the result is evidently antisymmetric in the variables and so will be a multiple of the Vandermonde determinant $p_A(\vec\lambda_A)$.
Since the degrees add up, $a(b-a) + \binom a 2 = a(b-1) - \binom a 2$ it follows at once that
\begin{equation*}
  \left( \prod_{i < j} \partial_{\lambda_{A,j} - \lambda_{A,i}} \right) \left( \prod_{i=1}^a \lambda_{A,i}^{b-1} \right)
  \propto
  \left( \prod_{i=1}^a \lambda_{A,i}^{b-a} \right) p_a(\vec\lambda_A)
\end{equation*}
Thus we conclude from the derivative principle \eqref{eq:dhmeasure/derivative principle} that the eigenvalue distribution of $\rho_A$ is proportional to
\begin{equation}
\label{eq:dhmeasure/lloyd-pagels single marginal}
  f_K(\vec\lambda_A) \propto \left( \prod_{i=1}^a \lambda_{A,i}^{b-a} \right) p^2_a(\vec\lambda_A)
\end{equation}
on $\Delta_{a,+} = \{ \vec\lambda_A \in \Delta_a : \lambda_{A,1} \geq \dots \geq \lambda_{A,a} \}$, the intersection of $\Delta_a$ with the positive Weyl chamber, which is the non-Abelian moment polytope for the action of $K = \U(a)$.
This is the well-known formula from \cite{LloydPagels88, ZyczkowskiSommers01}.
In mathematics, the distribution of $\rho_A$ is also known as a \emph{normalized Wishart ensemble}\index{Wishart ensemble!normalized} (see, e.g., \cite{AubrunSzarekYe14}).
Since $\rho_A$ and $\rho_B$ have the same non-zero eigenvalues (\autoref{lem:onebody/two}), it follows immediately from \eqref{eq:dhmeasure/lloyd-pagels single marginal} that the joint distribution of the marginal eigenvalues of both $\rho_A$ and $\rho_B$ is given by
\begin{equation}
\label{eq:dhmeasure/lloyd-pagels both marginals}
  \int d\Pspec \, f \propto \int_{\Delta_{a,+}} d\vec\lambda_A \left( \prod_{i=1}^a \lambda_{A,i}^{b-a} \right) p_a^2(\vec\lambda_A) \, f(\vec\lambda_A, \vec\lambda_A)
\end{equation}
for all test functions $f = f(\vec\lambda_A, \vec\lambda_B)$.
For $a = b = d$, this is precisely \eqref{eq:dhmeasure/two} in \autoref{sec:dhmeasure/dhmeasure}.

\section{Discussion}
\label{sec:dhmeasure/discussion}

Random ensembles of states have long been studied in quantum statistical physics.
In fact, Lloyd and Pagels had derived \eqref{eq:dhmeasure/lloyd-pagels single marginal} out of thermodynamic considerations \cite{LloydPagels88}.
More recently, the typical behavior of canonical states, i.e., states that are obtained by computing the reduced density matrix of the uniform state on a subspace (encoding the energy constraint) of the system--bath Hilbert space, have been considered \cite{PopescuShortWinter06, Lloyd06, GoldsteinLebowitzTumulkaEtAl06}, and the average von Neumann entropy of a subsystem \cite{Lubkin78, Page93} has featured in the analysis of the black hole entropy paradox \cite{HaydenPreskill07}.
Apart from their import from the perspective of the quantum marginal problem, random states of multipartite systems are also relevant in quantum information theory.
For example, conditional entropies and mutual informations are central quantities in entanglement theory.
Since they are functions of the eigenvalues only they can be studied in the framework of this chapter (e.g., in the case of tripartite pure states).
Remarkable recent progress has been made by analyzing the entanglement properties of the two-body reduced density matrix $\rho_{AB}$ of a randomly-chosen tripartite pure state \cite{HaydenLeungShorEtAl04, HaydenLeungWinter06, AubrunSzarekYe14, AubrunSzarekYe12, CollinsNechitaYe12}.
In all these applications, most known results are for high-dimensional Hilbert spaces, where the powerful concentration of measure phenomenon occurs.
In contrast, our exact algorithms require no such assumption and are instead well-suited for low-dimensional systems, which previously remained inaccessible.
It would be highly desirable to find a common meeting ground for both techniques that would allow for an interpolation between the combinatorial and the analytical regime (cf.\ \cite{BorodinGorin12, BorodinPetrov14}).
\chapter[Interlude: Multiplicities of Representations]{Interlude: Computing Multiplicities of Lie Group Representations}
\label{ch:multiplicities}

In this chapter we consider the classical branching or subgroup restriction problem in representation theory, which asks for the multiplicity of an irreducible representation of a subgroup $K \subseteq K'$ in the restriction of an irreducible representation of $K'$.
In the case of compact, connected Lie groups, we provide a polynomial-time algorithm for this problem---based on a finite-difference formula for the multiplicities and Barvinok's algorithm for counting integral points in polytopes.
Our algorithm is also applicable to the Kronecker coefficients of the symmetric group, which play an important role in the geometric complexity theory approach to the $\Pclass$ vs.\ $\NP$ problem. Whereas the computation of Kronecker coefficients is known to be $\SharpP$-hard for Young diagrams with an arbitrary number of rows, our algorithm computes them in polynomial time if the number of rows is bounded.
The finite-difference formula that we use is a discrete analogue of the derivative principle that was used in the preceding chapter.
This connection can be made precise, as the asymptotic growth rates of multiplicities are directly related to the measures that we had computed in the preceding chapter.
We complement our results by showing that in geometric complexity theory, such asymptotic information might not directly lead to complexity-theoretic obstructions beyond what can be obtained from moment polytopes.
Non-asymptotic information on the multiplicities, such as provided by our algorithm, may therefore be essential in order to find new obstructions in geometric complexity theory.

The results in this chapter have been obtained in collaborations with Matthias Christandl and Brent Doran, and they have appeared in \cite{ChristandlDoranWalter12}; the first part of \autoref{sec:mul/asymptotics} is adapted from the collaboration \cite{ChristandlDoranKousidisEtAl12}.

\section{Summary of Results}
\label{sec:mul/summary}

The decomposition of Lie group representations into irreducible sub-representations is a fundamental problem in mathematics with a variety of applications to the sciences.
In atomic and molecular physics as well as in high-energy physics, this problem has been studied extensively in the context of symmetries \cite{Weyl50, WignerGriffin59, Wigner73}, perhaps most famously in Ne'eman and Gell-Mann's ``eight-fold way'' of elementary particles \cite{Neeman61, Gell-Mann61, Gell-Mann62}.
In pure mathematics, the combinatorial resolution by Knutson and Tao of the problem of decomposing tensor products of irreducible representations of the unitary group has been a recent highlight with a long history of research \cite{Fulton00, KnutsonTao99}.
More recently, the representation theory of Lie groups has found novel applications in quantum information \cite{KeylWerner01, ChristandlMitchison06, Klyachko06, ChristandlSchuchWinter10} and quantum computation \cite{BaconChuangHarrow07, BaconChuangHarrow06, Jordan08}, as well as in the geometric complexity theory approach to the $\Pclass$ vs.\ $\NP$ problem in computer science \cite{MulmuleySohoni01, MulmuleySohoni08, BuergisserLandsbergManivelEtAl11}.
In this chapter, we study the problem of computing multiplicities of Lie group representations:

\begin{pro}[Subgroup Restriction Problem, Branching Problem]\index{subgroup restriction problem}\index{branching problem}
  \label{pro:mul/subgroup restriction}
  Let $\Phi \colon K \rightarrow K'$ be a fixed homomorphism of compact, connected Lie groups.
  What is the multiplicity $m^\lambda_\mu$ of an irreducible $K$-representation $V_{K,\mu}$ in an irreducible $K$-representation $V_{K',\lambda}$, when given as input the highest weights $\mu$ and $\lambda$?%
  \nomenclature[Rm^lambda_mu]{$m^\lambda_\mu$}{multiplicity of $V_{K,\mu}$ in $V_{K,\lambda}$}
\end{pro}

The term \emph{subgroup restriction problem} comes from the archetypical case where the map $\Phi$ is the inclusion of a subgroup $K \subseteq K'$.

The main result of this chapter is a polynomial-time algorithm for the subgroup restriction problem (\autoref{alg:mul/main}).
It takes as input bitstrings containing the coordinates of the highest weights with respect to fixed bases of fundamental weights.
As a direct consequence, for any fixed $\lambda$ and $\mu$ the stretching function $k \mapsto m^{k \lambda}_{k \mu}$ can be evaluated in polynomial time.
Another immediate corollary is that the positivity of the coefficients $m^\lambda_\mu$ can be decided in polynomial time for any fixed homomorphism $\Phi \colon K \rightarrow K'$.
Mulmuley has conjectured that deciding positivity of the multiplicities $m^\lambda_\mu$ should be possible in polynomial time even if the group homomorphism $\Phi$ is also part of the input \cite{Mulmuley07}.
This is known only for specific families of homomorphisms, such as those corresponding to the Littlewood--Richardson coefficients \cite{KnutsonTao99, MulmuleyNarayananSohoni12}, and our result can be regarded as supporting evidence for the conjecture.
However, any approach to deciding positivity that proceeds by computing the actual multiplicities is of course expected to fail, since the latter problem is well-known to be $\SharpP$-hard \cite{Narayanan06, BuergisserIkenmeyer08}.
Since representations of compact, connected Lie groups are in one-to-one correspondence with the rational representations of complex, reductive, connected algebraic groups (\autoref{sec:onebody/lie}), our results can also be interpreted in the latter context.

\bigskip

Our polynomial-time algorithm is based on a formula for the multiplicities $m^\lambda_\mu$ (\autoref{prp:mul/main}), which is obtained in three steps:
First, we restrict from the group $K'$ to its maximal torus $T'$.
The corresponding weight multiplicities can be computed efficiently by using the classical multiplicity formula of Kostant \cite{Kostant59, Cochet05} or, in fact, by evaluating a single vector partition function \cite{BilleyGuilleminRassart04, Bliem08, Bliem10} (\autoref{sec:mul/weights}).
Second, we restrict all weights to a maximal torus $T$ of $K$.
Third, we recover the multiplicity of an irreducible $K$-representation by using a finite-difference formula based on Weyl's character formula:
For any $K$-representation $V$, the highest weight multiplicity function $m_{K,V}$ of irreducible $K$-representations and the weight multiplicity function $m_{T,V}$ are related by
\begin{equation}
\label{eq:mul/finite difference formula intro}
  m_{K,V} = \left. \left[ \left(\prod_{\alpha \in R_{K,+}} - D_\alpha \right) m_{T,V} \right] \right|_{\Lambda^*_{K,+}},
\end{equation}
where $(D_\alpha m)(\lambda) = m(\lambda + \alpha) - m(\lambda)$ denotes the finite-difference operator in direction $\alpha$
(\autoref{sec:mul/finite difference formula}).
By carefully combining the first two steps, we obtain a formula for the multiplicities that reduces \autoref{pro:mul/subgroup restriction} to the problem of counting integral points in rational convex polytopes of bounded dimension (\autoref{sec:mul/subgroup restriction problem}).
The latter can be done efficiently by using Barvinok's algorithm \cite{Barvinok94} (see also \cite{Dyer91, CookHartmannKannanEtAl92, DyerKannan97, BarvinokPommersheim99, BaldoniBeckCochetEtAl06}) and we thus obtain \autoref{alg:mul/main}.
The multiplicity formula itself has intrinsic interest beyond its application to algorithmics.
One immediate insight is a piecewise quasi-polynomiality of the multiplicities and of the stretching function \cite{Mulmuley07}.

\bigskip

In \autoref{sec:mul/kronecker}, we turn to the computation of the \emph{Kronecker coefficients} $g_{\alpha,\beta,\gamma}$, which are commonly defined as the multiplicities in the decomposition of tensor products of irreducible representations of the symmetric group $S_k$.
Kronecker coefficients are notoriously difficult to study, and finding an appropriate combinatorial interpretation is one of the outstanding problems of classical representation theory.
Apart from the role for the quantum marginal problem, they appear naturally in geometric complexity theory, where their computation serves as a model problem which has been subject to a number of conjectures \cite{Mulmuley07}. %
In \autoref{sec:onebody/kronecker}, we have seen that the Kronecker coefficients can be equivalently defined in terms of the special linear or unitary groups.
For Young diagrams of bounded height, this amounts to an instance of \autoref{pro:mul/subgroup restriction} for a fixed group homomorphism $\Phi$.
Therefore the Kronecker coefficients can in this case be computed in polynomial time by using \autoref{alg:mul/main}.
We also get a clean closed-form expression for the Kronecker coefficients (\autoref{pro:mul/optimized kronecker}), which not only nicely illustrates the algorithm's effectiveness, but also implies directly that the problem of computing Kronecker coefficients with unbounded height is in $\GapP$, as first proved in \cite{BuergisserIkenmeyer08}.

In practice, our algorithm appears to work rather well as long as the rank of the Lie group $G$ is not too large.
In the case of Kronecker coefficients for Young diagrams with two rows, we can easily go up to $k=10^8$ boxes using commodity hardware.
In contrast, all other software packages known to the authors cannot go beyond only a moderate number of boxes ($k=10^2$ on the same hardware as used above).
Moreover, by distributing the computation of weight multiplicities onto several processors, we have been able to compute Kronecker coefficients for Young diagrams with three rows and $k=10^5$ boxes in a couple of minutes.
A preliminary implementation of the algorithm is available at \cite{Walter12a}.
We hope that our algorithm will provide a useful tool in experimental mathematics, theoretical physics, and geometric complexity theory.

\bigskip

We now turn towards studying the asymptotics of multiplicities.
Recall from the fundamental \autoref{thm:onebody/kirwan} that the moment polytope of a projective subvariety $\XX \subseteq \PP(\calH)$ can be described as
\begin{equation}
\label{eq:mul/moment polytope}
  \Delta_K(\XX) = \overline{ \{ \lambda/k : m_{K,R^*_k(\XX)}(\lambda) > 0 \} },
\end{equation}
where $m_{K,R^*_k(\XX)}(\lambda)$ denotes the multiplicity of the irreducible $K$-representation $V_{G,\lambda}$ in $R^*_k(\XX)$, the dual of the space of regular functions of degree $k$.
This suggests to consider the following measure defined as the weak limit of the atomic measures that encode the multiplicities in a fixed degree $k$,%
\nomenclature[Tnu_K,X]{$\nu_{K,\XX}$}{multiplicity measure for $K$-action on projective subvariety $\XX$}
\begin{equation}
\label{eq:mul/multiplicity measure intro}
  \nu_{K,\XX} := \wlim_{k \rightarrow \infty}  \frac 1 {k^{d_{K,\XX}}}  \sum_{\lambda \in \Lambda^*_{K,+}} m_{K,R^*_k(\XX)}(\lambda) \, \delta_{\lambda/k},
\end{equation}
where $d_{K,\XX} \in \ZZ_{\geq 0}$ is the appropriate exponent such that $\nu_{K,\XX}$ has finite, non-zero measure \cite{Okounkov96}.
It has a piecewise-polynomial density function $f_{K,\XX}$ with respect to Lebesgue measure on the moment polytope and its support is the entire moment polytope (both assertions follow from the main result of \cite{Okounkov96}, but probably go back to earlier work).
By definition, the measure $\nu_{K,\XX}$ encodes the \emph{asymptotic growth rate} of multiplicities in the ring of regular functions of the projective variety $\XX$.
If $\XX$ is smooth then $\nu_{K,\XX}$ can be given a remarkable geometric interpretation:
It is directly related to the Duistermaat--Heckman measure $\Prob_{K,\XX}$ that we had considered in \autoref{ch:dhmeasure}.
This fact holds more generally for smooth projective subvarieties and other symplectic manifolds that can be ``quantized'' in a certain technical sense \cite{Heckman82, GuilleminSternberg82a, Sjamaar95, Meinrenken96, Vergne98, MeinrenkenSjamaar99}.

In \autoref{sec:mul/asymptotics}, we first give a succinct proof of the relation between $\nu_{K,\XX}$ and $\Prob_{K,\XX}$ in the case where $\XX$ is the full projective space $\PP(\calH)$.
Our argument is based on the observation that the finite difference formula \eqref{eq:mul/finite difference formula intro} ``converges'' to the derivative principle \eqref{eq:dhmeasure/derivative principle} in the limit as $k \rightarrow \infty$, as was also pointed out to us by Allen Knutson.
Applied to the quantum marginal problem, we find that the distribution of the one-body reduced density matrices of a random quantum state, as computed in \autoref{ch:dhmeasure}, is directly related to the asymptotic growth of the Kronecker coefficients.

We then consider the applicability of the measures $\nu_{K,\XX}$ to geometric complexity theory.
In a nutshell, the main challenge in geometric complexity theory is to show that for certain pairs of projective subvarieties $\XX$ and $\YY$ one is not contained in the other; this would then imply complexity-theoretic lower bounds.
Both the permanent vs.\ determinant problem, which is equivalent to the $\VP$ vs.\ $\VNP$ problem \cite{Valiant79} (an algebraic version of the $\Pclass \neq \NP$ conjecture), as well as the complexity of matrix multiplication \cite{Strassen69} can be formulated in this framework \cite{MulmuleySohoni01, MulmuleySohoni08, BuergisserIkenmeyer11, BuergisserLandsbergManivelEtAl11}.
The starting point is the basic observation that
\begin{equation}
  \label{eq:mul/multiplicity criterion}
  \XX \subseteq \YY
  \;\Longrightarrow\;
  m_{K,R^*_k(\XX)}(\lambda) \leq m_{K,R^*_k(\YY)}(\lambda)
\end{equation}
for all $\lambda$ and $k \geq 0$.
Therefore, the existence of $\lambda$ and $k$ such that $m_{K,R^*_k(\XX)}(\lambda) > m_{K,R^*_k(\YY)}(\lambda)$ proves that $\XX \not\subseteq \YY$; such a pair $(\lambda,k)$ is called an \emphindex{obstruction} \cite{MulmuleySohoni08}.
Since computing multiplicities in general coordinate rings is a difficult problem, it is natural to instead study their asymptotic behavior.
Following an idea of Strassen \cite{Strassen06}, it has been proposed in \cite{BuergisserIkenmeyer11} to consider instead the moment polytopes \eqref{eq:mul/moment polytope}, whose geometric interpretation facilitates their computation, as we have seen in \autoref{ch:kirwan}. Clearly,
\begin{equation}
  \label{eq:mul/moment polytope criterion}
  \XX \subseteq \YY
  \;\Longrightarrow\;
  \Delta_K(\XX) \subseteq \Delta_K(\YY).
\end{equation}
However, preliminary results suggest that the right-hand side moment polytope $\Delta_K(\YY)$ might be trivially large for the varieties of interest in geometric complexity theory \cite{BuergisserChristandlIkenmeyer11a, BuergisserIkenmeyer11, Kumar11, BuergisserLandsbergManivelEtAl11}, which would imply that \eqref{eq:mul/moment polytope criterion} is insufficient for finding complexity-theoretic obstructions.
It has therefore been suggested recently to study the asymptotic growth of multiplicities as captured by the measures $\nu_{K,\XX}$ (e.g., \cite[\S 2.2]{GrochowRusek12}). In this context, our contribution is a ``no go'' result: We show that, in the situations of interest in geometric complexity theory, $\XX \subseteq \YY$ implies that $d_{K,\XX} < d_{K,\YY}$ for the exponents in the definition \eqref{eq:mul/multiplicity measure intro} of the measures (\autoref{cor:mul/dimension lemma}).
Therefore, we \emph{cannot} deduce from \eqref{eq:mul/multiplicity measure intro} and \eqref{eq:mul/multiplicity criterion} a criterion of the form
\begin{equation*}
  \text{``$\XX \subseteq \YY
  \;\Longrightarrow\;
  f_{K,\XX}(\lambda) \leq f_{K,\YY}(\lambda)$'',}
\end{equation*}
since we need to divide by different powers of $k$ in the definition of the weak limit.
In this sense, the measures $\nu_{K,\XX}$ do not directly give rise to new obstructions, indicating that a more refined understanding of the behavior of multiplicities in coordinate rings might be required.

\section{Weight Multiplicities}
\label{sec:mul/weights}

Throughout this chapter, we will use the notations and conventions of \autoref{sec:onebody/lie}; however, we will label all objects such as weight lattices, irreducible representations, etc.\ by the compact connected Lie group $K$, since we will mostly not need its complexification $G$ explicitly.
Recall that the irreducible representations of $K$ are labeled by their highest weight in $\Lambda^*_{K,+}$.
An arbitrary finite-dimensional representation $V$ can always be decomposed into irreducible sub-representations,
\begin{equation*}
  V \cong \bigoplus_{\lambda \in \Lambda^*_{K,+}} m_{K,V}(\lambda) \, V_{K,\lambda}.
\end{equation*}
We will call the function $m_{K,V}$ thus defined the \emph{highest weight multiplicity function}\index{multiplicity function!highest weight}\nomenclature[Rm_K,V]{$m_{K,V}$}{highest weight multiplicity function}.

We may similar decompose $V$ into irreducible representations of the maximal torus $T \subseteq K$, which we recall are one-dimensional and labeled by elements of the weight lattice $\Lambda^*_K$. We thus obtain the decomposition into weight spaces
\begin{equation*}
  V = \bigoplus_{\omega \in \Lambda^*_K} V_\omega.
\end{equation*}
We will call $m_{T,V}(\omega) = \dim V_\omega$ the \emph{weight multiplicity function}\index{multiplicity function!weight}\nomenclature[Rm_T,V]{$m_{T,V}$}{weight multiplicity function}.
An equivalent way of encoding the weight multiplicities is in terms of the \emph{(formal) character}\index{character!formal}\nomenclature[RchV]{$\ch V$}{formal character of representation $V$},
\begin{equation*}
  \ch V = \sum_\omega m_{T,V}(\omega) \, e^\omega,
\end{equation*}
Formally, $\ch V$ is an element of the group ring $\ZZ[\Lambda^*_K]$, which consists of (finite) linear combinations of basis elements $e^\omega$ for each weight $\omega \in \Lambda^*_K$; the multiplication is defined by $e^\omega e^{\omega'} = e^{\omega + \omega'}$ (see, e.g., \cite{Knapp02}).
We may identify elements of the group ring with integer-valued functions on the weight lattice $\Lambda^*_K$ that have finite support.
In this way, the character $\ch V$ and the weight multiplicity function $m_{T,V}$ become one and the same object.
We remark that $\ch V$ is directly related to the usual definition of a character as the trace of the representation:
we have that $\tr \Pi(\exp(X)) = \sum_\omega m_{T,V}(\omega) e^{(\omega,X)}$ for all $X \in \mathfrak t$.

\subsection*{Weyl and Kostant Formulas}

The character of an irreducible representation $V_{K,\lambda}$ is given by the famous \emphindex{Weyl character formula}
\begin{equation}
  \label{eq:mul/weyl character formula}
  \ch V_{K,\lambda} =
  \frac {\sum_{w \in W_K} (-1)^{l(w)} \, e^{w(\lambda + \rho_K) - \rho_K}} {\prod_{\alpha \in R_{K,+}} \left( 1 - e^{-\alpha} \right)}.
\end{equation}
where $\rho_K = \sum_{\alpha \in R_{K,+}} \alpha/2 \in \Lambda^*_K$ is known as the \emphindex{Weyl vector}.
To make sense of the right-hand side fraction, we need to work in the larger ring of functions that are supported in a finite number of ``cones'' of the form \[\{ \omega - \sum_\alpha x_\alpha \alpha : (x_\alpha) \in \ZZ^{R_{K,+}}_{\geq 0} \}.\]
This restriction ensures that multiplication is still well-defined \cite{Knapp02} (unlike for general functions!).
We find that the denominator in \eqref{eq:mul/weyl character formula} is indeed invertible, with inverse
\begin{equation*}
  \frac 1 {\prod_{\alpha \in R_{K,+}} \left( 1 - e^{-\alpha} \right)}
  = \prod_{\mathclap{\alpha \in R_{K,+}}} \left( 1 + e^{-\alpha} + e^{-2\alpha} + \dots \right)
  = \sum_{\mathclap{\omega \in \Lambda^*_K}} \phi_{R_{K,+}}(\omega) e^{-\omega}.
\end{equation*}
In the last equation we have introduced the \emph{Kostant partition function}
\begin{equation*}
  \phi_{R_{K,+}}(\omega) = \# \{ (x_\alpha) \in \ZZ^{R_{K,+}}_{\geq 0} : \sum_\alpha x_\alpha \alpha = \omega \},
\end{equation*}
which counts the number of ways that a weight can be written as a sum of positive roots (this number is always finite since the positive roots span a proper cone).
It follows directly from the above that
\begin{align*}
  \ch V_{K,\lambda}
  = &\sum_{w \in W_K} (-1)^{l(w)} \sum_{\omega \in \Lambda^*_K} \phi_{R_{K,+}}(\omega) e^{w(\lambda + \rho) - \rho - \omega} \\
  = &\sum_{\omega \in \Lambda^*_K} \sum_{w \in W_K} (-1)^{l(w)} \, \phi_{R_{K,+}}(w(\lambda + \rho) - \rho - \omega) e^\omega.
\end{align*}
Thus the multiplicity of a weight $\omega$ in an irreducible representation $V_{K,\lambda}$ is given by the well-known \emph{Kostant multiplicity formula} \cite{Kostant59},
\begin{equation}
  \label{eq:mul/kostant}
  m_{T,V_{K,\lambda}}(\omega) = \sum_{w \in W_K} (-1)^{l(w)} \, \phi_{R_{K,+}}(w(\lambda + \rho) - \rho - \omega).
\end{equation}

For any fixed group $K$, the Kostant partition function can be evaluated in polynomial time by using Barvinok's algorithm \cite{Barvinok94}, since it amounts to counting integral points in a convex polytope in an ambient space of fixed dimension.
Therefore, weight multiplicities for fixed groups $K$ can be computed efficiently.
This idea has been implemented by Cochet \cite{Cochet05} to compute weight multiplicities for the classical Lie algebras (but using the method of \cite{BaldoniBeckCochetEtAl06} instead of Barvinok's algorithm).
We note that the problem of computing weight multiplicities is of course a special case of \autoref{pro:mul/subgroup restriction} where $K$ is the maximal torus of $K'$.

\subsection*{Weight Multiplicities as a Single Partition Function}

If the group $K$ is semisimple, then it is known that we can find $s, t \in \ZZ_{\geq 0}$ and $\ZZ$-linear maps $A \colon \ZZ^s \rightarrow \ZZ^t$, $B \colon \Lambda^*_K \oplus \Lambda^*_K \rightarrow \ZZ^t$ such that
\begin{equation}
  \label{eq:mul/bliem}
  m_{T,V_{K,\lambda}}(\omega) = \phi_A \left( B \matrix{\lambda \\ \omega} \right)
\end{equation}
for all $\lambda \in \Lambda^*_{K,+}, \omega \in \Lambda^*_K$,
where $\phi_A$ is a \emphindex{vector partition function} defined by
\begin{equation}
  \phi_A(y) = \# \{ x \in \ZZ^s_{\geq 0} : A x = y \}.
\end{equation}
Note that this improves over the Kostant multiplicity formula \eqref{eq:mul/kostant}, where weight multiplicities are expressed as an alternating sum over several invocations of a vector partition function.
In particular, \eqref{eq:mul/bliem} is an evidently positive formula.
Billey, Guillemin, and Rassart have constructed \eqref{eq:mul/bliem} for the Lie group $K = \SU(d)$ \cite{BilleyGuilleminRassart04} by using Gelfand--Tsetlin patterns \cite{GelfandTsetlin88}; the general construction is due to Bliem \cite{Bliem08} based on Littelmann patterns \cite{Littelmann98}.

We note that the assumption of semisimplicity for \eqref{eq:mul/bliem} is not a restriction.
If $K$ is a general compact connected Lie group then its Lie algebra can always decomposed as
\begin{equation*}
  \mathfrak k = [\mathfrak k,\mathfrak k] \oplus \mathfrak z,
\end{equation*}
where $[\mathfrak k, \mathfrak k]$ is the Lie algebra of a compact connected semisimple Lie group $K_{\sesi}$ and $\mathfrak z$ the Lie algebra of the center $Z(K)$ of $K$.
By Schur's lemma \eqref{eq:onebody/schurs lemma}, each element of the center acts by a scalar on an irreducible representation $V_{K,\lambda}$.
Therefore, all weights that appear in $V_{K,\lambda}$ have the same restriction to $\mathfrak z$.
It follows that
\begin{equation}
  \label{eq:mul/semisimple}
    m_{T,V_{K,\lambda}}(\omega) =
    \begin{cases}
      m_{T_{\sesi},V_{K_{\sesi},\lambda_{\sesi}}}(\omega_{\sesi}) & \text{if $\omega_z = \lambda_z$},\\
      0 & \text{otherwise},
  \end{cases}
\end{equation}
where we write $\omega = \omega_{\sesi} \oplus \omega_{z}$ according to the corresponding decomposition of weight lattices $\Lambda^*_K = \Lambda^*_{K_{\sesi}} \oplus \Lambda^*_{Z(K)}$, and likewise for $\lambda$.
These multiplicities can thus be evaluated by using \eqref{eq:mul/bliem}.

\section{The Finite Difference Formula}
\label{sec:mul/finite difference formula}

Let $V$ be a finite-dimensional representation of the compact connected Lie group $K$.
In the preceding section we have seen that we can compute the weight multiplicity function $m_{T,V}$ from the highest weight multiplicity function $m_{K,V}$ by using one of the classical formulas \eqref{eq:mul/weyl character formula} and \eqref{eq:mul/kostant}, or by evaluating the vector partition function \eqref{eq:mul/bliem} of Bliem.
By ``inverting'' the Weyl character formula, the converse can also be achieved:

\begin{lem}
  \label{lem:mul/steinberg}
  \index{finite-difference formula}
  For any finite-dimensional $K$-representation $V$, we have that
  \begin{equation}
  \label{eq:mul/finite difference formula}
    m_{K,V} = \left. \left[ \left(\prod_{\alpha \in R_{K,+}} - D_\alpha \right) m_{T,V} \right] \right|_{\Lambda^*_{K,+}},
  \end{equation}
  where $(D_\alpha m)(\lambda) = m(\lambda + \alpha) - m(\lambda)$ is the \emphindex{finite-difference operator}\nomenclature[RD_alpha]{$D_\alpha$}{finite-difference operator in direction $\alpha$} in direction $\alpha$.
\end{lem}
\begin{proof}
  By linearity, it suffices to establish the lemma for an irreducible representation $V = V_{K,\lambda}$.
  The Weyl character formula \eqref{eq:mul/weyl character formula} can be rewritten in the form
  \begin{equation}
    \label{eq:mul/weyl for steinberg}
    \prod_{\alpha > 0} \left( 1 - e^{-\alpha} \right) \ch V_{K,\lambda} =
    \sum_{w \in W_K} \det(w) \, e^{w(\lambda + \rho_K) - \rho_K}.
  \end{equation}
  Under our identification of elements in $\ZZ[\Lambda^*_K]$ with functions on the weight lattice, multiplication with $\left(e^{-\alpha} - 1 \right)$ amounts to applying the finite-difference operator $D_\alpha$.
  Thus the left-hand side of \eqref{eq:mul/weyl for steinberg} can be identified with $\left( \prod_{\alpha \in R_{K,+}} \!- D_\alpha \right) m_{T,V_{K,\lambda}}$.

  Now consider the right-hand side of \eqref{eq:mul/weyl for steinberg}.
  Since $\lambda + \rho_K$ is a strictly dominant weight in $\Lambda^*_{K,>0} = \Lambda^*_K \cap i \mathfrak t^*_{>0}$, it is sent by any Weyl group element $w \neq 1$ to the interior of a different Weyl chamber.
  That is, for any $w \neq 1$ there exists a positive root $\alpha \in R_{K,+}$ such that $(w(\lambda + \rho_K), H_\alpha) < 0$.
  Therefore $w(\lambda + \rho_K) - \rho_K$ is a dominant weight if only if $w = 1$, in which case it is equal to $\lambda$.
  It follows that the right-hand side of \eqref{eq:mul/weyl for steinberg} identifies with a function on the weight lattice whose restriction to $\Lambda^*_{K,+}$ is equal to the indicator function of $\{\lambda\}$, i.e., to the highest weight multiplicity function $m_{K,V_{K,\lambda}} = \delta_\lambda$ of $V_{K,\lambda}$.
\end{proof}

The idea of using \eqref{eq:mul/weyl character formula} for determining multiplicities of irreducible representations goes back at least to Steinberg \cite{Steinberg61}, who proved a formula for the multiplicity $c^{\alpha,\beta}_\lambda$ of an irreducible representation $V_{K,\lambda}$ in the tensor product $V_{K,\alpha} \otimes V_{K,\beta}$.
These multiplicities $c^{\alpha,\beta}_\lambda$ are called the \emphindex{Littlewood--Richardson coefficients} for $K$ (cf.\ \autoref{sec:onebody/kronecker}).
Steinberg's formula involves an alternating sum over the Kostant partition function \eqref{eq:mul/kostant} and can therefore be evaluated efficiently as described by Cochet \cite{Cochet05}.
De Loera and McAllister give another method for computing Littlewood--Richardson coefficients \cite{DeMcAllister06}, which applies Barvinok's algorithm to a result by Berenstein and Zelevinsky \cite{BerensteinZelevinsky01}.
Since the tensor products of irreducible $K$-representations are just the irreducible representations of $K \times K$, the problem of computing Littlewood--Richardson coefficients is again a special case of \autoref{pro:mul/subgroup restriction}.
The Clebsch--Gordan series in quantum mechanics is routinely derived in a similar fashion (it corresponds to the case $K = \SU(2)$).

An immediate consequence of \autoref{lem:mul/steinberg} is that there exists a finite set of weights $\Gamma_K \subseteq \Lambda^*_K$ and coefficients $c_\gamma \in \ZZ$ such that
\begin{equation}
\label{eq:mul/steinberg concrete}
    m_{K,V}(\lambda) = \sum_{\gamma \in \Gamma_K} c_\gamma \, m_{T,V}(\lambda + \gamma).
\end{equation}
In fact, it is not hard to see that the coefficients can be defined by the expansion $\prod_{\alpha \in R_{K,+}} \!\left( 1 - e^{-\alpha} \right) = \sum_{\gamma \in \Gamma_K} c_\gamma e^{-\gamma}$.

\section{Multiplicities for the Subgroup Restriction Problem}
\label{sec:mul/subgroup restriction problem}

Now let $\Phi \colon K \rightarrow K'$ be a homomorphism of compact, connected Lie groups and $\phi \colon \mathfrak k \rightarrow \mathfrak{k'}$ the induced homomorphism of Lie algebras.%
\nomenclature[RPhi]{$\Phi$, $\phi$, $\phi^*$}{homomorphism of Lie groups, corresponding homomorphism of Lie algebras, dual map between weight lattices}
We choose maximal tori $T \subseteq K$ and $T' \subseteq K'$ such that $\Phi(T) \subseteq T'$.
Then $\phi(\mathfrak t) \subseteq \mathfrak{t'}$,
and the (complexified) dual map restricts to a $\ZZ$-linear map between the weight lattices, which we shall denote by
$\phi^* \colon \Lambda^*_{K'} \rightarrow \Lambda^*_K$.
Indeed, the following lemma follows readily from the definitions:

\begin{lem}
  \label{lem:mul/weight restriction}
  Let $V$ be a representation of $K'$ and $\ket\psi \in V$ a weight vector of weight $\omega \in \Lambda^*_{K'}$.
  If we restrict the action to $K$ via $\Phi \colon K \rightarrow K'$ then $\ket\psi$ is a weight vector of weight $\phi^*\omega \in \Lambda^*_K$.
\end{lem}

As alluded to in the introduction, our strategy for solving the subgroup restriction problem then is the following:
Given an irreducible representation $V_{K',\lambda}$ of $K'$, we can determine its weight multiplicities with respect to the maximal torus $T_{K'}$ by using any of the formulas from \autoref{sec:mul/weights}.
We then obtain the weight multiplicities for $T$ by restricting as in \autoref{lem:mul/weight restriction}.
Finally, we reconstruct the multiplicity of an irreducible representation $V_{K,\mu}$ by using the finite-difference formula from \autoref{sec:mul/finite difference formula}.
If this procedure were to be translated directly into an algorithm, the runtime would be polynomial in the coefficients of $\lambda$, i.e., exponential in their bitlength (cf.\ the branching formula by Straumann \cite{Straumann65}).
Indeed, the number of weights generically grows polynomially and can even be of the order of the dimension of the irreducible representation $V_{K',\lambda}$, which according to the \emphindex{Weyl dimension formula} is given by
\begin{equation}
\label{eq:mul/weyl dimension formula}
  \dim V_{K',\lambda} = \prod_{\alpha \in R_{K',+}} \frac {(\lambda + \rho_{K'}, H_\alpha)} {(\rho_{K'}, H_\alpha)}.
\end{equation}
It is however possible to combine \eqref{eq:mul/bliem} with the restriction map $\phi^*$ in a way that will subsequently give rise to an algorithm that runs in polynomial time in the bitlength of the input:

\begin{prp}
  \label{prp:mul/main}
  Let $\Phi \colon K \rightarrow K'$ be a homomorphism of compact connected Lie groups.
  Then we can construct $\ZZ$-linear maps $A' \colon \ZZ^{s + s'} \rightarrow \ZZ^u$ and $B' \colon \Lambda^*_{K'} \oplus \Lambda^*_K \rightarrow \ZZ^u$ with $s = O(r_{K'}^2)$, $s' = r_{K'_{\sesi}} \leq r_{K'}$ and $u = O(r_{K'}^2) + r_K$ that satisfy the following property:
  For any two highest weights $\lambda \in \Lambda^*_{K',+}$ and $\mu \in \Lambda^*_{K,+}$, the multiplicity $m^\lambda_\mu$ is given by
  \begin{equation}
  \label{eq:mul/main}
    m^\lambda_\mu =
    \sum_{\gamma \in \Gamma_K} c_\gamma \, \# \{ x' \in \ZZ^s_{\geq 0} \oplus \ZZ^{s'} : A' x' = B' \matrix{\lambda \\ \mu + \gamma} \},
  \end{equation}
  where the set $\Gamma_K$ and the coefficients $c_\gamma \in \ZZ$ are defined as in \eqref{eq:mul/steinberg concrete}.
\end{prp}
\begin{proof}
  We start with the finite-difference formula in the form \eqref{eq:mul/steinberg concrete},
  \begin{equation*}
    m^\lambda_\mu
    = m_{K,V_{K',\lambda}}(\mu)
    = \sum_{\gamma \in \Gamma_K} c_\gamma \, m_{T,V_{K',\lambda}}(\mu + \gamma).
  \end{equation*}
  According to \autoref{lem:mul/weight restriction}, weight multiplicities for $K$ can be expressed in terms of weight multiplicities for $K'$:
  \begin{equation*}
    m_{T,V_{K',\lambda}}(\omega)
    = \sum_{\mathclap{\substack{\omega' \in \Lambda^*_{K'} \\ \phi^*\omega' = \omega}}}
      \; m_{T',V_{K',\lambda}}(\omega')
    = \sum_{\mathclap{\substack{\omega'_{\sesi} \in \Lambda^*_{K'_{\sesi}} \\ \phi^*(\omega'_{\sesi} \oplus \lambda_z) = \omega}}}
      \; m_{T'_{\sesi},V_{K'_{\sesi},\lambda_{\sesi}}}(\omega'_{\sesi}),
  \end{equation*}
  where we have used \eqref{eq:mul/semisimple} to reduce to the semisimple part $K'_{\sesi} \subseteq K'$.
  According to \cite{Bliem08}, we can find $\ZZ$-linear maps $A \colon \ZZ^s \rightarrow \ZZ^t$, $B \colon \Lambda^*_{K'_{\sesi}} \oplus \Lambda^*_{K'_{\sesi}} \rightarrow \ZZ^t$ with $s, t = O(r_{K'}^2)$ such that weight multiplicities for $K'_{\sesi}$ are given as in \eqref{eq:mul/bliem}.
  Thus,
  \begin{equation*}
      \sum_{\mathclap{\substack{\omega'_{\sesi} \in \Lambda^*_{K'_{\sesi}} \\ \phi^*(\omega'_{\sesi} \oplus \lambda_z) = \omega}}}
      \; m_{T'_{\sesi},V_{K'_{\sesi},\lambda_{\sesi}}}(\omega'_{\sesi})
    = \sum_{\mathclap{\substack{\omega'_{\sesi} \in \Lambda^*_{K'_{\sesi}} \\ \phi^*(\omega'_{\sesi} \oplus \lambda_z) = \omega}}}
      \; \#\{ x \in \ZZ^s_{\geq 0} : Ax = B \matrix{\lambda_{\sesi} \\ \omega'_{\sesi}} \}.
  \end{equation*}
  Finally, decompose $B = B_1 \oplus B_2$ according to $\Lambda^*_{K'_{\sesi}} \oplus \Lambda^*_{K'_{\sesi}}$ and $\phi^* = \phi^*_{\sesi} \oplus \phi^*_z$ according to $\Lambda^*_K = \Lambda^*_{K_{\sesi}} \oplus \Lambda^*_{Z(K)}$. Then we can rewrite the above in the following way:
  \begin{equation}
  \label{eq:mul/constructive}
  \begin{aligned}
    &\#\{ (x, \omega'_{\sesi}) \in \ZZ^s_{\geq 0} \oplus \Lambda^*_{K'_{\sesi}} :
      Ax = B \matrix{\lambda_{\sesi} \\ \omega'_{\sesi}}, \phi^*(\omega'_{\sesi} \oplus \lambda_z) = \omega \} \\
  =\;&\#\{ (x, \omega'_{\sesi}) \in \ZZ^s_{\geq 0} \oplus \Lambda^*_{K'_{\sesi}} :
      Ax = B_1 \lambda_{\sesi} + B_2 \omega'_{\sesi}, \;
      \phi^*_{\sesi} \omega'_{\sesi} + \phi^*_z \lambda_z = \omega \} \\
  =\;&\#\{ (x, \omega'_{\sesi}) \in \ZZ^s_{\geq 0} \oplus \Lambda^*_{K'_{\sesi}} :
        \underbrace{\matrix{A & -B_2 \\ 0 & \phi^*_{\sesi}}}_{=: A'} \matrix{x \\ \omega'_{\sesi}}
      = \underbrace{\matrix{B_1 & 0 & 0 \\ 0 & -\phi^*_z & \Id}}_{=: B'} \matrix{\lambda_{\sesi} \\ \lambda_z \\ \omega} \}
  \end{aligned}
  \end{equation}
  After choosing a basis of the lattice $\Lambda^*_{K'_{\sesi}} \cong \ZZ^{r_{K'_{\sesi}}}$
  we arrive at the asserted formula (with $u = t + r_K$).
\end{proof}

We note that the proof of \autoref{prp:mul/main} is constructive:
The maps $A'$ and $B'$, whose existence is asserted by the proposition, are defined in \eqref{eq:mul/constructive} in terms of the maps $A$ and $B$ constructed explicitly in \cite[\S 4]{Bliem08} (or \cite[Proof of Theorem 2.1]{BilleyGuilleminRassart04} for $\mathfrak k = \su(d)$).
In \autoref{sec:mul/kronecker} we give an illustration in the context of the Kronecker coefficients.
If one uses the Kostant multiplicity formula \eqref{eq:mul/kostant} rather than Bliem's formula \eqref{eq:mul/bliem} in the proof of \autoref{prp:mul/main} then one obtains at a similar formula for the multiplicities $m^\lambda_\mu$.
After completion of this work, we have learned of \cite[(3.5)]{Heckman82}, which is derived precisely in this spirit (and attributed to Kostant).
\subsection*{Consequences}

Let us identify
\begin{equation}
\label{eq:mul/weight lattice identification}
  \Lambda^*_K \cong \ZZ^{r_K}
  \text{ and }
  \Lambda^*_{K'} \cong \ZZ^{r_{K'}}
\end{equation}
according to the bases of fundamental weights.
Then the $\ZZ$-linear maps $A'$ and $B'$ in \autoref{prp:mul/main} correspond to matrices with integer entries, which we shall denote by the same symbol.
In this way, we find that \eqref{eq:mul/main} in essence reduces the computation of the multiplicities $m^\lambda_\mu$ to counting the number of integral points in rational convex polytopes of the form
\begin{equation}
  \label{eq:mul/weight restriction polytope}
  \Delta_{A',B'}(y) = \{ x' \in \RR^s_{\geq 0} \oplus \RR^{s'} : A'x' = B'y \}.
\end{equation}
parametrized by $y \in \ZZ^{r_{K'} + r_K}$.
It is well-known that this number
\[n(y) = \#(\Delta_{A',B'}(y) \cap \ZZ^{s+s'})\]
is a \emphindex{piecewise quasi-polynomial} function of $y$ \cite{ClaussLoechner98}.
That is, there exists a decomposition of $\ZZ^{r_{K'}+r_K}$ into polyhedral chambers such that on each chamber $C$ the function $n(y)$ is given by a single quasi-polynomial, i.e., there exists a sublattice $L \subseteq \ZZ^{r_{K'}+r_K}$ of finite index and polynomials $p_z$ with rational coefficients, labeled by the finitely many points $z \in \ZZ^{r_{K'}+r_K} / L$, such that $n(y) = p_{[y]}(y)$ for all $y \in \ZZ^{r_{K'}+r_K}$ (cf.\ \cite[\S 2.2]{VerdoolaegeSeghirBeylsEtAl07}).
It follows that the multiplicities $m^\lambda_\mu$, too, are piecewise quasi-polynomial functions.
While the chambers for $n(y)$ are rational polyhedral cones, it does not follow that the same is true for $m^\lambda_\mu$ (due to the shifts in \eqref{eq:mul/main}).
However, what remains true is that the \emph{stretching function} $k \mapsto m^{k\lambda}_{k\mu}$ is a quasi-polynomial function for large $k$.
It is more involved to show that this is in fact true for all $k$, as has been observed in \cite{Mulmuley07}; we refer to \cite{MeinrenkenSjamaar99} for more general quasi-polynomiality results on polyhedral cones (cf.\ the discussion in \cite{BaldoniVergne10}).

\subsection*{Polynomial-Time Algorithm for the Subgroup Restriction Problem}

We will now formulate our polynomial-time algorithm for the subgroup restriction problem.
As we have just explained, \eqref{eq:mul/main} reduces the computation of the multiplicities $m^\lambda_\mu$ to counting the number of integral points in rational convex polytopes of the form \eqref{eq:mul/weight restriction polytope}.
We shall suppose that the highest weights $\lambda$ and $\mu$, which are the \emph{input} to our algorithm, are given in terms of bitstrings containing coordinates with respect to the identification \eqref{eq:mul/weight lattice identification}.
Clearly, for each of the finitely many $\gamma \in \Gamma_K$, the polytope $\Delta_{A',B'}(\lambda, \mu+\gamma)$ defined in \eqref{eq:mul/weight restriction polytope} can be described in polynomial size in the bitlength of the input (e.g., in terms of linear equalities and inequalities), and it can be produced from the input in polynomial time.
Therefore we may use \emphindex{Barvinok's algorithm} to compute the number of integral points in each of these polytopes in polynomial time \cite{Barvinok94}.
We thus obtain the following algorithm:

\begin{alg}
  \label{alg:mul/main}
  Let $\Phi \colon K \rightarrow K'$ be a homomorphism of compact connected Lie groups.
  Given as input two highest weights $\lambda \in \Lambda^*_{K'} \cong \ZZ^{r_{K'}}$ and $\mu \in \Lambda^*_K \cong \ZZ^{r_K}$, encoded as bitstrings containing their coordinates with respect to \eqref{eq:mul/weight lattice identification}, the following algorithm computes the multiplicity $m^\lambda_\mu$ in polynomial time in the bitlength of the input:
  \begin{algorithmic}
    \State $m \gets 0$
    \ForAll{$\gamma \in \Gamma_K$}
      \State $n \gets \# \left( \Delta_{A',B'}(\lambda, \mu + \gamma) \cap \ZZ^{s+s'} \right)$ \Comment{Barvinok's algorithm}
      \State $m \gets m + c_\gamma n$
    \EndFor
    \State \textbf{return} $m$
  \end{algorithmic}
  Here, $\Delta_{A',B'}(y)$ denotes the rational convex polytope \eqref{eq:mul/weight restriction polytope}, and the set $\Gamma_K \subseteq \Lambda^*_K$ as well as the coefficients $c_\gamma \in \ZZ$ are defined as in \eqref{eq:mul/steinberg concrete}.
\end{alg}

There are at least two software packages which have implemented Barvinok's algorithm, namely \textsc{LattE} \cite{DeDutraKoeppeEtAl11} and \textsc{barvinok} \cite{Verdoolaege14,VerdoolaegeSeghirBeylsEtAl07}.
In \autoref{sec:mul/summary} we have reported on the performance of our preliminary implementation \cite{Walter12a} of \autoref{alg:mul/main} for computing Kronecker coefficients using the latter package.

\section{Kronecker Coefficients}
\label{sec:mul/kronecker}

In this section we will describe precisely how the Kronecker coefficients can be computed using the general method.

Let $K = \U(a) \times \U(b) \times \U(c)$.
Recall from \autoref{sec:onebody/kronecker} that the \emphindex{Kronecker coefficients} $g_{\alpha,\beta,\gamma}$ can be defined as the multiplicity of an irreducible $K$-representation $V^a_\alpha \otimes V^b_\beta \otimes V^c_\gamma$ in the symmetric subspace $\Sym^k(\CC^{abc})$:
\begin{equation}
  g_{\alpha,\beta,\gamma} = m_{\alpha,\beta,\gamma}^{(k)},
\end{equation}
where we recall that $(k) = (k,0,\dots,0)$ is the highest weight of $\Sym^k(\CC^{abc})$.
We may think of $\alpha$, $\beta$ and $\gamma$ as Young diagrams with $k$ boxes each and no more than $a$, $b$, and $c$ rows, respectively; the Kronecker coefficients do not depend on the concrete choice of $a$, $b$ and $c$ (as long as they are chosen at least as large as the number of rows of the respective Young diagrams).
It follows that \autoref{alg:mul/main} can be used to compute Kronecker coefficients of Young diagrams with bounded height in polynomial time in the input size---that is, in time $O(\poly(\log k))$, where $k$ is the number of boxes of the Young diagrams.
We note again that the problem of computing Kronecker coefficients is known to be $\SharpP$-hard in general \cite{BuergisserIkenmeyer08}; hence we do not expect that there exists a polynomial-time algorithm without any assumption on the height of the Young diagrams.

When computing Kronecker coefficients using our method, we are only interested in the subgroup restriction problem for the symmetric subspaces rather than for arbitrary irreducible representations $V^{abc}_\lambda$.
By specializing the construction of \autoref{prp:mul/main} to this one-parameter family of representations, we obtain the following formulas:

\begin{prp}
  \label{pro:mul/optimized kronecker}
  The multiplicity of a weight $\omega = (\omega_A,\omega_B,\omega_C) \in \ZZ^a \oplus \ZZ^b \oplus \ZZ^c \cong \Lambda^*_K$ in $\Sym^k(\CC^{abc})$ is equal to
  \begin{equation}
  \label{eq:mul/contingency}
  \begin{aligned}
    &n(k, \omega) = \# \{
      (x_{i,j,k}) \in \ZZ^{abc}_{\geq 0} \; : \;
      \sum_{\mathclap{i,j,k}} x_{i,j,k} = k, \\
      &\qquad \sum_{\mathclap{j,k}} x_{i,j,k} = \omega_{A,i} \; (\forall i),
      \sum_{\mathclap{i,k}} x_{i,j,k} = \omega_{B,j} \; (\forall j),
      \sum_{\mathclap{i,j}} x_{i,j,k} = \omega_{C,k} \; (\forall k)
    \}.
  \end{aligned}
  \end{equation}
  Therefore, the Kronecker coefficient for Young diagrams $\lambda, \mu, \nu$ with $k$ boxes and no more than $a$, $b$ and $c$ rows, respectively, is given by the formula
  \begin{equation*}
    g_{\lambda,\mu,\nu} = \sum_{\gamma \in \Gamma_K} c_\gamma \, n(k, (\lambda, \mu, \nu) + \gamma),
  \end{equation*}
  where $\Gamma_K$ and $c_\gamma \in \ZZ$ are defined as in \eqref{eq:mul/steinberg concrete}.
\end{prp}
\begin{proof}
  The weights of the symmetric subspace $\Sym^k(\CC^{abc})$ considered as a representation of $\U(abc)$ are the ``bosonic occupation numbers''
  \begin{equation*}
    \{ (x_{i,j,k}) \in \ZZ^{abc}_{\geq 0} : \sum_{i,j,k} x_{i,j,k} = k \},
  \end{equation*}
  and the corresponding weight spaces are all one-dimensional.
  On the other hand, the dual map between the weight lattices induced by the tensor product embedding
  \begin{equation*}
    \Phi \colon \U(a) \times \U(b) \times \U(c) \rightarrow \U(abc), \quad (U,V,W) \mapsto U \otimes V \otimes W
  \end{equation*}
  is given by
  \begin{equation*}
    \phi^* \colon \ZZ^{abc} \rightarrow \ZZ^a \oplus \ZZ^b \oplus \ZZ^c \cong \Lambda^*_K, \quad
    (x_{i,j,k}) \mapsto (\sum_{j,k} x_{i,j,k}, \sum_{i,k} x_{i,j,k}, \sum_{j,k} x_{i,j}).
  \end{equation*}
  The formulas asserted in the proposition follow at once.
\end{proof}

\autoref{pro:mul/optimized kronecker} gives rise to a polynomial-time algorithm for computing Kronecker coefficients with a bounded number of rows which in practice is faster than \autoref{alg:mul/main} (as the ambient space $\RR^{abc}$ has a smaller dimension than what we would get from the general construction described in the proof of \autoref{prp:mul/main}).
The time complexity of the resulting algorithm can be extracted from the scaling of Barvinok's algorithm in the dimension of the ambient space \cite{BarvinokPommersheim99, PakPanova14}.
We remark that \eqref{eq:mul/contingency} counts the number of \emphindex{three-way contingency tables} with fixed marginal, which appear in statistics. It is known that counting two-way contingency tables is already $\SharpP$-complete \cite{DyerKannanMount97}.

\section{Asymptotics}
\label{sec:mul/asymptotics}
\index{semiclassical limit}

Let $G$ be a connected reductive algebraic group, $K \subseteq G$ a maximal compact subgroup, and $\Pi \colon G \rightarrow \GL(\calH)$ be a representation on a Hilbert space $\calH$ with $K$-invariant inner product.
Let $\XX \subseteq \PP(\calH)$ be a $G$-invariant projective subvariety.
As motivated in \autoref{sec:mul/summary}, we consider the following measure
\begin{equation}
\label{eq:mul/multiplicity measure}
  \nu_{K,\XX} := \wlim_{k \rightarrow \infty}  \frac 1 {k^{d_{K,\XX}}}  \sum_{\lambda \in \Lambda^*_{K,+}} m_{K,R^*_k(\XX)}(\lambda) \, \delta_{\lambda/k},
\end{equation}
which encodes the asymptotic growth of multiplicities in the ring of regular functions.
The exponent $d_{K,\XX}$ is chosen appropriately such that $\nu_{K,\XX}$ has finite, non-zero measure;
in this case $\nu_{K,\XX}$ is absolutely continuous with respect to Lebesgue measure on the moment polytope of $\XX$ \cite{Okounkov96}.
We remark that the subgroup restriction problem fits naturally into this more general setup.
Indeed, for $\XX = \calO_{K',\lambda}$ the coadjoint orbit through a dominant weight $\lambda \in \Lambda^*_{K'}$, the Borel--Weil theorem (\autoref{lem:onebody/borel weil}) implies that $m^{k\lambda}_{\mu}$ is equal to $m_{K,R^*_k(\calO_{K',\lambda})}(\mu)$.

\subsection*{Projective Space}

In \autoref{sec:mul/summary}, we had mentioned that the measures $\nu_{K,\XX}$ are closely related to the Duistermaat--Heckman measures as defined in \autoref{sec:dhmeasure/dhmeasure}.
We will now explain how this can be seen rather directly in the case of projective space, $\XX = \PP(\calH)$.
The starting point is the observation, already used in the proof of \autoref{pro:mul/optimized kronecker}, that the weights of $R_k^*(\calH) = \Sym^k(\calH)$ considered as a representation of $K' = \U(\calH)$ are given by ``bosonic occupation numbers'' $\vec x \in \ZZ^D_{\geq 0}$ with $\sum_{i=1}^D x_i = k$, where $D = \dim \calH$.
In other words, the weights correspond to the integral points of $k \Delta_D$, the $k$-times dilated standard simplex, and each weight occurs with multiplicity one.
Thus it is immediate that
\begin{equation*}
  \nu_{T',\PP(\calH)}
  = \wlim_{k \rightarrow \infty} \frac 1 {k^{D-1}} \; \sum_{\mathclap{\vec x \in k \Delta_D \cap \ZZ^D}} \; \delta_{\vec x/k}
  = \wlim_{k \rightarrow \infty} \frac 1 {k^{D-1}} \; \sum_{\mathclap{\vec p \in \Delta_D \cap \ZZ^D/k}} \; \delta_{\vec p}
\end{equation*}
is equal to Lebesgue measure on the standard simplex, normalized to total volume
\begin{equation*}
  \frac {\#(k \Delta_D \cap \ZZ^D)} {k^{D-1}}
  = \frac {\binom {k+D-1} k} {k^{D-1}} \rightarrow \frac 1 {(D-1)!}.
\end{equation*}
On the other hand, $\Prob_{T'}$ is equal to Lebesgue measure on the standard simplex normalized to probability one, so that:
\begin{equation}
\label{eq:mul/semiclassical simplex}
  \Prob_{T',\PP(\calH)} = (D-1)! \, \nu_{T',\PP(\calH)}.
\end{equation}
(Indeed, the measures $\nu_{K,\XX}$ are always normalized to the symplectic volume of $\XX$.)

Now consider $\calH$ as a representation of $T \subseteq T'$, the maximal torus of $K$.
By pushing forward both the left and the right-hand side of \eqref{eq:mul/semiclassical simplex} along the map $(x_i) \mapsto \sum_{k=1}^D x_i \omega_i$, we obtain that
\begin{equation}
\label{eq:mul/semiclassical abelian}
  \Prob_{T,\PP(\calH)} = (D-1)! \, \nu_{T,\PP(\calH)}.
\end{equation}
(\autoref{lem:dhmeasure/abelian} and \autoref{lem:mul/weight restriction}).
Thus the Abelian Duistermaat--Heckman measure encodes the asymptotics of the weight multiplicities.
In the case of the quantum marginal problem, this is plainly visible by comparing \eqref{eq:dhmeasure/qmp cmp} and \eqref{eq:mul/contingency}.

We now lift \eqref{eq:mul/semiclassical abelian} to the non-Abelian group $K$.
As in \autoref{sec:dhmeasure/dhmeasure} we assume that $\Delta_K \cap i \mathfrak t^*_{>0} \neq \emptyset$, so that the exponent in \eqref{eq:mul/multiplicity measure} is given by $d_{K,\PP(\calH)} = D-1 - \#R_{K,+}$ (see \autoref{lem:mul/exponent} below).
Let $g$ be a smooth test function on $i \mathfrak t^*$ that is compactly supported in $i \mathfrak t^*_{>0}$.
Using the finite-difference formula \eqref{eq:mul/finite difference formula},
\begin{align*}
    &\sum_{\lambda \in \Lambda^*_{K,+}} m_{K,\PP(\calH)}(\lambda) \, g(\lambda/k) \\
  = &\sum_{\lambda \in \Lambda^*_{K,+}} \left[ \left( \;\;\; \prod_{\mathclap{\alpha \in R_{K,+}}} - D_\alpha \right) m_{T,\PP(\calH)} \right](\lambda) \, g(\lambda/k) \\
  = &\sum_{\lambda \in \Lambda^*_{K,+}} m_{T,\PP(\calH)}(\lambda)
  \left[ \left( \;\;\; \prod_{\mathclap{\alpha \in R_{K,+}}} - D_{-\alpha/k} \right) g \right] (\lambda/k)
\end{align*}
In the limit as $k \rightarrow \infty$,
\begin{equation*}
  \frac 1 {k^{\#R_{K,+}}} \left( \;\;\; \prod_{\mathclap{\alpha \in R_{K,+}}} - D_{-\alpha/k} \right) g
  \rightarrow
  \left( \prod_{\alpha \in R_{K,+}} \partial_\alpha \right) g.
\end{equation*}
since the (suitably rescaled) finite differences converge to partial derivatives.
It follows by using the definition of the measures \eqref{eq:mul/multiplicity measure} that
\begin{equation*}
    \int d\nu_{K,\PP(\calH)} \, g
  = \int d\nu_{T,\PP(\calH)} \, \left( \prod_{\alpha \in R_{K,+}} \partial_\alpha \right) g
\end{equation*}
for all test functions $g$ that are compactly supported in $i \mathfrak t^*_{>0}$.
In other words, up to a trivial factor $p_K$, the measures $\nu_{K,\PP(\calH)}$ satisfy the same derivative principle as the Duistermaat--Heckman measures (\autoref{lem:dhmeasure/derivative principle}).
Thus we obtain at once from \eqref{eq:mul/semiclassical abelian} that the non-Abelian Duistermaat--Heckman measure encodes the asymptotics of the highest weight multiplicities:
\begin{equation}
\label{eq:mul/semiclassical non-abelian}
  \Prob_{K,\PP(\calH)} = (D-1)! \, p_K \, \nu_{K,\PP(\calH)}.
\end{equation}
This strengthens the characterization \eqref{eq:mul/moment polytope} of the moment polytope in terms of the non-vanishing of highest weight multiplicities (cf.\ \eqref{eq:onebody/semiclassical limit}).
In fact, Heckman studied the measures $\Prob_{K,\PP(\calH)}$ precisely to understand the asymptotic behavior of the multiplicities \cite{Heckman82}.
In the same spirit, Harish-Chandra's formula \eqref{eq:dhmeasure/harish-chandra} can be seen as the ``semiclassical limit'' of the Weyl character formula \eqref{eq:mul/weyl character formula}.

\bigskip

In the case of the one-body quantum marginal problem for three subsystems, the highest-weight multiplicities are given by the Kronecker coefficients $g_{\alpha,\beta,\gamma}$ (\autoref{sec:onebody/kronecker}).
Thus \eqref{eq:mul/semiclassical non-abelian} asserts that their asymptotics determines that the local eigenvalue distribution of a random pure states on $\calH = \CC^a \otimes \CC^b \otimes \CC^c$:
\begin{equation*}
  \Pspec =
  (abc-1)! \,
  p_a p_b p_c
  \wlim_{k \rightarrow \infty}
  \frac 1 {k^d}
  \sum_{\alpha, \beta, \gamma} g_{\alpha,\beta,\gamma} \, \delta_{(\alpha,\beta,\gamma)/k}
\end{equation*}
where $p_a,p_b,p_c$ are the Vandermonde determinants \eqref{eq:dhmeasure/kks volume unitary},
$d = abc-1-\binom a 2-\binom b 2-\binom c 3$,
and where the sum runs over all Young diagrams $\alpha, \beta, \gamma$ with $k$ boxes each and no more than $a$, $b$ and $c$ rows, respectively (if $\Delta_K \cap i \mathfrak t^*_{>0} \neq \emptyset$).

The link between representation theory and geometry is quite remarkable.
Not only can one read off the existence of a pure state with given local eigenvalues from the non-vanishing of the corresponding Kronecker coefficients $g_{k\lambda,k\mu,k\nu}$, but their asymptotic growth also encodes the probability of finding these eigenvalue spectra when the global state is chosen according to the invariant probability measure.
See \autoref{fig:dhmeasure/hedgehog} for an illustration of the convergence of the corresponding multiplicity measure.

\begin{figure}
  \centering
  \includegraphics[width=\linewidth]{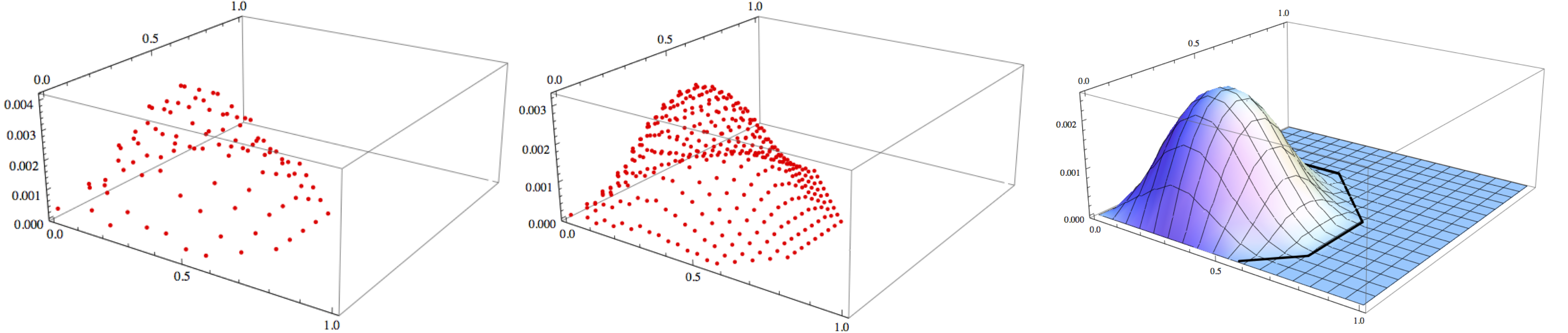}
  \begin{tabularx}{\linewidth}{*3{>{\centering\arraybackslash}X}}
  {\footnotesize (a)} &
  {\footnotesize (b)} &
  {\footnotesize (c)} \\
  \end{tabularx}
  \caption[Illustration of the semiclassical limit of the Kronecker coefficients]{
  \emph{Illustration of the semiclassical limit of the Kronecker coefficients.}
  (a) and (b) Suitably rescaled Kronecker coefficients for the mixed-state quantum marginal problem of two qubits with global spectrum $\vec\lambda_{AB} = (4/7, 2/7, 1/7, 0)$ and $k = 28$ and $56$ boxes (computed with the algorithm described in \autoref{sec:mul/kronecker}).
  (c) The limiting distribution \eqref{eq:mul/multiplicity measure}, whose support is the moment polytope cut out by Bravyi's inequalities \eqref{eq:kirwan/bravyi}.}
  \label{fig:dhmeasure/hedgehog}
\end{figure}

\bigskip

We conclude with an amusing application of \eqref{eq:mul/semiclassical non-abelian}.
Consider $\calH = \CC^a \otimes \CC^b$ with $a \leq b$.
Recall that Schur--Weyl duality implies that we have the decomposition \eqref{eq:onebody/schur-weyl bipartite},
\begin{equation*}
  \Sym^k(\CC^a \otimes \CC^b) = \bigoplus_{\alpha \vdash_a k} V^a_\alpha \otimes V^b_\alpha.
\end{equation*}
Restricting to the action of $K = \SU(a)$ on the first tensor factor, we find that
\begin{equation*}
  \mu_{K,\PP(\calH)} = \wlim_{k \rightarrow \infty} \frac 1 {k^{ab-1 - \binom a 2}} \sum_{\alpha \vdash_a k} \dim V^b_\alpha \, \delta_{\alpha/k}.
\end{equation*}
For large $k$, it follows readily from Weyl's dimension formula \eqref{eq:mul/weyl dimension formula} that
\begin{equation*}
  \frac {\dim V^b_\alpha} {k^{a(b-a) + \binom a 2}} \sim \frac 1 Z \left( \prod_{i=1}^a \left( \alpha_i / k \right)^{b-a} \right) p_a\left(\alpha / k \right)
\end{equation*}
for some constant $Z > 0$.
Together with \eqref{eq:mul/semiclassical non-abelian}, we recover \eqref{eq:dhmeasure/lloyd-pagels single marginal}, the formula for the density of the eigenvalue distribution of $\rho_A$ for a random bipartite pure state.

\subsection*{Order of Growth and Geometric Complexity Theory}

We now consider general projective subvarieties $\XX \subseteq \PP(\calH)$.
In the following, we study the exponent $d_{K,\XX}$ in the definition of the measure \eqref{eq:mul/multiplicity measure} more carefully and illustrate its implications in the context of geometric complexity theory.
We will write $f \sim g$ for the asymptotic equivalence $\lim_{k \rightarrow \infty} f(k)/g(k) = 1$.

\begin{lem}
  \label{lem:mul/exponent}
  The exponent $d_{K,\XX}$ is equal to $\dim \XX - R_\XX$, where
  \[R_\XX := \#\{ \alpha \in R_{K,+} : (\lambda, H_\alpha) > 0 \text{ for some } \lambda \in \Delta_K(\XX) \}\]
  is the number of positive roots that are not orthogonal to all points on the moment polytope (we write $H_\alpha \not\perp \Delta_K(\XX)$).
\end{lem}
\begin{proof}
  By the Hilbert--Serre theorem, the function $k \mapsto \dim R_k(\XX)$ is a polynomial of degree $\dim \XX$ for large $k$ \cite[Theorem I.7.5]{Hartshorne77}. Hence there exists a constant $A > 0$ such that
  \begin{align*}
    &A \, k^{\dim \XX}
    \sim \dim R_k(\XX) \\
    =~ &\sum_{\mathclap{\lambda \in \Lambda^*_{K,+}}} m_{K,R_k(\XX)}(\lambda) \, \dim V_{K,\lambda} \\
    =~ &\sum_{\mathclap{\lambda \in \Delta_K(\XX) \cap \frac 1 k \Lambda^*_{K,+}}} m_{K,R_k(\XX)}(k \lambda) \, \dim V_{K,k \lambda},
  \end{align*}
  where for the last equality we have used \eqref{eq:mul/moment polytope}.
  By the Weyl dimension formula \eqref{eq:mul/weyl dimension formula},
  \begin{equation*}
    \dim V_{K,k\lambda} =
    \prod_{\mathclap{\alpha \in R_{K,+}}} \frac {(k \lambda + \rho_K, H_\alpha)} {(\rho_K, H_\alpha)}
    = \underbrace{\left( ~~~~~~
        \prod_{\mathclap{\substack{\alpha \in R_{K,+} \\ H_\alpha \not\perp \Delta_K(\XX)}}} ~ \frac {(\lambda, H_\alpha)} {(\rho_K, H_\alpha)}
    \right)}_{=: p(\lambda)} k^{R_\XX} +
    O(k^{R_\XX-1})
  \end{equation*}
  for the representations that occur in $R(\XX)$.
  The coefficient $p(\lambda)$ of $k^{R_\XX}$ is a polynomial function in $\lambda$.
  We may therefore find a constant $C > 0$ such that $\dim V_{K,k\lambda} \leq C k^{R_\XX}$ for all $\lambda \in \Delta_K(\XX) \cap \frac 1 k \Lambda^*_{K,+}$.
  It follows that
  \begin{align*}
    &\sum_{\mathclap{\lambda \in \Delta_K(\XX) \cap \frac 1 k \Lambda^*_{K,+}}} m_{K,R_k(\XX)}(k \lambda) \, \dim V_{K,k \lambda} \\
    \leq\; &C k^{R_\XX} \sum_{\mathclap{\lambda \in \Delta_K(\XX) \cap \frac 1 k \Lambda^*_{K,+}}} m_{K,R_k(\XX)}(k \lambda) \\
    \sim\; &C k^{R_\XX + d_{K,\XX}} \int_{\Delta_K(\XX)} d\nu_{K,\XX}
  \end{align*}
  since the measure $\nu_{K,\XX}$ has non-zero volume.
  Thus we find that $\dim \XX \leq R_\XX + d_{K,\XX}$.

  On the other hand, since $\nu_{K,\XX}$ is Lebesgue-absolutely continuous with respect to Lebesgue measure on the moment polytope, we may find a compact set $\Delta$ contained in its relative interior which has positive measure with respect to $\nu_{K,\XX}$.
  We claim that $p(\lambda)$ is strictly positive for all $\lambda$ contained in the relative interior of the moment polytope.
  To see this, observe that for all positive roots $\alpha$ with $H_\alpha \not\perp \Delta_K(\XX)$ there exists by definition some $\mu \in \Delta_K(\XX)$ such that $(H_\alpha, \mu) > 0$;
  since we can always write $\lambda$ as a proper convex combination of $\mu$ and some other point $\mu' \in \Delta_K(\XX)$, it follows that $(\lambda, H_\alpha) > 0$.
  This implies that on the compact set $\Delta$ we can bound $p(\lambda)$ from below by a positive constant.
  Thus there exists a constant $D > 0$ (depending on $\Delta$) such that
  $\dim V_{K,k \lambda} \geq D \, k^{R_\XX}$ for all $\lambda \in \Delta \cap \frac 1 k \Lambda^*_{K,+}$.
  Consequently,
  \begin{align*}
    &\sum_{\mathclap{\lambda \in \Delta_K(\XX) \cap \frac 1 k \Lambda^*_{K,+}}} m_{K,R_k(\XX)}(k \lambda) \, \dim V_{K, k \lambda} \\
    \geq\;&D \, k^{R_\XX} \sum_{\mathclap{\lambda \in \Delta \cap \frac 1 k \Lambda^*_{K,+}}} m_{K,R_k(\XX)}(k \lambda) \\
    \sim\;&D \, k^{R_\XX + d_{K,\XX}} \int_\Delta d\nu_{K,\XX}.
  \end{align*}
  Since $\Delta$ has positive measure, we may also conclude that $\dim \XX \geq R_\XX + d_{K,\XX}$.
  Thus we have equality.
\end{proof}

Let us now elaborate on the applicability of the multiplicity measures \eqref{eq:mul/multiplicity measure} for finding new obstructions in geometric complexity theory.
For this, we consider a pair of projective subvarieties $\XX$ and $\YY$ with $\dim \XX < \dim \YY$, as is the case for the orbit closures of relevance to geometric complexity theory \cite{BuergisserLandsbergManivelEtAl11, BuergisserIkenmeyer11}.

\begin{lem}
  \label{lem:mul/polytope smaller}
  Let $\Delta_K(\XX) \subseteq \Delta_K(\YY)$ and $R_\XX < R_\YY$. Then, $\dim \Delta_K(\XX) < \dim \Delta_K(\YY)$.
\end{lem}
\begin{proof}
  Note that we have
  \begin{equation*}
    \dim \Delta_K(\XX) = \dim \aff \Delta_K(\XX) \leq \dim \aff \Delta_K(\YY) = \dim \Delta_K(\YY),
  \end{equation*}
  with equality if and only if the two affine hulls are equal.
  By assumption there exists a positive root $\alpha \in R_{K,+}$ such that $H_\alpha$ is orthogonal to all points in $\Delta_K(\XX)$ (i.e., $(\lambda, H_\alpha) = 0$ for all $\lambda \in \Delta_K(\XX)$), but not to all points in $\Delta_K(\YY)$.
  Then $H_\alpha$ is also orthogonal to all points in the affine hull of $\Delta_K(\XX)$, but not to all points in the affine hull of $\Delta_K(\YY)$.
  Since $\Delta_K(\XX) \subseteq \Delta_K(\YY)$ by assumption, it follows that $\aff \Delta_K(\XX) \subsetneq \aff \Delta_K(\YY)$, and hence that $\dim \Delta_K(\XX) < \dim \Delta_K(\YY)$.
\end{proof}

\begin{lem}
  \label{lem:mul/different lebesgue}
  Let $\dim \Delta_K(\XX) < \dim \Delta_K(\YY)$.
  Then, $\XX \subseteq \YY$ implies that $d_{K,\XX} < d_{K,\YY}$.
\end{lem}
\begin{proof}
  If $\XX \subseteq \YY$ then it is immediate from \eqref{eq:mul/multiplicity criterion} and the definition of the measures \eqref{eq:mul/multiplicity measure} that $d_{K,\XX} \leq d_{K,\YY}$.
  Let us suppose for a moment that in fact $d_{K,\XX} = d_{K,\YY}$. Then \eqref{eq:mul/multiplicity criterion} would imply that
  \begin{equation*}
    \int d\nu_{K,\XX} \, g \leq \int d\nu_{K,\YY} \, g
  \end{equation*}
  for any non-negative test function $g$.
  In particular, this inequality would hold for the indicator function of $\Delta_K(\XX)$, so that
  \begin{equation*}
    0 < \int_{\Delta_K(\XX)} d\nu_{K,\XX} \leq \int_{\Delta_K(\XX)} d\nu_{K,\YY}.
  \end{equation*}
  But this is clearly absurd---since $\nu_{K,\YY}$ is absolutely continuous with respect to Lebesgue measure on its moment polytope $\Delta_K(\YY)$, the assumption implies that $\Delta_K(\XX)$ is a set of measure zero for $\nu_{K,\YY}$; thus the right-hand side integral is zero.
\end{proof}

\begin{cor}
  \label{cor:mul/dimension lemma}
  Let $\dim \XX < \dim \YY$. Then, $\XX \subseteq \YY$ implies $d_{K,\XX} < d_{K,\YY}$.
\end{cor}
\begin{proof}
  Clearly, $\XX \subseteq \YY$ implies that $\Delta_K(\XX) \subseteq \Delta_K(\YY)$ and $R_\XX \leq R_\YY$.
  If $R_\XX = R_\YY$ then the assertion follows directly from \autoref{lem:mul/exponent}, since then
  \begin{equation*}
    d_{K,\XX} = \dim \XX - R_\XX < \dim \YY - R_\YY = d_{K,\YY}.
  \end{equation*}
  Otherwise, if $R_\XX < R_\YY$ then \autoref{lem:mul/polytope smaller} and \autoref{lem:mul/different lebesgue} together imply that $d_{K,\XX} < d_{K,\YY}$.
\end{proof}

As we had explained in the introduction, the significance of \autoref{cor:mul/dimension lemma} is that the multiplicity measures do not directly lead to a new criterion for obstructions that goes beyond what is provided by the moment polytope.
From a conceptual point of view, the exponents $d_{K,\XX}$ and $d_{K,\YY}$ each capture the order of growth, while the densities of the measures $\nu_{K,\XX}$ and $\nu_{K,\YY}$ encode (smoothed versions) of the corresponding leading-order coefficients---but the latter are incomparable if the orders of growth differ!

\section{Discussion}
\label{sec:mul/discussion}

In a recent preprint, our algorithm for computing Kronecker coefficients has been analyzed in some detail and it has also been shown that it is possible to decide positivity in linear time for Young diagrams of bounded height \cite{PakPanova14}.
The problem of efficiently deciding the positivity of Kronecker coefficients for general Young diagrams, though, is still wide open, partly because we do not know of an effective combinatorial description akin to the honeycomb model for the Littlewood--Richardson coefficients \cite{KnutsonTao99}.

\bigskip

In the past chapters, we have studied the one-body reduced density matrix in quantum mechanics from a variety of perspectives.
In \autoref{ch:kirwan}, we gave a new solution to the one-body quantum marginal in terms of ``Ressayre-type inequalities''.
In \autoref{ch:slocc}, we studied multipartite entanglement from the perspective of the one-body marginals.
In \autoref{ch:dhmeasure}, we showed how the joint distribution of the local eigenvalues can be computed by reducing to the distribution of diagonal entries; the discrete analogue of this reduction leads to an efficient algorithm for the subgroup restriction problem as we have seen in this \autoref{ch:multiplicities}.
Common to all our results is the remarkable role played by the maximal unipotent subgroups and the interplay between highest weights and weights, which in each case allowed us to reduce from a non-Abelian quantum-mechanical problem to an Abelian one, and thus in essence to the classical combinatorics of weights.

\chapter{The Search for Further Entropy Inequalities}
\label{ch:entropy}

In the second part of this thesis, we go beyond the study of one-body reduced density matrices and consider general quantum marginals.
Motivated by the fundamental role of entropy in physics and information theory, we start by introducing in this chapter the problem of determining the linear inequalities that constrain the von Neumann entropy of the marginals of a multipartite quantum state \cite{Pippenger03}.
The strong subadditivity of the von Neumann entropy is perhaps the most important such inequality \cite{LiebRuskai73}, and an indispensable tool in quantum statistical physics and quantum information theory \cite{OhyaPetz93}.
The discovery of any further entropy inequality would be considered a major breakthrough, and it would shed further light on the general quantum marginal problem.

The following introduction is partially adapted from \cite{GrossWalter13}.
We refer to \cite{CoverThomas06, Yeung02} and \cite{NielsenChuang04} for comprehensive introductions to classical and quantum entropy.

\subsection*{The Classical Entropy Cone}

Let us first consider the classical situation.

\begin{dfn}
  The \emph{Shannon entropy}\index{entropy!Shannon}\index{Shannon entropy}\nomenclature[QH(X)]{$H(X)$}{Shannon entropy of random variable $X$} of a random variable $X$ with finitely many outcomes is given by
  \[H(X) = H(p) = -\sum_x p_x \log p_x \geq 0,\]
  where $p = (p_x)$ is the probability distribution of $X$ (i.e., $p_x$ is the probability of an outcome $x$).
\end{dfn}

Given a collection of random variables $X_1, \dots, X_n$ defined on a common probability space, we can then consider the Shannon entropy $H(X_I)$ of any non-empty subset $X_I = (X_i)_{i \in I}$ of the variables.\nomenclature[QX_I]{$X_I$}{subset of random variables with indices in $I$}
These entropies are not independent: For example, \emph{monotonicity}\index{monotonicity!Shannon entropy} asserts that the Shannon entropy can never decrease if more
random variables are taken into account:
\begin{equation}
\label{eq:entropy/monotonicity}
  H(X_{I\cup J}) - H(X_I) \geq 0
\end{equation}
A second example is the \emph{strong subadditivity of the Shannon entropy}\index{strong subadditivity!Shannon entropy}:
\begin{equation}
\label{eq:entropy/ssa shannon}
  H(X_{I}) + H(X_J) - H(X_{I\cap J}) - H(X_{I \cup J}) \geq 0
\end{equation}
Equivalently, the conditional entropy and the conditional mutual information are always non-negative.
To study the entropies of subsystems systematically, we define the \emph{classical entropy region}\index{entropy region!classical}\nomenclature[QC_n]{$\calC_n$}{classical entropy region}
\begin{align*}
  \calC_n := \{ (H(X_I))_{\emptyset \neq I \subseteq \{1, \dots, n\}} :\; &X_1, \dots, X_n \text{ random variables} \\
  &\text{with finitely many outcomes} \} \subseteq \RR^{2^n-1}.
\end{align*}
In general, the set $\calC_n$ has a complicated geometrical structure \cite{ZhangYeung97}.
However, its closure $\overline{\calC_n}$ is a convex cone \cite{ZhangYeung97}, which we call the \emph{classical entropy cone}\index{entropy cone!classical}\nomenclature[QC_n2]{$\overline{\calC_n}$}{classical entropy cone}.
Like any closed convex cone, $\overline{\calC_n}$ can be described by linear inequalities.
To see this, consider the dual cone\nomenclature[QC_n3]{$\calC_n^*$}{dual cone of linear entropy inequalities}
\begin{equation*}
  \overline{\mathcal C_n}^* = \calC_n^* = \{ (\nu_I) \in \RR^{2^n-1} : \sum_I \nu_I H(X_I) \geq 0  \quad \forall X_1, \dots, X_n \}.
\end{equation*}
Each element $(\nu_I) \in \calC_n^*$ can be identified with an \emphindex{entropy inequality} $\sum_I \nu_I H(X_I) \geq 0$ that is satisfied by the Shannon entropy.
The bipolar theorem now asserts that the bidual cone $\calC_n^{**}$ is equal to $\overline{\calC_n}$ (e.g., \cite{Rockafellar72}).
Thus the set of entropy inequalities (or even just its extreme rays) determines the classical entropy region up to closure \cite{Pippenger86}.
Mat\'{u}\u{s} has shown that the classical entropy region contains the relative interior of the classical entropy cone, so that
$\relint \overline{\calC_n} \subseteq \calC_n \subseteq \overline{\calC_n}$.

Monotonicity \eqref{eq:entropy/monotonicity} and strong subadditivity \eqref{eq:entropy/ssa shannon} together span the polyhedral cone of \emph{entropy inequalities of Shannon-type}\index{entropy inequality!Shannon-type}.
For any given candidate inequality it can be automatically checked by using linear programming if it is an entropy inequality of Shannon-type \cite{YeungYan, Yeung97}. %
For a long time, these were the only known inequalities. Equivalently, it was not known if monotonicity and strong subadditivity were the only constraints on a non-negative vector $(h_I)$ to be approximable by the Shannon entropies of random variables.

In the seminal work \cite{ZhangYeung98}, Zhang and Yeung have shown that for $n\geq 4$ random variables there exist entropy inequalities that are \emph{not} of Shannon-type. In particular, they have proved that
\begin{equation}
  \label{eq:entropy/zhang-yeung}
  2 I(X_3 : X_4) \leq I(X_1 : X_2) + I(X_1 : X_{34}) + 3 I(X_3 : X_4 | X_1) + I(X_3 : X_4 | X_2)
\end{equation}
is a linear entropy inequality that is not of Shannon-type.
Here, $I(X_I : X_J) = S(X_I) + S(X_J) - S(X_{I \cup J})$ and $I(X_I : X_J | X_K) := H(X_{I \cup K}) + H(X_{J \cup K}) - H(X_{I \cup J \cup K}) - H(X_K)$ are the \emph{Shannon (conditional) mutual information}\index{mutual information!Shannon}\index{mutual information!Shannon conditional}.
In fact, there are infinitely many independent inequalities that are not of Shannon-type. Thus the classical entropy cone $\calC_n$ is not polyhedral for $n \geq 4$ \cite{Matus07,DoughertyFreilingZeger11}.
Its general structure is still only poorly understood, and any improved understanding should have direct applications to network information theory \cite{Yeung02, SunJafar13}.

\subsection*{The Quantum Entropy Cone}

We now consider the situation in quantum mechanics. Here, the natural analogue of the Shannon entropy is the von Neumann entropy:

\begin{dfn}
  The \emph{von Neumann entropy}\index{entropy!von Neumann}\nomenclature[QS(rho)]{$S(\rho)$}{von Neumann entropy of quantum state $\rho$} of a density matrix $\rho$ on a finite-dimensional Hilbert space is given by
  \[S(\rho) = -\tr \rho \log \rho \geq 0.\]
\end{dfn}
\noindent By the spectral theorem, the von Neumann entropy $S(\rho)$ is equal to $H(\vec\lambda)$, the Shannon entropy of the spectrum $\vec\lambda$ of $\rho$, which is a probability distribution. It is a concave function of $\rho$.

Now let $\rho$ be a density matrix describing the state of a quantum system of $n$ distinguishable particles with tensor-product Hilbert space $\calH = \bigotimes_{i=1}^n \calH_i$. The state of any subset $I \subseteq \{1,\dots,n\}$ of the particles is described by the reduced density matrix $\rho_I = \tr_{I^c} \rho$ formed by tracing out the Hilbert space of the other particles.
We define the \emph{quantum entropy region}\index{entropy region!quantum}\nomenclature[QQ_n]{$\calQ_n$}{quantum entropy region} to be \cite{Pippenger03}
\begin{align*}
  \calQ_n := \{ (S(\rho_I))_{\emptyset \neq I \subseteq \{1,\dots,n\}} \;:\; &\rho \text{ density matrix on an $n$-fold tensor product} \\
  &\text{of finite-dimensional Hilbert spaces} \} \subseteq \RR^{2^n-1}.
\end{align*}
Clearly, $\calQ_n \supseteq \calC_n$, since any joint probability distribution can be considered as a multipartite density matrix.
In general, $\calQ_n$ has a complex structure and in particular is neither closed nor convex \cite{LindenWinter05}.
But its closure is again a convex cone, called the \emph{quantum entropy cone}\index{entropy cone!quantum}\nomenclature[QQ_n2]{$\overline{\calQ_n}$}{quantum entropy cone}.
We recall a proof of this important fact:

\begin{lem}[\cite{Pippenger03}]
  The quantum entropy cone $\overline{\calQ_n}$ is indeed a convex cone.
\end{lem}
\begin{proof}
  (1) We first show that $\calQ_n + \calQ_n \subseteq \calQ_n$:
  Let $\rho$, $\rho'$ be quantum states on $\bigotimes_{i=1}^n \calH_i$ and $\bigotimes_{i=1}^n \calH'_i$, respectively. Then
  $\rho'' := \rho \otimes \rho'$ is a quantum state on $\bigotimes_{i=1}^n \calH_i \otimes \calH'_i$, and $\rho''_I = \rho_I \otimes \rho'_I$ for each subset $I \subseteq \{1,\dots,n\}$.
  By additivity of the von Neumann entropy, $S(\rho''_I) = S(\rho_I) + S(\rho'_I)$, which shows the claim.

  (2) We now show that $\lambda \calQ_n \subseteq \overline{\calQ_n}$ for all $\lambda \geq 0$:
  For this, let $\rho$ be a quantum state on $\bigotimes_{i=1}^n \calH_i$.
  Let $p \in [0,1]$ and $k \in \NN_{>0}$ and consider the quantum state $\rho' := (1-p) \proj 0^{\otimes n} + p \rho^{\otimes k}$ on $\bigotimes_{i=1}^n \left( \CC \ket 0 \oplus \calH_i^{\otimes k} \right)$. Then, $\rho'_I = (1-p) \proj 0^{\otimes I} + p \rho_I^{\otimes k}$, so that
  \begin{equation*}
    S(\rho'_I) = h(p) + p k S(\rho_I),
  \end{equation*}
  where $h(p) = -p \log p - (1-p) \log (1-p)$ is the binary entropy function.
  Let $\varepsilon > 0$.
  Since $h(p)$ is continuous and $h(0) = 0$, we may choose $k$ large enough such that $h(p) \leq \varepsilon$ for $p = \lambda / k$, and hence
  \begin{equation*}
    \abs{S(\rho'_I) - \lambda S(\rho_I)} \leq \varepsilon.
  \end{equation*}
  Since $\varepsilon > 0$ was arbitrary, we may conclude that $(\lambda S(\rho_I)) \in \overline{\calQ_n}$.
  This establishes the second claim.

  By taking limits, points (1) and (2) together imply that $\overline{\calQ_n}$ is a convex cone.
\end{proof}

The most immediate difference to the classical case is that the von Neumann entropy is no longer monotonic:
Global quantum states can exhibit less entropy than their reductions (a signature of entanglement), which also shows that $\calQ_n \supsetneq \calC_n$.
Instead, the von Neumann entropy satisfies \emph{weak monotonicity}\index{weak monotonocity!von Neumann entropy}:
\begin{equation}
\label{eq:entropy/weak monotonicity}
  S(\rho_I) + S(\rho_J) - S(\rho_{I\setminus J}) - S(\rho_{J \setminus I}) \geq 0
\end{equation}
\emph{Strong subadditivity}\index{strong subadditivity!von Neumann entropy}, however, famously remains valid for quantum entropies \cite{LiebRuskai73}:
\begin{equation}
\label{eq:entropy/ssa von neumann}
  S(\rho_I) + S(\rho_J) - S(\rho_{I\cap J}) - S(\rho_{I \cup J}) \geq 0
\end{equation}
In fact, \eqref{eq:entropy/weak monotonicity} and \eqref{eq:entropy/ssa von neumann} are easily shown to be equivalent by the process of purification \cite{Lieb75}. Since purification is a non-linear construction, this does not imply equivalence on the level of the entropy regions; but see the discussion in \cite{LindenMatusRuskaiEtAl13}.

Just as in the classical case, it is of fundamental interest to determine the linear inequalities satisfied by the von Neumann entropies $S(\rho_I)$ of subsystems.
In fact, it is a major open problem to decide if there are any further entropy inequalities besides weak monotonicity and strong subadditivity \cite{Pippenger03}.
Let us describe some partial progress from the literature:
For vectors $(s_I) \in \RR^{2^n-1}$ that are permutation-invariant under relabeling of the subsystems, it has been shown that $(s_I) \in \overline{\calQ_n}$ if and only if weak monotonicity and strong subadditivity are satisfied \cite{Pippenger03}.
In \cite{LindenWinter05,CadneyLindenWinter12} a class of so-called \emph{constrained inequalities} has been established, which in particular showed that $\calQ_n \subsetneq \overline{\calQ_n}$.
We also refer to \cite{Ibinson07} for numerical evidence that the Zhang--Yeung inequality \eqref{eq:entropy/zhang-yeung} might also hold for the von Neumann entropy.

\subsection*{Balanced Entropy Inequalities}

Instead of directly determining the quantum entropy cone, it is natural to approach the problem by asking which of the classical entropy inequalities might continue to hold for the von Neumann entropy.
Since the latter is no longer monotonic, we need to identify those classical entropy inequalities which do not ``involve'' monotonicity.
An interesting class of entropy inequalities introduced by Chan does precisely that:

\begin{dfn}
  An entropy inequality $\sum_I \nu_I H(X_I) \geq 0$ or $\sum_I \nu_I S(\rho_I) \geq 0$ is called \emph{balanced}\index{entropy inequality!balanced} \cite{Chan03}
  (also \emph{correlative} \cite{Han75})
  if
  \begin{equation*}
    \sum_{I \ni i} \nu_I = 0  \quad (\forall i=1,\dots,n).
  \end{equation*}
\end{dfn}
For example, strong subadditivity is balanced, while monotonicity is not.
Chan has shown that the classical entropy cone is determined by the set of balanced entropy inequalities together with monotonicity.
More precisely, he has proved that any Shannon entropy inequality can be decomposed uniquely into the form
\begin{equation}
\label{eq:entropy/balanced decomposition}
  \sum_I \nu_I H(X_I) + \sum_{i=1}^n r_i H(X_i|X_{i^c}) \geq 0,
\end{equation}
where the left-hand side is a \emph{balanced} entropy inequality and the right-hand side a conic combination of conditional entropies, i.e.\ all $r_i \geq 0$ \cite{Chan03}.
In other words, the dual cone $\mathcal C^*_n$ of Shannon entropy inequalities is a direct sum of the cone of balanced entropy inequalities and the cone spanned by monotonicity \eqref{eq:entropy/monotonicity}.
It is thus natural to consider the following problem.

\begin{pro} %
\label{pro:entropy/balanced}
  Which balanced Shannon entropy inequalities also hold true for the von Neumann entropy?
\end{pro}

We remark that the Zhang--Yeung inequality \eqref{eq:entropy/zhang-yeung} as well as the infinite families in \cite{Matus07,DoughertyFreilingZeger11} are balanced, since they are linear combinations of conditional mutual informations.
On the other hand, we note that there are unbalanced entropy inequalities that hold for the von Neumann entropy, e.g.\ weak monotonicity.
It can be argued that there is no direct quantum counterpart of the decomposition \eqref{eq:entropy/balanced decomposition} \cite{Majenz14}, which perhaps makes it unlikely that the problem of determining all linear entropy inequalities satisfied by the von Neumann entropy can be reduced to \autoref{pro:entropy/balanced}.

\subsection*{Discussion}

We conclude this introduction by mentioning some related avenues of investigation.
Entropic constraints are of interest not only for the Shannon and von Neumann entropy, but also for other kinds of entropies, e.g.\ differential entropies \cite{Chan03} and R\'{e}nyi entropies \cite{LindenMosonyiWinter13,CadneyHuberLindenEtAl13}.
Another interesting direction is to relax the notion of subsystems beyond the tensor-product case.
Instead of only considering partial traces over the factors of a tensor-product Hilbert space, we may consider marginals with respect to more general subalgebras of observables \cite{OhyaPetz93}. This leads directly to the study of entropic uncertainty relations \cite{BertaChristandlColbeckEtAl10, MaassenUffink88}, entropy power inequalities \cite{KoenigSmith14}, and entropy in space-time and might thus serve as a unifying framework for studying entropy in general quantum systems.
\chapter{Entropy Inequalities from Phase Space}
\label{ch:stabs}

In this chapter we study the entropy inequalities satisfied by two classes of quantum states---namely, \emph{stabilizer states} and \emph{Gaussian states} (the latter can be seen as continuous-variable counterparts of the former).
Both classes of states can exhibit intrinsically quantum features, such as multi-particle entanglement, but they possess enough structure to allow for a concise and computationally efficient description and so have proven to be extremely useful in quantum information theory and beyond. For example, the stabilizer formalism is a basic tool for constructing quantum error-correcting codes, while Gaussian states and transformations are routinely used in quantum optics (see, e.g., \cite{NielsenChuang04, WeedbrockPirandolaGarcia-PatronEtAl12}).
\emph{Quantum phase-space methods} have been built around both classes of states, and it is this point of view we will exploit here.

The results in this chapter have been obtained in collaboration with David Gross, and they have appeared in \cite{GrossWalter13}.

\section{Summary of Results}

The starting point for our work is the Wigner function, which for Gaussian states as well as for stabilizer states in odd dimensions $d$ is a bona fide probability distribution on the classical phase space (the case of even $d$ requires some more care, see \autoref{thm:stabs/main theorem} below).
For $n$ bosonic modes, this phase space is given by $\RR^{2n}$, while for $n$ systems of local dimension $d$, it is the finite group $\ZZ_d^{2n}$.
In both cases, it is the direct sum of the single-particle or single-mode phase spaces.
In Sections~\ref{sec:stabs/discrete phase space} and \ref{sec:stabs/stabilizer states} we give a self-contained account of stabilizer states in the discrete phase-space formalism.
We may thus define random variables $X_1,\dots,X_n$ on the phase space, jointly distributed according to the Wigner function of the given quantum state $\rho$. Here, $X_i$ denotes the component in the phase space of the $i$-th particle or mode.
The random variables $X_1, \dots, X_n$ constitute our \emph{classical model}.
This construction is compatible with reduction:
The marginal probability distribution of a subset of variables $X_I = (X_i)_{i \in I}$ is given precisely by the Wigner function of the reduced density matrix $\rho_I$.\nomenclature[Qrho_I]{$\rho_I$}{reduced density matrix of subsystems $I$}

Our crucial observation then is that certain quantum entropies are simple functions of corresponding classical entropies.
More precisely, in the case of stabilizer states (\autoref{sec:stabs/classical model}), we find that
\begin{equation}
  \label{eq:stabs/stab entropies}
  S(\rho_I) = H(X_I) - \abs I \log(d).
\end{equation}
Therefore, if $\sum_I \nu_I H(X_I) \geq 0$ is a \emph{balanced} entropy inequality satisfied by the Shannon entropies of the random variables $X_I$ then the same inequality is also satisfied by the von Neumann entropies of the quantum states $\rho_I$:
\begin{equation*}
  \sum_I \nu_I S(\rho_I) = \sum_I \nu_I H(X_I) - \underbrace{\sum_I \nu_I \abs I}_{=0} \log(d) \geq 0
\end{equation*}
In particular, \emph{the von Neumann entropy of stabilizer states respects all balanced Shannon entropy inequalities}, such as the inequalities of non-Shannon type found in \cite{ZhangYeung98,Matus07,DoughertyFreilingZeger11}.
This completely solves \autoref{pro:entropy/balanced} for the class of stabilizer states.

Our construction can also be understood in the group-theoretical framework of \cite{ChanYeung02}.
Here it is well-known that there are inequalities for the Shannon entropy which do not hold for arbitrary random variables, but only for random variables constructed from certain classes of subgroups \cite{LiChong07}.
By analyzing the construction above, we show that the von Neumann entropy for stabilizer states similarly respects a further entropy inequality which does \emph{not} hold for arbitrary random variables (and hence quantum states)---namely the \emphindex{Ingleton inequality} \cite{LiChong07}, which is the balanced inequality
\begin{equation}
  \label{eq:stabs/ingleton}
  I_\rho(I:J|K) + I_\rho(I:J|L) + I_\rho(K:L) - I_\rho(I:J) \geq 0.
\end{equation}
Here, $I_\rho(I:J) = S(\rho_I) + S(\rho_J) - S(\rho_{I \cup J})$ and $I_\rho(I:J|K) = S(\rho_{I \cup K}) + S(\rho_{J \cup K}) - S(\rho_K) - S(\rho_{I \cup J \cup K})$ are the \emph{quantum (conditional) mutual information}\index{mutual information!quantum}\index{mutual information!quantum conditional}.

We find it instructive to understand how the above classical model manages to respect monotonicity \eqref{eq:entropy/monotonicity}, while the quantum state may violate it.
For example, since stabilizer states can be entangled (even maximally so, see \autoref{exl:stabs/maximally entangled} below), $S(\rho_1) = S(\rho_2) = 1$ and $S(\rho_{12})=0$ are perfectly valid entropies of a stabilizer state which certainly violate monotonicity.
Equation~\eqref{eq:stabs/stab entropies} states that the classical model is \emph{more highly mixed} than the quantum one, in the sense that the entropy associated with a subset $I$ is higher by an amount of $\abs I$ $d$its.
That is precisely the maximal amount by which quantum mechanics can violate monotonicity.

\bigskip

In the case of Gaussian states (\autoref{sec:stabs/gaussian states}), the random variables $X_1,\dots,X_n$ have a multivariate normal distribution on $\RR^{2n}$, and we show that
\begin{equation*}
  S_2(\rho_I) = h_\alpha(X_I) - \abs I \left( \log \pi -  \frac {\log \alpha} {1-\alpha} \right),
\end{equation*}
where $S_2(\rho) = - \log \tr \rho^2$ is the \emph{quantum R\'{e}nyi-2 entropy}\index{R\'{e}nyi entropy!quantum}\nomenclature[QS_2(rho)]{$S_2(\rho)$}{R\'{e}nyi-2 entropy of quantum state $\rho$}; $h_\alpha$ is the \emph{differential R\'{e}nyi-$\alpha$ entropy}\index{R\'{e}nyi entropy!differential}\nomenclature[Qh_alpha(X)]{$h_\alpha(X)$}{differential R\'{e}nyi-$\alpha$ entropy of random variable $X$}, defined by $h_\alpha(X) = (1-\alpha)^{-1} \log \int p(x)^\alpha dx$ for any positive $\alpha \neq 1$, with $p(x)$ the probability density of the random variable $X$ with respect to Lebesgue measure.
In the limiting case $\alpha \rightarrow 1$, we recover a formula involving the \emph{differential Shannon entropy}\index{Shannon entropy!differential} $h(X) = -\int p(x) \log p(x)$,
which has previously appeared in \cite{AdessoGirolamiSerafini12}, attributed to Stratonovich:
\begin{equation*}
  S_2(\rho_I) = h(X_I) - \abs I \log ( \pi e )
\end{equation*}
Thus, R\'{e}nyi-2 entropies of Gaussian states respect all balanced entropy inequalities that hold for the Shannon entropies of multivariate normal distributions. The latter have been investigated in the literature (see, e.g., \cite{HoltzSturmfels07,ShadbakhtHassibi11}).

\bigskip

In \autoref{sec:stabs/discussion} we discuss the relation between our results for stabilizer states and Gaussian states and point towards further avenues of investigations.

\subsection*{Related Work}

Independently of the work presented in this chapter, Linden, Ruskai, and Winter have published an analysis of the entropy cone generated by stabilizer states \cite{LindenMatusRuskaiEtAl13}.
Their methods -- focusing on group-theoretical constructions -- are conceptually complementary to our phase-space approach.
\cite{LindenMatusRuskaiEtAl13} contains a complete characterization of the entropy cone generated by four-party stabilizer states.
The paper also lists further examples of inequalities which, like the Ingleton Inequality, are respected by stabilizer states, even though there are classical distributions violating it.
While not originally stated explicitly, their results also imply that all balanced inequalities remain valid for stabilizer states (see Theorem 11 in \cite{LindenMatusRuskaiEtAl13} and discussion thereafter).

\section{Discrete Phase Space}
\label{sec:stabs/discrete phase space}

In this section, we present a self-contained account of Weyl operators and stabilizer states in the discrete phase-space picture.
This section does not contain original results. All statements could be found in some form in \cite{Gottesman96, Appleby05, Gross06, Beaudrap13, KuengGross13}, albeit not in a unified language.

\subsection*{Discrete Symplectic Geometry}
Let $d > 1$ be an integer and let $\ZZ_d=\ZZ/d\ZZ$ be the ring of (congruence classes of) integers modulo $d$.
The \emph{discrete phase space}\index{phase space!discrete} for $n$ particles with local dimension $d$ is by definition $V = \bigoplus_{i=1}^n V_i = \ZZ_d^{2n}$, the free $\ZZ_d$-module of rank $2n$. Given a point $v \in V$, we write $v_i = (p_i,q_i) \in V_i = \ZZ_d^2$ for its components. Consider the bilinear form $\omega \colon V \times V \rightarrow \ZZ_d$ defined by
\[\omega(v, v') = \sum_{i=1}^n p_i q_i' - q_i p_i'.\]
It is non-degenerate and totally isotropic, i.e.\ $\omega(v, v) = 0$ for all $v \in V$.
If $d$ is prime then $V$ is simply a symplectic vector space over the finite field $\FF_d = \ZZ_d$. We will also in the general case refer to $\omega$ as the \emph{symplectic form}\index{symplectic form!$\ZZ^{2n}_d$}\nomenclature[Tomega]{$\omega(v, v')$}{symplectic form of $\ZZ^{2n}_d$}.

A \emph{character}\index{character!finite Abelian group} of a finite Abelian group $G$ is a group homomorphism $G \rightarrow \U(1)$. Denote by $\widehat G$ the set of characters, which is again an Abelian group with the operation of pointwise multiplication. It is called the \emph{(Pontryagin) dual}\index{Pontryagin dual}\nomenclature[RGhat]{$\widehat G$}{Pontryagin dual of finite Abelian group $G$} of $G$. It is well-known that $G \cong \widehat G$, although not canonically. For the cyclic group $G = \ZZ_d$, all characters are powers of $\chi_d(x) = e^{\frac{2\pi \mathrm{i}}{d} x}$.\nomenclature[Rxchar]{$\chi_d$}{character of cyclic group $\ZZ_d$}

\begin{lem}
  \label{lem:stabs/phase space characters}
  The characters of the additive group of the phase space $V$ are
  $\widehat V = \{ \chi_d(\omega(v, -)) : v \in V \} \cong V$.
\end{lem}
\begin{proof}
  By injectivity of $\chi_d \colon \ZZ_d \rightarrow \U(1)$ and non-degeneracy of the symplectic form, each $v$ determines a different character.
  Thus we have found all $\abs{\widehat V} = \abs V$ many characters.
\end{proof}

The \emphindex{symplectic complement}\nomenclature[TMomega]{$M^\omega$}{symplectic complement of subspace $M \subseteq \ZZ_d^{2n}$} of a submodule $M \subseteq V$ is the submodule
$M^\omega = \{ v \in V : \omega(v, m) = 0 \;(\forall m \in M) \}$.
In the case of prime $d$, it is well-known that $\dim M + \dim M^\omega = \dim V$---however, in general the dimension (or rank) might not even be well-defined. Still there is an important analogue that holds in the general case:
\begin{lem}
  \label{lem:stabs/perp counting}
  $\abs M \abs{M^\omega} = \abs V$.
\end{lem}
\begin{proof}
  We show that the group homomorphism
  \begin{equation*}
    \Phi \colon M^\omega \rightarrow \widehat{V/M},
    \quad
    v \mapsto \left( [w] \mapsto \chi_d(\omega(v, w)) \right)
  \end{equation*}
  is both injective and surjective (it is certainly well-defined).
  Injectivity follows immediately from the non-degeneracy of the symplectic form.
  For surjectivity, let $\tau \in \widehat{V/M}$. Then $w \mapsto \tau([w])$ is a character of $V$.
  By \autoref{lem:stabs/phase space characters}, there exists $v \in V$ such that $\tau([w]) = \chi_d(\omega(v, w))$. Since $\tau$ vanishes on $M$, $v \in M^\omega$.
  Thus $\Phi$ is an isomorphism, and we find that
  \begin{equation*}
    \abs{M^\omega} = \abs{\widehat{V/M}} = \abs{V/M} = \frac {\abs V} {\abs M}.
    \qedhere
  \end{equation*}
\end{proof}
The following important corollary follows from \autoref{lem:stabs/perp counting} and $M \subseteq (M^\omega)^\omega$:
\begin{cor}
  $(M^\omega)^\omega = M$.%
\end{cor}

We call a submodule $M \subseteq V$ an \emphindex{isotropic submodule} if $M \subseteq M^\omega$, i.e.~if $\omega(m, m') = 0$ for all $m, m' \in M$.
Finally consider $V_I = \bigoplus_{i \in I} V_i$, the phase space of particles $I \subseteq \{ 1, \dots, n\}$. There is a natural way of \emph{restricting} a submodule $M$ to $V_I$: we set
\begin{equation*}
  M_I := M \cap V_I,
\end{equation*}
where $V_I$ is identified with a submodule of $V$ in the natural way.

\subsection*{Weyl Representation}

Following \cite{Appleby05,Beaudrap13}, we first define
\emphindex{Weyl operators}\nomenclature[QWPQ]{$W(P,Q)$}{Weyl operator for integers $P$, $Q$} for general integers $(P,Q) \in \ZZ^2$, not necessarily
in the range $\{0, \dots, d-1\}$. These are the unitaries
on $L^2(\ZZ_d) \cong \CC^d$ given by
\begin{equation*}
  (W({P,Q})\psi)(x) = \tau_{2d}(-P Q) \, \chi_d(P x) \, \psi(x - Q),
\end{equation*}
where $\tau_{2d}(R) = \chi_{2d}((d^2+1) R)$. %
For example, $W({1,0})$ is the $Z$-operator $\ket x \mapsto e^{\frac
{2\pi\mathrm{i}} d x} \ket x$, while $W({0,1})$ is the $X$-operator $\ket
x \mapsto \ket{x+1 \pmod d}$.
By direct computation \cite{Beaudrap13},
\begin{align}
  W({P,Q})W({P',Q'})  &=  \tau_{2d}(P Q'-Q P') \, W({P+P',Q+Q'}),  \label{eq:stabs/Z heisenberg} \\
  W({P,Q})^{-1}  &=  W({P,Q})^\dagger = W({-P,-Q}), \label{eq:stabs/Z adjoint} \\
  W({P,Q}) W({P',Q'})  &=  \chi_d(P Q'- Q P') \, W({P',Q'}) W({P,Q}). \label{eq:stabs/Z commutator}
\end{align}

We now introduce the Weyl operators $w(p,q)$ for congruence classes
$(p,q)\in\ZZ_d^2$.  It is here that the treatment of the odd and the
even-dimensional case diverges.

For $d$ odd, $\tau_{2d}(1) = \chi_d\big(\frac{d^2 + 1} 2\big)$ is a $d$-th root of unity, so that $W(P+d, Q) = W(P, Q+d) = W(P,Q)$.
In other words, $W$ is constant on congruence classes modulo $d$, so
we can directly define $w(p,q) := W(P,Q)$.
In fact, $2^{-1}:=\frac{d^2+1}2\in\ZZ$ is
the multiplicative inverse of $2$ modulo $d$, so that
\begin{align}
  \nonumber
  (w(p, q) \psi)(x) = \chi_d(p x - 2^{-1} p q) \, \psi(x - q) \\
  w(v) w(v') = \chi_d(2^{-1} \omega(v, v')) \, w(v+v').
  \label{eq:stabs/odd heisenberg}
\end{align}

For $d$ even, $\tau_{2d}(1) = \chi_{2d}(1)$ is a primitive $2d$-th root of unity (e.g., in the case of qubits $\tau_{2d}(1)=i$).
Equation~\eqref{eq:stabs/Z heisenberg} then implies that $W(P+d,Q)$ and $W(P,Q+d)$ are either $W(P,Q)$ or $-W(P,Q)$.
In order to fix the sign, we choose $w(p,q):=W(P,Q)$ where $(P,Q)$ is the unique preimage in $\{0,\dots,d-1\}^2 \subseteq \ZZ^2$.
Because $w$ and $W$ differ at most by a phase, \eqref{eq:stabs/Z heisenberg} still implies that $(p,q)\mapsto w(p,q)$ defines a projective representation of the additive structure of the phase space $\ZZ_d^2$.

In both the even and the odd case, it now follows
from \eqref{eq:stabs/Z adjoint} and \eqref{eq:stabs/Z commutator}
that
\begin{align}
  \label{eq:stabs/adjoint}
  w(v)^{-1} &= w(v)^\dagger = \pm w(-v), \\
  \label{eq:stabs/commutator}
  w(v) w(v') &= \chi_d(\omega(v, v')) \, w(v') w(v).
\end{align}

For $n$-particles, the phase space is the direct sum $V = \bigoplus_{i=1}^n V_i = \ZZ_d^{2n}$. We define its \emph{Weyl representation}\index{Weyl representation!discrete}\nomenclature[Qwv]{$w(v)$}{discrete Weyl representation of $\ZZ^{2n}_d$} on $(\CC^d)^{\otimes n}$ by the tensor product of the single-particle representations, $w(v) = \bigotimes_{i=1}^n w(v_i)$. In this way, the relations \eqref{eq:stabs/adjoint} and \eqref{eq:stabs/commutator} continue to hold. Moreover, it is easy to verify that
\begin{equation}
  \label{eq:stabs/trace}
  \tr w(v) = d^n \, \delta_{v,0}.
\end{equation}

\section{Stabilizer States in Phase Space}
\label{sec:stabs/stabilizer states}

To define stabilizer states, we start with a \emphindex{stabilizer group} $G$, i.e.\ a finite Abelian group whose elements are scalar multiples of Weyl operators on $(\CC^d)^{\otimes n}$, such that the only multiple of $\Id=w(0)$ contained in $G$ is $\Id$ itself.
With such a group we associate the operator
\begin{equation*}
  P = \frac 1 {\abs G} \sum_{g \in G} g.
\end{equation*}
From the fact that $G$ is a group, we deduce that $P^2=P$;
since all elements of $G$ are unitaries, $P=P^\dagger$;
and \eqref{eq:stabs/trace} implies that $\tr P=d^n/|G|$.
Hence $P$ projects onto a $\big(d^n/|G|\big)$-dimensional subspace, called the \emphindex{stabilizer code} of $G$.
Note that the stabilizer code is the subspace of all vectors that are stabilized by $G$.
The corresponding \emphindex{stabilizer state} is $\rho = \frac 1 {d^n} \sum_g g$.
We refer to \cite[\S{}10.5]{NielsenChuang04} for an introduction to the stabilizer formalism from the perspective of quantum information theory.

We now prove the central theorem that provides a description of stabilizer states in terms of discrete phase space:

\begin{thm}[Stabilizers in phase space]
  \label{thm:stabs/phase space}
  Let $V=\bigoplus_{i=1}^n V_i=\ZZ^{2n}_d$ be the phase space for $n$ particles with
  local dimension $d$, where $d > 1$ is an arbitrary integer.
  There is a one-to-one correspondence between isotropic submodules $M \subseteq V$ and equivalence classes $[\rho(M)]$
  of stabilizer states on $(\CC^d)^{\otimes n}$ under conjugation with Weyl operators. Moreover,
  \begin{align}
    [\rho(M)_I] &= [\rho(M_I)],
    \label{eq:stabs/reductions stabilizers} \\
    S([\rho(M)_I]) &= \abs I \log(d) - \log\ \abs{M_I}
    \label{eq:stabs/von neumann entropy stabilizers}
  \end{align}
  for all subsets $I \subseteq \{ 1, \dots, n \}$.
  If $d$ is odd then there is a canonical element $\rho(M)$\nomenclature[QrhoM]{$\rho(M)$}{stabilizer state corresponding to isotropic submodule $M$} in each equivalence class, given by
  \begin{equation}
    \label{eq:stabs/stabilizer state}
    \rho(M)
    = \frac 1 {d^n} \sum_{m \in M} w(m).
  \end{equation}
  It is compatible with reductions, i.e.\ $\rho(M)_I = \rho(M_I)$.
\end{thm}
\begin{proof}
  \emph{(1) From isotropic submodules to classes of stabilizer
  states:}
  Let $M \subseteq \ZZ_d^{2n}$ be an isotropic submodule.
  Since $M$ is a finite Abelian group, it can be written as a direct sum of cyclic groups,
  $M \cong \ZZ_{d_1} \oplus \dots \oplus \ZZ_{d_k}$.
  Let $m_j \in M$ be a generator of the $j$-th cyclic subgroup.
  Since $w(m_j)^{d_j} \propto w(0) = \Id$, we can choose phases $\lambda_j$ such that $(\lambda_j w(m_j))^{d_j} = \Id$.
  Define
  \begin{equation*}
    G = \{ \underbrace{\prod_{j=1}^k (\lambda_j w(m_j))^{x_j}}_{=:
    \mu_m \, w(m)} : m = \sum_j x_j m_j \in M \}, %
  \end{equation*}
  where the phases $\mu_m$ exist since $w(m)$ is a projective representation.
  Since $M$ is isotropic, \eqref{eq:stabs/commutator} implies that the Weyl operators $\{ w(m) : m \in M \}$ all commute.
  It follows that $G$ is closed under multiplication:
  we have that $\mu_m w(m) \mu_{m'} w(m') = \mu_{m+m'} w(m+m')$.
  Moreover, the only multiple of the identity in $G$ is $\Id$ itself.
  Thus the data $(M, \mu)$ defines a stabilizer group of cardinality $\abs M$ with corresponding stabilizer state $\rho(M, \mu) = \frac 1 {d^n} \sum_{m \in M} \mu_m w(m)$.
  This state depends on the phases $\mu_m$, which in turn resulted from our choice of generators $m_j$ and phases $\lambda_j$.
  A different choice would have resulted in another stabilizer group
  $G' = \{\nu_m w(m) \,|\, m\in M\}$
  and we have yet to show that the two groups are related by conjugation with a Weyl operator.
  To this end, we note that likewise $\nu_m w(m) \nu_{m'} w(m') = \nu_{m+m'} w(m+m')$. %
  It follows that $\tau(m) := \nu_m / \mu_m$ defines a character of $M$.
  Since $V$ is an Abelian group, this character can be extended to all of $V$ \cite[Proposition~2.2.1]{Bump10} and is therefore of the form $\tau(m) = \chi_d(\omega(v, m))$ for some $v \in V$ (\autoref{lem:stabs/phase space characters}).
  But then it follows from \eqref{eq:stabs/commutator} that
  \begin{equation*}
    w(v) \, \mu_m w(m) \, w(v)^\dagger
    = \tau(m) \mu_m w(m)
    = \nu_m w(m).
  \end{equation*}
  Consequently, $\rho(M, \nu)$ and $\rho(M, \mu)$ are related by
  conjugation with the Weyl operator $w(v)$.

  If $d$ is odd, then \eqref{eq:stabs/odd heisenberg} implies that $w(m) w(m')
  = w(m+m')$. It follows that $G := \{ w(m) : m \in M \}$ is a
  stabilizer group of cardinality $\abs M$, with corresponding
  stabilizer state
  $
    \rho(M) = \frac 1 {d^n} \sum_{m \in M} w(m).
  $
  This is the canonical representative \eqref{eq:stabs/stabilizer state} of
  the equivalence class of states associated with $M$.

  \emph{(2) Surjectivity:}
  Here, we show that
  our map from isotropic submodules to equivalence classes of
  stabilizer states is surjective. Let $G$ be a stabilizer group with
  corresponding stabilizer state $\rho$.  Equation~\eqref{eq:stabs/trace}
  implies that for each $g \in G$ there exists a unique $m_g \in V$
  such that $g \propto w(m_g)$.  Conversely, no two $m_g$ can be
  equal---
  otherwise, two group elements in $G$ would differ only by a phase
  and hence there would be a non-trivial multiple of $\Id$ in $G$.
  Define $M := \{ m_g \}$. Then $M$ is a submodule of $V$, since
  $m_g + m_h = m_{gh}$.
  Since $G$ is Abelian, all $w(m_g)$ commute and \eqref{eq:stabs/commutator} shows that $M$ is isotropic.
  Then $M$ is indeed a preimage of $[\rho]$, since by its very construction there exists a choice of phases by which we recover $G$ (namely $\mu_{m_g} = g \, w(m_g)^{-1}$).

  \emph{(3) Injectivity:}
  Suppose that $\rho(M,\mu)$ and $\rho(M',\mu')$ are two equivalent stabilizer states. As we saw, conjugating with a Weyl operator only changes the phases, so we may in fact assume that states are equal. Now assume that $M \neq M'$, so that there exists, e.g., $m \in M \setminus M'$. Then, \eqref{eq:stabs/trace} shows that
  \begin{equation*}
    0 \neq \tr w(m) \rho(M,\mu) = \tr w(m) \rho(M',\mu') = 0,
  \end{equation*}
  which is the desired contradiction.

  \emph{(4) Reduction:}
  We now show that our construction is compatible with reduction. For this, observe that
  \begin{equation*}
    \tr_{I^c} w(m) =
    \left( \bigotimes_{i \in I} w(m_i) \right) \left( \prod_{i \in I^c} \tr w(m_i) \right) =
    \left( \bigotimes_{i \in I} w(m_i) \right) d^{\abs I^c} \, \delta_{M_I}(m).
  \end{equation*}
  Since any valid assignment of phases $\mu_m$ restricts to the submodule $M_I = M \cap V_I$,
  it follows that $[\rho(M)_I] = [\rho(M_I)]$.
  It is also immediate that the canonical element \eqref{eq:stabs/stabilizer state} is compatible with reduction.

  \emph{(5) Entropy:}
  In view of the preceding point, it suffices to show \eqref{eq:stabs/von neumann entropy stabilizers} for $I = \{1,\dots,n\}$.
  Recall that the cardinalities of $M$ and of the corresponding stabilizer group $G$ agree.
  We have already seen that the dimension of the stabilizer code is equal to $d^n / \abs G$.
  Thus,
  \begin{equation*}
    S([\rho(M)]) = n \log(d) - \log \, \abs G = n \log(d) - \log \, \abs M. \qedhere
  \end{equation*}
\end{proof}

\begin{exl}[Maximally entangled states]
\label{exl:stabs/maximally entangled}
\index{maximally entangled states}
  Let $\ZZ_2^4$ be the discrete phase space for $n=2$ particles with local dimension $d=2$, and consider the ``diagonal'' subspace
  \[M = \{ (v,v) : v = (p,q) \in \ZZ_2^2 \} \subseteq \ZZ_2^4,\]
  which is isotropic since
  $\omega((v,v), (w,w)) = 2 \omega(v, w) = 0$ in $\ZZ_2$.
  Hence the general theory shows that corresponding Weyl operators $w(v,v)$ commute with each other.
  To see this concretely, observe that $w(v,v) = w(v) \otimes w(v)$, where the single-particle Weyl operator $w(v)$ is one of the Pauli matrices $X := w(0,1)$, $Y = w(1,1)$, or $Z := w(1,0)$; since the latter anti-commute with each other, their second tensor powers commute with each other.
  A corresponding stabilizer subgroup is
  \[G = \{ \Id, X \otimes X, -Y \otimes Y, Z \otimes Z \}\]
  (observe the judicious choice of signs).
  The stabilizer code is one-dimensional, since $d^n / \abs G = 2^2 / 4 = 1$, and spanned by the vector
  $\ket{\Psi^+} = \left( \ket{00} + \ket{11} \right) / \sqrt 2$ which is stabilized by all operators in $G$.
  Thus the stabilizer state $\rho$ is the maximally entangled state $\proj{\Psi^+}$.
  In accordance with \autoref{thm:stabs/phase space}, its entropies are given by
  \begin{equation*}
    S(\rho) = 2 - \log \abs M = 0, \quad
    S(\rho_1) = 1 - \log \abs{M_1} = 1, \quad
    S(\rho_2) = 1 - \log \abs{M_2} = 1,
  \end{equation*}
  since $M_1 = M_2 = \{ 0 \}$.
  Another choice of stabilizer subgroup is
  \[G = \{ \Id, -X \otimes X, Y \otimes Y, Z \otimes Z \}.\]
  The corresponding stabilizer state is $\ket{\Psi^-} = \left( \ket{00} - \ket{11} \right) / \sqrt 2$, which can be obtained from the previous one by conjugation with the Weyl operator $Z \otimes \Id$.

\medskip

  More generally, we may obtain maximally entangled states on $\CC^d \otimes \CC^d$ by taking the isotropic subspace $M = \{ ((p,q), (-p,q)) : (p,q) \in \ZZ_d^2 \} \subseteq \ZZ_d^4$.
  This shows that stabilizer states can be highly entangled.
\end{exl}

\section{A Classical Model for Stabilizer Entropies}
\label{sec:stabs/classical model}

If the local dimension $d$ is odd, there exists a \emph{discrete Wigner function}\index{Wigner function!discrete}\nomenclature[QW_rho]{$W_\rho$}{discrete Wigner function of quantum state $\rho$} that replicates many properties of its better-known continuous-variable counterpart \cite{Gross06}.
It is the function on phase space defined by
\begin{equation}
\label{eq:stabs/discrete wigner function}
  W_\rho(v) =
  \frac 1 {d^{2n}} \sum_{v'\in V} e^{-\frac{2\pi\mathrm{i}}{d} 2^{-1} \omega(v, v')} \tr\big(w(v')^\dagger
  \rho\big).
\end{equation}
The central observation is that in the case of stabilizer states, the Wigner
function $W_{\rho(M)}$ is %
a \emph{probability distribution on phase space},
i.e.\ it attains only non-negative values which sum to one.
In fact \cite{Gross05,Gross06},
\begin{equation}
\label{eq:stabs/phase space dist}
\begin{aligned}
  W_{\rho(M)}(v)
  =&\frac 1 {d^{2n}} \sum_{v'\in V} e^{-\frac{2\pi\mathrm{i}}{d} 2^{-1} \omega(v, v')} \delta_M(v') \\
  =&\frac 1 {d^{2n}} \sum_{v'\in M} e^{-\frac{2\pi\mathrm{i}}{d} 2^{-1} \omega(v, v')} \\
  =&\frac {\abs M} {d^{2n}} \delta_{M^\omega}(v) \\
  =&\frac 1 {\abs{M^\omega}} \delta_{M^\omega}(v).
\end{aligned}
\end{equation}
Thus the Wigner function of the stabilizer state with isotropic submodule $M \subseteq V$ is given by the uniform distribution on $M^\omega \subseteq V$.

We now show that this construction defines a classical model which reproduces the entropies of the given stabilizer state and its reduced states, up to a certain constant. In fact, by phrasing the construction solely in terms of the symplectic complement (hence without recourse to the Wigner function), this result can be established for arbitrary local dimension, even or odd:

\begin{thm}[Classical model for stabilizer states]
  \label{thm:stabs/main theorem}
  Let $V = \bigoplus_{i=1}^n V_i = \ZZ^{2n}_d$ be the discrete phase space
  for $n$ particles with local dimension $d$, where $d > 1$ is an arbitrary integer.
  Let $\rho$ be a stabilizer state with isotropic submodule $M \subseteq V$,
  and define a random variable $X = (X_1, \dots, X_n)$ that takes values uniformly in the symplectic
  complement $M^\omega \subseteq V$. Then,
  \begin{equation}
    \label{eq:stabs/main eqn}
    S(\rho_I) = H(X_I) - \abs I \log(d),
  \end{equation}
  and the same conclusion holds if we replace the Shannon and von Neumann entropy
  by any R\'{e}nyi entropy.
  If $d$ is odd then the above construction can also be obtained by
  interpreting the Wigner function $W_{\rho}$ as the probability
  distribution of the random variable $X$.
\end{thm}
\begin{proof}
  To prove \eqref{eq:stabs/main eqn}, denote by $\pi_I \colon V \rightarrow
  V_I$ the projection onto the phase space of parties $I \subseteq
  \{1,\dots,n\}$.  It will be convenient to consider $V_I$ as a symplectic submodule of $V$ in the natural way.
  To avoid any notational ambiguity, we denote by $\omega_I$ the restriction of the symplectic form to $V_I$ and by $X^{\omega_I}$ the symplectic complement of a subspace $X$ taken correspondingly within $V_I$.
  Observe that
  \begin{equation}
  \label{eq:stabs/perp lemma a}
    \pi_I(M^\omega) \subseteq M_I^{\omega_I}.
  \end{equation}
  Indeed, if $v\in M^\omega$ and $m_I \in M_I$, then
  $\omega_I(\pi_I(v), m_I) = \omega(v, m_I) = 0$.
  On the other hand, we find that
  \begin{equation}
  \label{eq:stabs/perp lemma b}
    \pi_I(M^\omega)^{\omega_I}
    \subseteq
    M_I.
  \end{equation}
  To see this, consider a vector $v_I \in V_I$ and note that
  if $\omega_I(v_I, \pi_I(M^\omega)) = 0$ then $\omega(v_I, M^\omega) = 0$,
  hence $v_I \in M \cap V_I = M_I$ since $(M^\omega)^\omega = M$.
  We conclude from \eqref{eq:stabs/perp lemma a} and \eqref{eq:stabs/perp lemma b} that
  \begin{equation}
  \label{eq:stabs/perp lemma}
    \pi_I(M^\omega) = M_I^{\omega_I}.
  \end{equation}

  Note that
  $X_I = \pi_I(X)$. Since $\pi_I$ is a group homomorphism, it follows
  that $X_I$ is distributed uniformly on its range, so that
  \begin{align*}
    H(X_I)
    &= \log\ \abs{\pi_I(M^\omega)}
    = \log\ \abs{M_I^{\omega_I}}
    = \log \frac {d^{2 \abs I}} {\abs {M_I}} \\
    &= 2 \abs I \log(d) - \log\ \abs{M_I}
    = \abs I \log(d) + S(\rho_I),
  \end{align*}
  where we have used \eqref{eq:stabs/von neumann entropy stabilizers} in the last step.
  We have thus established \eqref{eq:stabs/main eqn}.
  The same result holds if we replace the Shannon and von Neumann
  entropy by a classical and quantum R\'{e}nyi entropy, respectively.
  This is because the random variables $X_I$ are distributed uniformly on their range
  and the stabilizer states $\rho(M)_I$ are normalized projectors, so that the respective entropies coincide.

  Finally, it is clear from \eqref{eq:stabs/phase space dist} that for odd
  $d$ the distribution of $X$ coincides with the Wigner function
  $W_{\rho(M)}$ of the stabilizer state.
  It remains to show that the Wigner function $W_{\rho_I}$ of a
  reduced state $\rho_I$ is obtained by marginalizing the full Wigner
  function
  (in other words: the
  quantum and the classical way of reducing to subsystems commute):
  \begin{equation}
    \label{eq:stabs/marginals}
    W_{\rho_I}(v) = \sum_{w\, : \,w_I = v} W_\rho(w)
  \end{equation}
  for all $v \in V_I$. While this can easily be proved in full
  generality from the definition of the Wigner function,
  it is also true that for the special case of stabilizer states
  \eqref{eq:stabs/marginals} follows directly from \eqref{eq:stabs/phase space dist} and \eqref{eq:stabs/perp lemma}.
\end{proof}

The following corollary completely solves \autoref{pro:entropy/balanced} in the case of stabilizer states:

\begin{cor}
\label{cor:stabs/stabilizer}
  The von Neumann entropies of stabilizer states satisfy all balanced entropy inequalities satisfied by the Shannon entropy.
  Moreover, they satisfy the Ingleton inequality \eqref{eq:stabs/ingleton}.
\end{cor}
\begin{proof}
  As described in the introduction, the first claim follows immediately from
  \eqref{eq:stabs/main eqn}. This is because for any balanced information inequality
  $\sum_I \nu_I H(X_I) \geq 0$ we necessarily have that \cite{ShadbakhtHassibi11}
  \begin{equation*}
    \sum_I \nu_I \abs I = \sum_I \left( \sum_{i \in I} \nu_I \right) = \sum_i \left( \sum_{I \ni i} \nu_I \right) = 0.
  \end{equation*}
  Hence the correction term in \eqref{eq:stabs/main eqn} cancels as we sum over all
  subsystems:
  \begin{equation*}
    \sum_I \nu_I S(\rho_I) = \sum_I \nu_I H(X_I) - \underbrace{\sum_I \nu_I \abs I}_{=0} \log(d) \geq 0
  \end{equation*}

  For the second claim, we recall that the random variable $X$ is uniformly distributed on $M^\omega$, which is a group, and that the projections $\pi_I \colon V \rightarrow V_I$ are group homomorphisms (also when restricted to $M^\omega$).
  In this situation, it was shown in \cite{LiChong07} that the Ingleton inequality \eqref{eq:stabs/ingleton} holds for the random variables $X_I = \pi_I(X)$.
  (In the language of \cite{ChanYeung02}, the entropy vector $(H(X_I))$ can be characterized by the normal subgroups $\ker(\pi_I) \cap M^\omega \subseteq M^\omega$.)
  Since the Ingleton inequality is balanced, the same argument that we used above shows that the Ingleton inequality also holds for the von Neumann entropies of stabilizer states.
\end{proof}

Pure stabilizer states correspond to maximally isotropic submodules $M \subseteq V$.
Such submodules are called \emph{La\-grang\-ian}, and they satisfy $\abs M = d^n$ and $M = M^\omega$.
Thus in this case our classical model can also be defined by choosing $X \in M$ uniformly at random.
Furthermore, since $\pi_I(M) \cong M / (\ker \pi_I \cap M)$, we may also define $X_I$ to be the coset of $X$
modulo $\ker(\pi_I) \cap M = M \cap V_{I^c} = M_{I^c}$.
In this way, we recover the construction of Theorem 11 in \cite{LindenMatusRuskaiEtAl13}.

\section{Gaussian States}
\label{sec:stabs/gaussian states}

We sketch the corresponding result for Gaussian states of continuous-variable systems (see, e.g., \cite{Holevo82} for a mathematical approach).
Recall that the \emph{Weyl representation}\index{Weyl representation!continuous-variable} of the single-particle phase space $\RR^2$ on the single-particle Hilbert space $L^2(\RR)$ is given by
\begin{equation*}
  (w(p, q) f)(x) = e^{-ipq/2} e^{ixp} f(x-q) ,
\end{equation*}
where $v = (p,q) \in \RR^2$.
For $n$ particles, the classical \emph{phase space}\index{phase space!continuous-variable} is $\RR^{2n}$ and the Weyl representation on $L^2(\RR)^{\otimes n} = L^2(\RR^n)$ is defined by the tensor product $w(v) = \otimes_{i=1}^n w(v_i)$ for all $v = (v_i) \in \RR^{2n}$.\nomenclature[Qwv]{$w(v)$}{continuous-variable Weyl representation of $\RR^{2n}$}
The \emph{Wigner function}\index{Wigner function!continuous-variable}\nomenclature[QW_rho]{$W_\rho$}{continuous-variable Wigner function of quantum state $\rho$} of a quantum state $\rho$ on $L^2(\RR^n)$ is the following function on classical phase space $\RR^{2n}$,
\begin{equation*}
  W_\rho(v) = \frac 1 {(2\pi)^{2n}} \int_{\RR^{2n}} dv' \, e^{-i \omega(v, v')} \tr ( w(v')^\dagger \rho ),
\end{equation*}
where $\omega(v,v') = \sum_{i=1}^n p_i q'_i - q_i p'_i$ denotes the usual \emph{symplectic form}\index{symplectic form!$\RR^{2n}$}\nomenclature[Tomega]{$\omega(v, v')$}{symplectic form of $\RR^{2n}$} on $\RR^{2n}$.
That is, just like its discrete analogue \eqref{eq:stabs/discrete wigner function}, the Wigner function is defined as the symplectic Fourier transform of the so-called characteristic function, which maps points in phase space to the expectation value of Weyl operators. It follows from the ``non-commutative Parseval identity'' for the latter \cite[Theorem~V.3.2]{Holevo82} that
\begin{equation}
\label{eq:stabs/parseval}
  \tr \rho^2
  = (2\pi)^n \norm{W_\rho}_{L^2(\RR^{2n})}^2.
\end{equation}

A state $\rho$ on $L^2(\RR^n)$ is a called a \emphindex{Gaussian quantum state} if its Wigner function is of the form
\begin{equation}
\label{eq:stabs/gaussian wigner}
  W_\rho(v) = \frac1{(2\pi)^{n} (\det\Sigma)^{\frac12}} e^{-\frac12 (v-\mu)^T \Sigma^{-1} (v-\mu)}.
\end{equation}
Here, $\Sigma$ is a real, positive definite $n \times n$-matrix called the \emph{covariance matrix}\index{Gaussian quantum state!covariance matrix}, and $\mu$ is the vector of \emph{first moments}\index{Gaussian quantum state!first moments}.
Conversely, for every covariance matrix $\Sigma$ that satisfies the \emph{uncertainty relation} $\Sigma + i\Omega \geq 0$, where $\Omega$ is the symplectic matrix, there exists a corresponding Gaussian quantum state.
Note that \eqref{eq:stabs/gaussian wigner} is the probability density of a random vector $X = (X_1,\dots,X_{2n})$ with multivariate normal distribution of mean $\mu$ and covariance matrix $\Sigma$. Using \eqref{eq:stabs/parseval}, it follows that the R\'{e}nyi-2 entropy of the quantum state, $S_2(\rho) = - \log \tr \rho^2$, is directly related to the differential R\'{e}nyi-2 entropy of the continuous random variable $X$, $h_2(X) = - \log \int W^2_\rho(x)\ dx$:
\begin{equation}
\label{eq:stabs/gaussian qc overall}
  S_2(\rho) = h_2(X) - n \log(2 \pi)
\end{equation}
The reduced state $\rho_I$ for some subset of modes $I \subseteq \{1,\dots,n\}$ is again a Gaussian state, and its covariance matrix is equal to the corresponding submatrix of $\Sigma$. Thus the Wigner function of $\rho_I$ is given by the marginal probability density of the variables $X_I = (X_i)_{i\in I}$, and using \eqref{eq:stabs/gaussian qc overall} we find that
\begin{equation}
\label{eq:stabs/gaussian qc}
  S_2(\rho_I) = h_2(X_I) - \abs I \log(2 \pi).
\end{equation}
Equation~\eqref{eq:stabs/gaussian qc} states that the R\'{e}nyi-2 entropy of a Gaussian quantum state is always lower than the phase space entropy of its classical model, as given by the Wigner function. It is so by a precise amount, namely by $\log(2\pi)$ bits per mode.

\begin{thm}[Classical model for Gaussian states]
  \label{thm:stabs/gaussian}
  Let $\rho$ be a Gaussian state with covariance matrix $\Sigma$, and
  define a random variable $X = (X_1,\dots,X_n)$ with probability density
  given by the Wigner function $W_\rho(x)$. Then, for any positive $\alpha \neq 1$,
  \begin{equation*}
    S_2(\rho_I) = h_\alpha(X_I) - \abs I \left( \log \pi -  \frac {\log \alpha} {1-\alpha} \right),
  \end{equation*}
  where $h_\alpha(X) = (1-\alpha)^{-1} \log \int W_\rho^\alpha(x)\ dx$ is the differential R\'{e}nyi-$\alpha$ entropy.
  In the limit $\alpha \rightarrow 1$, we recover
  \begin{equation}
    \label{eq:stabs/shannon gaussian}
    S_2(\rho_I) = h(X_I) - \abs I \log ( \pi e),
  \end{equation}
  where $h(X) = - \int W_\rho(x) \log W_\rho(x)\ dx$ is the differential Shannon entropy.
\end{thm}
\begin{proof}
  By Gaussian integration, the differential R\'{e}nyi-$\alpha$ entropy of the random variable $X_I$ is given by
  \begin{equation*}
    h_\alpha(X_I) = \frac 1 2 \log \det \Sigma + \abs I \left( \log 2\pi - \frac {\log \alpha} {1-\alpha} \right).
  \end{equation*}
  The assertions of the theorem follow from this and \eqref{eq:stabs/gaussian qc}.
\end{proof}

Equation~\eqref{eq:stabs/shannon gaussian} has been previously used in \cite{AdessoGirolamiSerafini12}, where the formula is attributed to Stratonovich.
Just as in the discrete case, we immediately get the following corollary:

\begin{cor}
\label{cor:stabs/gaussian}
  The R\'{e}nyi-2 entropies of Gaussian states satisfy all balanced entropy inequalities that are valid for the differential Shannon entropies of multivariate normal distributions.
\end{cor}

Interestingly, it is not clear whether \autoref{cor:stabs/gaussian} holds for the von Neumann entropy.
We remark that Gaussian states can violate the Ingleton inequality \eqref{eq:stabs/ingleton} (in contrast to stabilizer states, cf.\ \autoref{cor:stabs/stabilizer}).
Indeed, this is well-known for multivariate normal distributions \cite{ShadbakhtHassibi11}, and the counterexample presented in \cite{ShadbakhtHassibi11} can be readily adapted:
\begin{equation*}
  \Sigma = \Id_{p,q} \otimes \begin{bmatrix}10 & 2.5 & 5 & 5 \\ 2.5 & 10 & 5 & 5 \\ 5 & 5 & 10 & 0 \\ 5 & 5 & 0 & 10\end{bmatrix}
\end{equation*}
is a covariance matrix on the four-particle phase space that satisfies the uncertainty relation $\Sigma + i \Omega \geq 0$, and the corresponding four-partite Gaussian quantum state violates the Ingleton inequality.

\section{Discussion}
\label{sec:stabs/discussion}

Theorems~\ref{thm:stabs/main theorem} and \ref{thm:stabs/gaussian} can be phrased in a unified language by observing that all reductions $\rho_I$ of a stabilizer state are proportional to projectors, while the corresponding random variables $X_I$ are uniformly distributed on their support.
As noted in \autoref{thm:stabs/main theorem}, this implies that we may replace the von Neumann and Shannon entropy in \eqref{eq:stabs/main eqn} by any quantum and classical R\'{e}nyi-$\alpha$ entropy, respectively.
For discrete random variables the latter are defined by $H_\alpha(X) := (1-\alpha)^{-1} \log \sum_x p_x^\alpha$.
In particular, we find that
\begin{equation}
\label{eq:stabs/qc 2 norm}
  S_2(\rho_I) = H_2(X_I) - \abs I \log(d),
\end{equation}
which is a perfect analogue of the identity \eqref{eq:stabs/gaussian qc} for Gaussian states.
Equation~\eqref{eq:stabs/qc 2 norm} can be interpreted as the unitarity (up to a constant factor) of the transformation that sends a quantum state to its Wigner function (if the latter exists).
This suggests that the fundamental quantities in this context are in fact the R\'{e}nyi-2 entropies or, equivalently, the Hilbert--Schmidt and $\ell^2$/$L^2$-norms.

\medskip

We conclude this chapter with a few remarks.
Our work uses the classical model provided by the Wigner function as a
tool for proving statements that do not, a priori, seem to be connected
to phase-space distributions. This point of view has been employed
before, e.g.\ to construct quantum expanders \cite{GrossEisert08},
to establish simulation algorithms \cite{VeitchFerrieGrossEtAl12, MariEisert12, VeitchWiebeFerrieEtAl13}, and
to demonstrate the onset of contextuality \cite{HowardWallmannVeitchEtAl14}.  It
would be interesting to see further applications.
While it is known that the Wigner function approach cannot be
straight-forwardly translated to non-stabilizer states
\cite{Hudson74, Gross06, Gross07}, our discussion suggests searching for other
maps from quantum states to probability distributions that reproduce
entropies faithfully, up to state-independent additive constants.

In order to establish the Ingleton inequality \eqref{eq:stabs/ingleton}, we have used
the group-theoretical approach to classical information inequalities \cite{ChanYeung02}.
It would be highly desirable to find a quantum-mechanical analogue of this work
(see \cite{ChristandlMitchison06} and \autoref{ch:strong6j} for partial results towards this goal,
motivated by the quantum marginal problem).
\tikzstyle epr=[fill=black]
\tikzstyle adjoint epr=[fill=black]
\tikzstyle adjoint=[]
\tikzstyle arrowstyle=[scale=1.25]
\tikzset{->-/.style={decoration={markings,mark=at position #1 with {\arrow[arrowstyle]{stealth};}},postaction={decorate}}}
\tikzset{-<-/.style={decoration={markings,mark=at position #1 with {\arrowreversed[arrowstyle]{stealth};}},postaction={decorate}}}

\newcommand{\strongsixjCG}[4]{
  \begin{tikzpicture}[baseline=0cm]
    \draw[-<-=.3] (0,-0.2) -- (0,-0.7);
    \draw[->-=.6] (0.177,0.1) -- (0.766, 0.5);
    \draw[->-=.6] (-0.177,0.1) -- (-0.766,0.5);
    \draw (0,0) circle[radius=0.2];
    \draw (0,0) node {#4};
    \draw (-0.766,0.5) node[above] {#2};
    \draw (0.766, 0.5) node[above] {#3};
    \draw (0,-0.7) node[below] {#1};
  \end{tikzpicture}
}

\newcommand{\strongsixjCGadjoint}[4]{
  \begin{tikzpicture}[baseline=0cm]
    \draw[->-=.6] (0,-0.2) -- (0,-0.7);
    \draw[-<-=.35] (0.177,0.1) -- (0.766, 0.5);
    \draw[-<-=.35] (-0.177,0.1) -- (-0.766,0.5);
    \draw[adjoint] (0,0) circle[radius=0.2];
    \draw (0,0) node {#4};
    \draw (-0.766,0.5) node[above] {#2};
    \draw (0.766, 0.5) node[above] {#3};
    \draw (0,-0.7) node[below] {#1};
  \end{tikzpicture}
}

\newcommand{\strongsixjCGinvariant}[1]{
    \begin{tikzpicture}[baseline=0cm]
      \draw[->-=.8] (0,-0.45) -- (0,-0.2);
      \draw[->-=.65] (0,-0.6) -- (0,-0.9);
      \draw[->-=.6] (0.177,0.1) -- (0.766, 0.5);
      \draw[->-=.6] (-0.177,0.1) -- (-0.766,0.5);
      \draw (0,0) circle[radius=0.2];
      \draw (0,0) node {#1};
      \draw[epr] (0,-0.52) circle[radius=0.1];
      \draw (-0.766,0.5) node[above] {$\alpha$};
      \draw (0.766, 0.5) node[above] {$\beta$};
      \draw (0,-0.8) node[below] {$\lambda$};
    \end{tikzpicture}
}

\newcommand{\strongsixjCGadjointinvariant}[1]{
    \begin{tikzpicture}[baseline=0cm]
      \draw[->-=.6] (0,-0.2) -- (0,-0.7);
      \draw[->-=.7] (0.57,0.35) -- (0.786,0.5);
      \draw[->-=.7] (0.413, 0.25) -- (0.177,0.1);
      \draw[->-=.7] (-0.57,0.35) -- (-0.786,0.5);
      \draw[->-=.7] (-0.413, 0.25) -- (-0.177,0.1);
      \draw[adjoint] (0,0) circle[radius= 0.2];
      \draw[epr] (0.491, 0.3) circle[radius=0.1];
      \draw[epr] (-0.491, 0.3) circle[radius=0.1];
      \draw (0,0) node {#1};
      \draw (-0.766,0.5) node[above] {$\alpha$};
      \draw (0.766, 0.5) node[above] {$\beta$};
      \draw (0,-0.7) node[below] {$\lambda$};
    \end{tikzpicture}
}

\newcommand{\strongsixjAsymmetric}{
  \begin{tikzpicture}[baseline=0cm]
    \draw[->-=.6] (0.173,0.1) -- (0.693, 0.4);
    \draw[->-=.6] (-0.173, 0.1) -- (-0.693, 0.4);
    \draw[-<-=.35] (0,-0.2) -- (0,-0.8);
    \draw[adjoint] (0.866, 0.5) circle[radius= 0.2];
    \draw[adjoint] (-0.866, 0.5) circle[radius= 0.2];
    \draw (0, -1) circle[radius= 0.2];
    \draw (0,0) circle[radius= 0.2];
    \draw[->-=.6] (0.19,-0.98) arc [radius=1, start angle= -77, end angle= 18];
    \draw[->-=.6] (0.748,0.662) arc [radius=1, start angle= 43, end angle= 138];
    \draw[->-=.6] (-0.967,0.319) arc [radius=1, start angle= 163, end angle= 259];
    \draw (0,-1) node {$i$};
    \draw (0,0) node {$j$};
    \draw (0.866, 0.5) node {$k$};
    \draw (-0.866, 0.5) node {$l$};
    \draw (-0.53,0.3) node[below] {$\alpha$};
    \draw (0.53,0.3) node[below] {$\beta$};
    \draw (0.866,-0.5) node[below] {$\gamma$};
    \draw (0, -0.5) node[right] {$\mu$};
    \draw (0,1) node[below] {$\nu$};
    \draw (-0.866,-0.5) node[below] {$\lambda$};
  \end{tikzpicture}
}

\newcommand{\strongsixjAsymmetricTeleported}{
  \begin{tikzpicture}[baseline=0cm, scale=1.4]
    \draw (-0.53,0.2) node[below] {$\alpha$};
    \draw[->-=.55] (-0.54,0.31) -- (-0.866, 0.5);
    \draw[->-=.66] (-0.54,0.31) -- (-0.33, 0.18);
    \draw[->-=.7] (0,0) -- (-0.33, 0.18);
    \draw[adjoint epr] (-0.56,0.32) circle[radius=0.1/1.4];
    \draw[adjoint epr] (-0.31,0.17) circle[radius=0.1/1.4];

    \draw (0.53,0.3) node[below] {$\beta$};
    \draw[->-=.55] (0.54,0.31) -- (0.866, 0.5);
    \draw[->-=.66] (0.54,0.31) -- (0.33, 0.18);
    \draw[->-=.7] (0,0) -- (0.33, 0.18);
    \draw[adjoint epr] (0.56,0.32) circle[radius=0.1/1.4];
    \draw[adjoint epr] (0.31,0.17) circle[radius=0.1/1.4];

    \draw (0, -0.5) node[right] {$\mu$};
    \draw[->-=.75] (0,-1) -- (0,-0.65);
    \draw[->-=.65] (0,-0.35) -- (0,-0.65);
    \draw[->-=.5] (0,-0.35) -- (0,0);
    \draw[adjoint epr] (0,-0.355) circle[radius=0.1/1.4];
    \draw[adjoint epr] (0,-0.65) circle[radius=0.1/1.4];

    \draw (0,1) node[below] {$\nu$};
    \centerarc[->-=0.57](0,0)(35:70:1);
    \centerarc[-<-=0.35](0,0)(70:120:1);
    \centerarc[->-=0.5](0,0)(120:145:1);
    \draw[adjoint epr] (0.35,0.92) circle[radius=0.1/1.4];
    \draw[adjoint epr] (-0.35,0.92) circle[radius=0.1/1.4];

    \draw (0.866,-0.5) node[right] {$\gamma$};
    \centerarc[->-=0.75](0,0)(-90:-50:1);
    \centerarc[-<-=0.35](0,0)(-40:-13:1);
    \centerarc[->-=0.65](0,0)(-10:25:1);
    \draw[adjoint epr] (1,-0.15) circle[radius=0.1/1.4];
    \draw[adjoint epr] (0.7,-0.7) circle[radius=0.1/1.4];

    \draw (-0.866,-0.5) node[left] {$\lambda$};
    \centerarc[-<-=0.6](0,0)(270:230:1);
    \centerarc[->-=0.6](0,0)(220:193:1);
    \centerarc[-<-=0.4](0,0)(190:150:1);
    \draw[adjoint epr] (-1,-0.15) circle[radius=0.1/1.4];
    \draw[adjoint epr] (-0.7,-0.7) circle[radius=0.1/1.4];

    \draw[fill=white] (0,-1) circle[radius=0.2/1.4];
    \draw (0,-1) node {$i$};
    \draw[fill=white] (0,0) circle[radius=0.2/1.4];
    \draw (0,0) node {$j$};
    \draw[fill=white] (0.866, 0.5) circle[radius=0.2/1.4];
    \draw (0.866, 0.5) node {$k$};
    \draw[fill=white] (-0.866, 0.5) circle[radius=0.2/1.4];
    \draw (-0.866, 0.5) node {$l$};
  \end{tikzpicture}
}

\newcommand{\strongsixjSymmetric}{
  \begin{tikzpicture}[baseline=0cm]
    \draw (0,-1) node {$i$};
    \draw (0,-1) circle[radius= 0.2];
    \draw (0,0) node {$j$};
    \draw (0,0) circle[radius= 0.2];
    \draw (0.866, 0.5) node {$k$};
    \draw (0.866, 0.5) circle[radius= 0.2];
    \draw (-0.866, 0.5) node {$l$};
    \draw (-0.866, 0.5) circle[radius= 0.2];

    \draw (-0.53,0.2) node[below] {$\alpha$};
    \draw[->-=.56] (-0.693,0.4) -- (-0.433, 0.25);
    \draw[->-=.56] (-0.173,0.1) -- (-0.433, 0.25);
    \draw[adjoint epr] (-0.43,0.24) circle[radius=0.1];

    \draw (0.53,0.3) node[below] {$\beta$};
    \draw[->-=.56] (0.693,0.4) -- (0.433, 0.25);
    \draw[->-=.56] (0.173,0.1) -- (0.433, 0.25);
    \draw[adjoint epr] (0.43,0.24) circle[radius=0.1];

    \draw (0, -0.5) node[right] {$\mu$};
    \draw[->-=.56] (0,-0.2) -- (0,-0.5);
    \draw[->-=.56] (0,-0.8) -- (0,-0.48);
    \draw[adjoint epr] (0,-0.5) circle[radius=0.1];

    \draw (0,1) node[below] {$\nu$};
    \draw[->-=.56] (0.748,0.662) arc [radius=1, start angle=43, end angle=90]; %
    \draw[-<-=.56] (0,1) arc [radius=1, start angle=90, end angle=138];
    \draw[adjoint epr] (0,0.99) circle[radius=0.1];

    \draw (0.866,-0.5) node[below] {$\gamma$};
    \draw[->-=.56] (0.19,-0.98) arc [radius=1, start angle=-77, end angle=-30]; %
    \draw[-<-=.56] (0.83,-0.5) arc [radius=1, start angle=-30, end angle=18];
    \draw[adjoint epr] (0.86,-0.45) circle[radius=0.1];

    \draw (-0.866,-0.5) node[below] {$\lambda$};
    \draw[->-=.56] (-0.967,0.319) arc [radius=1, start angle=163, end angle=210]; %
    \draw[-<-=.56] (-0.866,-0.5) arc [radius=1, start angle=210, end angle=259];
    \draw[adjoint epr] (-0.87,-0.45) circle[radius=0.1];
  \end{tikzpicture}
}
\chapter{Entropy Inequalities from Recoupling Coefficients}
\label{ch:strong6j}

In this chapter we consider entropy inequalities for general quantum states.
This requires us to go beyond the mathematical techniques of the first part of this thesis, which only gave us control over non-overlapping marginals.
To this end, we unveil a novel link between the existence of multipartite quantum states with given marginal eigenvalues and the representation theory of the symmetric group $S_k$.
We use this link to give a new proof of the strong subadditivity and weak monotonicity of the von Neumann entropy, and propose an approach to finding further entropy inequalities based on studying representation-theoretic symbols and their symmetry properties.

The results in this chapter have been obtained in collaboration with Matthias Christandl and Burak \c{S}ahino\u{g}lu, and they have appeared in the preprint \cite{ChristandlSahinogluWalter12} (cf.\ \cite{Sahinoglu12} for an earlier version).

\section{Summary of Results}

To establish the link to representation theory, we consider the \emph{recoupling coefficients} of the symmetric group $S_k$, which measure the overlap of two ways of decomposing a triple tensor product of irreducible representations of $S_k$ (\autoref{sec:strong6j/recoupling}).
We find that in the ``semiclassical limit'' $k \rightarrow \infty$, the recoupling coefficients' norm decreases at most polynomially for a sequence of Young diagrams of $k$ boxes converging to the eigenvalues of a given tripartite quantum state $\rho_{ABC}$ and its reduced states $\rho_A$, $\rho_B$, $\rho_C$, $\rho_{AB}$, $\rho_{BC}$; conversely, if there exists no such quantum state then the coefficients decrease exponentially in norm (\autoref{thm:strong6j/main theorem} in \autoref{sec:strong6j/overlapping marginals}).
This is a first general result for the quantum marginal problem with overlaps. It extends significantly the characterization of the triples $\rho_A$, $\rho_B$, $\rho_{AB}$ by the Kronecker coefficient of the symmetric group that we had reviewed in \autoref{sec:onebody/kronecker} \cite{ChristandlMitchison06, Klyachko04, DaftuarHayden04, ChristandlHarrowMitchison07, ChristandlDoranKousidisEtAl12}.
Our result can be generalized to an arbitrary number of particles and linearly many reduced states, and may thus be regarded as a first step towards a quantum-mechanical version of Chan and Yeung's description of the set of compatible Shannon entropies in terms of group theory \cite{ChanYeung02}.

To illustrate the power of this characterization, we show that symmetry properties of the recoupling coefficients alone imply the strong subadditivity and weak monotonicity of the von Neumann entropy (\autoref{sec:strong6j/ssa}).
This symmetry is particularly transparent when the coefficients are expressed in a graphical calculus for tensor categories (\autoref{sec:strong6j/symmetry}).
Our strategy of proof suggests a hitherto unexplored route towards establishing further entropy inequalities by exploiting the symmetries of higher-order representation-theoretic objects (\autoref{sec:strong6j/discussion}).

Our result is inspired by Wigner's seminal work on the semiclassical behavior of quantum spins, which are described by the representation theory of the group $\SU(2)$ \cite{WignerGriffin59}.
The recoupling coefficients of $\SU(2)$, known as the \emph{Wigner $6j$-symbols} in their rescaled, more symmetric form \cite{WignerGriffin59}, describe the relation between individual spins $j_A$, $j_B$, $j_C$, their total spin $j_{ABC}$ and the intermediate spins $j_{AB}$ and $j_{BC}$ (the Racah W-coefficients \cite{Racah42} are also closely related).
As first noted by Wigner, there is a dichotomy in the semiclassical limit where all spins are simultaneously large: the $6j$-symbol decays polynomially if there exists a tetrahedron with side lengths $j_A$, $j_B$, $j_C$, $j_{AB}$, $j_{BC}$, $j_{ABC}$, and exponentially otherwise \cite[\S{}27]{WignerGriffin59} (\autoref{fig:strong6j/tetrahedron}).

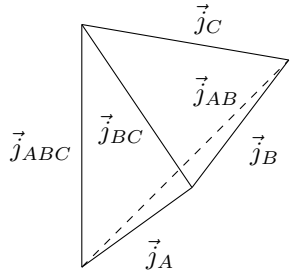
\begin{figure}
  \begin{center}
    \begin{tikzpicture}
    \draw [-] (0,3.2) -- (0,0);
    \draw [-] (0,0) -- (1.455,1.05);
    \draw [-] (1.455,1.05) -- (2.732, 2.732);
    \draw [-] (2.732, 2.732) -- (0,3.2);
    \draw [-] (1.455,1.05)-- (0,3.2);
    \draw [-, dashed] (0,0) -- (2.732, 2.732);
    \node [left] at (0, 1.6) {$\vec j_{ABC}$};
    \node [below right] at (0.727, 0.5) {$\vec j_A$};
    \node [below right] at (2.1, 1.866) {$\vec j_B$};
    \node [above right] at (1.366, 2.95) {$\vec j_C$};
    \node [above left] at (2.2, 2) {$\vec j_{AB}$};
    \node [above right] at (0.1, 1.5) {$\vec j_{BC}$};
    \end{tikzpicture}
  \end{center}
  \caption[Wigner's tetrahedron]{\emph{Wigner's tetrahedron}, corresponding to polynomial decay of the $\SU(2)$ recoupling coefficients.}
  \label{fig:strong6j/tetrahedron}
\end{figure}

That the asymptotics are in both cases guided by the existence of a geometric object---for Wigner, a tetrahedron with certain side lengths, for us, a quantum state with certain spectral properties---is not an accident. Via Schur--Weyl duality, our limit $k \rightarrow \infty$ can similarly be understood as a semiclassical limit of representation-theoretic coefficients of unitary groups (\autoref{sec:strong6j/semiclassicality}).
In \autoref{sec:strong6j/sums of matrices} we show that Wigner's scenario is in fact a special case of the quantum marginal problem considered above:
For every tetrahedron, we construct a tripartite quantum state in a faithful, i.e., side length-encoding way, and for every recoupling coefficient for $\SU(2)$ we construct a corresponding recoupling coefficient of the symmetric group.
The existence of Wigner's tetrahedron can be understood as an instance of the more general problem of characterizing the eigenvalues of certain partial sums of matrices \cite{Backens10}, and our construction generalizes readily. This extends the connection between the non-overlapping quantum marginal problem and Horn's conjecture mentioned in \autoref{sec:onebody/kronecker} \cite{Klyachko04}.

\section{Recoupling Coefficients}
\label{sec:strong6j/recoupling}

Recall from \autoref{sec:onebody/kronecker} that the finite-dimensional irreducible representations of the symmetric group $S_k$ are labeled by \emph{Young
diagrams} with $k$ boxes, that is, ordered partitions $\lambda_1 \geq \dots \geq \lambda_l > 0$ of $\sum_i \lambda_i = k$.
As before, we write $\lambda \vdash k$ for such a partition and $[\lambda]$ for the associated irreducible unitary representation of $S_k$; we denote its dimension by $d_\lambda := \dim [\lambda]$.
Any finite-dimensional representation $V$ of $S_k$ can be decomposed into a direct sum of irreducible representations, and if $V$ is a unitary representation then this decomposition can also be made unitary.
Concretely, consider the space of $S_k$-linear maps
$\Ha^V_\lambda := \Hom_{S_k}([\lambda], V)$ defined as in \eqref{eq:onebody/linear maps} and equip $\Ha^V_\lambda$ with the re-scaled Hilbert--Schmidt inner product
\begin{equation}
\label{eq:strong6j/rescaled hs}
  \braket{\psi, \phi}_\lambda := \frac 1 {d_\lambda} \tr \psi^\dagger \phi.
\end{equation}

\begin{lem}
\label{lem:strong6j/decomposition}
  (1) Any unit vector in $\Ha^V_\lambda$ is an $S_k$-linear isometry, and any two orthogonal vectors have orthogonal range.

  (2) The direct sum of the $S_k$-linear maps
  \begin{equation*}
    \Phi^V_\lambda \colon [\lambda] \otimes \Ha^V_\lambda \rightarrow V,
    \quad
    v \otimes \phi \mapsto \phi(v)
  \end{equation*}
  defines an $S_k$-linear unitary isomorphism $\bigoplus_{\lambda \vdash k} [\lambda] \otimes \Ha^V_\lambda \cong V$.
\end{lem}
\begin{proof}
  (1) Let $\phi, \psi \in \Ha^V_\lambda$. By Schur's lemma \eqref{eq:onebody/schurs lemma}, $\psi^\dagger \phi \propto \Id_{[\lambda]}$, so that
  \begin{equation*}
    \psi^\dagger \phi = \braket{\psi, \phi}_\lambda \Id_{[\lambda]}.
  \end{equation*}
  In particular, $\phi^\dagger \phi = \Id$ if $\phi$ is a unit vector, and $\psi^\dagger \phi = 0$ if $\phi$ and $\psi$ are orthogonal.

  (2) This follows from (1) and Schur's lemma.
\end{proof}

In particular, we may decompose a tensor product $[\alpha] \otimes [\beta]$ of two irreducible representations:

\begin{dfn}
\label{dfn:strong6j/clebsch gordan}
  The \emph{Clebsch--Gordan isomorphism of the symmetric group}\index{Clebsch--Gordan isomorphism!symmetric group|textbf} is the $S_k$-linear unitary isomorphism
  \begin{equation}
  \label{eq:strong6j/clebsch gordan}
    \bigoplus_{\lambda \vdash k} [\lambda] \otimes \Ha^{\alpha\beta}_{\lambda} \stackrel{\cong}{\longrightarrow} [\alpha] \otimes [\beta]
  \end{equation}
  defined as in \autoref{lem:strong6j/decomposition}.
  Its components are $S_k$-linear isometric embeddings; they will be denoted by
  $\Phi^{\alpha\beta}_{\lambda} \colon [\lambda] \otimes \Ha^{\alpha\beta}_{\lambda} \rightarrow [\alpha] \otimes [\beta]$.%
  \nomenclature[RPhi_lambda^alphabeta]{$\Phi_\lambda^{\alpha\beta}$}{Clebsch--Gordan embeddings for the symmetric group}%
  \nomenclature[RH_lambda^alphabeta]{$\Ha_\lambda^{\alpha\beta}$}{multiplicity spaces in the Clebsch--Gordan decomposition for the symmetric group}
  According to \eqref{eq:onebody/kronecker asymmetric}, the dimension of $\Ha^{\alpha\beta}_\lambda$ is equal to the Kronecker coefficient $g_{\alpha,\beta,\lambda}$.
\end{dfn}

Now we consider a triple tensor product. Since the tensor product is associative, we have
\[\left( [\alpha] \otimes [\beta] \right) \otimes [\gamma] \cong [\alpha] \otimes \left( [\beta] \otimes [\gamma] \right).\]
Decomposing accordingly using \eqref{eq:strong6j/clebsch gordan}, we get an isomorphism
\begin{equation}
\label{eq:strong6j/6j decomposition}
  \bigoplus_{\lambda \vdash k} [\lambda] \otimes \bigoplus_{\mu \vdash k} \Ha_{\lambda}^{\mu\gamma} \otimes \Ha_{\mu}^{\alpha\beta}
  \cong
  \bigoplus_{\lambda \vdash k} [\lambda] \otimes \bigoplus_{\nu \vdash k} \Ha_{\lambda}^{\alpha\nu} \otimes \Ha_{\nu}^{\beta\gamma}.
\end{equation}
By Schur's lemma, this allows us to identify the multiplicity spaces for each fixed $\lambda$,
\begin{equation}
\label{eq:strong6j/6j decomposition after schur}
  \bigoplus_{\mu \vdash k} \Ha_{\lambda}^{\mu\gamma} \otimes \Ha_{\mu}^{\alpha\beta}
  \cong
  \bigoplus_{\nu \vdash k} \Ha_{\lambda}^{\alpha\nu} \otimes \Ha_{\nu}^{\beta\gamma}.
\end{equation}

\begin{dfn}
\label{dfn:strong6j/recoupling}
  The \emph{recoupling coefficients of the symmetric group}\index{recoupling coefficients!symmetric group|textbf}\nomenclature[Rsixj]{$\smallsixj \alpha \beta \gamma \mu \nu \lambda$}{recoupling coefficients of the symmetric group} are the components of the isomorphism \eqref{eq:strong6j/6j decomposition after schur} for fixed $\mu$ and $\nu$, denoted by
  \begin{equation*}
    \sixj{\alpha}{\beta}{\gamma}{\mu}{\nu}{\lambda} \colon
    \Ha_{\lambda}^{\mu  \gamma} \otimes \Ha_{\mu}^ {\alpha  \beta} \rightarrow \Ha_{\lambda}^{\alpha  \nu} \otimes \Ha_{\nu}^{\beta  \gamma}.
  \end{equation*}
  In other words, they are defined by the relation
  \begin{equation}
  \label{eq:strong6j/6j via clebsch gordan}
    \Id_{\lambda} \otimes \sixj{\alpha}{\beta}{\gamma}{\mu}{\nu}{\lambda}
    =
    \left(
    \left( \Id_{\alpha} \otimes \Phi^{\beta\gamma}_{\nu} \right)
    \Phi^{\alpha\nu}_\lambda
    \right)^\dagger
    \left( \Phi^{\alpha\beta}_{\mu} \otimes \Id_{\gamma} \right)
    \Phi^{\mu\gamma}_\lambda
  \end{equation}
  in terms of the Clebsch--Gordan maps from \autoref{dfn:strong6j/clebsch gordan}.
\end{dfn}

In contrast to the case of $\SU(2)$, these recoupling coefficients are linear maps rather than scalars, since the ``Clebsch--Gordan series'' for $S_k$ is \emph{not} multiplicity-free.
Their size is thus measured by an operator norm, and it will be convenient to employ the \emphindex{Hilbert--Schmidt norm}\nomenclature[Q<X_2]{$\normHS{X}$}{Hilbert--Schmidt norm, or Schatten-2 norm of operator $X$} $\normHS{X}^2 := \tr X^\dagger X$ as well as the \emphindex{operator norm}\nomenclature[Q<X_infty]{$\norm{X}_\infty$}{operator norm, or Schatten-$\infty$ norm of operator $X$} $\norm{X}_\infty := \sup_{\norm \psi = 1} \norm{X \psi}$ (cf.\ \autoref{lem:strong6j/poly eqv}).

\subsection*{Schur--Weyl Duality}

We now consider a setup in which both the symmetric group and quantum states are at home: the vector space $V = \left( \CC^d \right)^{\otimes k}$.
The symmetric group acts on $V$ by permuting the tensor factors and the special unitary group $\SU(d)$ acts diagonally; both actions commute.
Recall that \emphindex{Schur--Weyl duality} \eqref{eq:onebody/schur-weyl} asserts that we have a decomposition
\begin{equation}
\label{eq:strong6j/schur-weyl}
  \left( \CC^d \right)^{\otimes k} \cong
  \; \smashoperator{\bigoplus_{\lambda \vdash_d k}} \; [\lambda] \otimes V^d_\lambda,
\end{equation}
where the direct sum runs over all Young diagram with $k$ boxes and at most $d$ rows.
In the following we shall denote by $P_\lambda$ the orthogonal projector onto a direct summand $[\lambda] \otimes V^d_\lambda$ in \eqref{eq:strong6j/schur-weyl}.

In order to study tripartite quantum states, we consider $\CC^{abc} \cong \CC^a \otimes \CC^b \otimes \CC^c$. Applying Schur--Weyl duality separately to the $k$-th tensor powers of $\CC^a$, $\CC^b$ and $\CC^c$ and inserting \eqref{eq:strong6j/6j decomposition}, we find that
\begin{equation}
  \label{eq:strong6j/tripartite schur-weyl}
  \begin{aligned}
  \left( \CC^{abc} \right)^{\otimes k} &\cong
  \; \smashoperator{\bigoplus_{\alpha \vdash_a k, \, \beta \vdash_b k, \, \gamma \vdash_c k}} \; [\alpha] \otimes [\beta] \otimes [\gamma] \otimes V^a_\alpha \otimes V^b_\beta \otimes V^c_\gamma \\
  &\cong
  \;\smashoperator{\bigoplus_{\alpha \vdash_a k, \dotsc, \lambda \vdash_{abc} k}}\;
  [\lambda] \otimes \Ha^{\mu \gamma}_{\lambda} \otimes \Ha^{\alpha \beta}_{\mu} \otimes V^a_\alpha \otimes V^b_\beta \otimes V^c_\gamma\\
  &\cong
  \;\smashoperator{\bigoplus_{\alpha \vdash_a k, \dotsc, \lambda \vdash_{abc} k}}\;
  [\lambda] \otimes \Ha_{\lambda}^{\alpha \nu} \otimes \Ha_{\nu}^{\beta \gamma} \otimes V^a_\alpha \otimes V^b_\beta \otimes V^c_\gamma.
  \end{aligned}
\end{equation}
We denote by
\begin{equation}
\label{eq:strong6j/pq projectors}
\begin{aligned}
  P &:= P^{\alpha\beta\gamma}_{\lambda,\mu} := (P_\alpha \otimes P_\beta \otimes P_\gamma) (P_\mu \otimes P_\gamma) P_\lambda \\
  Q &:= Q^{\alpha\beta\gamma}_{\lambda,\nu} := (P_\alpha \otimes P_\beta \otimes P_\gamma) (P_\alpha \otimes P_\nu) P_\lambda
\end{aligned}
\end{equation}
the orthogonal projectors onto the respective direct summands in the second and third line of \eqref{eq:strong6j/tripartite schur-weyl} (observe that each is defined as a product of commuting projectors).
The following lemma connects the operator norm of their product to the operator norm of the corresponding recoupling coefficient.

\begin{lem}
  \begin{equation}
    \label{eq:strong6j/PQ vs 6j}
    \norm{P Q}_\infty =
    \norm{Q P}_\infty =
    \norm{\sixj \alpha \beta \gamma \mu \nu \lambda}_\infty
  \end{equation}
\end{lem}
\begin{proof}
Recall from \eqref{eq:strong6j/6j via clebsch gordan} that the recoupling coefficients are defined by the identity
\begin{equation*}
  \Id_{\lambda} \otimes \sixj \alpha \beta \gamma \mu \nu \lambda = F^\dagger E,
\end{equation*}
where $E := \left( \Phi^{\alpha\beta}_\mu \otimes \Id_{\gamma} \right) \Phi^{\mu\gamma}_\lambda$
and $F := \left( \Id_{\alpha} \otimes \Phi^{\beta\gamma}_\nu \right) \Phi^{\alpha\nu}_\lambda$
are compositions of Clebsch--Gordan isometries. In particular,
\begin{equation*}
  \norm{\sixj \alpha \beta \gamma \mu \nu \lambda}_\infty =
  \norm{F^\dagger E}_\infty =
  \norm{F F^\dagger E E^\dagger}_\infty,
\end{equation*}
where the last equality holds because both $E$ and $F$ are isometries.
Now note that $E E^\dagger$ is precisely equal to the orthogonal projector onto
\begin{equation*}
  [\lambda] \otimes \Ha_\lambda^{\mu\gamma} \otimes \Ha_\mu^{\alpha\beta} \subseteq [\alpha] \otimes [\beta] \otimes [\gamma].
\end{equation*}
On the other hand, $P$ as defined in \eqref{eq:strong6j/pq projectors} is the orthogonal projector onto
\begin{align*}
  &[\lambda] \otimes \Ha_\lambda^{\mu\gamma} \otimes \Ha_\mu^{\alpha\beta} \otimes V^a_\alpha \otimes V^b_\beta \otimes V^c_\gamma \\
  \subseteq
  &[\alpha] \otimes [\beta] \otimes [\gamma] \otimes V^a_\alpha \otimes V^b_\beta \otimes V^c_\gamma
  \subseteq
  \left( \mathbb C^{abc} \right)^{\otimes k}.
\end{align*}
It follows that
$P = E E^\dagger \otimes \Id_{V^a_\alpha \otimes V^b_\beta \otimes V^c_\gamma}$
where we slightly abuse notation by considering the right-hand side as an operator on $\left( \mathbb C^{abc} \right)^{\otimes k}$.
Likewise, we find that
$Q = F F^\dagger \otimes \Id_{V^a_\alpha \otimes V^b_\beta \otimes V^c_\gamma}.$
Together we conclude that
\begin{equation*}
  \norm{F F^\dagger E E^\dagger}_\infty = \norm{Q P}_\infty. \qedhere
\end{equation*}
\end{proof}

For the following it will be useful to relate the recoupling coefficients' operator norm to their Hilbert--Schmidt norm:

\begin{lem}
\label{lem:strong6j/poly eqv}
  \begin{equation}
    \label{eq:strong6j/6j lower upper}
    \norm{\begin{bmatrix} \alpha & \beta & \mu \\ \gamma & \lambda & \nu \end{bmatrix}}_\infty \\
    \leq \normHS{\begin{bmatrix} \alpha & \beta & \mu \\ \gamma & \lambda & \nu \end{bmatrix}}\\
    \leq \poly(k) \norm{\begin{bmatrix} \alpha & \beta & \mu \\ \gamma & \lambda & \nu \end{bmatrix}}_\infty
    \leq \poly(k)
  \end{equation}
\end{lem}
\noindent Here and in the following, $\poly(k)$ denotes a polynomial in $k$ that depends only on the maximal number of rows of the Young diagrams.
Equivalently, it depends only on $a$, $b$ and $c$---the dimensions of the local Hilbert spaces $\CC^a$, $\CC^b$ and $\CC^c$ in \eqref{eq:strong6j/tripartite schur-weyl}.
\begin{proof}
  It suffices to show that the multiplicity spaces $\Ha^{\alpha\beta}_{\lambda}$ in the Clebsch--Gordan decomposition \eqref{eq:strong6j/clebsch gordan} are of dimension $\poly(k)$.
  This dimension is equal to the Kronecker coefficient $g_{\alpha,\beta,\lambda}$.
  In \autoref{sec:onebody/kronecker} we have seen that we may equivalently define $g_{\alpha,\beta,\lambda}$ as the multiplicity of the irreducible $\SU(d) \times \SU(d) \times \SU(d)$-representation $V^d_\alpha \otimes V^d_\beta \otimes V^d_\gamma$ in $\Sym^k(\CC^{d^3})$, where $d$ has to be chosen at least as large as the number of rows in the Young diagrams $\alpha$, $\beta$ and $\gamma$.
  But the Young diagrams that appear in \eqref{eq:strong6j/tripartite schur-weyl} all have at most $d = abc$ rows.
  It follows that we have the simple polynomial upper bound %
  \begin{equation*}
    g_{\alpha,\beta,\gamma} \leq \dim \Sym^k(\CC^{(abc)^3}) = \poly(k).
    \qedhere
  \end{equation*}
\end{proof}

Schur--Weyl duality also leads to an alternative definition of the recoupling coefficients in terms of unitary groups (\autoref{sec:strong6j/semiclassicality}).

\section{Overlapping Marginals and Recoupling Coefficients}
\label{sec:strong6j/overlapping marginals}

For a density operator $\rho_{ABC}$ on $\CC^a \otimes \CC^b \otimes \CC^c$, we consider the reduced density operators $\rho_{AB}=\tr_{C} (\rho_{ABC})$, $\rho_A= \tr_{BC} (\rho_{ABC})$ etc., and denote by $\vecr_{ABC}$, $\vecr_{AB}$, $\vecr_A$, etc., the corresponding spectra (each ordered non-increasingly, e.g.~$r_{ABC,1} \geq r_{ABC,2} \geq \dots$).%
\nomenclature[Qr_I]{$\vecr_I$}{vector of eigenvalues $r_{I,1} \geq r_{I,2} \geq \dots$ of reduced density matrix $\rho_I$}
We define the \emph{normalization}\index{Young diagram!normalization}\nomenclature[Rlambda_bar]{$\bar\lambda$}{normalization of Young diagram $\lambda$} of a Young diagram $\lambda \vdash k$ by $\bar\lambda := \lambda / k$.
Then we have the following theorem relating the recoupling coefficients' asymptotic behavior to the existence of a tripartite quantum state with given marginal eigenvalues:

\begin{thm}
  \label{thm:strong6j/main theorem}
  If there exists a quantum state $\rho_{ABC}$ with eigenvalues $\vecr_A$, $\vecr_B$, $\vecr_C$, $\vecr_{AB}$, $\vecr_{BC}$, $\vecr_{ABC}$ then there exists a sequence of Young diagrams $\alpha, \beta, \gamma, \mu, \nu, \lambda \vdash k$ with $k \rightarrow \infty$ boxes and at most $a$, $b$, etc.\ rows such that
    \begin{equation}\label{eq:strong6j/assumption-conv}
      \lim_{k \rightarrow \infty}(\bar\alpha, \bar\beta, \bar\gamma, \bar\mu, \bar\nu, \bar\lambda)= (\vecr_A, \vecr_B, \vecr_C, \vecr_{AB}, \vecr_{BC}, \vecr_{ABC})
    \end{equation}
    and
    \begin{equation} \label{eq:strong6j/assumption-hs}
      \normHS{\begin{bmatrix} \alpha & \beta & \mu \\ \gamma & \lambda & \nu \end{bmatrix}}
      \geq \frac 1 {\mathrm{poly}(k)}.
    \end{equation}
  Conversely, if $(\vecr_A, \vecr_B, \vecr_C, \vecr_{AB}, \vecr_{BC}, \vecr_{ABC})$ is not associated to any tripartite density matrix then for every sequence of Young diagrams satisfying \eqref{eq:strong6j/assumption-conv} we have
  \begin{equation*}
    \normHS{\begin{bmatrix} \alpha & \beta & \mu \\ \gamma & \lambda & \nu \end{bmatrix}} \leq \exp(-\Omega(k)).
  \end{equation*}
\end{thm}

\begin{proof}
For both directions of the proof we use the spectrum estimation theorem \cite{KeylWerner01} (cf.~\cite{AlickiRudnickiSadowski88, HayashiMatsumoto02, ChristandlMitchison06}), which says that $k$ copies of a density operator $\rho$ on $\CC^d$ are mostly supported on the subspaces $[\lambda] \otimes V^d_\lambda \subseteq (\CC^d)^{\otimes k}$ satisfying $\bar\lambda = \lambda/k \approx \spec \rho = \vec r$. More precisely,
\begin{equation}
\label{eq:strong6j/keyl werner}
\tr(P_\lambda \rho^{\otimes k}) \leq \poly(k) \exp(-k \norm{\bar{\lambda} - \vec r}_1^2 / 2),
\end{equation}
where $\norm{x}_1= \sum_i \abs{x_i}$ is the $\ell_1$-norm.

We start with the proof of the ``if'' statement.
Define $\widetilde P$ as the sum of the projectors $P_{\alpha\beta\gamma}^{\lambda,\mu}$ -- defined as in \eqref{eq:strong6j/pq projectors} -- for which $\norm{\bar{\alpha}- \vecr_A}_1 \leq \delta$, $\norm{\bar{\beta} - \vecr_B}_1 \leq \delta$, etc.; $\widetilde Q$ is defined accordingly.
By \eqref{eq:strong6j/keyl werner} and the fact that there are only $\poly(k)$ many Young diagrams with a bounded number of rows,
\begin{equation*}
\tr(\widetilde P \rho^{\otimes k}_{ABC}) \geq 1 - \varepsilon, \quad \tr(\widetilde Q \rho^{\otimes k}_{ABC}) \geq 1 - \varepsilon,
\end{equation*}
where $\varepsilon = \poly(k) \exp(-k \delta^2 / 2)$.
Now we use
\begin{equation}
\label{eq:strong6j/nc union bound}
  \abs{\tr(\widetilde P \widetilde Q \sigma)} \geq \tr(\widetilde P \sigma) - \sqrt{1 - \tr(\widetilde Q \sigma)},
\end{equation}
which holds for arbitrary projectors $\widetilde P$, $\widetilde Q$ and quantum states $\sigma$,%
\footnote{For pure states $\sigma = \proj\phi$, $\tr(PQ \sigma)= \braket{\phi | PQ | \phi} \geq \braket{\phi | P |\phi} - \abs{\braket{\phi | P (1 - Q) | \phi}} \geq \braket{\phi | P | \phi} - \norm{(1 - Q) \ket\phi} = \tr (P \sigma) - \sqrt{1 - \tr (Q \sigma)}$ by the Cauchy--Schwarz inequality; the general statement follows by considering a purification of $\sigma$.} and obtain
\begin{equation*}
  \norm{\widetilde P \widetilde Q}_\infty \geq
  \abs{\tr(\widetilde P \widetilde Q \rho^{\otimes k}_{ABC})} \geq 1 - 2 \sqrt{\varepsilon}.
\end{equation*}
Using \eqref{eq:strong6j/PQ vs 6j}, \eqref{eq:strong6j/6j lower upper} and the triangle inequality, we find that
\begin{equation*}
  \displaystyle\sum_{\alpha, \beta, \gamma, \mu, \nu, \lambda}\normHS{\begin{bmatrix} \alpha & \beta & \mu \\ \gamma & \lambda & \nu \end{bmatrix}}
  \geq  \norm{\widetilde P \widetilde Q}_\infty
  \geq 1- 2 \sqrt\varepsilon,
\end{equation*}
where the sum extends over Young diagrams whose normalization is $\delta$-close to the eigenvalues associated to $\rho_{ABC}$ and its marginals, respectively (as specified above). Since the number of terms in this sum is again upper-bounded by $\poly(k)$, we can find sequences of Young diagrams satisfying \eqref{eq:strong6j/assumption-conv} and \eqref{eq:strong6j/assumption-hs}.

\smallskip

We now prove the converse statement.
For this, we consider a sequence of Young diagrams \eqref{eq:strong6j/assumption-conv} whose limit $(\vecr_A, \vecr_B, \dots)$ is \emph{not} associated to any tripartite density matrix on $\CC^{abc}$.
For each $k$, we choose quantum states $\sigma_{(ABC)^k}$ such that
\[
  \norm{P Q}_\infty^2
  =
  \norm{(Q P) (P Q)}_\infty
  =
  \tr(P Q \sigma_{(ABC)^k} Q P )
,\]
where the projectors $P$ and $Q$ are defined as in \eqref{eq:strong6j/pq projectors}.
Both projectors $P$ and $Q$ commute with the action of $S_k$ on $(\CC^{abc})^{\otimes k}$, so we may assume $\sigma_{(ABC)^k}$ to be permutation-invariant.
Then we can use the bound \eqref{eq:dhmeasure/postselection}
\begin{equation}
\label{eq:strong6j/postselection}
  \sigma_{(ABC)^k} \leq \mathrm{poly}(k) \int d\rho_{ABC} \, \rho_{ABC}^{\otimes k},
\end{equation}
where $d\rho_{ABC}$ is the Hilbert--Schmidt probability measure\index{Hilbert--Schmidt probability measure} on the set of density matrices on $\CC^{abc}$.
The right-hand side of \eqref{eq:strong6j/postselection} commutes with the action of $S_k$ as well as with unitaries of the form $U^{\otimes k}$, $U \in \SU(abc)$.
In view of Schur--Weyl duality \eqref{eq:strong6j/schur-weyl}, Schur's Lemma thus implies that it is a linear combination of the isotypical projectors $P_\lambda$, and so commutes with $P$ and $Q$.
It follows that
\begin{align*}
    \norm{P Q}_\infty^2
    &\leq \mathrm{poly}(k) \tr (P Q \int d\rho_{ABC} \, \rho_{ABC}^{\otimes k} Q P) \\
    &= \mathrm{poly}(k) \tr (P \int d\rho_{ABC} \, \rho_{ABC}^{\otimes k} Q) \\
    &= \mathrm{poly}(k) \int d\rho_{ABC} \, \tr(P \rho_{ABC}^{\otimes k} Q)
\end{align*}
by linearity and cyclicity of the trace.
Since $d\rho_{ABC}$ is a probability measure, it follows, using \eqref{eq:strong6j/6j lower upper} and \eqref{eq:strong6j/PQ vs 6j}, that for each $k$ there exists at least one quantum state $\rho_{ABC, k}$ on $\CC^{abc}$ such that
\begin{align*}
  \normHS{\begin{bmatrix} \alpha & \beta & \mu \\ \gamma & \lambda & \nu \end{bmatrix}}^2
  \leq
  \mathrm{poly}(k) \abs{\tr(P \rho_{ABC, k}^{\otimes k} Q)}.
\end{align*}
Using the Cauchy--Schwarz inequality, the right-hand side can be upper-bounded by the square roots of each of the six traces $\tr(P_\alpha \rho_A^{\otimes k})$, $\tr(P_\beta \rho_B^{\otimes k})$, etc., which in turn can be upper-bounded via \eqref{eq:strong6j/keyl werner}. Thus we find
\begin{equation*}
  \normHS{\begin{bmatrix} \alpha & \beta & \mu \\ \gamma & \lambda & \nu \end{bmatrix}}^2 \leq \poly(k) \exp(-k D^2 / 4),
\end{equation*}
where $D := \min_{(\vecs_A, \vecs_B, \dots)} \max \{ \norm{\bar\alpha - \vecs_A}_1, \norm{\bar\beta - \vecs_B}_1, \dots \}$ quantifies the distance of $(\vecr_A, \vecr_B, \dots)$ to the closed set of spectra $(\vecs_A, \vecs_B, \dots)$ that correspond to tripartite quantum states on $\CC^{abc}$.
By assumption, $D > 0$, so we get the exponential decay asserted in the theorem.
\end{proof}

For pure quantum states $\rho_{ABC}$, the Schmidt decomposition \eqref{eq:onebody/schmidt} implies that necessarily $\vecr_{AB} = \vecr_C$ and $\vecr_A = \vecr_{BC}$. Therefore, we can discard the two-body spectra, and the problem reduces to a one-body quantum marginal problem.
On the level of representation theory, it suffices to consider single-row Young diagrams $\lambda = (k)$, corresponding to the trivial representation; hence, $\mu = \gamma$ and $\alpha = \nu$ according to \eqref{eq:onebody/schur-weyl bipartite}, and it can be shown easily that $\normHS{\smallsixj \alpha \beta \gamma \mu \nu \lambda}^2 = \dim \Ha^{\alpha\beta}_\gamma$, which is the \emph{Kronecker coefficient} of the symmetric group.
In this way, \autoref{thm:strong6j/main theorem} specializes to the results of \cite{ChristandlMitchison06, Klyachko04, ChristandlHarrowMitchison07} discussed in \autoref{sec:onebody/kronecker}.
This also shows that recoupling coefficients can grow with $k$.

\bigskip

Our result can be generalized to more than three parties by considering the following quantity: as in \eqref{eq:strong6j/6j decomposition}, successively decompose a tensor product of irreducible representations in two inequivalent ways; the corresponding ``generalized recoupling coefficients'' then are the components of the resulting isomorphism for fixed intermediate labels $\mu_j$ and $\nu_k$, and an analogous result can be established for these coefficients.
They are in the same way related to Wigner's $3nj$-symbols for $\SU(2)$ as the recoupling coefficients are related to the $6j$-symbol.
Just as \autoref{thm:strong6j/main theorem} does not cover the eigenvalues of $\rho_{AC}$, in general only a linear number of the exponentially many reduced states can be controlled in this fashion (e.g., the nearest-neighbor reduced states in a linear chain of particles).
However, the ``if'' part of \autoref{thm:strong6j/main theorem} immediately generalizes to an arbitrary number of marginals, since it only relies on the spectrum estimation theorem \eqref{eq:strong6j/keyl werner} and the ``union bound'' \eqref{eq:strong6j/nc union bound}.
What the representation-theoretic quantities involved in controlling all marginal spectra should be is an intriguing question, with possible ramifications for the search for new entropy inequalities of the von Neumann entropy, as we detail in the next section.

\section{Symmetry Properties of the Recoupling Coefficients}
\label{sec:strong6j/symmetry}

In the following we use a graphical calculus for symmetric monoidal categories to deduce symmetry properties of the recoupling coefficients for the symmetric group (see, e.g., the reviews \cite{Coecke10} and \cite{Selinger11}, \cite[\S{}2]{Turaev10}, or \cite{Preskill04} for a more physical introduction).
An alternative, purely algebraic proof is given at the end of this section.
Both the strong subadditivity of the von Neumann entropy as well as its weak monotonicity can then be understood in terms of the coefficients' symmetries (\autoref{sec:strong6j/ssa}).

Recall from \autoref{sec:strong6j/recoupling} that the multiplicity spaces $\Ha^{\alpha\beta}_\lambda$ are given by the space of $S_k$-linear maps from $[\lambda]$ to $[\alpha] \otimes [\beta]$.
In each multiplicity space, let us choose maps $\Phi^{\alpha\beta}_{\lambda,i}$ that form an orthonormal basis with respect to the inner product \eqref{eq:strong6j/rescaled hs}.
We will represent them by
\begin{equation}
\label{eq:strong6j/graphical isometries}
  \Phi^{\alpha\beta}_{\lambda,i} = \strongsixjCG{$\lambda$}{$\alpha$}{$\beta$}{$i$}
  \quad \text{and} \quad
  (\Phi^{\alpha\beta}_{\lambda,i})^{\dagger}= \strongsixjCGadjoint{$\lambda$}{$\alpha$}{$\beta$}{$i$}
\end{equation}
in the graphical calculus.
The maps $\Phi^{\alpha\beta}_{\lambda,i}$ are nothing but components of the Clebsch--Gordan isometries $\Phi_\lambda^{\alpha\beta}$ (\autoref{lem:strong6j/decomposition}).
It follows from \eqref{eq:strong6j/6j via clebsch gordan} that
\begin{equation*}
  \Id_{\lambda} \otimes \begin{bmatrix} \alpha & \beta & \mu \\ \gamma & \lambda & \nu \end{bmatrix}^{kl}_{ij} = \begin{tikzpicture}[baseline=0cm]
    \draw[->-=.6] (0,1.4) -- (0,1.8);
    \draw[->-=.6] (0.1,0.1) -- (0.46,0.46);
    \draw[->-=.6] (1.1,0.1) -- (0.74,0.46);
    \draw[->-=.6] (0.46,0.74) -- (0.14,1.06);
    \draw[-<-=.25] (-0.14,1.06) -- (-1.1,0.1);
    \draw[-<-=.25] (-0.1,-0.1) -- (-0.46,-0.46);
    \draw[-<-=.25] (-1.1,-0.1) -- (-0.74,-0.46);
    \draw[-<-=.25] (-0.46,-0.74) -- (-0.14,-1.06);
    \draw[->-=.6] (0.14,-1.06) -- (1.1,-0.1);
    \draw[-<-=.25] (0,-1.4) -- (0,-1.8);
    \draw[adjoint] (0.6,0.6) circle[radius= 0.2];
    \draw[adjoint] (0,1.2) circle[radius= 0.2];
    \draw (-0.6,-0.6) circle[radius= 0.2];
    \draw (0,-1.2) circle[radius= 0.2];
    \draw (0,1.2) node {$l$};
    \draw (0.6,0.6) node {$k$};
    \draw (-0.6,-0.6) node {$j$};
    \draw (0,-1.2) node {$i$};
    \draw (0,0) node {$\beta$};
    \draw (-1.2,0) node {$\alpha$};
    \draw (1.2,0) node {$\gamma$};
    \draw (-0.3,-0.9) node[below left] {$\mu$};
    \draw (0.3,0.9) node[above right] {$\nu$};
    \draw (0,1.6) node[right] {$\lambda$};
    \draw (0,-1.6) node[right] {$\lambda$};
  \end{tikzpicture}.
\end{equation*}
By taking the trace over $[\lambda]$ and deforming the above graphic, we obtain the following graphical expression for the matrix elements of the recoupling coefficients.

\begin{lem}
  \begin{equation}
  \label{eq:strong6j/6j graphical expression}
    \begin{bmatrix} \alpha & \beta & \mu \\ \gamma & \lambda & \nu \end{bmatrix}^{kl}_{ij}
    =
    \dfrac{1}{d_\lambda} \strongsixjAsymmetric
  \end{equation}
\end{lem}

Our goal is to transform the right-hand side expression in \eqref{eq:strong6j/6j graphical expression} into a form that renders its symmetries apparent.
For this, recall from \autoref{sec:onebody/kronecker} that the irreducible representations of the symmetric group are self-dual, i.e., $[\lambda] \cong [\lambda]^*$, because they can be defined over the reals.
We saw that this implied that $\Ha^{\lambda\lambda}_{\mathbf 1}$ is one-dimensional, i.e., there exists a single copy of the trivial representation $\mathbf{1}$ in each tensor product $[\lambda] \otimes [\lambda]$.
We shall denote the corresponding basis vector by
\begin{equation}
\label{eq:strong6j/epr graphical}
  \begin{tikzpicture}[baseline=-0.1cm]
    \draw[->-=.6] (-0.1,0) -- (-0.8, 0);
    \draw[->-=.6] (0.1,0) -- (0.8,0);
    \draw[epr] (0,0) circle[radius=0.1];
    \draw (-0.9,0) node[left] {$\lambda$};
    \draw (0.9,0) node[right] {$\lambda$};
  \end{tikzpicture}
  := \strongsixjCG{$\mathbf{1}$}{$\lambda$}{$\lambda$}{},
\end{equation}
omitting the leg corresponding to the identity object $\mathbf 1$ as is usual in the graphical calculus.
It can be concretely written as a maximally entangled state $\sum_i \ket{\lambda,i} \otimes \ket{\lambda,i} / {\sqrt {d_\lambda}}$ in any real orthonormal basis $\ket{\lambda,i}$ of $[\lambda]$ (i.e., in a basis such that $S_k$ acts by real orthogonal matrices).
We denote the adjoint of \eqref{eq:strong6j/epr graphical} by reversing arrows.
It is then easy to see that we have the ``teleportation identity''
\begin{equation}
\label{eq:strong6j/bell teleport graphical}
  \begin{tikzpicture}[baseline=-0.1cm]
    \draw[->-=.6] (-0.35,0) -- (0.35, 0);
    \draw[-<-=.35] (0.55,0) -- (1.25,0);
    \draw[->-=.6] (-0.55,0) -- (-1.25,0);
    \draw[adjoint epr] (0.45,0) circle[radius=0.1];
    \draw[epr] (-0.45,0) circle[radius=0.1];
    \draw (-1.35,0) node[left] {$\lambda$};
    \draw (1.35,0) node[right] {$\lambda$};
  \end{tikzpicture}
  =
  \frac 1 {d_\lambda}
  \begin{tikzpicture}[baseline=-0.1cm]
    \draw[-<-=.4] (0,0) -- (0.8,0);
    \draw (0,0) node[left] {$\lambda$};
    \draw (0.8,0) node[right] {$\lambda$};
  \end{tikzpicture}
  .
\end{equation}

We can use \eqref{eq:strong6j/epr graphical} and its adjoint to raise and lower indices, i.e., to reverse the direction of arrows. We thus obtain the following important property of the Clebsch--Gordan isometries (cf.\ \cite[(7-205a)]{Hamermesh89}):

\begin{lem}
  \label{lem:strong6j/invariant bases}
  Both sets
  \begin{equation*}
    \left\{ \strongsixjCGinvariant{$i$} \right\}
    \quad\text{and}\quad
    \left\{
      \sqrt{ \frac {d_\alpha d_\beta} {d_\lambda} }
      \strongsixjCGadjointinvariant{$i$}
    \right\}
  \end{equation*}
  form orthonormal bases of the space $([\alpha] \otimes [\beta] \otimes [\lambda])^{S_k}$ of $S_k$-invariant vectors in the triple tensor product.
\end{lem}
\begin{proof}
  Since the dimensions of $([\alpha] \otimes [\beta] \otimes [\lambda])^{S_k}$ and of $\Ha^{\alpha\beta}_\lambda$ agree by self-duality of $[\lambda]$, it suffices to show that both sets of vectors are orthonormal.
  For the first set, observe that it follows from the teleportation identity \eqref{eq:strong6j/bell teleport graphical} that
  \begin{equation*}
    \left\langle \strongsixjCGinvariant{$i$} | \strongsixjCGinvariant{$i'$} \right\rangle
    =
    \begin{tikzpicture}[baseline=0cm]
      \draw[->-=.56] (-0.2,-0.6) -- (-0.2,0.6);
      \draw[->-=.56] (0.2,-0.6) -- (0.2,0.6);

      \draw[->-=.6] (0.8,-0.3) -- (0.8,0.3);
      \draw[->-=0.4] (0.8,0.8) -- (0.8,0.5);
      \draw[-<-=0] (0.8,-0.8) -- (0.8,-0.5);
      \draw[adjoint epr] (0.8,0.4) circle[radius=0.1];
      \draw[epr] (0.8,-0.4) circle[radius=0.1];

      \draw[adjoint] (0, 0.6) circle[radius= 0.2];
      \draw (0, -0.6) circle[radius= 0.2];
      \draw (0, 0.8) arc[radius= 0.4, start angle= 180, end angle= 0];
      \draw (0, -0.8) arc[radius= 0.4, start angle= -180, end angle= 0];
      \draw (0,0.6) node {$i$};
      \draw (0,-0.6) node {$i'$};
      \draw (-0.2, 0) node[left] {$\alpha$};
      \draw (0.2, 0) node[right] {$\beta$};
      \draw (0.2, 0.9) node[right] {$\lambda$};
    \end{tikzpicture}
    =
    \frac 1 {d_\lambda}
    \begin{tikzpicture}[baseline=0cm]
      \draw[->-=.56] (-0.2,-0.6) -- (-0.2,0.6);
      \draw[->-=.56] (0.2,-0.6) -- (0.2,0.6);
      \draw[->-=.56] (0.8,0.8) -- (0.8,-0.8);
      \draw[adjoint] (0, 0.6) circle[radius= 0.2];
      \draw (0, -0.6) circle[radius= 0.2];
      \draw (0, 0.8) arc[radius= 0.4, start angle= 180, end angle= 0];
      \draw (0, -0.8) arc[radius= 0.4, start angle= -180, end angle= 0];
      \draw (0,0.6) node {$i$};
      \draw (0,-0.6) node {$i'$};
      \draw (-0.2, 0) node[left] {$\alpha$};
      \draw (0.2, 0) node[right] {$\beta$};
      \draw (0.2, 0.9) node[right] {$\lambda$};
    \end{tikzpicture}
    \;.
  \end{equation*}
  This is in turn equal to
  \begin{equation*}
    \frac {\tr (\Phi^{\alpha\beta}_{\lambda,i})^\dagger \Phi^{\alpha\beta}_{\lambda,i'}} {d_\lambda}
    = \braket{\Phi^{\alpha\beta}_{\lambda,i}, \Phi^{\alpha\beta}_{\lambda,i'}}_\lambda
    = \delta_{i,i'},
  \end{equation*}
  since $\phi_i$ is an orthonormal basis with respect to the inner product \eqref{eq:strong6j/rescaled hs}.

  For the second set, we find similarly that
  \begin{equation*}
    \left\langle \strongsixjCGadjointinvariant{$i$} | \strongsixjCGadjointinvariant{$i'$} \right\rangle
    =
    \begin{tikzpicture}[baseline=0cm]
      \draw[->-=.56] (0.8,-0.8) -- (0.8,0.8);

      \draw[->-=.85] (0.2, -0.1) -- (0.2,0.1);
      \draw[->-=.9] (0.2,0.6) -- (0.2,0.3);
      \draw[->-=.56] (0.2,-0.3) -- (0.2,-0.6);
      \draw[adjoint epr] (0.2,0.2) circle[radius=0.1];
      \draw[epr] (0.2,-0.2) circle[radius=0.1];

      \draw[->-=.85] (-0.2, -0.1) -- (-0.2,0.1);
      \draw[->-=.9] (-0.2,0.6) -- (-0.2,0.3);
      \draw[->-=.56] (-0.2,-0.3) -- (-0.2,-0.6);
      \draw[adjoint epr] (-0.2,0.2) circle[radius=0.1];
      \draw[epr] (-0.2,-0.2) circle[radius=0.1];

      \draw (0, 0.6) circle[radius= 0.2];
      \draw[adjoint] (0, -0.6) circle[radius= 0.2];
      \draw (0, 0.8) arc[radius= 0.4, start angle= 180, end angle= 0];
      \draw (0, -0.8) arc[radius= 0.4, start angle= -180, end angle= 0];
      \draw (0,0.6) node {$i$};
      \draw (0,-0.6) node {$i'$};
      \draw (-0.2, 0) node[left] {$\alpha$};
      \draw (0.2, 0) node[right] {$\beta$};
      \draw (0.2, 0.9) node[right] {$\lambda$};
    \end{tikzpicture}
    = \frac 1 {d_\alpha d_\beta}
    \begin{tikzpicture}[baseline=0cm]
      \draw[->-=.56] (0.8,-0.8) -- (0.8,0.8);
      \draw[->-=.56] (0.2,0.5) -- (0.2,-0.6);
      \draw[->-=.56] (-0.2,0.5) -- (-0.2,-0.6);
      \draw (0, 0.6) circle[radius= 0.2];
      \draw[adjoint] (0, -0.6) circle[radius= 0.2];
      \draw (0, 0.8) arc[radius= 0.4, start angle= 180, end angle= 0];
      \draw (0, -0.8) arc[radius= 0.4, start angle= -180, end angle= 0];
      \draw (0,0.6) node {$i$};
      \draw (0,-0.6) node {$i'$};
      \draw (-0.2, 0) node[left] {$\alpha$};
      \draw (0.2, 0) node[right] {$\beta$};
      \draw (0.2, 0.9) node[right] {$\lambda$};
    \end{tikzpicture}
    = \frac {d_\lambda} {d_\alpha d_\beta} \delta_{i,i'}.
  \end{equation*}
\end{proof}

We finally introduce the symmetric notation:
\begin{equation}
\label{eq:strong6j/invariant notation}
  \begin{tikzpicture}[baseline=0cm]
    \draw[->-=.56] (0,-0.2) -- (0,-0.7);
    \draw[->-=.56] (0.177,0.1) -- (0.766, 0.5);
    \draw[->-=.56] (-0.177,0.1) -- (-0.766,0.5);
    \draw (0,0) circle[radius=0.2];
    \draw (0,0) node {$i$};
    \draw (-0.766,0.5) node[above] {$\alpha$};
    \draw (0.766, 0.5) node[above] {$\beta$};
    \draw (0,-0.7) node[below] {$\lambda$};
  \end{tikzpicture}
  :=
  \strongsixjCGinvariant{$i$}
\end{equation}
We note that the vectors \eqref{eq:strong6j/invariant notation} depend on the choice of arrow that was reversed.
However, by \autoref{lem:strong6j/invariant bases} any such choice gives rise to unitarily equivalent bases of the space of $S_k$-invariants!
We thus obtain the following result:

\begin{prp}
  \label{prp:strong6j/invariant object}
  We have
  \begin{equation}
  \label{eq:strong6j/invariant norm}
    \frac 1 {d_\mu d_\nu} \norm{\sixj \alpha \beta \gamma \mu \nu \lambda}_{\text{HS}}^2
  = d_\alpha d_\beta d_\gamma d_\lambda d_\mu d_\nu \, \norm{\strongsixjSymmetric}^2_2,
  \end{equation}
  where $\norm{(x_{i,j,k,l})}_2 := \sqrt{\sum_{i,j,k,l} \abs{x_{i,j,k,l}}^2}$ denotes the $\ell_2$-norm of a tensor with indices $i,j,k,l$.
  The right-hand side does not depend on the choice of arrow that was reversed in the definition \eqref{eq:strong6j/invariant notation}.
\end{prp}
\begin{proof}
  Recall from \eqref{eq:strong6j/6j graphical expression} that
  \begin{equation*}
      \normHS{\sixj \alpha \beta \gamma \mu \nu \lambda}^2
    = \frac 1 {d_\lambda^2} \norm{\strongsixjAsymmetric}^2_2.
  \end{equation*}
  By inserting the teleportation identity \eqref{eq:strong6j/bell teleport graphical} once for each of the six arrows, we obtain
  \begin{equation*}
    \frac 1 {d_\lambda^2} d_\alpha^2 d_\beta^2 d_\gamma^2 d_\mu^2 d_\nu^2 d_\lambda^2 \norm{\strongsixjAsymmetricTeleported}^2_2.
  \end{equation*}
  By first applying the unitary transformation that relates the second orthonormal basis in \autoref{lem:strong6j/invariant bases} to the first (at the vertices \circled{$k$} and \circled{$l$}) and then using definition \eqref{eq:strong6j/invariant notation} (at all four vertices), this is in turn equal to
  \begin{align*}
    \frac 1 {d_\lambda^2} d_\alpha^2 d_\beta^2 d_\gamma^2 d_\mu^2 d_\nu^2 d_\lambda^2 \frac {d_\lambda} {d_\alpha d_\nu} \frac {d_\nu} {d_\beta d_\gamma} \norm{\strongsixjSymmetric}^2_2,
  \end{align*}
  as the $\ell_2$-norm is unitarily invariant.
  From this we obtain the desired expression by deforming the diagram and simplifying the prefactor.
  The unitary invariance of the $\ell_2$-norm also shows that the expression is insensitive to the choice of which arrow is reversed in \eqref{eq:strong6j/invariant notation}.
\end{proof}

The right-hand side of \eqref{eq:strong6j/invariant norm} is the symmetric group analogue of \emphindex{Wigner's $6j$-symbol}, which can be obtained in the same way from the recoupling coefficients of $\SU(2)$.
It is immediately apparent from the graphical calculus that it has the symmetries of a tetrahedron.
We remark that for our purposes it was important to study the recoupling coefficients in \autoref{sec:strong6j/overlapping marginals}, since the dimensions of irreducible $S_k$-representations grow exponentially with $k$ and thus affect the asymptotics.
We record the following consequence, which has a well-known counterpart for $\SU(2)$; cf.\ \cite[(B4)]{LevinWen05}.

\begin{cor}
  \label{cor:strong6j/symmetry}
  The quantities
  \begin{equation}
    \frac 1 {d_\mu d_\nu} \norm{\sixj \alpha \beta \gamma \mu \nu \lambda}_{\text{HS}}^2
  \end{equation}
  are invariant under exchanging the columns $(\beta,\lambda) \leftrightarrow (\mu,\nu)$ as well as $(\alpha,\gamma) \leftrightarrow (\mu,\nu)$.
\end{cor}
\begin{proof}
  This is an immediate consequence of \autoref{prp:strong6j/invariant object}, since the right-hand side norm in \eqref{eq:strong6j/invariant norm} is invariant under reflection of the diagram by the axes through the edges labeled by $\alpha$ and $\beta$, respectively.
\end{proof}

\subsection*{Algebraic Proof}

We now give an alternative, algebraic proof of \autoref{prp:strong6j/invariant object} and \autoref{cor:strong6j/symmetry} that follows along the same lines as the graphical proof.

In quantum information theory, \emph{maximally entangled states}\index{maximally entangled states|textbf}\nomenclature[QketPsi^+_H]{$\ket{\Psi^+_\calH}$}{maximally entangled state on $\calH \otimes \cal H$} on a Hilbert space $\calH \otimes \calH$ are defined by the formula
\begin{equation}
\label{eq:strong6j/epr}
  \ket{\Psi^+_\calH} := \frac 1 {\sqrt{\dim \mathcal H}} \sum_i \ket i \otimes \ket i
\end{equation}
with respect to an orthonormal basis $\ket i$.
They satisfy the fundamental identity
\begin{equation}
  \label{eq:strong6j/epr transpose}
  (X \otimes \Id) \ket{\Psi^+_\calH} = (\Id \otimes X^T) \ket{\Psi^+_\calH}
\end{equation}
for any operator $X$ on $\calH$, where $X^T$ denotes the transpose in the basis $\ket i$.
Thus they are invariant under operations of the form $U \otimes \overline{U}$, where $U \in \U(\calH)$ is a unitary and where $\overline{U}$ denotes its complex conjugate with respect to the basis $\ket i$ \cite{HorodeckiHorodecki99}.
In particular, this implies that for any basis $\ket{\lambda,i}$ of $[\lambda]$ in which $S_k$ acts by orthogonal transformations,
\begin{equation*}
  \ket{\Psi^+_\lambda} := \frac 1 {\sqrt{\dim [\lambda]}} \sum_j \ket{\lambda,j} \otimes \ket{\lambda,j}
\end{equation*}
is the (unique up to phase) invariant vector in $[\lambda] \otimes [\lambda]$---as we had asserted before (cf.\ \eqref{eq:strong6j/epr graphical}).\nomenclature[RPsi^+_lambda]{$\ket{\Psi^+_\lambda}$}{$S_k$-invariant vector in $[\lambda] \otimes [\lambda]$}
By using \eqref{eq:strong6j/epr transpose} it is straightforward to verify that the following two well-known properties hold:
\begin{itemize}
\item We have the ``teleportation identity''
  \begin{equation}
  \label{eq:strong6j/bell teleport}
    \left( \bra{\Psi^+_\calH} \otimes \Id_{\mathcal H} \right)
   \left( \Id_{\mathcal H} \otimes \ket {\Psi^+_\calH} \right)
   = \frac 1 {\dim \mathcal H} \Id_{\mathcal H}.
  \end{equation}
  This is the algebraic version of \eqref{eq:strong6j/bell teleport graphical}.
  It follows that for any two operators $X \colon \calK \rightarrow \calK' \otimes \calH$ and $Y \colon \calH \otimes \calL \rightarrow \calL'$ we have the relation
  \begin{equation}
  \label{eq:strong6j/bell teleport operators}
  \begin{aligned}
      &\left( \Id_{\calK'} \otimes \bra{\Psi^+_\calH} \otimes \Id_{\calL'} \right)
      \left( X \otimes \Id_\calH \otimes Y \right)
      \left( \Id_\calK \otimes \ket{\Psi^+_\calH} \otimes \Id_\calL \right) \\
    = &\frac 1 {\dim \calH} (\Id_{\calK'} \otimes Y) (X \otimes \Id_\calL)
  \end{aligned}
  \end{equation}
  (which is best understood graphically).
\item The normalized trace of any operator $X$ can be written as
  \begin{equation}
  \label{eq:strong6j/bell trace}
    \braket{\Psi^+_\calH | X \otimes \Id_{\mathcal H} | \Psi^+_\calH} =
    \frac 1 {\dim \mathcal H} \tr X.
  \end{equation}
\end{itemize}
For any $\alpha$, $\beta$ and $\lambda$, we shall consider the following sets of vectors in $([\alpha] \otimes [\beta] \otimes [\lambda])^{S_k}$,
\begin{align}
  \ket{\alpha\beta\lambda,i} &:= (\Phi^{\alpha\beta}_{\lambda,i} \otimes \Id_\lambda) \ket{\Psi^+_\lambda}, \nonumber \\
  \ket{\widetilde{\alpha\beta\lambda,i}} &:= \sqrt {\frac {d_\alpha d_\beta} {d_\lambda}} (\Phi^{\alpha\beta}_{\lambda,i} \otimes \Id_\alpha \otimes \Id_\beta)^\dagger (\ket{\Psi^+_\alpha} \otimes \ket{\Psi^+_\beta}), \label{eq:strong6j/alternative basis}
\end{align}
constructed as in \autoref{lem:strong6j/invariant bases}.
We now prove algebraically that each set forms an orthonormal basis.
For the first,
\begin{align*}
    &\braket{\alpha\beta\lambda,i | \alpha\beta\lambda,i'}
  = \braket{\Psi^+_\lambda | (\Phi^{\alpha\beta}_{\lambda,i})^\dagger \Phi^{\alpha\beta}_{\lambda,i'} \otimes \Id_\lambda | \Psi^+_\lambda}
  = \frac 1 {d_\lambda} \tr (\Phi^{\alpha\beta}_{\lambda,i})^\dagger \Phi^{\alpha\beta}_{\lambda,i'} \\
  = &\braket{\Phi^{\alpha\beta}_{\lambda,i}, \Phi^{\alpha\beta}_{\lambda,i'}}_\lambda
  = \delta_{i,i'}
\end{align*}
by \eqref{eq:strong6j/bell trace} and the definition of the inner product \eqref{eq:strong6j/rescaled hs}.
For the second set of vectors,
\begin{align*}
    &\braket{\widetilde{\alpha\beta\lambda,i} | \widetilde{\alpha\beta\lambda,i'}}
  = \frac {d_\alpha d_\beta} {d_\lambda}
    \braket{\Psi^+_\alpha \otimes \Psi^+_\beta |
      \Phi^{\alpha\beta}_{\lambda,i} (\Phi^{\alpha\beta}_{\lambda,i'})^\dagger \otimes \Id_\alpha \otimes \Id_\beta |
      \Psi^+_\alpha \otimes \Psi^+_\beta} \\
  = &\frac 1 {d_\lambda} \tr \Phi^{\alpha\beta}_{\lambda,i} (\Phi^{\alpha\beta}_{\lambda,i'})^\dagger
  = \braket{\Phi^{\alpha\beta}_{\lambda,i'}, \Phi^{\alpha\beta}_{\lambda,i}}_\lambda
  = \delta_{i',i}.
\end{align*}

We now consider the recoupling coefficients. First, \eqref{eq:strong6j/6j via clebsch gordan} and \eqref{eq:strong6j/bell trace} give
\begin{align*}
    &\sixj\alpha \beta\gamma\mu\nu\lambda_{ij}^{kl}
  = \frac 1 {d_\lambda} \tr
    (\Phi^{\alpha\nu}_{\lambda,k})^\dagger \left( \Id_\alpha \otimes \Phi^{\beta\gamma}_{\nu,l} \right)^\dagger
    \left( \Phi^{\alpha\beta}_{\mu,j} \otimes \Id_{\gamma} \right) \Phi^{\mu\gamma}_{\lambda,i} \\
  = &\bra{\Psi^+_\lambda}
     \left( (\Phi^{\alpha\nu}_{\lambda,k})^\dagger \otimes \Id_\lambda \right)
     \left( \Id_\alpha \otimes \Phi^{\beta\gamma}_{\nu,l} \otimes \Id_\lambda \right)^\dagger
     \left( \Phi^{\alpha\beta}_{\mu,j} \otimes \Id_\gamma \otimes \Id_\lambda \right)
     \underbrace{\left( \Phi^{\mu\gamma}_{\lambda,i} \otimes \Id_\lambda \right)
     \ket{\Psi^+_\lambda}}_{=\ket{\mu\gamma\lambda,i}}.
\end{align*}
We may now apply \eqref{eq:strong6j/bell teleport operators} to $X = \ket{\mu\gamma\lambda,i}$ and $Y = \Phi^{\alpha\beta}_{\mu,j}$ in order to rewrite %
\begin{align*}
    &\left( \Phi^{\alpha\beta}_{\mu,j} \otimes \Id_\gamma \otimes \Id_\lambda \right) \ket{\mu\gamma\lambda,i} \\
  = &d_\mu \left( \Id_\gamma \otimes \Id_\lambda \otimes \bra{\Psi^+_\mu} \otimes \Id_\alpha \otimes \Id_\beta \right)
    \left( \ket{\mu\gamma\lambda,i} \otimes \Id_\mu \otimes \Phi^{\alpha\beta}_{\mu,j} \right)
    \ket{\Psi^+_\mu} \\
  = &d_\mu \left( \Id_\gamma \otimes \Id_\lambda \otimes \bra{\Psi^+_\mu} \otimes \Id_\alpha \otimes \Id_\beta \right)
    \left( \ket{\mu\gamma\lambda,i} \otimes \ket{\alpha\beta\mu,j} \right).
\end{align*}
Continuing in this way and using definition \eqref{eq:strong6j/alternative basis}, we obtain the following expression for the matrix elements of the recoupling coefficient:
\begin{align*}
  \sixj\alpha \beta\gamma\mu\nu\lambda_{ij}^{kl}
  = &d_\mu d_\alpha d_\beta d_\gamma d_\nu \sqrt{\frac {d_\lambda} {d_\alpha d_\nu}} \sqrt{\frac {d_\nu} {d_\beta d_\gamma}} \\
  \times &\left( \bra{\Psi^+_\alpha} \otimes \bra{\Psi^+_\beta} \otimes \bra{\Psi^+_\gamma} \otimes \bra{\Psi^+_\mu} \otimes \bra{\Psi^+_\nu} \otimes \bra{\Psi^+_\lambda} \right) \\
  &\left( \ket{\mu\gamma\lambda,i} \otimes \ket{\alpha\beta\mu,j} \otimes \ket{\widetilde{\alpha\nu\lambda,k}} \otimes \ket{\widetilde{\beta\gamma\nu,l}} \right)
\end{align*}
The sum of their absolute values squared over all indices $i$, $j$, $k$ and $l$ is equal to
\begin{equation}
\label{eq:strong6j/algebraic recoupling norm squared}
\begin{aligned}
    \frac 1 {d_\mu d_\nu} \normHS{\sixj\alpha\beta\gamma\mu\nu\lambda}^2
  =\ &d_\alpha d_\beta d_\gamma d_\lambda d_\mu d_\nu \\
  \times &\tr \left( P^{\alpha\alpha} \otimes P^{\beta\beta} \otimes P^{\gamma\gamma} \otimes P^{\mu\mu} \otimes P^{\nu\nu} \otimes P^{\lambda\lambda} \right) \\
  &\quad \left( P^{\mu\gamma\lambda} \otimes P^{\alpha\beta\mu} \otimes P^{\alpha\nu\lambda} \otimes P^{\beta\gamma\nu} \right),
\end{aligned}
\end{equation}where $P^{\alpha\beta\mu}$ denotes the orthogonal projection onto $([\alpha] \otimes [\beta] \otimes [\mu])^{S_k}$,
$P^{\alpha\alpha} = \proj{\Psi^+_\alpha}$, etc.
Equation~\eqref{eq:strong6j/algebraic recoupling norm squared} is the algebraic analogue of \autoref{prp:strong6j/invariant object}.
As before, \autoref{cor:strong6j/symmetry} is a direct consequence of its symmetries.

\section{Application: Strong Subadditivity of the von Neumann Entropy}
\label{sec:strong6j/ssa}

We now prove the strong subadditivity and weak monotonicity of the von Neumann entropy as direct consequences of \autoref{thm:stabs/main theorem} and the symmetry properties of the recoupling coefficients.
We refer to \autoref{sec:strong6j/discussion} for a discussion of the general technique in the context of the search for new entropy inequalities.
We start by noting that it follows from the first invariance asserted in \autoref{cor:strong6j/symmetry} and the polynomial upper bound \eqref{eq:strong6j/6j lower upper} that
\begin{equation*}
  \normHS{\sixj \alpha \beta \gamma \mu \nu \lambda}^2
  = \dfrac{d_\mu d_\nu} {d_\beta d_\lambda}
  \normHS{\sixj \alpha \mu \gamma \beta \lambda \nu}^2
  \leq \poly(k) \dfrac{d_\mu d_\nu} {d_\beta d_\lambda}.
\end{equation*}
Hence, if $\rho_{ABC}$ is a tripartite quantum state then \autoref{thm:strong6j/main theorem} implies that
\begin{equation}
  \label{eq:strong6j/subadd ratio cor}
  \dfrac{d_\mu d_\nu} {d_\beta d_\lambda}
  \geq
  \frac 1 {\mathrm{poly}(k)}
\end{equation}
for sequences of normalized Young diagrams that converge to the respective spectra of the reduced density matrices.
Since for large $k$,  $\frac 1 {k} \log_2 \dim [\lambda] \rightarrow H(\bar\lambda) = \sum_i -\bar\lambda_i \log_2 \bar\lambda_i$ \cite{ChristandlMitchison06}, we conclude that the von Neumann entropy is strongly subadditive:

\begin{thm}[\cite{LiebRuskai73}]
  For all quantum states $\rho_{ABC}$,
  \begin{equation*}
  S(\rho_{AB}) + S(\rho_{BC}) \geq S(\rho_{B}) + S(\rho_{ABC}).
  \end{equation*}
\end{thm}

For $[\beta]$ the trivial representation, this proof of strong subadditivity reduces to the proof of subadditivity given in \cite{ChristandlMitchison06}.
Weak monotonicity,
\[S(\rho_{AB}) + S(\rho_{BC}) \geq S(\rho_A) + S(\rho_C),\]
follows similarly by swapping the columns $(\alpha, \gamma) \leftrightarrow (\mu, \nu)$ in accordance with \autoref{cor:strong6j/symmetry}.

\section{Semiclassicality}
\label{sec:strong6j/semiclassicality}

In \cite{WignerGriffin59}, Wigner studied the asymptotics of the recoupling coefficients of $\SU(2)$ which can be defined in complete analogy to \autoref{dfn:strong6j/recoupling}. Given three particles of spin $j_A$, $j_B$, $j_C$ such that the total spin of the first two particles is $j_{AB}$ and of all three particles $j_{ABC}$, the absolute value squared of the $\SU(2)$ recoupling coefficient can be interpreted as the probability of observing that particles two and three have total spin $j_{BC}$. In the semiclassical limit of simultaneously large spins, Wigner showed that this probability oscillates around the inverse volume of the tetrahedron whose edges have length equal to the six spins---if such a tetrahedron exists (\autoref{fig:strong6j/tetrahedron}). In particular, it then decays polynomially with $j$. If no such tetrahedron exists then the $6j$-symbol decays exponentially.
This result is understood to mean that ``classical'' configurations are exponentially more likely than all others in the limit of large quantum numbers.
A more precise formula has been given by Ponzano and Regge \cite{PonzanoRegge68} and only fully proved in \cite{Roberts99}.
We remark that the labeling of Wigner's tetrahedron in \autoref{fig:strong6j/tetrahedron} is dual to the $\SU(2)$ analogue of the diagram \eqref{eq:strong6j/6j graphical expression} \cite{Roberts99}.

The recoupling coefficients for the symmetric groups $S_k$ that we consider in this chapter can also be defined in terms of unitary groups. This follows from Schur--Weyl duality \eqref{eq:strong6j/schur-weyl}, which implies that the projectors $P_\alpha$, $P_\beta$, etc.\ in \eqref{eq:strong6j/PQ vs 6j} can be equivalently defined as the isotypical projectors for the unitary groups $\SU(a)$, $\SU(b)$, etc.
From this perspective, the decompositions in \eqref{eq:strong6j/tripartite schur-weyl} arise by restricting the $\SU(abc)$-representation $(\CC^{abc})^{\otimes k}$ to $\SU(a) \times \SU(b) \times \SU(c)$ via either of the ``intermediate subgroups'' $\SU(ab) \times \SU(c)$ or $\SU(a) \times \SU(bc)$ (cf.\ \autoref{fig:strong6j/unitary}). Thus, it is suggestive to consider the number of boxes $k$ in the Young diagrams as the semiclassical parameter in our setup. In this sense, tripartite quantum states $\rho_{ABC}$ are the formal analogues of Wigner's tetrahedra---they are the geometric objects that describe the ``classical'' configurations, corresponding to polynomial decay in the limit $k \rightarrow \infty$.

\begin{figure}
  \begin{center}
    \includegraphics[width=0.9\linewidth]{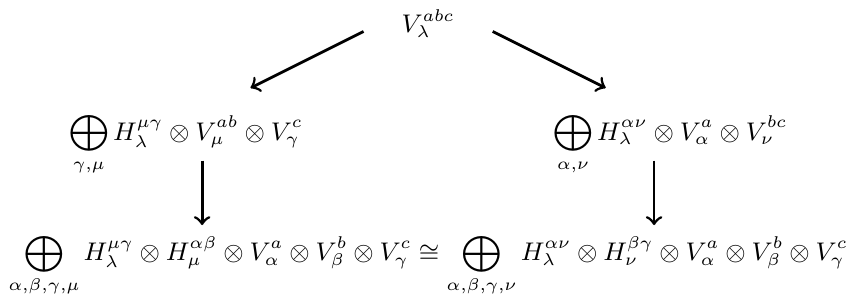}
  \end{center}
  \caption[Recoupling coefficients via unitary groups]{%
  \emph{``Dual'' definition of the recoupling coefficients of the symmetric group in terms of unitary groups.}
  An irreducible $\SU(abc)$-representation $V^{abc}_\lambda$ can be restricted along either of the chains of homomorphisms $\SU(a) \times \SU(b) \times \SU(c) \rightarrow \SU(ab) \times \SU(c) \rightarrow \SU(abc)$ and $\SU(a) \times \SU(b) \times \SU(c) \rightarrow \SU(a) \times \SU(bc) \rightarrow \SU(abc)$, resulting in two isomorphic decompositions. %
  The recoupling coefficients of the symmetric group can be obtained from the components of this isomorphism in the same way as in \eqref{eq:strong6j/6j decomposition}--\eqref{eq:strong6j/6j via clebsch gordan}.
  }
  \label{fig:strong6j/unitary}
\end{figure}

\section{Sums of Matrices and Quantum Marginals}%
\label{sec:strong6j/sums of matrices}

Wigner's result and our \autoref{thm:strong6j/main theorem} are in fact closely related on a precise mathematical level. Before we describe the construction, we note that the existence of the tetrahedron is equivalent to the existence of vectors $\vec j_1$, $\vec j_2$, $\vec j_3$ such that $\abs{\vec j_1} = j_1$, $\abs{\vec j_2} = j_2$, $\abs{\vec j_3} = j_3$, $\abs{\vec j_1 + \vec j_2} = j_{12}$, $\abs{\vec j_2 + \vec j_3} = j_{23}$, and $\abs{\vec j_1 + \vec j_2 + \vec j_3} = j_{123}$ (\autoref{fig:strong6j/tetrahedron}).
By the usual identification of $\RR^3$ with the Lie algebra $\mathfrak{su}_2$, we may assign to each vector $\vec j \in \RR^3$ the Hermitian $2 \times 2$-matrix $\vec j \cdot \vec \sigma$, where $\vec\sigma = (\sigma_x, \sigma_y, \sigma_z)$ is the vector of Pauli matrices. Then the above becomes an instance of the following general problem \cite{Backens10}: %

\begin{pro}[Partial Sums of Matrices]\index{sums of matrices!partial}
\label{pro:strong6j/generalized weyl}
  Do there exist Hermitian $d \times d$-matrices $A$, $B$ and $C$ with given prescribed eigenvalues for $A$, $B$, $C$, $A+B$, $B+C$ and $A+B+C$?
\end{pro}

This is a natural generalization of the problem of determining the relation between the eigenvalues of $A$, $B$ and $A+B$ that goes back at least to Weyl (\autoref{sec:onebody/kronecker}).
In \cite{Klyachko04}, it was shown how the one-body quantum marginal problem degenerates to Weyl's problem in an appropriate limit (cf.\ \cite{Ruskai07} for another connection in the context of the $N$-representability problem).
We will now show that \autoref{pro:strong6j/generalized weyl} can similarly be considered as a special case of the quantum marginal problem with overlapping marginals as discussed in this chapter---both on the level of geometry and on the level of representation theory.

\bigskip

Let $A$, $B$, $C$ be Hermitian $d \times d$-matrices.
Without loss of generality, we may assume that $A, B, C \geq 0$ and that $1 - \tr (A + B + C) \geq \norm{A+B+C}_\infty$ (else, we may add suitable multiples of the identity and rescale).
Generalizing a construction from \cite{Christandl08}, we define a tripartite quantum state $\rho_{ABC}$ in terms of its purification
\begin{equation}
\label{eq:strong6j/purification embedding}
\begin{aligned}
  \ket{\psi}_{ABCD}
=\;&\sum_{i=1}^d \sqrt{A} \ket i_A \otimes \ket{00i}_{BCD}
+ \sum_{j=1}^d \sqrt{B} \ket j_B \otimes \ket{00j}_{ACD} \\
+\;&\sum_{k=1}^d \sqrt{C} \ket k_C \otimes \ket{00k}_{ABD}
+ \sqrt{1 - \tr (A + B + C)} \ket{0000}_{ABCD}
\end{aligned}
\end{equation}
in $(\CC^{d+1})^{\otimes 4}$, where we consider $A$, $B$ and $C$ as acting on $\CC^d$, and $\CC^{d+1} = \CC \ket 0 \oplus \CC^d$.
We record the following properties:

\begin{lem}
\label{lem:strong6j/purification embedding}
  Let $\rho_{ABC}$ be the quantum state with purification \eqref{eq:strong6j/purification embedding}.
  Then the non-zero eigenvalues of $\rho_{ABC}$ and \emph{all} its reduced density matrices are given by
  \begin{align*}
    \spec \rho_{ABC} &= (1 - \tr (A+B+C), \spec (A+B+C)), \\
    \spec \rho_{AB} &= (1 - \tr (A+B), \spec (A+B)), \\
    \spec \rho_{A} &= (1 - \tr A, \spec A),
  \end{align*}
  etc.
\end{lem}
\begin{proof}
  Observe that $\ket{\psi_{ABCD}}$ is built from a sum of (unnormalized) maximally entangled states \eqref{eq:strong6j/epr} on $AD$, $BD$ and $CD$, respectively.
  By using \eqref{eq:strong6j/epr transpose} and the orthogonality properties of the construction \eqref{eq:strong6j/purification embedding}, we thus find that
  \begin{equation*}
    \rho_D = \left( 1-\tr(A+B+C) \right) \proj 0 + \left( \overline{A+B+C} \right).
  \end{equation*}
  Thus the density matrix $\rho_D$ is block-diagonal with respect to $\CC \proj 0 \oplus \CC^d$.
  These blocks can be jointly diagonalized, and our assumption $1 - \tr(A+B+C) \geq \norm{A+B+C}_\infty$ implies that $1 - \tr(A+B+C)$ is the largest eigenvalue.
  Since moreover the spectrum of a Hermitian matrix is invariant under conjugation, this shows the first claim, as $\rho_{ABC}$ and $\rho_D$ have the same non-zero eigenvalues.

  If we only trace out the first two systems, then we instead get a block decomposition of the form
  \begin{equation*}
    \rho_{CD} = \proj 0_C \otimes \left( \overline{A+B} \right) + \proj\phi_{CD},
  \end{equation*}
  where $\ket\phi_{CD} = \sum_{k=1}^d \sqrt C \ket k_C \otimes \ket k_D + \sqrt{1 - \tr (A + B + C)} \ket{00}_{CD}$.
  Using \eqref{eq:strong6j/bell trace}, we find that $\braket{\phi_{CD} | \phi_{CD}} =  1 - \tr(A+B)$, so that the second claim follows as above.
  The last claim follows from
  \begin{align*}
    \rho_A
    = &\sum_i \sqrt A \proj i_A \sqrt A
    \;+\; \big( \tr B + \tr C + \left( 1 - \tr (A + B + C) \right) \big) \proj 0_A \\
    = &A + \left( 1 - \tr A \right) \proj 0_A,
  \end{align*}
  which is established similarly.
  All other marginal spectra can be computed in the same way.
\end{proof}

We have thus obtained an embedding of triples of matrices into the space of tripartite quantum states that preserves the eigenvalue information.
We remark that \autoref{lem:strong6j/purification embedding} can be used to obtain entropy inequalities for sums of Hermitian matrices from entropy inequalities for multipartite quantum states.

\begin{cor}
  The state $\rho_{ABC}$ has rank at most $d+1$ and it satisfies the polygonal inequality \eqref{eq:kirwan/polygonal} with equality:
  \begin{equation}
  \label{eq:strong6j/polygonal equality}
    r_{A,1} + r_{B,1} + r_{C,1} = 2 + r_{ABC,1},
  \end{equation}
  where $r_{I,1}$ denotes the maximal eigenvalue of the reduced density matrix $\rho_I$.
\end{cor}

Note that \eqref{eq:strong6j/polygonal equality} implies that the polygonal inequalities for $AB:C$, $A:BC$ and $AC:B$ are likewise satisfied with equality (i.e., $r_{AB,1} + r_{C,1} = 1 + r_{ABC,1}$, etc.).
We now show the following converse statement.

\begin{prp}
\label{prp:strong6j/embedding converse}
  Let $\rho_{ABC}$ be a tripartite quantum state on $(\CC^{d+1})^{\otimes 3}$ of rank at most $d+1$ that satisfies the equality \eqref{eq:strong6j/polygonal equality}. Then $\rho_{ABC}$ can up to local unitaries be purified in the form \eqref{eq:strong6j/purification embedding}.
\end{prp}
\begin{proof}
  For this, we recall the following proof of the polygonal inequality:
  Let $\ket 0_A$, $\ket 0_B$, and $\ket 0_C$ denote maximal eigenvectors of $\rho_A$, $\rho_B$ and $\rho_C$, respectively.
  Let $P_A$, $P_B$ and $P_C$ denote the corresponding orthogonal projectors, and set $P_A^\perp := \Id - P_A$, etc.
  Then,
  \begin{align*}
      &r_{A,1} + r_{B,1} + r_{C,1} \\
    =\;&\tr \rho_{ABC} (P_A \otimes \Id_{BC} + P_B \otimes \Id_{AC} + P_C \otimes \Id_{AB}) \\
    =\;&2 \tr \rho_{ABC} \Id_{ABC}
    + \tr \rho_{ABC} (P_A \otimes P_B \otimes P_C) \\
    -\;&2 \tr \rho_{ABC} (P_A^\perp \otimes P_B^\perp \otimes P_C^\perp)
    - \tr \rho_{ABC} (P_A \otimes P_B^\perp \otimes P_C^\perp) \\
    -\;&\tr \rho_{ABC} (P_A^\perp \otimes P_B \otimes P_C^\perp)
    - \tr \rho_{ABC} (P_A^\perp \otimes P_B^\perp \otimes P_C) \\
    \leq\;&2 + \braket{000 | \rho_{ABC} | 000} \leq 2 + r_{ABC,1}.
  \end{align*}
  The first inequality is obtained by omitting the terms with negative signs, and the second by using the variational principle for the maximal eigenvalue of $\rho_{ABC}$.
  It is thus immediate that we have equality if and only if $\ket{000}_{ABC}$ is a maximal eigenvector of $\rho_{ABC}$ and
  \begin{equation}
  \label{eq:strong6j/negative sign projectors}
  \begin{aligned}
    &\tr \rho_{ABC} (P_A^\perp \otimes P_B^\perp \otimes P_C^\perp) =
    \tr \rho_{ABC} (P_A^\perp \otimes P_B^\perp \otimes P_C) \\
    =\;&\tr \rho_{ABC} (P_A^\perp \otimes P_B \otimes P_C^\perp) =
    \tr \rho_{ABC} (P_A \otimes P_B^\perp \otimes P_C^\perp) = 0.
  \end{aligned}
  \end{equation}
  Let us now assume that this is the case. Since the rank of $\rho_{ABC}$ was assumed to be at most $d+1$, we can find a purification on $(\CC^{d+1})^{\otimes 4}$.
  Since $\ket{000}_{ABC}$ is a maximal eigenvector, we can arrange for the first term of the Schmidt decomposition \eqref{eq:onebody/schmidt} to be $\sqrt{r_{ABC,1}} \ket{000}_A \otimes \ket{0}_D$. In other words, the purification can be chosen of the form
  \begin{equation*}
    \ket{\psi_{ABCD}} = \sqrt{r_{ABC,1}} \ket{0000}_{ABCD} + \sum_{i,j,k,l} \psi_{ijkl} \ket{ijkl}_{ABCD}.
  \end{equation*}
  A priori, the right-hand side can run over all indices $(i,j,k) \neq (0,0,0)$ and $l \neq 0$ by orthogonality of the bases in the Schmidt decomposition.
  But \eqref{eq:strong6j/negative sign projectors} implies that in fact precisely two out of the three indices $(i,j,k)$ have to be zero, so that we obtain
  \begin{align*}
    \ket{\psi_{ABCD}} = \sqrt{r_{ABC,1}} \ket{0000}_{ABCD}
    + \sum_{i,l=1}^d \psi_{i00l} \ket{i00l}_{ABCD} \\
    + \sum_{j,l=1}^d \psi_{0j0l} \ket{0j0l}_{ABCD}
    + \sum_{k,l=1}^d \psi_{00kl} \ket{00kl}_{ABCD}.
  \end{align*}
  Thus we may define $d \times d$-matrices $X_A$, $X_B$ and $X_C$ such that
  \begin{align*}
    \ket{\psi_{ABCD}} = \sqrt{r_{ABC,1}} \ket{0000}_{ABCD}
    + \sum_i X_A \ket i_A \otimes \ket{00i}_{BCD} \\
    + \sum_j X_B \ket j_B \otimes \ket{00j}_{ACD}
    + \sum_k X_C \ket k_C \otimes \ket{00k}_{ABD}.
  \end{align*}
  Finally, we use the polar decomposition to write $X_A = U_A \abs{X_A}$, etc., and set $\sqrt A := \abs{X_A}$, etc. Then \eqref{eq:strong6j/purification embedding} is indeed a purification of the quantum state
  $(U^\dagger_A \otimes U^\dagger_B \otimes U^\dagger_C) \rho_{ABC} (U_A \otimes U_B \otimes U_C)$, which is locally unitarily equivalent to $\rho_{ABC}$.
\end{proof}

The following theorem shows that \autoref{pro:strong6j/generalized weyl} -- in particular, the existence of Wigner's tetrahedra -- is in a precise mathematical sense a special case of the quantum marginal problem with overlapping marginals covered by \autoref{thm:strong6j/main theorem}.
This generalizes the corresponding result for the one-body quantum marginal problem in \cite[\S 6.2]{Klyachko04}, and in particular gives a geometric proof of the latter.

\begin{thm}
\label{thm:strong6j/weyl wigner embedding}
  Let $\vecs_A$, $\vecs_B$, $\vecs_C$, $\vecs_{A+B}$, $\vecs_{B+C}$, $\vecs_{A+B+C}$ be vectors in $\RR^d_{\geq 0}$ with weakly decreasing entries.
  Assume that $\norm{\vecs_A}_1 + \norm{\vecs_B}_1 = \norm{\vecs_{A+B}}_1$, etc., and that $1 - \norm{\vecs_{A+B+C}}_1 \geq s_{A+B+C,1}$, etc.
  Then the following are equivalent:
  \begin{enumerate}
    \item[(1)] There exist Hermitian $d \times d$-matrices $A$, $B$, and $C$ with
    $\spec (A+B+C) = \vecs_{A+B+C}$,
    $\spec (A+B) = \vecs_{A+B}$,
    $\spec A = \vecs_A$,
    etc.\ as their partial sums.
    \item[(2)] There exists a quantum state $\rho_{ABC}$ on $(\CC^{d+1})^{\otimes 3}$ with non-zero eigenvalues
    $\spec \rho_{ABC} = (1-\norm{\vecs_{A+B+C}}_1, \vecs_{A+B+C})$,
    $\spec \rho_{AB} = (1-\norm{\vecs_{A+B}}_1, \vecs_{A+B})$,
    $\spec \rho_{A} = (1-\norm{\vecs_{A}}_1, \vecs_{A})$,
    etc.\ for their reduced density matrices.
  \end{enumerate}
\end{thm}
\begin{proof}
  $(1) \Rightarrow (2)$ is the content of \autoref{lem:strong6j/purification embedding}.
  For $(2) \Rightarrow (1)$, \autoref{prp:strong6j/embedding converse} implies that there exist Hermitian $d \times d$-matrices $A$, $B$, $C$
  such that $\rho_{ABC}$ is locally unitarily equivalent to the state $\rho'_{ABC}$ with purification \eqref{eq:strong6j/purification embedding}.
  Since the spectra of $\rho_{ABC}$ and its reduced density matrices are left invariant by local unitaries, \autoref{lem:strong6j/purification embedding} implies that the partial sums of these matrices $A$, $B$ and $C$ have the desired spectra.
\end{proof}

Similar statements can be proved for all marginal spectra (i.e., including $\vecs_{A+C}$, since \autoref{lem:strong6j/purification embedding} holds for all reduced density matrices) as well as for an arbitrary number of summands. Thus the quantum marginal problem with overlaps is a precise generalization of the problem of characterizing the eigenvalues of partial sums of Hermitian matrices.

\subsection*{Recoupling Coefficients of the Unitary and Symmetric Groups}%

We now show an analogous statement to \autoref{thm:strong6j/weyl wigner embedding} on the level of representation theory---namely, that the recoupling coefficients of the unitary group can be obtained as special recoupling coefficients of the symmetric group.

To see this, let $\lambda$ be a Young diagram.
In \cite{Nishiyama00}, the restriction of an irreducible $\U(k)$-representation $V^k_\lambda$ to the \emph{subgroup} of permutation matrices $S_k \subseteq U(k)$ has been computed:
\begin{equation}
\label{eq:strong6j/nishiyama}
  V^k_\lambda \big|_{S_k \subseteq \U(k)}
  = \bigoplus_{\mu \vdash_k \abs\lambda}
    \Ind_{S_\alpha \times S_{k - \abs\alpha}}^{S_k}
    \left( [\lambda]^{S_\mu} \otimes \mathbf 1 \right)
\end{equation}
In the right-hand side of \eqref{eq:strong6j/nishiyama}, $\Ind$ denotes an \emphindex{induced representation} and $\alpha$ is the unique Young diagram with number of boxes equal to the number of rows of $\mu$ such that
\[N_{S_{\abs\mu}}(S_\mu) / S_\mu \cong S_\alpha,\]
where $S_\mu := S_{\mu_1} \times S_{\mu_2} \times \dots \subseteq S_{\abs\mu}$ is the Young subgroup corresponding to $\mu$,
$N_{S_{\abs\mu}}(S_\mu)$ its normalizer in $S_{\abs\mu}$,
and $S_\alpha \subseteq S_{\abs\alpha}$ the Young subgroup of $\alpha$.
Note that $S_\alpha$ indeed acts on the subspace $[\lambda]^{S_\mu}$.

\begin{lem}
  If $k - \abs\lambda \geq \lambda_1$, then $\lambda' := (k-\abs\lambda,\lambda)$ is again a Young diagram, and
  \begin{equation}
  \label{eq:strong6j/nishiyama lemma}
    V^k_\lambda \big|_{S_k \subseteq \U(k)} \cong [\lambda'] \oplus \dots.
  \end{equation}
  Here and in the following, we write ``\dots'' for a sum of irreducible $S_k$-representations whose Young diagrams have \emph{longer} first rows than all the preceding ones.

  Otherwise, if $k - \abs\lambda < \lambda_1$ then the first row of any Young diagram that appears in the restriction of $V^k_\lambda$ is longer than $k - \abs\lambda$. %
\end{lem}
\begin{proof}
  Since induction is transitive, we can rewrite \eqref{eq:strong6j/nishiyama} as
  \begin{equation}
  \label{eq:strong6j/nishiyama double induction}
  V^k_\lambda \big|_{S_k \subseteq \U(k)}
  = \bigoplus_{\mu \vdash_k \abs\lambda}
    \Ind_{S_{\abs\alpha} \times S_{k - \abs\alpha}}^{S_k}
    \left( \Ind_{S_\alpha}^{S_{\abs\alpha}} \left( [\lambda]^{S_\mu} \right) \otimes \mathbf 1 \right).
  \end{equation}
  The \emphindex{Pieri formula} asserts that the $S_k$-representation induced from a tensor product of an irreducible $S_{\abs\alpha}$-representation $[\nu]$ with the trivial $S_{k-\abs\alpha}$-representation $\mathbf 1$ is given by the sum over all irreducible $S_k$-representations with a Young diagram that can be obtained by adding $k-\abs\alpha$ boxes to $\nu$, with no two in the same column (see, e.g., \cite[\S{}2.2, (4)]{Fulton97}).
  The first row of any such Young diagram is of length at least $k-\abs\alpha$.
  As $\abs\alpha$ is equal to the number of rows of $\mu$, we obtain the lower bound
  \begin{equation}
  \label{eq:strong6j/nishiyama lower bound}
    k-\abs\alpha \geq k-\abs\mu = k-\abs\lambda
  \end{equation}
  on the length of the first row of any irreducible $S_k$-representation that occurs in the restriction of $V^k_\lambda$.

  Equality in \eqref{eq:strong6j/nishiyama lower bound} can occur only if each row of $\mu$ contains a single box, i.e., for $\mu=(1,\ldots,1,0,\ldots,0)$, such that $\abs\alpha = \abs\mu = \abs\lambda$.
  Then $S_\mu$ is the trivial group, $S_\alpha = S_{\abs\alpha} = S_{\abs\lambda}$, and the corresponding summand in \eqref{eq:strong6j/nishiyama double induction} is equal to
  \begin{equation}
  \label{eq:strong6j/nishiyama only way}
    \Ind_{S_{\abs\lambda} \times S_{k - \abs\lambda}}^{S_k} \left( [\lambda] \otimes \mathbf 1 \right).
  \end{equation}
  By the Pieri formula, \eqref{eq:strong6j/nishiyama only way} contains an irreducible $S_k$-representation with first row of length $k-\abs\lambda$ if and only if $\lambda_1 \leq k-\abs\lambda$ (since we only add boxes to $\lambda$).
  Moreover, if this condition is satisfied then there is only a single option, namely to place one box in each of the $k-\abs\lambda$ leftmost columns, resulting in the Young diagram $\lambda' = (k-\abs\lambda,\lambda)$.
\end{proof}

We now consider the decomposition of a tensor product of irreducible $\U(k)$-representations,
\begin{equation*}
  V^k_\alpha \otimes V^k_\beta = \bigoplus_\lambda c^{\alpha,\beta}_\lambda V^k_\lambda,
\end{equation*}
where we assume that $k - \abs\alpha \geq \alpha_1$ and $k - \abs\beta \geq \beta_1$.
The multiplicities $c^{\alpha,\beta}_\lambda$ are known as the \emph{Littlewood--Richardson coefficients}\index{Littlewood--Richardson coefficients|textbf}\nomenclature[Rc^alpha,beta_gamma]{$c^{\alpha,\beta}_\gamma$}{Littlewood--Richardson coefficients}, and they are independent of the choice of $k$ (if $k$ is at least as large as the number of rows in the Young diagrams involved) \cite{JamesKerber81}.
Moreover, $c^{\alpha,\beta}_\lambda$ is non-zero only if $\abs\alpha + \abs\beta = \abs\lambda$.
It follows from the points above that
\begin{equation*}
    V^k_\alpha \otimes V^k_\beta \big|_{S_k \subseteq \U(k)}
  = \bigoplus_{\mathclap{\substack{\abs\alpha + \abs\beta = \abs\lambda}}} c^{\alpha,\beta}_\lambda V^k_\lambda  \big|_{S_k \subseteq \U(k)}
  = \left(\;\;\;\;\;\bigoplus_{\mathclap{\substack{\abs\alpha + \abs\beta = \abs\lambda, \\ k - \abs\lambda \geq \lambda_1}}} c^{\alpha,\beta}_\lambda [\lambda'] \right) \oplus \dots
\end{equation*}
On the other hand, by applying \eqref{eq:strong6j/nishiyama lemma} to the individual tensor factors we find that
\begin{align*}
  &V^k_\alpha \otimes V^k_\beta \big|_{S_k \subseteq \U(k)}
  = \left( [\alpha'] \oplus \dots \right) \otimes \left( [\beta'] \oplus \dots \right)
  = \left(\;\;\;\;\;\bigoplus_{\mathclap{\substack{\abs\lambda = \abs\alpha + \abs\beta \\ k - \abs\lambda \geq \lambda_1}}} g_{\alpha',\beta',\lambda'} [\lambda'] \right) \oplus \dots,
\end{align*}
where $g_{\alpha',\beta',\lambda'}$ are the Kronecker coefficients. %
In the last inequality, we have used that $g_{\alpha',\beta',\lambda'} > 0$ only if $\abs\lambda \leq \abs\alpha + \abs\beta$ \cite[Theorem 2.9.22]{JamesKerber81}.
By comparing coefficients we find that
$c^{\alpha,\beta}_\lambda = g_{\alpha',\beta',\lambda'}$ for all triples of Young diagrams with $\abs\alpha + \abs\beta = \abs\gamma$ and $k$ large enough.
We thus recover a well-known result due to Littlewood and Murnaghan that states that the Littlewood--Richardson coefficients are a special case of the Kronecker coefficients \cite{Littlewood58, Murnaghan55}.
What is more, the argument shows that the Clebsch--Gordan embeddings $\Phi^{\alpha'\beta'}_{\lambda'}$ for $S_k$ can be obtained by restricting the ones of $\U(k)$.\index{Clebsch--Gordan isomorphism!symmetric group}\index{Clebsch--Gordan isomorphism!unitary group}
In view of \eqref{eq:strong6j/6j via clebsch gordan}, this implies directly that the recoupling coefficients are the same, since they are built solely from the action on the multiplicity spaces\index{recoupling coefficients!symmetric group}\index{recoupling coefficients!unitary group}.
Again, the recoupling coefficients for $\U(k)$ do not depend on the choice of $k$ (if $k$ is at least as large as the number of rows in the Young diagrams involved).

\section{Discussion}
\label{sec:strong6j/discussion}

From the perspective of representation theory, the one-body quantum marginal problem can be characterized in terms of the decomposition of tensor products of irreducible representations of the symmetric group.
\autoref{thm:strong6j/main theorem} generalizes this description:
It shows that the overlap between two such decompositions -- as captured by the recoupling coefficients -- similarly characterizes the quantum marginal problem with two overlapping marginals.
It would be of great interest to find a geometric explanation of this result in the framework of \autoref{sec:onebody/git}, which might also lead to a more refined understanding of the asymptotics (along the lines of \cite{Roberts99} for Wigner's $6j$-symbols).
Mathematically, this is related to the ``intersection'' of moment maps or to simultaneous Hamiltonian reduction for non-commuting group actions.

\medskip

In \autoref{sec:strong6j/ssa}, we have given a novel proof of the strong subadditivity and weak monotonicity of the von Neumann entropy.
It is markedly different from previous proofs in the literature, which are built on operator convexity \cite{LiebRuskai73,NielsenPetz05,Ruskai07a,Effros09} or asymptotic equipartition \cite{Renner05,Gromov13} (cf.\ the review \cite{Ruskai05}).
In our approach, we interpret an entropy inequality as the asymptotic shadow of a dimensional relation such as \eqref{eq:strong6j/subadd ratio cor}.
We establish the latter by exploiting the symmetries of a corresponding representation-theoretic object -- the recoupling coefficients -- together with a lower bound from spectrum estimation.
The generality of this approach suggests an intriguing route towards establishing new entropy inequalities---namely, by constructing novel representation-theoretic objects (e.g., by composing Clebsch--Gordan maps) and uncovering their symmetries (as can conveniently be done using the graphical calculus).

\medskip

Finally, we speculate that the surprising connection established in \autoref{sec:strong6j/sums of matrices} between tetrahedra and quantum states as well as between the corresponding recoupling coefficients may help to understand and connect the study of spin foams and spin networks in the context of quantum gravity \cite{Ooguri92, ReisenbergerRovelli97, FreidelLouapre03, BarrettSteele03, Gurau08, AquilantiHaggardHedemanEtAl12} and condensed matter physics \cite{LevinWen05} to quantum information theory.
\cleardoublepage
\fancypagestyle{manual}{\fancyhf{}\fancyhead[LE,RO]{\thepage}\fancyhead[RE]{\textit{Bibliography}}\fancyhead[LO]{\textit{Bibliography}}}
\pagestyle{manual}
\phantomsection\addcontentsline{toc}{chapter}{Bibliography}
\bibliographystyle{alpha-eprint-url}
\bibliography{thesis}

\cleardoublepage
\fancypagestyle{manual}{\fancyhf{}\fancyhead[LE,RO]{\thepage}\fancyhead[RE]{\textit{List of Symbols}}\fancyhead[LO]{\textit{List of Symbols}}}
\pagestyle{manual}
\phantomsection\addcontentsline{toc}{chapter}{List of Symbols}
\label{list of symbols first}
\printnomenclature[2.5cm]
\label{list of symbols last}

\cleardoublepage
\fancypagestyle{manual}{\fancyhf{}\fancyhead[LE,RO]{\thepage}\fancyhead[RE]{\textit{Index}}\fancyhead[LO]{\textit{Index}}}
\pagestyle{manual}
\phantomsection\addcontentsline{toc}{chapter}{Index}
\printindex

\cleardoublepage
\fancypagestyle{manual}{\fancyhf{}\fancyhead[LE,RO]{\thepage}\fancyhead[RE]{\textit{List of Publications}}\fancyhead[LO]{\textit{List of Publications}}}
\pagestyle{manual}
\phantomsection\addcontentsline{toc}{chapter}{List of Publications}
\chapter*{List of Publications}

\begin{enumerate}
\item[9.] \textit{A Heisenberg Limit for Quantum Region Estimation},\\
with J.~M.~Renes.\\
\emph{Proc.~IEEE Inter.~Symp.~Inform.~Theory (ISIT'14)}, 1126--1130 (2014).

\item[8.] \textit{Lower Bounds for Quantum Parameter Estimation},\\
with J.~M.~Renes.\\
Preprint \href{http://arxiv.org/abs/1310.2155}{arXiv:1310.2155};
accepted for publication in \emph{IEEE Trans.~Inf.~Theory}.

\item[7.] \textit{Stabilizer information inequalities from phase space distributions},\\
with D.~Gross.\\
\emph{J.~Math.~Phys.} \textbf{54} (8), 082201 (2013).

\item[6.] \textit{Recoupling Coefficients and Quantum Entropies},\\
with M.~Christandl and M.~B.~\c{S}ahino\u{g}lu.\\
Preprint \href{http://arxiv.org/abs/1210.0463}{arXiv:1210.0463}; presented at \emph{QIP'13}.

\item[5.] \textit{Entanglement Polytopes: Multiparticle Entanglement from Single-Particle Information},\\
with B.~Doran, D.~Gross, and M.~Christandl.\\
\emph{Science} \textbf{340} (6137), 1205--1208 (2013); presented at \emph{QIP'13}.

\item[4.] \textit{When is a pure state of three qubits determined by its single-particle reduced density matrices?},\\
with A.~Sawicki and M.~Ku\'{s}.\\
\emph{J.~Phys.~A} \textbf{46}, 055304 (2013).

\item[3.]
\textit{Computing Lie Group Multiplicities},\\
 with M.~Christandl and B.~Doran.\\
\emph{Proc.~IEEE Ann.~Symp.~Found.~Comput. Sci. (FOCS'12)}, 639--648 (2012).

\item[2.]
\textit{Eigenvalue Distributions of Reduced Density Matrices},\\
with M.~Christandl, B.~Doran, and S.~Kousidis.\\
\emph{Commun.~Math.~Phys.} \textbf{332}, 1--52 (2014).

\item[1.]
\textit{Equivariant geometric $K$-homology for compact Lie group actions},\\
with P.~Baum, H.~Oyono-Oyono, and T.~Schick.\\
\emph{Abh.~Math.~Sem.~Univ.~Hamburg} \textbf{80}, 149--173 (2010).
\end{enumerate}

\end{document}